\documentclass{amsart}

\usepackage[T1]{fontenc}
\usepackage[margin=1in]{geometry}
\usepackage{color}

\usepackage{amsmath,amsthm,amssymb}
\usepackage{amstext}
\usepackage{tikz}
\usetikzlibrary{decorations.markings,decorations.pathreplacing,calc,fadings, shadings, patterns}
\usepackage{comment}
\usepackage{mathtools}
\usepackage{mdframed}
\usepackage[lofdepth,lotdepth,labelformat=simple]{subfig}

\usepackage[comma,square,numbers,sort]{natbib}
\usepackage{thm-restate}
\usepackage{nicefrac}
\usepackage{xspace}
\usepackage{enumitem}
\usepackage[foot]{amsaddr}
\usepackage[colorlinks]{hyperref}
\mdfsetup{nobreak=true,skipabove=8pt,skipbelow=8pt,innertopmargin=8pt,innerbottommargin=8pt,frametitleaboveskip=8pt,frametitlebelowskip=0pt}

\DeclareFontEncoding{FML}{}{}%
\DeclareFontSubstitution{FML}{fncmi}{m}{it}%
\DeclareSymbolFont{fouriernc}{FML}{fncmi}{m}{it}%
\DeclareMathAccent{\fvec}{0}{fouriernc}{"7E}
\renewcommand{\vec}[1]{\fvec{#1}}

\DeclareCaptionListOfFormat{myformat}{\textsc{#2}}
\numberwithin{equation}{section}
\numberwithin{figure}{section}

\theoremstyle{definition}
\newtheorem{definition}{Definition}
\newtheorem{problem}{Problem}
\newtheorem{example}{Example}
\theoremstyle{plain}
\newtheorem{theorem}{Theorem}
\newtheorem{lemma}{Lemma}
\newtheorem{fact}{Fact}

\newcommand{\eq}[1]{\hyperref[eq:#1]{(\ref*{eq:#1})}}
\renewcommand{\sec}[1]{\hyperref[sec:#1]{Section~\ref*{sec:#1}}}
\newcommand{\app}[1]{\hyperref[app:#1]{Appendix~\ref*{app:#1}}}
\newcommand{\thm}[1]{\hyperref[thm:#1]{Theorem~\ref*{thm:#1}}}
\newcommand{\lem}[1]{\hyperref[lem:#1]{Lemma~\ref*{lem:#1}}}
\newcommand{\defn}[1]{\hyperref[defn:#1]{Definition~\ref*{defn:#1}}}
\newcommand{\ex}[1]{\hyperref[ex:#1]{Example~\ref*{ex:#1}}}
\newcommand{\fct}[1]{\hyperref[fct:#1]{Fact~\ref*{fct:#1}}}
\newcommand{\fig}[1]{\hyperref[fig:#1]{Figure~\ref*{fig:#1}}}
\newcommand{\figs}[1]{\hyperref[fig:#1]{Figures~\ref*{fig:#1}}}
\newcommand{\subfig}[1]{\protect\subref{fig:#1}}

\newcommand{\ocl}{\hyperref[lem:oc]{Occupancy Constraints Lemma}\xspace}
\newcommand{\npl}{\hyperref[lem:npl]{Nullspace Projection Lemma}\xspace}

\newcounter{mycount}

\DeclareMathOperator{\poly}{poly}
\DeclareMathOperator{\spn}{span}
\DeclareMathOperator{\Sym}{Sym}
\DeclareMathOperator{\AP}{AP}
\DeclareMathOperator{\In}{In}
\DeclareMathOperator{\new}{new}
\DeclareMathOperator{\Cube}{Cube}
\DeclareMathOperator{\diam}{diam}

\newcommand{\C}{\mathbb{C}}

\newcommand{\id}{1}

\newcommand{\wit}{\mathrm{wit}}
\newcommand{\prop}{\mathrm{prop}}
\newcommand{\CNOT}{\mathrm{CNOT}}
\newcommand{\SWAP}{\mathrm{SWAP}}
\newcommand{\occ}{\mathrm{occ}}
\newcommand{\inn}{\mathrm{in}}
\newcommand{\out}{\mathrm{out}}
\newcommand{\legal}{\mathrm{legal}}
\newcommand{\illegal}{\mathrm{illegal}}
\newcommand{\clock}{\mathrm{clock}}
\newcommand{\anc}{\mathrm{anc}}
\newcommand{\bnd}{\mathrm{bnd}}

\newcommand{\gxoc}{G_{X}^{\occ}}
\newcommand{\gx}{G_{X}}
\newcommand{\goc}{G^{\occ}}

\newcommand{\hopping}{t_{\mathrm{hop}}}

\newlength{\vdotheight}
\newcommand{\mystrut}{\rule{0pt}{\vdotheight}} 

\makeatletter
\def\@tocline#1#2#3#4#5#6#7{\relax
 \ifnum #1>\c@tocdepth 
 \else
   \par \addpenalty\@secpenalty\addvspace{#2}%
   \begingroup \hyphenpenalty\@M
   \@ifempty{#4}{%
     \@tempdima\csname r@tocindent\number#1\endcsname\relax
   }{%
     \@tempdima#4\relax
   }%
   \parindent\z@ \leftskip#3\relax \advance\leftskip\@tempdima\relax
   \rightskip\@pnumwidth plus4em \parfillskip-\@pnumwidth
   #5\leavevmode\hskip-\@tempdima
     \ifcase #1
      \or\or \hskip 1.75em \or \hskip 2em \else \hskip 3em \fi%
     #6\nobreak\relax
   \dotfill\hbox to\@pnumwidth{\@tocpagenum{#7}}\par
   \nobreak
   \endgroup
 \fi}
\makeatother

\begin{document}

\title{The Bose-Hubbard model is QMA-complete}

\author{Andrew M. Childs$^{1,2}$}
\author{David Gosset$^{1,2}$}
\author{Zak Webb$^{2,3}$}

\address{$^1$
Department of Combinatorics \& Optimization,
University of Waterloo}

\address{$^2$
Institute for Quantum Computing,
University of Waterloo}

\address{$^3$
Department of Physics \& Astronomy,
University of Waterloo}

\email{amchilds@uwaterloo.ca, dngosset@gmail.com, zakwwebb@gmail.com}

\renewcommand{\abstractname}{\normalsize Abstract}
\begin{abstract}
The Bose-Hubbard model is a system of interacting bosons that live on the vertices of a graph. The particles can move between adjacent vertices and experience a repulsive on-site interaction.  The Hamiltonian is determined by a choice of graph that specifies the geometry in which the particles move and interact.  We prove that approximating the ground energy of the Bose-Hubbard model on a graph at fixed particle number is QMA-complete. In our QMA-hardness proof, we encode the history of an $n$-qubit computation in the subspace with at most one particle per site (i.e., hard-core bosons). This feature, along with the well-known mapping between hard-core bosons and spin systems, lets us prove a related result for a class of $2$-local Hamiltonians defined by graphs that generalizes the XY model.  By avoiding the use of perturbation theory in our analysis, we circumvent the need to multiply terms in the Hamiltonian by large coefficients.
\end{abstract}

\maketitle

\section{Introduction}
\label{sec:intro}
The problem of approximating the ground energy of a given Hamiltonian is a natural quantum analog of classical constraint satisfaction. Many authors have considered the computational complexity of such quantum ground state problems. For a variety of classes of Hamiltonians and a suitable notion of approximation, this task is complete for the complexity class QMA, the quantum version of NP with two-sided error (see reference \cite{Boo12} for a recent review). These results provide evidence that approximating the ground energy of such quantum systems is likely intractable.

The first such example is the Local Hamiltonian problem introduced by Kitaev \cite{KSV02}. A $k$-local Hamiltonian acts on a system of $n$ qubits and can be written as a sum of terms, each acting nontrivially on at most $k$ qubits. The $k$-Local Hamiltonian problem is a promise problem related to the task of approximating the ground energy of a $k$-local Hamiltonian. Given such a Hamiltonian and two thresholds $a$ and $b$, one is asked to determine if the ground energy is below $a$ or above $b$ (promised that one of these conditions holds). Kitaev's original work showed that the $5$-local Hamiltonian problem is QMA-complete \cite{KSV02}; subsequent works proved QMA-completeness of the $3$-local Hamiltonian problem \cite{KR03}, the $2$-local Hamiltonian problem \cite{KKR04}, and the $2$-local Hamiltonian problem with interactions between qubits restricted to a two-dimensional lattice \cite{OT05}.

The complexity of similar computational problems related to other classes of Hamiltonians has also been considered.  These include Hamiltonians in one dimension \cite{AGIK09,GI09}, frustration-free Hamiltonians \cite{Bra06,GN13}, and stoquastic Hamiltonians (Hamiltonians with no ``sign problem'') \cite{BDOT08,BT09}, among others.

The QMA-hardness of ground energy problems for local Hamiltonians acting on qubits has implications for Hamiltonians acting on indistinguishable particles (bosons or fermions) due to formal mappings between these systems.  By applying such mappings to the Local Hamiltonian problem, one can show that certain bosonic \cite{WMN10} and fermionic \cite{LCV07} Hamiltonian problems are QMA-hard.  A more restrictive class of QMA-complete fermionic Hamiltonians was considered by Schuch and Verstraete, who showed that the Hubbard model with a site-dependent magnetic field is QMA-complete \cite{SV09}. This is a specific model of interacting electrons (i.e., spin-$\frac{1}{2}$ fermions) on a two-dimensional lattice, with a magnetic field that may take different values and point in different directions (in three dimensions) at distinct sites of the lattice.

Many of the QMA-complete problems considered previously have the property that the form of the terms in the Hamiltonian is part of the specification of the instance.  For example, a $2$-local Hamiltonian is specified by a $2$-local Hermitian operator for each pair of qubits. In the Hubbard model considered in reference \cite{SV09}, there is a similar freedom in the choice of magnetic field at each site.

In contrast, here we consider a system of interacting bosons with fixed movement and interaction terms.  Specifically, we consider the Bose-Hubbard model, which has one of the simplest interactions between particles that conserves total particle number. Although the Bose-Hubbard model is traditionally defined on a lattice \cite{FWGF89}, here we consider its natural extension to a general graph.

We consider undirected graphs without multiple edges and with at most one self-loop per vertex. Any such graph $G$ (with vertex set $V$) can be specified by its adjacency matrix, a symmetric $0$-$1$ matrix denoted $A(G)$. The Bose-Hubbard model on $G$ with hopping strength $\hopping$ and interaction strength $J_{\mathrm{int}}$ has the Hamiltonian
\begin{equation}
  H_{G}
  =\hopping\sum_{i,j \in V} A(G)_{ij}a_{i}^{\dagger}a_{j}
  +J_{\mathrm{int}}\sum_{k\in V}n_{k}\left(n_{k}-1\right)
\label{eq:Bose-Hubbard_Ham_intro}
\end{equation}
where $a_{i}^{\dagger}$ creates a boson at vertex $i$ and $n_{i}=a_{i}^{\dagger}a_{i}$ counts the number of bosons at vertex $i$. Our results apply to the Bose-Hubbard model for any fixed positive hopping strength $\hopping > 0$ and any fixed positive (i.e., repulsive) interaction strength $J_\mathrm{int}>0$.  Unlike in other QMA-completeness results, in our work the coefficients $\hopping,J_{\mathrm{int}}$ are not inputs to the problem; rather, each fixed choice defines a computational problem and we prove QMA-completeness for each of them.

Observe that the Bose-Hubbard Hamiltonian \eq{Bose-Hubbard_Ham_intro} conserves the total number of particles $N=\sum_{k\in V}n_k$. We focus on the space of $N$-particle states, which can be identified with the symmetric subspace of $(\C^{|V|})^{\otimes N}$ (as we discuss in more detail in \sec{Definitions-and-Results}). The first term in \eq{Bose-Hubbard_Ham_intro} allows particles to move between vertices; the second term is an interaction between particles that assigns an energy penalty for each vertex occupied by more than one particle. The Bose-Hubbard model is an example of a multi-particle quantum walk, a generalization of quantum walk to systems with more than one walker.

Recently we showed that the Bose-Hubbard model on a graph can perform efficient universal quantum computation \cite{CGW13}. Sometimes universality goes hand-in-hand with QMA-completeness, e.g., for local Hamiltonians, whose dynamics are BQP-complete \cite{Fey85} and whose ground energy problem is QMA-complete \cite{KSV02}. However, not all classes of Hamiltonians with universal dynamics have QMA-complete ground energy problems. For example, the dynamics of stoquastic local Hamiltonians are BQP-complete (as follows from \cite{JW06} and time-reversal symmetry), whereas the corresponding ground energy problem is in AM \cite{BDOT08}, which is presumably smaller than QMA.  Similarly, the ground energy problem for a Bose-Hubbard model with $\hopping<0$ is also in AM \cite{BDOT08}, whereas the dynamics of such Hamiltonians are universal \cite{CGW13}. The ferromagnetic Heisenberg model on a graph provides an even starker contrast: its dynamics are BQP-complete (as can be inferred from \cite{CGW13} using a correspondence between spins and hard-core bosons) but its ground energy problem is trivial since the ground space is the symmetric subspace.

\subsection{Overview of results}

In this paper we define the Bose-Hubbard Hamiltonian problem and characterize its complexity. In this problem one is given a graph $G$ and a number of particles $N$ and asked to approximate the ground energy of the Bose-Hubbard Hamiltonian \eq{Bose-Hubbard_Ham_intro} in the $N$-particle sector (in a precise sense described in \sec{Definitions-and-Results}). We prove that this problem is QMA-complete. 

To prove QMA-hardness of the Bose-Hubbard Hamiltonian problem, we show that in fact a notable special case of this problem, called Frustration-Free Bose Hubbard Hamiltonian, is QMA-hard. In this problem one is asked (roughly) to determine if the ground energy of the Bose-Hubbard Hamiltonian \eq{Bose-Hubbard_Ham_intro} in the $N$-particle sector is close to $N$ times its single-particle ground energy (i.e., $N$ times the smallest eigenvalue of the adjacency matrix $A(G)$). This is always a lower bound on the $N$-particle energy, and when it is achieved we say the $N$-particle ground states are frustration free. A frustration-free state has the special property that it has minimal energy for both terms in \eq{Bose-Hubbard_Ham_intro}, and in particular it is annihilated by the interaction term. Frustration-free states therefore live in the subspace of \emph{hard-core} bosons, where no more than one boson can occupy each site of the graph.

Furthermore, we prove a reduction from Frustration-Free Bose-Hubbard Hamiltonian to an eigenvalue problem for a class of $2$-local Hamiltonians defined by graphs. The two problems are related by a well-known mapping between hard-core bosons and spin systems. Specifically, given a graph $G$  (with vertex set $V$)  we consider the Hamiltonian
\begin{equation}
 \sum_{\substack{A(G)_{ij}=1 \\ i\neq j}}\big(|01\rangle\langle 10|+|10\rangle\langle 01|\big)_{ij} +\sum_{A(G)_{ii}=1} |1\rangle\langle1|_i=\sum_{\substack{A(G)_{ij}=1 \\ i\neq j}}\frac{\sigma_x^i \sigma_x^j+\sigma_y^i \sigma_y^j}{2}+\sum_{A(G)_{ii}=1}\frac{1-\sigma_z^{i}}{2}.\label{eq:xymodel}
\end{equation}
Note that this Hamiltonian commutes with the magnetization operator $M_z=\sum_{i=1}^{|V|}\frac{1-\sigma_z^{i}}{2}$ and has a sector for each of its eigenvalues $M_z\in\{0,1,\ldots,|V|\}$. We reduce Frustration-Free Bose-Hubbard Hamiltonian (with $N$ particles on a graph $G$) to the problem of approximating the smallest eigenvalue of \eq{xymodel} within the sector with magnetization $M_z=N$. We call this the XY Hamiltonian problem because of its connection to the XY model from condensed matter physics. Since this problem is contained in QMA, our reduction shows it to be QMA-complete.

We also obtain a related result that may be of independent interest. In \app{complexity_smallest_graph_eig} we give a self-contained proof that computing the smallest eigenvalue of a sparse, efficiently row-computable \cite{AT03} symmetric $0$-$1$ matrix (the adjacency matrix of a graph) is QMA-complete. This can alternatively be viewed as a result about the QMA-completeness of a single-particle quantum walk on a graph with at most one self-loop per vertex. To prove this,  we use a mapping from circuits to graphs that is also used in our main result.  Note that Janzing and Wocjan used a similar construction to design a BQP-complete problem \cite{JW06}.

\subsection{Proof techniques}

We prove our main result by direct reduction from quantum circuit satisfiability.  We introduce several new techniques in order to do this using the Bose-Hubbard model on an unweighted graph.

Kitaev's original proof of QMA-hardness of the Local Hamiltonian problem encodes a QMA verification circuit using ideas from a computationally universal Hamiltonian proposed by Feynman \cite{Fey85}.  This Hamiltonian uses a ``clock register'' to record the progress of the computation; in an appropriate basis, the Hamiltonian can be seen as a quantum walk on a path whose vertices represent the steps of the computation. Other proofs of QMA-hardness have used alternative encodings of the temporal structure of a verification circuit into a quantum state.  In our construction, we encode the history of an $n$-qubit verification circuit in the state of $n$ interacting particles on a graph, where each particle encodes a single qubit.

Our construction uses a class of graphs we define called \emph{gate graphs}.  Gate graphs are built from a basic subgraph whose single-particle ground states encode the history of a simple single-qubit computation. By suitably combining copies of this basic unit, we define gadgets with other functionality.  (Note that these gadgets realize some desired behavior exactly; they are not ``perturbative gadgets'' in the sense of \cite{KKR04,JF08}.)  In particular, we design gadgets for two-qubit gates such that each ground state of the two-particle Bose-Hubbard model encodes a two-qubit computation. We now give a high-level description of how these gadgets work and how we use them to construct a graph for a QMA verification circuit.

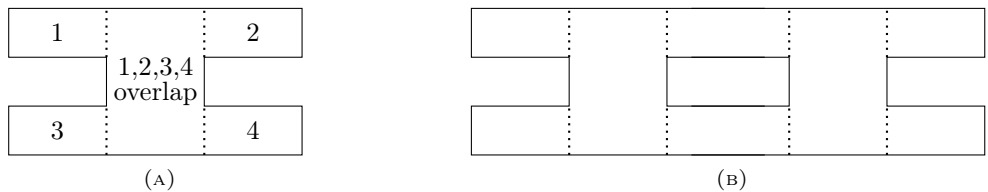
\begin{figure}
\subfloat[][]{\begin{tikzpicture}[scale=0.65]\label{fig:Vgraphstruct}
\draw (2,2)--(2,1)--(0,1)--(0,0)--(4,0)--(6,0)--(6,1)--(4,1)--(4,2)--(6,2)--(6,3)--(4,3)--(0,3)--(0,2)--(2,2);
\node at (1,0.5){3};
\node at (1,2.5){1};
\node at (5,0.5){4};
\node at (5,2.5){2};
\node at (3,1.75) {1,2,3,4};
\node at (3,1.25){overlap};
\draw[thick,dotted] (2,0)--(2,1);
\draw[thick,dotted] (2,2)--(2,3);
\draw[thick,dotted] (4,0)--(4,1);
\draw[thick,dotted] (4,2)--(4,3);
\end{tikzpicture}
}
\hspace{2cm}
\subfloat[][]{\begin{tikzpicture}[scale=0.65]\label{fig:circuitstruct}
\draw (2,2)--(2,1)--(0,1)--(0,0)--(4,0)--(6,0);
\draw (6,1)--(4,1)--(4,2)--(6,2);
\draw (6,3)--(4,3)--(0,3)--(0,2)--(2,2);
\draw[thick,dotted] (2,0)--(2,1);
\draw[thick,dotted] (2,2)--(2,3);
\draw[thick,dotted] (4,0)--(4,1);
\draw[thick,dotted] (4,2)--(4,3);
\begin{scope}[xshift=4.5cm]
\draw (2,2)--(2,1)--(0,1);
\draw (0,0)--(4,0)--(6,0)--(6,1)--(4,1)--(4,2)--(6,2)--(6,3)--(4,3)--(0,3);
\draw (0,2)--(2,2);
\draw[thick,dotted] (2,0)--(2,1);
\draw[thick,dotted] (2,2)--(2,3);
\draw[thick,dotted] (4,0)--(4,1);
\draw[thick,dotted] (4,2)--(4,3);
\end{scope}
\end{tikzpicture}
}
\caption{We design graphs for two-qubit gates with overlapping regions as shown in \subfig{Vgraphstruct}. Regions 1 and 2 are associated with the first encoded qubit and regions 3 and 4 are associated with the second encoded qubit. One could imagine designing a graph for a circuit with two-qubit gates $U_1$ followed by $U_2$ by connecting the corresponding gadgets as in \subfig{circuitstruct}. In the text we describe a challenge with this approach.}
\end{figure}

For each two-qubit gate $U$ from a fixed universal set, we design a graph $G_U$ that can be divided into four overlapping regions as shown schematically in \fig{Vgraphstruct}. (The specific graphs we use for two-qubit gates each have 4096 vertices and are described using the gate graph formalism.) The two-particle Bose-Hubbard model on this graph has ground states that encode the two-qubit computation. To describe them it is helpful to first consider the single-particle ground states, i.e., the ground states of the adjacency matrix $A(G_U)$. This matrix has 16 orthonormal single-particle ground states $|\rho^{i,U}_{z,a}\rangle$. Each index $i\in \{1,2,3,4\}$ is associated with the corresponding region in the graph, as $|\rho^{i,U}_{z,a}\rangle$ is supported entirely within region $i$.  The index $z\in\{0,1\}$ corresponds to the computational basis states of a single encoded qubit. Note that, since $A(G_U)$ is a real matrix, the complex conjugate of any eigenstate is also an eigenstate with the same eigenvalue.  The index $a\in\{0,1\}$ is associated with this freedom, i.e., $|\rho^{i,U}_{z,1}\rangle=|\rho^{i,U}_{z,0}\rangle^*$.  The ground space of the two-particle Bose-Hubbard model on $G_U$ is spanned by 16 states, indexed by two choices $z_1,z_2\in\{0,1\}$ of computational basis states for the encoded qubits and two bits $a_1,a_2\in\{0,1\}$ associated with complex conjugation. These states can be represented as symmetric states in the Hilbert space $\C^{4096}\otimes\C^{4096}$; they are
\[
\frac{1}{2}(|\rho^{1,U}_{z_1,a_1}\rangle|\rho^{3,U}_{z_2,a_2}\rangle+|\rho^{3,U}_{z_2,a_2}\rangle|\rho^{1,U}_{z_1,a_1}\rangle)+\frac{1}{2}\sum_{x_1,x_2\in{0,1}} U(a_1)_{x_1,x_2,z_1,z_2}(|\rho^{1,U}_{x_1,a_1}\rangle|\rho^{3,U}_{x_2,a_2}\rangle+|\rho^{3,U}_{x_2,a_2}\rangle|\rho^{1,U}_{x_1,a_1}\rangle)
\]
where $U(0)=U$ is the two-qubit gate of interest and $U(1)=U^*$ is its elementwise complex conjugate. Observe that each of these states is a superposition of a term where both particles are on the left-hand side of the graph, encoding a two-qubit input state $|z_1\rangle|z_2\rangle$, and a term where both particles are on the right-hand side of the graph, encoding the two-qubit output state $U(a_1)|z_1\rangle|z_2\rangle$ where either $U$ or its complex conjugate has been applied. In other words, the particles ``move together'' through the graph as the gate $U(a)$ is applied. While we might prefer the ground states to only encode the computation corresponding to $U$, we must include the possibility of $U^*$ because the Hamiltonian is real. The same issue arises for $n$-qubit verification circuits. Fortunately, the complex conjugate of a circuit is equally useful for QMA verification.

It is natural to attempt to construct a graph for an $n$-qubit verification circuit by combining gadgets for each of the two-qubit gates. However, there is an obstacle to this approach, as illustrated by the example of a two-qubit circuit consisting of only two gates $U_1$ and $U_2$. One could construct a graph for such a circuit as shown schematically in \fig{circuitstruct}, where the two-qubit gadgets for $U_1$ and $U_2$ are connected in some unspecified way in the middle. However, not every ground state of the two-particle Bose-Hubbard model on such a graph encodes a computation. For example, there could be a ground state where one of the particles is in the single-particle state $|\rho^{1,U_1}_{z,a}\rangle$ localized on the left side of the graph and the other particle is in the state $|\rho^{2,U_2}_{z,a}\rangle$ with support on a disjoint region of the graph on the right-hand side. To eliminate such spurious ground states, we develop a method to enforce \emph{occupancy constraints} on the locations of particles in gate graphs using the Bose-Hubbard interaction. Although this interaction only directly penalizes simultaneous occupation of the same vertex, we show how to simulate terms that penalize simultaneous occupation of different regions of the graph. We formalize this method by proving an ``\ocl'' for gate graphs. 

In summary, our construction of the graph for an $n$-qubit verification circuit proceeds in two steps. We first construct a graph $G$ by connecting two-qubit gadgets for each of the gates in the circuit. As discussed above, the ground space of the $n$-particle Bose-Hubbard model on $G$ includes a subspace of states that encode computations and a subspace of states that do not. We construct a set of occupancy constraints that are only satisfied by states in the former subspace. We then apply the \ocl to obtain a gate graph $G^\square$ where each $N$-particle ground state encodes a computation.

Unlike many previous works, we do not use perturbation theory in our analysis.  Instead, we use a ``\npl'' that characterizes the smallest nonzero eigenvalue of a sum of two positive semidefinite matrices $H_A+H_B$ in terms of the smallest nonzero eigenvalue of $H_A$ and the smallest nonzero eigenvalue of $H_B$ restricted to the nullspace of $H_A$.  This Lemma allows us to establish an eigenvalue promise gap (i.e., to bound the ground energies of yes instances away from those of no instances) without having to multiply terms in the Hamiltonian by large coefficients, something that is not allowed in the setting of the Bose-Hubbard model on a graph.  Whereas QMA-hardness proofs such as those of \cite{KR03,KKR04,OT05,SV09} require multiplying terms in the Hamiltonian by unphysical, problem-size dependent coefficients, our approach avoids this.

Note that the \npl was used implicitly in \cite{MLM99}, which claimed to give a simple proof of the computational universality of adiabatic evolution. That paper encoded a circuit into the multi-particle ground state of a fermionic Hamiltonian in a way that shares some features with our approach. Unfortunately, although the encoding from \cite{MLM99} is novel and interesting, the analysis of the resulting Hamiltonian appears to lack details that are crucial to proving the stated result.

\subsection{Extensions and open questions}
\label{sec:openquestions}

Our result shows that approximating the ground energy of the Bose-Hubbard model on a graph at fixed particle number is likely intractable.  In showing this, we introduce techniques that we expect will be useful in other contexts.  Here we briefly discuss some related questions for future work.

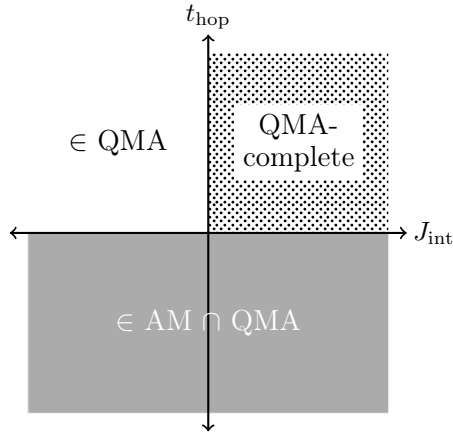
\begin{figure}
\begin{tikzpicture}[scale=0.8]
\draw[white, pattern=crosshatch dots] (3,0) rectangle (6,3);
\draw[white, fill=black!33!white] (0,0) rectangle +(6,-3);
\draw [<->, thick] (-.3,0)--(6.3,0);
\draw [<->, thick] (3,-3.3)--(3,3.3);
\node at (3,-1.5){\large{\color{white}{$\in$ AM $\cap$ QMA}}};
\node[fill=white, text width=42, align=center] at (4.5, 1.5){\large{QMA-complete}};
\node at (1.5, 1.5){\large{$\in$ QMA}};
\node at (6.75,0){$J_\mathrm{int}$};
\node at (3,3.6){$\hopping$};
\end{tikzpicture}
\caption{Each choice for the coefficients $\hopping$ and $J_\mathrm{int}$ defines a computational problem. We prove the problem is QMA-complete for any  positive hopping strength $\hopping>0$ and repulsive interaction strength $J_\mathrm{int}>0$.  For $\hopping<0$ the problem is contained in AM \cite{BDOT08}. For any values of $J_\mathrm{int}$ and $\hopping$ the problem is contained in QMA (see \sec{containment_in_QMA}). \label{fig:complexity}}
\end{figure}

One might consider the complexity of variants of the Bose-Hubbard Hamiltonian problem. For example, one could consider the problem with negative hopping (i.e., $\hopping < 0$), with attractive interactions (i.e., $J_{\mathrm{int}}<0$), or both.  For negative hopping, the results of \cite{BDOT08} show that the problem is in AM; we do not know if it is AM-hard.  Containment in AM suggests that the problem is easier with $\hopping<0$ than with $\hopping,J_{\mathrm{int}}>0$.  For attractive interactions, the problem is clearly in QMA (the verification procedure described in \sec{containment_in_QMA} applies independent of the signs of $\hopping,J_{\mathrm{int}}$), but again we do not know the true complexity. \fig{complexity} summarizes our knowledge of the complexity of the Bose-Hubbard Hamiltonian problem with different choices of $\hopping$ and $J_{\mathrm{int}}$.

One can define other variants of the Bose-Hubbard Hamiltonian problem by lifting the restriction to fixed particle number.

One could also consider other classes of graphs. The graphs we consider in this paper are described by symmetric $0$-$1$ matrices and have at most one self-loop per vertex. We do not know if the model remains QMA-hard on simple graphs, i.e., without any self-loops. Our reduction from Frustration-Free Bose-Hubbard Hamiltonian to the ground energy problem for the spin model \eq{xymodel} gives a stronger result for simple graphs, in which case the second term in \eq{xymodel} vanishes.  Thus, if the Frustration-Free Bose-Hubbard Hamiltonian problem for simple graphs is QMA-hard, then approximating the ground energy of the XY model on a (simple) graph at fixed magnetization is QMA-complete.

There are many open questions concerning the complexity of the ground energy problem for other quantum systems defined by graphs.  For example, one could consider fermions or bosons on a graph with nearest-neighbor interactions. One could also consider quantum spin models such as the XY model or the antiferromagnetic Heisenberg model defined on graphs. Both of these examples correspond to Hamiltonians that conserve magnetization, so one could consider the ground energy problem with or without a restriction to a fixed-magnetization sector. This would complement existing results about the complexity of computing the lowest-energy configuration of classical spin models defined by graphs (for example, the antiferromagnetic Ising model on a graph is NP-complete, as it is equivalent to Max Cut).

As emphasized previously, the Hamiltonians we consider are determined entirely by a choice of graph, with the same type of movement and interaction terms applied throughout the graph.  It might be interesting to find other QMA-complete problems with similar features, such as a version of Local Hamiltonian with only one type of local term.  Analogous classical constraint satisfaction problems with a fixed type of constraint are well known (e.g., Exact Cover and Not-All-Equal SAT) and have been widely studied.

\subsection{Outline of the paper}

The remainder of this paper is organized as follows.  In \sec{Definitions-and-Results} we define the class QMA, introduce the Bose-Hubbard model, formally state our main result, and describe a connection to spin models.  In \sec{containment_in_QMA} we explain why the ground state problem for the Bose-Hubbard model is contained in QMA. In \sec{A-class-of} we define the notion of gate graphs, analyze frustration-free ground states of gate graphs, and introduce the idea of occupancy constraints.  In \sec{Gadgets} we design and analyze gadgets for implementing gates.  In \sec{From-circuits-to} we construct the graph that we use to show QMA-hardness of the Bose-Hubbard model.  In \sec{Proof-of-Theorem} we perform the spectral analysis needed to prove our main result.  Some additional results and technical details appear in the appendices.  In \app{complexity_smallest_graph_eig} we prove that computing the smallest eigenvalue of a (succinctly specified) graph is QMA-complete. In \app{XY} we reduce Frustration-Free Bose-Hubbard Hamiltonian to XY Hamiltonian. In \app{Occupancy-Constraints-Lemma} we prove the \ocl.  In \app{graph_gadgets} we analyze gadgets for two-qubit gates.  Finally, in \app{tech_support} we provide various necessary technical results, including the \npl.


\newpage
\newgeometry{margin=1in,bottom=.5in}
\tableofcontents
\newpage
\newgeometry{margin=1in}

\section{Definitions and results}
\label{sec:Definitions-and-Results}

In this Section we define the complexity class QMA, introduce the
Bose-Hubbard model and a related spin model, and formally state our results.

\subsection{Quantum Merlin-Arthur}

Quantum Merlin-Arthur, or QMA, is a class of promise problems. Informally, a promise problem is in QMA if, for any yes instance $X$, there exists a quantum state $|\psi_{\wit}\rangle$ (with a number of qubits polynomial in the input size $|X|$) that ``proves'' $X$ is a yes instance. We imagine that all-powerful Merlin prepares the quantum proof $|\psi_{\wit}\rangle$ and hands it to polynomially-bounded Arthur, who checks it using his quantum computer.

Arthur checks Merlin's proof using a verification circuit $\mathcal{\mathcal{C}}_{X}$ that is uniformly generated: $\mathcal{C}_{X}$ is computed from $X$ using a deterministic polynomial-time (as a function of $|X|$) classical algorithm. The circuit $\mathcal{C}_{X}$ has an $n_{\inn}$-qubit input register, $n-n_{\inn}$ ancilla qubits initialized to $|0\rangle$, and one of the $n$ qubits designated as an output qubit. The number of qubits $n$ and the number of gates in the circuit are upper bounded by a polynomial function of $|X|$. We write $U_{\mathcal{C}_{X}}$ for the unitary operation that the verification circuit implements.

The acceptance probability of $\mathcal{C}_{X}$ given some $n_{\inn}$-qubit input state $|\phi\rangle$ is the probability of measuring the output qubit to be in the state $|1\rangle$ after the circuit is applied, namely
\[
\AP(\mathcal{C}_{X},|\phi\rangle)=\left\Vert |1\rangle\langle1|_{\out}U_{\mathcal{C}_{X}}|\phi\rangle|0\rangle^{\otimes n-n_{\inn}}\right\Vert^{2}.
\]
Arthur applies the verification circuit to the input state $|\psi_{\wit}\rangle$. If he measures the output qubit to be $|1\rangle$, he concludes $X$ is a yes instance. QMA is the class of problems for which reliable verification circuits exist, so that Arthur can trust this conclusion with reasonable probability. In other words, if $X$ is a yes instance then there exists a state $|\psi_{\wit}\rangle$ that $\mathcal{C}_{X}$ accepts with high probability, whereas if $X$ is a no instance then any state $|\phi\rangle$ has low probability of being accepted. To define the class we must specify which probabilities are considered ``high'' and which are ``low.'' A standard choice uses thresholds of $\frac{2}{3}$ and $\frac{1}{3}$ (called completeness and soundness, respectively).

\begin{definition}[QMA]\label{defn:QMA}A promise problem $L_{\mathrm{yes}}\cup L_{\mathrm{no}}\subset\{0,1\}^*$ is contained in QMA if there exists a uniform polynomial-size circuit family $\mathcal{C}_{X}$ with the following two properties. If $X\in L_{\mathrm{yes}}$, there exists an input state $|\psi_{\mathrm{wit}}\rangle$ such that 
\[
\AP(\mathcal{C}_{X},|\psi_{\mathrm{wit}}\rangle)\geq\frac{2}{3}\quad\text{(completeness).}
\]
If $X\in L_{\mathrm{no}}$ then all input states $|\phi\rangle$ satisfy
\[
\AP(\mathcal{C}_{X},|\phi\rangle)\leq\frac{1}{3}\quad{\text{(soundness).}}
\]
\end{definition}

It is a nontrivial fact that many other choices for the completeness and soundness thresholds lead to equivalent definitions of QMA \cite{KSV02,MW05}. In particular, we obtain the same class QMA if we change the completeness threshold $\frac{2}{3}$ in the above definition to $1-\frac{1}{2^{|X|}}$.

\subsection{The Bose-Hubbard model on a graph}

As discussed in \sec{intro}, we consider the Bose-Hubbard model on a graph $G$, with Hamiltonian
\begin{equation}
  H_{G}
  =\hopping\sum_{i,j\in V} A(G)_{ij}a_{i}^{\dagger}a_{j}
  +J_{\mathrm{int}}\sum_{k\in V}n_{k}\left(n_{k}-1\right)
\label{eq:Bose-Hubbard_Ham}
\end{equation}
where $a_{i}^{\dagger}$ creates a boson at vertex $i$ and $n_{i}=a_{i}^{\dagger}a_{i}$ counts the number of particles at vertex $i$. In this second-quantized formulation of the Bose-Hubbard model, the Hamiltonian $H_{G}$ acts on the Fock space with orthonormal basis vectors
\[
|l_{1},l_{2},\ldots,l_{|V|}\rangle
\]
where $l_{j}\in\{0,1,2,\ldots\}$ specifies the number of bosons at vertex $j\in V$. The operator $a_{i}$ is defined by its action in this basis: 
\begin{equation}
a_{i}|l_{1},l_{2},\ldots,l_{i},\ldots,l_{|V|}\rangle=\sqrt{l_{i}}|l_{1},l_{2},\ldots,l_{i}-1,\ldots,l_{|V|}\rangle.\label{eq:a_mat_els}
\end{equation}
The Hamiltonian \eq{Bose-Hubbard_Ham} conserves the total particle number $N=n_{1}+n_{2}+\cdots+n_{|V|}$. In the $N$-particle sector, the Bose-Hubbard model acts on the finite-dimensional Hilbert space
\begin{equation}
\spn\{|l_{1},\ldots,l_{|V|}\rangle\colon\sum_{j}l_{j}=N\}.\label{eq:occupation_num_states}
\end{equation}
The dimension of this space is 
\begin{equation}
D_{N}=\binom{N+|V|-1}{|V|-1}.\label{eq:DN}
\end{equation}

Our results apply to the Bose-Hubbard model for any strictly positive hopping and interaction strengths. Henceforth we set $\hopping=J_{\mathrm{int}}=1$ for convenience, but all our complexity-theoretic results hold for any fixed $\hopping,J_{\mathrm{int}}>0$.

An equivalent (first-quantized) formulation of the Bose-Hubbard model is as follows. Consider the Hilbert space 
\begin{equation}
(\C^{|V|})^{\otimes N}=\spn\{|i_{1},i_{2},\ldots,i_{N}\rangle\colon i_{1},i_{2},\ldots,i_{N}\in V\}\label{eq:disting_n_particle}
\end{equation}
where each basis state corresponds to an $N$-tuple of vertices in the graph. Define the linear operator Sym that symmetrizes over all $N!$ permutations of the $N$ particles: 
\[
  \Sym(|i_{1}\rangle|i_{2}\rangle\ldots|i_{N}\rangle)
  = \frac{1}{\sqrt{N!}} \sum_{\pi\in S_{N}}|i_{\pi(1)} 
    \rangle|i_{\pi(2)}\rangle\ldots|i_{\pi(N)}\rangle.
\]
Note that $\Sym$ does not in general preserve the norm (for example, any antisymmetric state is mapped to zero). Every state in the Hilbert space \eq{occupation_num_states} can be identified with a state in 
\[
\mathcal{Z}_{N}(G)=\spn\{\Sym(|i_{1},i_{2},\ldots,i_{N}\rangle)\colon i_{1},i_{2},\ldots,i_{N}\in V\}
\]
and vice versa since the two spaces have the same dimension. It is natural to identify states in the two Hilbert spaces by the following linear mapping, defined by its action on basis states. We identify each basis state $|l_{1},\ldots,l_{|V|}\rangle$ with the normalized state
\begin{equation}
  \frac{1}{\sqrt{l_{1}!\, l_{2}!\,\ldots\, l_{|V|}!}} \,
  \Sym\bigg(
    \underbrace{|1\rangle|1\rangle\ldots|1\rangle}_{l_{1}}
    \underbrace{|2\rangle|2\rangle\ldots|2\rangle}_{l_{2}}\ldots
    \underbrace{||V|\rangle||V|\rangle\ldots||V|\rangle}_{l_{|V|}}\bigg).
\label{eq:occup_num_symmetrized}
\end{equation}

The Hamiltonian of the Bose-Hubbard model in the $N$-particle sector acts as an operator $H_{G}^{N}$ on the space $\mathcal{Z}_{N}(G)$ (see for example \cite[\S 64]{LL94}): 
\begin{equation}
H_{G}^{N}=\sum_{w=1}^{N} A(G)^{(w)}+\sum_{k\in V}\hat{n}_{k}\left(\hat{n}_{k}-1\right)\label{eq:HGn}
\end{equation}
(recall we set $\hopping=J_{\mathrm{int}}=1)$ where the number operator is
\begin{equation}
\hat{n}_{i}=\sum_{w=1}^{N} |i\rangle\langle i|^{(w)}.\label{eq:n_hat}
\end{equation}
Here (and throughout the paper) we use the notation 
\[
M^{(w)}=\id^{\otimes w-1}\otimes M\otimes\id^{\otimes N-w}
\]
to indicate that the operator $M$ acts on subsystem $w$.

While $H_{G}^{N}$ is defined as a $|V|^{N}\times|V|^{N}$ matrix in the space \eq{disting_n_particle}, we consider its restriction
\[
\bar{H}_{G}^{N} = H_{G}^{N}\Big|_{\mathcal{Z}_{N}(G)}
\]
to the bosonic $N$-particle subspace $\mathcal{Z}_{N}(G)$ (throughout the paper we write $H|_{\mathcal{W}}$ for the restriction of a Hermitian operator $H$ to a subspace $\mathcal{W}$, which can be written as a $\dim\mathcal{W}\times\dim\mathcal{W}$ matrix). It is convenient to add a term proportional to the identity to obtain a positive semidefinite operator. Letting $\mu(G)$ denote the smallest eigenvalue of the adjacency matrix $A(G)$, we consider 
\[
H(G,N)=\bar{H}_{G}^{N}-N\mu(G).
\]
Clearly $H(G,N)\geq0$ since the interaction term is positive semidefinite. Also note that, given the graph $G$, the smallest eigenvalue $\mu(G)$ of its adjacency matrix can be efficiently approximated using a classical polynomial-time algorithm, so the complexity of approximating the ground energy of $H(G,N)$ is equivalent to the complexity of approximating $\bar{H}_{G}^{N}$.
(Note that here the graph is specified explicitly by its adjacency matrix. In other contexts one might consider a graph specified compactly, e.g., by a circuit that computes rows of its adjacency matrix. Then the situation is more complex since the input size can be much smaller than the number of vertices in the graph. Indeed, we prove in \app{complexity_smallest_graph_eig} that approximating the smallest eigenvalue of such a graph is QMA-complete.)

We write 
\[
0\leq\lambda_{N}^{1}(G)\leq\lambda_{N}^{2}(G)\leq\ldots\leq\lambda_{N}^{D_{N}}(G)
\]
for the eigenvalues of $H(G,N)$ and $\{|\lambda_{N}^{j}(G)\rangle\}$ for the corresponding normalized eigenvectors. 

When $\lambda_{N}^{1}(G)=0$, the ground energy of the $N$-particle Bose-Hubbard model $\bar{H}_{G}^{N}$ is equal to $N$ times the one-particle energy $\mu(G)$. In this case we say that the $N$-particle Bose-Hubbard model is frustration free. We also define frustration freeness for $N$-particle states.

\begin{definition}
\label{defn:FF_states}If $|\psi\rangle\in\mathcal{Z}_{N}(G)$ satisfies
$H(G,N)|\psi\rangle=0$ then we say $|\psi\rangle$ is an $N$-particle
frustration-free state for $G$.
\end{definition}

We now present two basic properties of $H(G,N)$. The following Lemma states that the ground energy is non-decreasing as a function
of the number of particles $N$.

\begin{restatable}{lemma}{incrementN}
\label{lem:increase_part_number}For all $N\geq1$, $\lambda_{N+1}^{1}(G)\geq\lambda_{N}^{1}(G)$.
\end{restatable}

In this paper we will encounter disconnected graphs $G$. In the cases of interest, the smallest eigenvalue of the adjacency matrix for each component is the same. The following Lemma shows that the eigenvalues of $H(G,N)$ on such a graph can be written as sums of eigenvalues for the components. In this Lemma (and throughout the paper), we let $[k] = \{1,2,\ldots,k\}$.

\begin{restatable}{lemma}{disc}
\label{lem:BH_disconnected_graphs}
Suppose $G=\bigcup_{i=1}^{k}G_{i}$ with $\mu(G_{1})=\mu(G_{2})=\cdots=\mu(G_{k})$. The eigenvalues of $H(G,N)$ are 
\[
\sum_{i\in[k]\colon N_{i}\neq0}\lambda_{N_{i}}^{y_{i}}(G_{i})
\]
where $N_{1},\ldots,N_{k}\in\{0,1,2,\ldots\}$ with $\sum_{i}N_{i}=N$ and $y_{i}\in[D_{N_{i}}].$ The corresponding eigenvectors are (up to normalization) 
\begin{equation}
\Sym\Bigg(\bigotimes_{i\in[k]\colon N_{i}\neq0}|\lambda_{N_{i}}^{y_{i}}(G_{i})\rangle\Bigg).
\label{eq:eigvecs_disconnected}
\end{equation}
\end{restatable}

Proofs of \lem{increase_part_number} and \lem{BH_disconnected_graphs} appear in \sec{basic_properties}.

\subsection{Complexity of the Bose-Hubbard model}\label{sec:complexity}

Given a $K$-vertex graph $G$ and a number of particles $N$, how hard is it to approximate the ground energy of the $N$-particle Bose-Hubbard model $\bar{H}_{G}^{N}$ on $G$? We consider the following decision version of this computational problem.

\begin{mdframed}
\begin{problem}
[\textbf{Bose-Hubbard Hamiltonian}]
We are given a $K$-vertex graph $G$, a number of particles $N$, a real number $c$, and a precision parameter $\epsilon=\frac{1}{T}$. The positive integers $N$ and $T$ are provided in unary; the graph is specified by its adjacency matrix, which can be any $K\times K$ symmetric $0$-$1$ matrix. We are promised that either the smallest eigenvalue of $\bar{H}_{G}^{N}$ is at most $c$ (yes instance) or is at least $c+\epsilon$ (no instance) and we are asked to decide which is the case.
\end{problem}
\end{mdframed}

In this problem $c$ is provided in a straightforward manner, with enough precision to resolve $\epsilon$, i.e., using $\mathcal{O}(\log |c|+\log T)$ bits. The input size is therefore $\Theta(K^{2}+T+N+\log |c|)$ bits. We prove that this problem is QMA-complete, providing evidence that approximating the ground energy of the $N$-particle Bose-Hubbard model on a graph $G$ is intractable.

\begin{theorem}
Bose-Hubbard Hamiltonian is QMA-complete.
\end{theorem}

The proof of this Theorem has two parts. 

The easy part is to show that Bose-Hubbard Hamiltonian is contained in QMA. The basic strategy of Arthur's verification protocol is to measure the energy of the Bose-Hubbard Hamiltonian in the state $|\phi\rangle$ given to him by Merlin, using phase estimation and Hamiltonian simulation. Arthur accepts if the energy is small enough and rejects otherwise. We give a more detailed description of the verification procedure in \sec{containment_in_QMA}.

The more involved part is to show that Bose-Hubbard Hamiltonian is QMA-hard. For this we show that any instance of a QMA problem can be converted (in deterministic polynomial time on a classical computer) into an equivalent instance of Bose-Hubbard Hamiltonian. In fact, our reduction proves a slightly stronger result, namely that a notable extremal case of Bose-Hubbard Hamiltonian is already QMA-hard. We now discuss this special case.

Recall from the previous section that the ground energy of the $N$-particle Bose-Hubbard model is at least $N$ times the single-particle ground energy $\mu(G)$, i.e., $\lambda_{N}^{1}(G)\geq0$. We can ask if this inequality is close to equality, i.e., is the $N$-particle Bose-Hubbard model close to being frustration free?

\begin{mdframed}
\begin{problem}
[\textbf{Frustration-Free Bose-Hubbard Hamiltonian}]
We are given a $K$-vertex graph $G$, a number of particles $N\leq K$, and a precision parameter $\epsilon=\frac{1}{T}$. The integer $T \ge 4K$ is provided in unary; the graph is specified by its adjacency matrix, which can be any $K\times K$ symmetric $0$-$1$ matrix. We are promised that either $\lambda_{N}^{1}(G)\leq\epsilon^3$ (yes instance) or $\lambda_{N}^{1}(G)\geq \epsilon+\epsilon^3$ (no instance) and we are asked to decide which is the case. 
\end{problem}
\end{mdframed}

For concreteness, we have made some specific choices in defining this problem. Our proof that it is QMA-hard also applies, for example, to variants of the problem where $\epsilon^3$ is replaced (in both places it appears) by $\epsilon^\alpha$ for any constant $\alpha\in \{1,2,3,\ldots\}$. In \sec{reduction}, we use the version with $\alpha=3$ as stated above to facilitate a reduction to the XY Hamiltonian problem.

The requirement $T\geq 4K$ ensures that $\epsilon$ is small so that, for a yes instance, the system is very close to being frustration free. We choose the specific threshold $4K$ for concreteness.

The restriction $N\leq K$ is without loss of generality since the problem is trivial otherwise. To see this, note that any state with more than $K$ particles is orthogonal to the nullspace of the interaction term since there are always two or more particles located at one vertex; hence $\lambda_{N}^{1}(G)\geq2$ whenever $N\geq K+1$.

Frustration-Free Bose-Hubbard Hamiltonian is a special case of Bose-Hubbard Hamiltonian with $c=N\mu(G)+\epsilon^3$. To prove that Bose-Hubbard Hamiltonian is QMA-hard, it therefore suffices to prove that Frustration-Free Bose-Hubbard Hamiltonian is QMA-hard. The bulk of this paper is concerned with the proof of this fact.

\subsection{Complexity of the XY Hamiltonian problem}
\label{sec:reduction}

We reduce Frustration-Free Bose-Hubbard Hamiltonian to an eigenvalue problem for a class of $2$-local Hamiltonians defined by graphs. The reduction is based on a well-known mapping between hard-core bosons and spin systems, which we now review.

We define the subspace $\mathcal{W}_N(G)\subset \mathcal{Z}_N (G)$ of $N$ hard-core bosons on a graph $G$ to consist of the states where each vertex of $G$ is occupied by either $0$ or $1$ particle, i.e., 
\[
  \mathcal{W}_N(G)=\text{span}\{\text{Sym}(|i_1,i_2,\ldots,i_N\rangle) \colon
  i_1,\ldots,i_N\in V,\; i_j\neq i_k \text{ for distinct } j,k\in[N] \}.
\]
A basis for $\mathcal{W}_N(G)$ is the subset of occupation-number states \eq{occup_num_symmetrized} labeled by bit strings $l_1\ldots l_{|V|}\in \{0,1\}^{|V|}$ with Hamming weight $\sum_{j\in V}l_j=N$.  The space $\mathcal{W}_N(G)$ can thus be identified with the weight-$N$ subspace
\[
\text{Wt}_N(G)=\text{span}\{|z_1,\ldots, z_{|V|}\rangle:\; z_i\in\{0,1\},\; \sum_{i=1}^{|V|} z_i =N\}
\]
of a $|V|$-qubit Hilbert space. We consider the restriction of $H_G^N$ to the space $\mathcal{W}_N(G)$, which can equivalently be written as a $|V|$-qubit Hamiltonian $O_G$ restricted to the space $\text{Wt}_N(G)$.  In particular,
\begin{equation}
H_G^N\big|_{\mathcal{W}_N(G)}=O_G\big|_{\text{Wt}_N(G)}
\label{eq:restriction_equality}
\end{equation}
where 
\begin{align*}
O_G & = \sum_{\substack{A(G)_{ij}=1 \\ i\neq j}}\big(|01\rangle\langle 10|+|10\rangle\langle 01|\big)_{ij} +\sum_{A(G)_{ii}=1} |1\rangle\langle1|_i\\
&=\sum_{\substack{A(G)_{ij}=1 \\ i\neq j}}\frac{\sigma_x^i \sigma_x^j+\sigma_y^i \sigma_y^j}{2}+\sum_{A(G)_{ii}=1}\frac{1-\sigma_z^{i}}{2}.
\end{align*}
Note that the Hamiltonian $O_G$ conserves the total magnetization $M_z=\sum_{1=1}^{|V|}\frac{1-\sigma_z^{i}}{2}$ along the $z$ axis. 

We define $\theta_N(G)$ to be the smallest eigenvalue of \eq{restriction_equality}, i.e., the ground energy of $O_G$ in the sector with magnetization $N$. We show that approximating this quantity is QMA-complete.

\begin{mdframed}
\begin{problem}
[\textbf{XY Hamiltonian}]
We are given a $K$-vertex graph $G$, an integer $N\leq K$, a real number $c$, and a precision parameter $\epsilon=\frac{1}{T}$. The positive integer $T$ is provided in unary; the graph is specified by its adjacency matrix, which can be any $K\times K$ symmetric $0$-$1$ matrix. We are promised that either $\theta_N(G)\leq c$ (yes instance) or else $\theta_N(G)\geq c+\epsilon$ (no instance) and we are asked to decide which is the case.
\end{problem}
\end{mdframed}

\begin{theorem}\label{thm:XY}
XY Hamiltonian is QMA-complete.
\end{theorem}

We prove QMA-hardness of XY Hamiltonian by reduction from Frustration-Free Bose-Hubbard Hamiltonian.  The proof of \thm{XY} appears in \app{XY}.

\section{Bose-Hubbard Hamiltonian is contained in QMA}
\label{sec:containment_in_QMA}

To prove that Bose-Hubbard Hamiltonian is contained in QMA, we provide a verification algorithm satisfying the requirements of \defn{QMA}. In the Definition this algorithm is specified by a circuit involving only one measurement of the output qubit at the end of the computation. The procedure we describe below, which contains intermediate measurements in the computational basis, can be converted into a verification circuit of the desired form by standard techniques.

We are given an instance specified by $G$, $N$, $c$, and $\epsilon$. We are also given an input state $|\phi\rangle$ of $n_{\inn}$ qubits, where $n_{\inn}=\lceil \log_{2}D_{N}\rceil $ and $D_{N}$ is the dimension of $\mathcal{Z}_{N}(G)$ as given in equation \eq{DN}. Note, using the inequality $\binom{a}{b} \leq a^{b}$ in equation \eq{DN}, that $n_{\inn}=\mathcal{O}(K\log\left(N+K\right))$, where $K=|V|$ is the number of vertices in the graph $G$. We embed $\mathcal{Z}_{N}(G)$ into the space of $n_{\inn}$ qubits straightforwardly as the subspace spanned by the first $D_{N}$ standard basis vectors (with lexicographic ordering, say). The first step of the verification procedure is to measure the projector onto this space $\mathcal{Z}_{N}(G)$. If the measurement outcome is $1$ then the resulting state $|\phi^{\prime}\rangle$ is in $\mathcal{Z}_{N}(G)$ and we continue; otherwise we reject.

In the second step of the verification procedure, the goal is to measure $\bar{H}_{G}^{N}$ in the state $|\phi^{\prime}\rangle$. The Hamiltonian $\bar{H}_{G}^{N}$ is sparse and efficiently row-computable, with norm
\[
\left\Vert \bar{H}_{G}^{N}\right\Vert \leq\left\Vert H_{G}^{N}\right\Vert \leq N\left\Vert A(G)\right\Vert +\left\Vert \sum_{k\in V}\hat{n}_{k}\left(\hat{n}_{k}-1\right)\right\Vert \leq NK+N^{2}.
\]
We use phase estimation (see for example \cite{CEMM98}) to estimate the energy of $|\phi^{\prime}\rangle$, using sparse Hamiltonian simulation \cite{AT03} to approximate evolution according to $\bar{H}_{G}^{N}$. We choose the parameters of the phase estimation so that, with probability at least $\frac{2}{3}$, it produces an approximation $E$ of the energy with error at most $\frac{\epsilon}{4}$. This can be done in time $\poly(N,K,\frac{1}{\epsilon})$. If $E\leq c+\frac{\epsilon}{2}$ then we accept; otherwise we reject.

We now show that this verification procedure satisfies the completeness and soundness requirements of \defn{QMA}. For a yes instance, an eigenvector of $\bar{H}_{G}^{N}$ with eigenvalue $e\leq c$ is accepted by this procedure as long as the energy $E$ computed in the phase estimation step has the desired precision. To see this, note that we measure $\left|E-e\right|\leq\frac{\epsilon}{4}$, and hence $E\leq c+\frac{\epsilon}{4}$, with probability at least $\frac{2}{3}$.  For a no instance, write $|\phi^{\prime}\rangle\in\mathcal{Z}_{N}(G)$ for a state obtained after passing the first step. The value $E$ computed by the subsequent phase estimation step satisfies $E\geq c+\frac{3\epsilon}{4}$ with probability at least $\frac{2}{3}$, in which case the state is rejected. From this we see that the probability of accepting a no instance is at most $\frac{1}{3}$.

\section{Gate graphs}
\label{sec:A-class-of}

In this Section we define a class of graphs (\emph{gate graphs}) and a diagrammatic notation for them (\emph{gate diagrams}). We also discuss the Bose-Hubbard model on these graphs.

Every gate graph is constructed using a specific $128$-vertex graph $g_{0}$ as a building block. This graph is shown in \fig{g_0} and discussed in \sec{Encoding-a-Computation}. In \sec{Gate-graphs-and} we define gate graphs and gate diagrams. A gate graph is obtained by adding edges and self-loops (in a prescribed way) to a collection of disjoint copies of $g_{0}$.

In \sec{FF_State} we discuss the ground states of the Bose-Hubbard model on gate graphs. For any gate graph $G$, the smallest eigenvalue $\mu(G)$ of the adjacency matrix $A(G)$ satisfies $\mu(G)\geq-1-3\sqrt{2}$. It is convenient to define the constant
\begin{equation}
e_{1}=-1-3\sqrt{2}.\label{eq:e1_defn}
\end{equation}
When $\mu(G)=e_{1}$ we say $G$ is an $e_{1}$-gate graph. We focus on the frustration-free states of $e_1$-gate graphs (recall from \defn{FF_states} that $|\phi\rangle\in \mathcal{Z}_N(G)$ is frustration free if and only if $H(G,N)|\phi\rangle=0$). We show that all such states live in a convenient subspace (called $\mathcal{I}(G,N)$) of the $N$-particle Hilbert space. This subspace has the property that no two (or more) particles ever occupy vertices of the same copy of $g_{0}$. The restriction to this subspace makes it easier to analyze the ground space.

In \sec{Occupancy-constraints} we consider a class of subspaces that, like $\mathcal{I}(G,N)$, are defined by a set of constraints on the locations of $N$ particles in an $e_{1}$-gate graph $G$. We state an ``\ocl'' (proven in \app{Occupancy-Constraints-Lemma}) that relates a subspace of this form to the ground space of the Bose-Hubbard model on a graph derived from $G$.

\subsection{The graph $g_{0}$}
\label{sec:Encoding-a-Computation}

The graph $g_{0}$ shown in \fig{g_0} is closely related to a single-qubit circuit $\mathcal{C}_{0}$ with eight gates $U_{j}$ for $j\in[8]$, where 
\begin{align*}
U_{1} & =U_{2}=U_{7}=U_{8}=H\qquad U_{3}=U_{5}=HT\qquad U_{4}=U_{6}=\left(HT\right)^{\dagger}
\end{align*}
with 
\[
H=\frac{1}{\sqrt{2}}\begin{pmatrix}
1 & 1\\
1 & -1
\end{pmatrix}\qquad T=\begin{pmatrix}
1 & 0\\
0 & e^{i\frac{\pi}{4}}
\end{pmatrix}.
\]
In this section we map this circuit to the graph $g_{0}$. The mapping we use can be generalized to map an arbitrary quantum circuit with any number of qubits to a graph, but for simplicity we focus here on $g_{0}$. In \app{complexity_smallest_graph_eig} we discuss the more general mapping and use it to prove that computing (in a certain precise sense specified in the Appendix) the smallest eigenvalue of a sparse, efficiently row-computable symmetric $0$-$1$ matrix is QMA-complete.

\begin{figure}
\centering
\begin{tikzpicture}[scale = 0.75]

  \foreach \sym in {1,-1}{  
  \begin{scope}[xscale = \sym]
  \foreach \theta/\shiftaa /\shiftab /\shiftba /\shiftbb in {
      0/4/4/4/0,
      -45/4/4/4/0,
      -90/4/4/5/1,
      -135/4/3/4/7}{
  \begin{scope}[rotate = \theta]
    \foreach \zstart / \zend /\shift in {6/6/\shiftaa,6/1.5/\shiftab,1.5/6/\shiftba,1.5/1.5/\shiftbb}{
    \foreach \w in {0,...,7}{
      \draw[draw=black!70] let \n1={int(mod(\w + \shift,8)) /2 + \zend} in (90:{\w/2 + \zstart} ) -- (45: \n1 cm);
    }}
  \end{scope}}
  \end{scope}}

  \foreach \theta in {0,45,...,315}{
  \foreach \yshift in {1.5,6}{
  \begin{scope}[rotate = \theta,yshift=\yshift cm]
  \begin{scope}[draw=black!40]
    \draw (0,3.5) to[out=210,in = 150] (0,1);
    \draw (0,3.5) to[out=225,in = 135] (0,1.5);
    \draw (0,3.5) to[out=240,in=120] (0,2);
  
    \draw (0,3) to[out=330,in = 30] (0,.5);
    \draw (0,3) to[out=315,in=45] (0,1);
    \draw (0,3) to[out=300,in=60] (0,1.5);
  
    \draw (0,0) to[out=30,in=330] (0,2.5);
    \draw (0,0) to[out=45,in=315] (0,2);
    \draw (0,0) to[out=60,in=300] (0,1.5);
  
    \draw (0,.5) to[out=135,in=225] (0,2.5);
    \draw (0,.5) to[out=120,in=240] (0,2);
  
    \draw (0,2.5) to[out=240,in=120] (0,1);
  \end{scope}
   \foreach \y in {0,.5,...,3.5}{
    \draw[fill = black,draw=black] (0,\y cm) circle (.33mm);
  }
  
  \end{scope}}}
  
  \foreach \t/\gate in {1/H,2/H,3/HT,4/{(HT)^\dag},5/HT,6/{(HT)^\dag},7/H,8/H}{
    \node at ({135 - \t * 45} :10.25) {$t = \t$};
    \draw[->,draw=black] ({120-\t*45} :10.25) arc[radius=10.25, start angle = {120-\t*45} ,end angle= {105-\t*45}];
    \node[fill=white] at ({112.5 - \t * 45} :10.25) {$\gate$ };
  }
\end{tikzpicture}

\caption{The graph $g_{0}$.\label{fig:g_0}}
\end{figure}

Starting with the circuit $\mathcal{C}_{0}$, we apply the Feynman-Kitaev circuit-to-Hamiltonian mapping \cite{Fey85,KSV02} (up to a constant term and overall multiplicative factor) to get the Hamiltonian
\begin{equation}
-\sqrt{2}\sum_{t=1}^{8}\left(U_{t}^{\dagger}\otimes|t\rangle\langle t+1|+U_{t}\otimes|t+1\rangle\langle t|\right).\label{eq:single_qubit_ham}
\end{equation}
This Hamiltonian acts on the Hilbert space $\C^{2}\otimes\C^{8}$, where the second register (the ``clock register'') has periodic boundary conditions (i.e., we let $|8+1\rangle=|1\rangle$). The ground space of \eq{single_qubit_ham} is spanned by so-called history states
\begin{align*}
  |\phi_{z}\rangle & 
  =\frac{1}{\sqrt{8}} \big(|z\rangle(|1\rangle+|3\rangle+|5\rangle+|7\rangle)
    +H|z\rangle(|2\rangle+|8\rangle)
    +HT|z\rangle(|4\rangle+|6\rangle)\big),
  \quad z\in\{0,1\},
\end{align*}
that encode the history of the computation where the circuit $\mathcal{C}_{0}$ is applied to $|z\rangle$. One can easily check that $|\phi_{z}\rangle$ is an eigenstate of the Hamiltonian with eigenvalue $-2\sqrt{2}$. 

Now we modify \eq{single_qubit_ham} to give a symmetric $0$-$1$ matrix. The trick we use is a variant of one used in references \cite{JW06,JGL10} for similar purposes. 

The nonzero standard basis matrix elements of \eq{single_qubit_ham} are integer powers of $\omega=e^{i\frac{\pi}{4}}$. Note that $\omega$ is an eigenvalue of the $8\times8$ shift operator 
\[
S=\sum_{j=0}^{7}|j+1\bmod 8\rangle\langle j|
\]
with eigenvector 
\[
|\omega\rangle=\frac{1}{\sqrt{8}}\sum_{j=0}^{7}\omega^{-j}|j\rangle.
\]
For each operator $-\sqrt{2}H,-\sqrt{2}HT,$ or $-\sqrt{2}(HT)^{\dagger}$ appearing in equation \eq{single_qubit_ham}, define another operator acting on $\C^{2}\otimes\C^{8}$ by replacing nonzero matrix elements with powers of the operator $S$, namely $\omega^k \mapsto S^k$. Write $B(U)$ for the operator obtained by making this replacement in $U$, e.g.,
\begin{align*}
-\sqrt{2}HT & = \begin{pmatrix}
\omega^{4} & \omega^{5}\\
\omega^{4} & \omega
\end{pmatrix} 
\mapsto B(HT)
=\begin{pmatrix}
S^{4} & S^{5}\\
S^{4} & S
\end{pmatrix}.
\end{align*}
We adjoin an 8-level ancilla and we make this replacement in equation \eq{single_qubit_ham}. This gives
\begin{align}
H_{\prop} & =\sum_{t=1}^{8}\left(B(U_{t})^{\dagger}_{13}\otimes|t\rangle\langle t+1|_2+B(U_{t})_{13}\otimes|t+1\rangle\langle t|_2\right),\label{eq:Hprop}
\end{align}
a symmetric $0$-1 matrix acting on $\C^{2}\otimes\C^{8}\otimes\C^{8}$, where the second register is the clock register and the third register is the ancilla register on which the $S$ operators act (the subscripts indicate which registers are acted upon). It is an insignificant coincidence that the clock and ancilla registers have the same dimension. 

Note that $H_{\prop}$ commutes with $S$ (acting on the $8$-level ancilla) and therefore is block diagonal with eight sectors. In the sector where $S$ has eigenvalue $\omega$, it is identical to the Hamiltonian we started with, equation \eq{single_qubit_ham}. There is also a sector (where $S$ has eigenvalue $\omega^*$) where the Hamiltonian is the entrywise complex conjugate of the one we started with. We add a term to $H_{\prop}$ that assigns an energy penalty to states in any of the other six sectors, ensuring that none of these states lie in the ground space of the resulting operator.

Now we can define the graph $g_{0}$. Each vertex in $g_0$ corresponds to a standard basis vector in the Hilbert space $\C^{2}\otimes\C^{8}\otimes\C^{8}$. We label the vertices $(z,t,j)$ with $z\in\{0,1\}$ describing the state of the computational qubit, $t\in[8]$ giving the state of the clock, and $j\in\{0,\ldots,7\}$ describing the state of the ancilla. The adjacency matrix is
\[
A(g_{0})=H_{\prop}+H_{\text{penalty}}
\]
where the penalty term
\[
H_{\text{penalty}}=\id\otimes\id\otimes\left(S^{3}+S^{4}+S^{5}\right)
\]
acts nontrivially on the third register. The graph $g_{0}$ is shown in \fig{g_0}.

Now consider the ground space of $A(g_{0})$. Note that $H_{\prop}$ and $H_{\text{penalty}}$ commute, so they can be simultaneously diagonalized. Furthermore, $H_{\text{penalty}}$ has smallest eigenvalue $-1-\sqrt{2}$ (with eigenspace spanned by $|\omega\rangle$ and $|\omega^*\rangle$) and first excited energy $-1$. The norm of $H_{\prop}$ satisfies $\left\Vert H_{\prop}\right\Vert \leq4$, which follows from the fact that $H_{\prop}$ has four ones in each row and column (with the remaining entries all zero).

The smallest eigenvalue of $A(g_{0})$ lives in the sector where $H_{\text{penalty}}$
has eigenvalue $-1-\sqrt{2}$ and is equal to 
\begin{equation}
  -2\sqrt{2}+(-1-\sqrt{2})=-1-3\sqrt{2}=-5.24\ldots.\label{eq:e_1_definition}
\end{equation}
This is the constant $e_{1}$ from equation \eq{e1_defn}. To see this, note that in any other sector $H_{\text{penalty}}$ has eigenvalue at least $-1$ and every eigenvalue of $A(g_{0})$ is at least $-5$ (using the fact that $H_{\prop}\geq-4$). An orthonormal basis for the ground space of $A(g_{0})$ is furnished by the states
\begin{align}
|\psi_{z,0}\rangle & =\frac{1}{\sqrt{8}}\big(|z\rangle(|1\rangle+|3\rangle+|5\rangle+|7\rangle)+H|z\rangle(|2\rangle+|8\rangle)+HT|z\rangle(|4\rangle+|6\rangle)\big)|\omega\rangle\label{eq:psi0m}\\
|\psi_{z,1}\rangle & =|\psi_{z,0}\rangle^{*}\label{eq:psi1m}
\end{align}
where $z\in \{0,1\}$.

Note that the amplitudes of $|\psi_{z,0}\rangle$ in the above basis contain the result of computing either the identity, Hadamard, or $HT$ gate acting on the ``input'' state $|z\rangle$.

\subsection{Gate graphs and gate diagrams}
\label{sec:Gate-graphs-and}

We use three different schematic representations of the graph $g_{0}$ (defined in \sec{Encoding-a-Computation}), as depicted in \fig{diagram_elements}. We call these Figures \emph{diagram elements}; they are also the simplest examples of \emph{gate diagrams}, which we define shortly.

\begin{figure}
\centering
\subfloat[][]{ 
\label{fig:diagram_elementH}
\begin{tikzpicture}
  \draw[rounded corners=2mm,thick] (0,0) rectangle (3.24 cm, 2cm);

 \foreach \y in {.33,.66,1.33,1.66}
{ 
 \foreach \x/\color in {0/black,3.24/gray}
{      
\draw[fill=\color,draw=\color] (\x cm, \y cm) circle (.66mm);   
}
}   

\foreach \z in {0,1}
{    
\node[right] at (0,1.66-\z) {\footnotesize $(\z,1)$}; 
\node[right] at (0,1.33-\z) {\footnotesize $(\z,3)$};    

   \node[left] at (3.24,1.66-\z) {\footnotesize $(\z,2)$}; 
   \node[left] at (3.24,1.33-\z) {\footnotesize $(\z,8)$};  
}   

  \node at (1.62,1) {\huge $H$};   
\end{tikzpicture} 
} 
\qquad
\subfloat[][]{
\label{fig:diagram_elementHT}
\begin{tikzpicture}
  \draw[rounded corners=2mm,thick] (0,0) rectangle (3.24 cm, 2cm);
   
\foreach \y in {.33,.66,1.33,1.66}
{  
\foreach \x/\color in {0/black,3.24/gray}
{     
\draw[fill=\color,draw=\color] (\x cm, \y cm) circle (.66mm);   
}
}     

\foreach \z in {0,1}
{    
\node[right] at (0,1.66-\z) {\footnotesize $(\z,1)$}; 
   \node[right] at (0,1.33-\z) {\footnotesize $(\z,3)$};     
  
  \node[left] at (3.24,1.66-\z) {\footnotesize $(\z,4)$}; 
   \node[left] at (3.24,1.33-\z) {\footnotesize $(\z,6)$}; 
 }    
  
\node at (1.62,1) {\huge $HT$}; 
 \end{tikzpicture}
} 
\qquad 
\subfloat[][]
{
\label{fig:diagram_element1}
\begin{tikzpicture}
  \draw[rounded corners=2mm,thick] (0,0) rectangle (3.24 cm, 2cm);

\foreach \y in {.33,.66,1.33,1.66}
{   
\foreach \x/\color in {0/black,3.24/gray}
{     
\draw[fill=\color,draw=\color] (\x cm, \y cm) circle (.66mm);   
}
}    

\foreach \z in {0,1}
{     
\node[right] at (0,1.66-\z) {\footnotesize $(\z,1)$};   
 \node[right] at (0,1.33-\z) {\footnotesize $(\z,3)$};    

 \node[left] at (3.24,1.66-\z) {\footnotesize $(\z,5)$};  
  \node[left] at (3.24,1.33-\z) {\footnotesize $(\z,7)$};   }    
  
\node at (1.62,1) {\huge $1$};  
\end{tikzpicture}  
}

\caption{Diagram elements from which a gate diagram is constructed. Each diagram element is a schematic representation of the graph $g_{0}$ shown
in \fig{g_0}. \label{fig:diagram_elements}}
\end{figure}
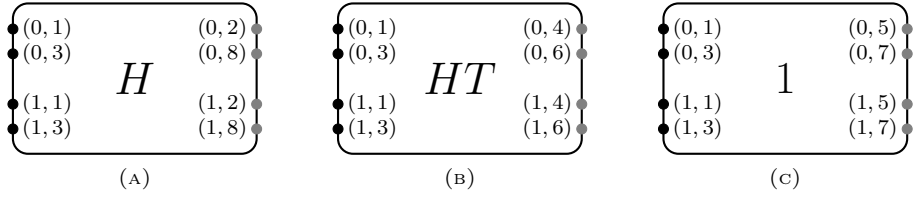

The black and grey circles in a diagram element are called ``nodes.'' Each node has a label $(z,t)$. The only difference between the three diagram elements is the labeling of their nodes. In particular, the nodes in the diagram element $U\in\{\id,H,HT\}$ correspond to values of $t\in[8]$ where the first register in equation \eq{psi0m} is either $|z\rangle$ or $U|z\rangle$. For example, the nodes for the $H$ diagram element have labels with $t\in\{1,3\}$ (where $|\psi_{z,0}\rangle$ contains the ``input'' $|z\rangle$) or $t=\{2,8\}$ (where $|\psi_{z,0}\rangle$ contains the ``output'' $H|z\rangle$). We draw the input nodes in black and the output nodes in grey.

The rules for constructing gate diagrams are simple. A gate diagram consists of some number $R \in \{1,2,\ldots\}$ of diagram elements, with self-loops attached to a subset $\mathcal{S}$ of the nodes and edges connecting a set $\mathcal{E}$ of pairs of nodes. A node may have a single edge or a single self-loop attached to it, but never more than one edge or self-loop and never both an edge and a self-loop. Each node in a gate diagram has a label $(q,z,t)$ where $q \in [R]$ indicates the diagram element it belongs to. An example is shown in \fig{simple_gate_diagram}. Sometimes it is convenient to draw the input nodes on the right-hand side of a diagram element; e.g., in \fig{W_gadget} the node closest to the top left corner is labeled $(q,z,t)=(3,0,2)$.

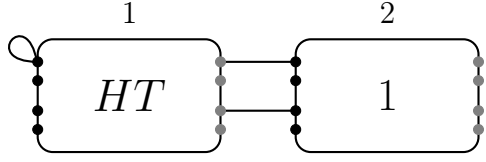
\begin{figure}
\centering \begin{tikzpicture}

 \foreach \y in {1.2,.55}
{   
 \draw[thick] (2.43,\y) -- (3.43,\y);  
}      
\draw[thick,looseness = 200] (0,1.2) to [out = 180, in = 100] (-.01,1.2);

 \foreach \offset/\unitary/\label in {0/HT/1,3.43/1/2}
{ 
 \begin{scope}[xshift=\offset cm] 
 \draw[rounded corners=2mm,thick] (0,0) rectangle (2.43cm,1.5 cm);
 \foreach \x /\color in {0/black,2.43/gray}
{     
\foreach \y in {1.2,.95,.55,.3}
{      
\draw[fill=\color,draw=\color] (\x cm, \y cm) circle (.66mm);  
  }
}    
\node at (1.22cm, .75cm) {\huge $\unitary$};   
 \node at (1.22cm, 1.85cm){\Large \label};  
\end{scope}
}    
\end{tikzpicture}
\caption{A gate diagram with two diagram elements labeled $q=1$ (left) and $q=2$ (right).
\label{fig:simple_gate_diagram}}
\end{figure}

To every gate diagram we associate a \emph{gate graph} $G$ with
vertex set 
\[
\left\{ (q,z,t,j)\colon q\in[R],\, z\in\{0,1\},\, t\in[8],\, j\in\{0,\ldots,7\}\right\} 
\]
and adjacency matrix 
\begin{align}
A(G) & =\id_{q}\otimes A(g_{0})+h_{\mathcal{S}}+h_{\mathcal{E}}\label{eq:adj_gate_graph}\\
h_{\mathcal{S}} & =\sum_{\mathcal{S}}|q,z,t\rangle\langle q,z,t|\otimes\id_{j}\label{eq:h_loops}\\
h_{\mathcal{E}} & =\sum_{\mathcal{E}}\left(|q,z,t\rangle+|q^{\prime},z^{\prime},t^{\prime}\rangle\right)\left(\langle q,z,t|+\langle q^{\prime},z^{\prime},t^{\prime}|\right)\otimes\id_{j}.\label{eq:h_edges}
\end{align}
The sums in equations \eq{h_loops} and \eq{h_edges} run over the set of nodes with self-loops $(q,z,t)\in\mathcal{S}$ and the set of pairs of nodes connected by edges $\{(q,z,t),(q^{\prime},z^{\prime},t^{\prime})\}\in\mathcal{E}$, respectively. We write $\id_{q}$ and $\id_{j}$ for the identity operator on the registers with variables $q$ and $j$, respectively. We see from the above expression that each self-loop in the gate diagram corresponds to $8$ self-loops in the graph $G$, and an edge in the gate diagram corresponds to $8$ edges and $16$ self-loops in $G$.

Since a node in a gate graph never has more than one edge or self-loop attached to it, equations \eq{h_loops} and \eq{h_edges} are sums of orthogonal Hermitian operators.  Therefore 
\begin{align}
\left\Vert h_{\mathcal{S}}\right\Vert  & =\max_{\mathcal{S}}\left\Vert |q,z,t\rangle\langle q,z,t|\otimes\id_{j}\right\Vert =1\quad\text{if }\mathcal{S}\neq\emptyset\label{eq:h_S_bound}\\
\left\Vert h_{\mathcal{E}}\right\Vert  & =\max_{\mathcal{E}}\left\Vert \left(|q,z,t\rangle+|q^{\prime},z^{\prime},t^{\prime}\rangle\right)\left(\langle q,z,t|+\langle q^{\prime},z^{\prime},t^{\prime}|\right)\otimes\id_{j}\right\Vert =2\quad\text{if }\mathcal{E}\neq\emptyset\label{eq:h_E_bound}
\end{align}
for any gate graph. (Of course, this also shows that $\|{h_{\mathcal{S}^\prime}}\|=1$ and $\|{h_{\mathcal{E}^\prime}}\|=2$ for any nonempty subsets $\mathcal{S}^\prime\subseteq \mathcal{S}$ and $\mathcal{E}^\prime\subseteq \mathcal{E}$.)

\subsection{Frustration-free states on $e_{1}$-gate graphs}
\label{sec:FF_State}

Consider the adjacency matrix $A(G)$ of a gate graph $G$, and note (from equation \eq{adj_gate_graph}) that its smallest eigenvalue $\mu(G)$ satisfies
\[
\mu(G)\geq e_{1}
\]
since $h_{\mathcal{S}}$ and $h_{\mathcal{E}}$ are positive semidefinite and $A(g_{0})$ has smallest eigenvalue $e_{1}$. In the special case where $\mu(G)=e_{1}$, we say $G$ is an $e_{1}$-gate graph.

\begin{definition}
An $e_{1}$-gate graph is a gate graph $G$ such that the smallest eigenvalue of its adjacency matrix is $e_{1}=-1-3\sqrt{2}$.
\end{definition}

When $G$ is an $e_{1}$-gate graph, a single-particle ground state $|\Gamma\rangle$ of $A(G)$ has minimal energy for each term in \eq{adj_gate_graph}, i.e., it satisfies 
\begin{align}
\left(\id\otimes A(g_{0})\right)|\Gamma\rangle & =e_{1}|\Gamma\rangle\label{eq:Gamma_disc}\\
h_{\mathcal{S}}|\Gamma\rangle & =0\label{eq:h_s_0}\\
h_{\mathcal{E}}|\Gamma\rangle & =0.\label{eq:h_e_0}
\end{align}
Indeed, to show that a given gate graph $G$ is an $e_{1}$-gate graph, it suffices to find a state $|\Gamma\rangle$ satisfying these conditions. Note that equation \eq{Gamma_disc} implies that $|\Gamma\rangle$ can be written as a superposition of the states
\[
  |\psi_{z,a}^{q}\rangle=|q\rangle|\psi_{z,a}\rangle,\quad
  z,a\in\{0,1\},\, q\in[R]
\]
where $|\psi_{z,a}\rangle$ is given by equations \eq{psi0m} and \eq{psi1m}. The coefficients in the superposition are then constrained by equations \eq{h_s_0} and \eq{h_e_0}.
 
\begin{example}\label{ex:As-an-example}
As an example, we show the gate graph in \fig{simple_gate_diagram} is an $e_{1}$-gate graph. As noted above, equation \eq{Gamma_disc} lets us restrict our attention to the space spanned by the eight states $|\psi_{z,a}^{q}\rangle$ with $z,a\in \{0,1\}$ and $q\in \{1,2\}$. In this basis, the operators $h_{\mathcal{S}}$ and $h_{\mathcal{E}}$ only have nonzero matrix elements between states with the same value of $a\in\{0,1\}$. We therefore solve for the $e_{1}$ energy ground states with $a=0$ and those with $a=1$ separately. Consider a ground state of the form
\[
\left(\tau_{1}|\psi_{0,a}^{1}\rangle+\nu_{1}|\psi_{1,a}^{1}\rangle\right)+\left(\tau_{2}|\psi_{0,a}^{2}\rangle+\nu_{2}|\psi_{1,a}^{2}\rangle\right)
\]
and note that in this case \eq{h_s_0} implies $\tau_{1}=0$. Equation \eq{h_e_0} gives
\[
\begin{pmatrix}
\tau_{2}\\
\nu_{2}
\end{pmatrix}=\begin{cases}
HT\begin{pmatrix}
-\tau_{1}\\
-\nu_{1}
\end{pmatrix} & a=0\\
(HT)^{*}\begin{pmatrix}
-\tau_{1}\\
-\nu_{1}
\end{pmatrix} & a=1.
\end{cases}
\]
We find two orthogonal $e_{1}$-energy states, which are (up to normalization)
\begin{align}
|\psi_{1,0}^{1}\rangle-\frac{e^{i\frac{\pi}{4}}}{\sqrt{2}}\left(|\psi_{0,0}^{2}\rangle-|\psi_{1,0}^{2}\rangle\right)\label{eq:example_ffstate_1}\\
|\psi_{1,1}^{1}\rangle-\frac{e^{-i\frac{\pi}{4}}}{\sqrt{2}}\left(|\psi_{0,1}^{2}\rangle-|\psi_{1,1}^{2}\rangle\right) & .\label{eq:example_ffstate_2}
\end{align}
We interpret each of these states as encoding a qubit that is transformed at each set of input/output nodes in the gate diagram in \fig{simple_gate_diagram}. The encoded qubit begins on the input nodes of the first diagram element in the state 
\[
\begin{pmatrix}
\tau_{1}\\
\nu_{1}
\end{pmatrix}=\begin{pmatrix}
0\\
1
\end{pmatrix}
\]
because the self-loop penalizes the basis vectors $|\psi_{0,a}^{1}\rangle$. On the output nodes of diagram element $1$, the encoded qubit is in the state where either $HT$ (if $a=0$) or its complex conjugate (if $a=1$) has been applied. The edges in the gate diagram ensure that the encoded qubit on the input nodes of diagram element 2 is minus the state on the output nodes of diagram element $1$.
\end{example}

In this example, each single-particle ground state encodes a single-qubit computation. Later we show how $N$-particle frustration-free states on $e_{1}$-gate graphs can encode computations on $N$ qubits. Recall from \defn{FF_states} that a state $|\Gamma\rangle\in\mathcal{Z}_{N}(G)$ is said to be frustration free if and only if $H(G,N)|\Gamma\rangle=0.$ Note that $H(G,N)\geq0$, so an $N$-particle frustration-free state is necessarily a ground state. Putting this together with \lem{increase_part_number}, we see that the existence of an $N$-particle frustration-free state implies
\[
\lambda_{N}^{1}(G)=\lambda_{N-1}^{1}(G)=\cdots=\lambda_{1}^{1}(G)=0,
\]
i.e., there are $N^{\prime}$-particle frustration-free states for all $N^{\prime}\leq N$. 

We prove that the graph $g_{0}$ has no two-particle frustration-free states. By \lem{increase_part_number}, it follows that $g_0$ has no $N$-particle frustration-free states for $N\geq 2$.

\begin{lemma}
\label{lem:2particle}$\lambda_{2}^{1}(g_{0})>0$.
\end{lemma}

\begin{proof}
Suppose (for a contradiction) that $|Q\rangle\in\mathcal{Z}_{2}(g_{0})$ is a nonzero vector in the nullspace of $H(g_{0},2)$, so 
\[
H_{g_{0}}^{2}|Q\rangle=\bigg(A(g_{0})\otimes\id+\id\otimes A(g_{0})+2\sum_{v\in g_{0}}|v\rangle\langle v|\otimes|v\rangle\langle v|\bigg)|Q\rangle=2e_{1}|Q\rangle.
\]
This implies 
\[
A(g_{0})\otimes\id|Q\rangle=\id\otimes A(g_{0})|Q\rangle=e_{1}|Q\rangle
\]
since $A(g_0)$ has smallest eigenvalue $e_1$ and the interaction term is positive semidefinite. We can therefore write 
\[
|Q\rangle=\sum_{z,a,x,y\in\{0,1\}}Q_{za,xy}|\psi_{z,a}\rangle|\psi_{x,y}\rangle
\]
with $Q_{za,xy}=Q_{xy,za}$ (since $|Q\rangle\in\mathcal{Z}_{2}(g_{0})$) and 
\begin{equation}
\left(|v\rangle\langle v|\otimes|v\rangle\langle v|\right)|Q\rangle=0\label{eq:eqn_twoparticleannihilate}
\end{equation}
for all vertices $v=(z,t,j)\in g_{0}.$ Using this equation with
$|v\rangle=|0,1,j\rangle$ gives 
\begin{align*}
&Q_{00,00}\langle0,1,j|\psi_{0,0}\rangle^{2}  +2Q_{01,00}\langle0,1,j|\psi_{0,1}\rangle\langle0,1,j|\psi_{0,0}\rangle+Q_{01,01}\langle0,1,j|\psi_{0,1}\rangle^{2}\\
 &\quad =\frac{1}{64}\left(Q_{00,00}i^{-j}+2Q_{01,00}+Q_{01,01}i^{j}\right)\\
 &\quad =0
\end{align*}
for each $j\in\{0,\ldots,7\}$. The only solution to this set of equations is $Q_{00,00}=Q_{01,00}=Q_{01,01}=0$. The same analysis, now using $|v\rangle=|1,1,j\rangle$, gives $Q_{10,10}=Q_{11,10}=Q_{11,11}=0$. Finally, using equation \eq{eqn_twoparticleannihilate} with $|v\rangle=|0,2,j\rangle$ gives
\begin{align*}
 & \frac{1}{64}\langle0|H|1\rangle\langle0|H|0\rangle\left(2Q_{10,00}i^{-j}+2Q_{10,01}+2Q_{11,00}+2Q_{11,01}i^{j}\right)\\
 &\quad =\frac{1}{64}\left(Q_{10,00}i^{-j}+Q_{10,01}+Q_{11,00}+Q_{11,01}i^{j}\right)\\
 &\quad =0
\end{align*}
for all $j\in\{0,\ldots,7\}$, which implies that $Q_{10,00}=Q_{11,01}=0$ and $Q_{11,00}=-Q_{10,01}$. Thus, up to normalization, 
\[
|Q\rangle=|\psi_{1,0}\rangle|\psi_{0,1}\rangle+|\psi_{0,1}\rangle|\psi_{1,0}\rangle-|\psi_{11}\rangle|\psi_{00}\rangle-|\psi_{00}\rangle|\psi_{11}\rangle.
\]
Now applying equation \eq{eqn_twoparticleannihilate} with $|v\rangle=|0,4,j\rangle$, we see that the quantity 
\begin{align*}
\frac{1}{64}\left(2\langle0|HT|1\rangle\langle0|(HT)^{*}|0\rangle-2\langle0|(HT)^{*}|1\rangle\langle0|HT|0\rangle\right) & =\frac{1}{64}\left(e^{i\frac{\pi}{4}}-e^{-i\frac{\pi}{4}}\right)
\end{align*}
must be zero, which is a contradiction. Hence we conclude that the nullspace of $H(g_{0},2)$ is empty.
\end{proof}
We now characterize the space of $N$-particle frustration-free states on an $e_{1}$-gate graph $G$. Define the subspace $\mathcal{I}(G,N)\subset\mathcal{Z}_{N}(G)$ where each particle is in a ground state of $A(g_{0})$ and no two particles are located within the same diagram element: 
\begin{equation}
  \mathcal{I}(G,N)=\spn\{
  \Sym(|\psi_{z_{1},a_{1}}^{q_{1}}\rangle
  \ldots|\psi_{z_{N},a_{N}}^{q_{N}}\rangle)\colon 
  z_{i},a_{i}\in\{0,1\},\; q_{i}\in[R],\; 
  q_{i}\neq q_{j}\;\text{whenever}\; i\neq j\}.\label{eq:Ign}
\end{equation}

\begin{lemma}\label{lem:FF_characterization}
Let $G$ be an $e_{1}$-gate graph. A state $|\Gamma\rangle\in\mathcal{Z}_{N}(G)$ is frustration free if and only if 
\begin{align}
\left(A(G)-e_{1}\right)^{(w)}|\Gamma\rangle & =0\;\text{ for all }w\in[N]\label{eq:ff_condition1}\\
|\Gamma\rangle & \in\mathcal{I}(G,N).\label{eq:ff_condition2}
\end{align}
\end{lemma}

\begin{proof}
First suppose that equations \eq{ff_condition1} and \eq{ff_condition2} hold. From \eq{ff_condition2} we see that $|\Gamma\rangle$ has no support on states where two or more particles are located at the same vertex. Hence 
\begin{equation}
\sum_{k\in V}\hat{n}_{k}\left(\hat{n}_{k}-1\right)|\Gamma\rangle=0.\label{eq:ff_condition3}
\end{equation}
Putting together equations \eq{ff_condition1} and \eq{ff_condition3}, we get 
\[
H(G,N)|\Gamma\rangle=\left(H_{G}^{N}-Ne_{1}\right)|\Gamma\rangle=0,
\]
so $|\Gamma\rangle$ is frustration free.

To complete the proof, we show that if $|\Gamma\rangle$ is frustration free, then conditions \eq{ff_condition1} and \eq{ff_condition2} hold. By definition, a frustration-free state $|\Gamma\rangle$ satisfies 
\begin{equation}
H(G,N)|\Gamma\rangle=\left(\sum_{w=1}^{N}\left(A(G)-e_{1}\right)^{(w)}+\sum_{k\in V}\hat{n}_{k}\left(\hat{n}_{k}-1\right)\right)|\Gamma\rangle=0.\label{eq:defn_frustr}
\end{equation}
Since both terms in the large parentheses are positive semidefinite, they must both annihilate $|\Gamma\rangle$; similarly, each term in the first summation must be zero. Hence equation \eq{ff_condition1} holds. Let $G_{\mathrm{rem}}$ be the graph obtained from $G$ by removing all of the edges and self-loops in the gate diagram of $G$. In other words,
\[
A(G_{\mathrm{rem}})=\sum_{q=1}^{R}|q\rangle\langle q|\otimes A(g_{0})=\id\otimes A(g_{0}).
\]
Noting that 
\[
H(G,N)\geq H(G_{\mathrm{rem}},N)\geq0,
\]
we see that equation \eq{defn_frustr} also implies 
\begin{equation}
H(G_{\mathrm{rem}},N)|\Gamma\rangle=0.\label{eq:Grem_gamma}
\end{equation}
Since each of the $R$ components of $G_{\mathrm{rem}}$ is an identical copy of $g_{0}$, the eigenvalues and eigenvectors of $H(G_{\mathrm{rem}},N)$ are characterized by \lem{BH_disconnected_graphs} (along with knowledge of the eigenvalues and eigenvectors of $g_{0}$). By \lem{2particle} and \lem{increase_part_number}, no component has a two- (or more) particle frustration-free state. Combining these two facts, we see that in an $N$-particle frustration-free state, every component of $G_{\mathrm{rem}}$ must contain either $0$ or $1$ particles, and the nullspace of $H(G_{\mathrm{rem}},N)$ is the space $\mathcal{I}(G,N).$ From equation \eq{Grem_gamma} we get $|\Gamma\rangle\in\mathcal{I}(G,N)$.
\end{proof}

Note that if $\mathcal{I}(G,N)$ is empty then \lem{FF_characterization} says that $G$ has no $N$-particle frustration-free states. For example, this holds for any $e_{1}$-gate graph $G$ whose gate diagram has $R<N$ diagram elements.

A useful consequence of \lem{FF_characterization} is the fact that every $k$-particle reduced density matrix of an $N$-particle frustration-free state $|\Gamma\rangle$ on an $e_{1}$-gate graph $G$ (with $k\leq N$) has all of its support on $k$-particle frustration-free states. To see this, note that for any partition of the $N$ registers into subsets $A$ (of size $k$) and $B$ (of size $N-k$), we have
\[
\mathcal{I}(G,N)\subseteq\mathcal{I}(G,k)_{A}\otimes\mathcal{Z}_{N-k}(G)_{B}.
\]
Thus, if condition \eq{ff_condition2} holds, then all $k$-particle reduced density matrices of $|\Gamma\rangle$ are contained in $\mathcal{I}(G,k)$. Furthermore, \eq{ff_condition1} is a statement about the single-particle reduced density matrices, so it also holds for each $k$-particle reduced density matrix. From this we see that each reduced density matrix of $|\Gamma\rangle$ is frustration free.

\subsection{Occupancy constraints}
\label{sec:Occupancy-constraints}

\lem{FF_characterization} says in particular that a frustration-free state on an $e_{1}$-gate graph has no support on states where multiple particles occupy the same diagram element. Indeed, the Lemma allows us to restrict our attention to the subspace $\mathcal{I}(G,N)$ when solving for $N$-particle frustration-free states. Here we consider restrictions to other subspaces corresponding to more general constraints on the locations of particles in a gate graph.

For any $e_{1}$-gate graph $G$ with $R$ diagram elements and any simple $R$-vertex graph $\goc$ with edge set $E(\goc)$, we define a subspace 
\begin{equation}
  \mathcal{I}(G,\goc,N)
  =\spn\{ \Sym(|\psi_{z_{1},a_{1}}^{q_{1}}\rangle
  \ldots|\psi_{z_{N},a_{N}}^{q_{N}}\rangle)\colon
  z_{i},a_{i}\in\{0,1\},\; q_{i}\neq q_{j} \text{ for } i \ne j,\; 
  \{q_{i},q_{j}\}\notin E(\goc)\}.
  \label{eq:occup_space_defn}
\end{equation}
In this subspace, no two particles occupy the same diagram element and no two particles occupy diagram elements connected by an edge in $G_{\occ}$. Note that 
\[
\mathcal{I}(G,\goc,N)\subseteq\mathcal{I}(G,N)\subseteq\mathcal{Z}_{N}(G).
\]

We say that $\goc$ specifies a set of \emph{occupancy constraints} on $G$. If $\mathcal{I}(G,\goc,N)$ is not empty, we define the restriction 
\[
  H(G,\goc,N)=H(G,N)\big|_{\mathcal{I}(G,\goc,N)}
\]
and write $\lambda_{N}^{1}(G,\goc)$ for its smallest eigenvalue. 

We now explain how the subspace $\mathcal{\mathcal{I}}(G,\goc,N)$ relates to the Bose-Hubbard model on a graph derived from $G$ and $\goc$. Specifically, the following Lemma states that we can construct another graph $G^{\square}$ (depending on $G$ and $\goc$) so that the ground state of the Bose-Hubbard model on $G^{\square}$ effectively ``simulates'' the restriction to $\mathcal{I}(G,\goc,N)$. The ground energy $\lambda_{N}^{1}(G,\goc)$ of $H(G,\goc,N)$ is zero if and only if the ground energy $\lambda_{N}^{1}(G^{\square})$ of $H(G^{\square},N)$ is zero. The Lemma also shows that the two energies are related even if they are nonzero.

\begin{restatable}[\textbf{Occupancy Constraints Lemma}]{lemma}{occup}\label{lem:oc}
Let $G$ be an $e_{1}$-gate graph specified as a gate diagram with $R\geq2$ diagram elements. Let $N\in[R]$, let $\goc$ specify a set of occupancy constraints on $G$, and suppose the subspace $\mathcal{I}(G,\goc,N)$ is nonempty. Then there exists an efficiently computable $e_{1}$-gate graph $G^{\square}$ with at most $7R^{2}$ diagram elements such that
\begin{enumerate}
\item If $\lambda_{N}^{1}(G,\goc)\leq a$ then $\lambda_{N}^{1}(G^{\square})\leq\frac{a}{R}$.
\item If $\lambda_{N}^{1}(G,\goc)\geq b$ with $b\in[0,1]$, then $\lambda_{N}^{1}(G^{\square})\geq \frac{b}{\left(13R\right)^{9}}$.
\end{enumerate}
\end{restatable}

Note that the Lemma stipulates $N\leq R$ without loss of generality, since $\mathcal{\mathcal{I}}(G,\goc,N)$ is empty otherwise.  The proof of this Lemma appears in \app{Occupancy-Constraints-Lemma}. In the proof we show how to construct the gate graph $G^\square$ from $G$ and $\goc$.

\section{Gadgets}
\label{sec:Gadgets}

In \ex{As-an-example} we saw how a single-particle ground state can encode a single-qubit computation. In this Section we see how a two-particle frustration-free state on a suitably designed $e_{1}$-gate graph can encode a two-qubit computation. We design specific $e_{1}$-gate graphs (called \emph{gadgets}) that we use in \sec{From-circuits-to} to prove that Bose-Hubbard Hamiltonian is QMA-hard. For each gate graph we discuss, we show that the smallest eigenvalue of its adjacency matrix is $e_{1}$ and we solve for all of the frustration-free states.

We first design a gate graph where, in any two-particle frustration-free state, the locations of the particles are synchronized. This ``move-together'' gadget is presented in \sec{move-together}. In \sec{Gadgets-for-two-qubit}, we design gadgets for two-qubit gates using four move-together gadgets, one for each two-qubit computational basis state. Finally, in \sec{Other-gate-graph} we describe a small modification of a two-qubit gate gadget called the ``boundary gadget.''

The circuit-to-gate graph mapping described in \sec{From-circuits-to} uses a two-qubit gate gadget for each gate in the circuit, together with boundary gadgets in parts of the graph corresponding to the beginning and end of the computation.

\subsection{The move-together gadget}
\label{sec:move-together}

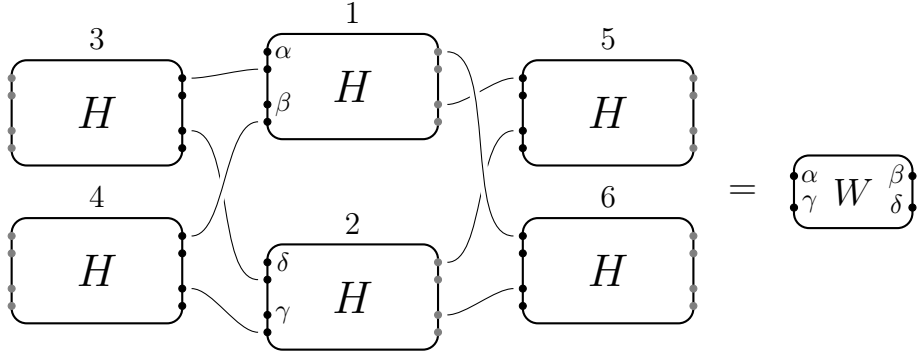
\begin{figure}
\centering \begin{tikzpicture}[scale=.7]

\foreach \startpoint/\endpoint in {{(-1.618,1.16)}/{(0,.33)},{(-1.618,4.16)}/{(0,1.33)},{(-1.618,2.16)}/{(0,4.33)}, {(-1.618,5.16)}/{(0,5.33)},{(3.24,.66)}/{(4.85,1.16)},{(3.24,1.66)}/{(4.85,4.16)},{(3.24,4.66)}/{(4.85,5.16)},{(3.24,5.66)}/{(4.85,2.16)}}
{   
\node (a) at \startpoint {}; 
 \node (b) at \endpoint {};
  \draw[looseness=.66,line width=4pt,color=white] (a) to [out=0,in=180] (b);
  \draw[looseness=.66] (a) to [out=0,in=180] (b); 
}

\foreach \xshift / \yshift /\xscale /\lab in  {0/0/1/2,  0/4/1/1,  4.85/.5/1/6,  4.85/3.5/1/5,  -1.618/.5/-1/4,  -1.618/3.5/-1/3}{   \begin{scope}[shift={(\xshift,\yshift)},xscale=\xscale]   
 \draw[rounded corners=2mm,thick] (0,0) rectangle (3.24cm, 2cm);
 \node at (1.62,2.4) {\Large \lab};   
 \foreach \y in {.66,.33,1.33,1.66}
{
    \foreach \x /\color in {0/black,3.24 /gray}
{       
		\draw[fill=\color,draw=\color] (\x cm, \y cm) circle (.66mm);   
}
}
  \node at (1.618,1) {\huge$H$}; \end{scope}
}  

\foreach \point/\name in {{(0,.66)}/\gamma,{(0,1.66)}/\delta,{(0,4.66)}/\beta,{(0,5.66)}/\alpha}
{ 
\begin{scope}[shift=\point] 
 \node at (3mm,0mm){$\name$}; 
\end{scope} 
}
\begin{scope}[shift={(10,3)},scale=1.5]
  \node at (-.65 ,0) {\LARGE $=$};
  \draw[rounded corners=2mm,thick] (0,-.46) rectangle (1.5cm,.46cm); 
    \node at (.75,0) {\LARGE $W$};
  \foreach \x/\y  in {0,1.5}
{  
\foreach \y in {-.2,.2}
{    
\draw[fill=black] (\x,\y) circle (.45mm);  
}
}      
\node at (.2,.2) {$\alpha$};  
\node at (.2,-.15) {$\gamma$}; 
 \node at (1.3,.2) {$\beta$}; 
 \node at (1.3,-.15) {$\delta$}; 
  \end{scope} 
\end{tikzpicture}

\caption{The gate diagram for the move-together gadget. \label{fig:W_gadget}}
\end{figure}

The gate diagram for the \emph{move-together gadget} is shown in \fig{W_gadget}. Using equation \eq{adj_gate_graph}, we write the adjacency matrix of the corresponding gate graph $G_W$ as 
\begin{equation}
A(G_W)=\sum_{q=1}^{6}|q\rangle\langle q|\otimes A(g_{0})+h_{\mathcal{E}}\label{eq:move_together_adj}
\end{equation}
where $h_{\mathcal{E}}$ is given by \eq{h_edges} and $\mathcal{E}$ is the set of edges in the gate diagram (in this case $h_{\mathcal{S}}=0$ as there are no self-loops).

We begin by solving for the single-particle ground states, i.e., the eigenvectors of \eq{move_together_adj} with eigenvalue $e_{1}=-1-3\sqrt{2}$. As in \ex{As-an-example}, we can solve for the states with $a=0$ and $a=1$ separately, since
\[
\langle\psi_{x,1}^{j}|h_{\mathcal{E}}|\psi_{z,0}^{i}\rangle=0
\]
for all $i,j\in\{1,\ldots,6\}$ and $x,z\in\{0,1\}$. We write a single-particle ground state as
\[
\sum_{i=1}^{6}\left(\tau_{i}|\psi_{0,a}^{i}\rangle+\nu_{i}|\psi_{1,a}^{i}\rangle\right)
\]
and solve for the coefficients $\tau_{i}$ and $\nu_{i}$ using equation \eq{h_e_0} (in this case equation \eq{h_s_0} is automatically satisfied since $h_{\mathcal{S}}=0$). Enforcing \eq{h_e_0} gives eight equations, one for each edge in the gate diagram:
\begin{align*}
  \tau_{3}&=-\tau_{1} & 
  \frac{1}{\sqrt{2}}(\tau_{1}+\nu_{1})&=-\tau_{6} \\
  \tau_{4}&=-\nu_{1} & 
  \frac{1}{\sqrt{2}}(\tau_{1}-\nu_{1})&=-\tau_{5}\\
  \nu_{3}&=-\tau_{2} & 
  \frac{1}{\sqrt{2}}(\tau_{2}+\nu_{2})&=-\nu_{5}\\
  \nu_{4}&=-\nu_{2} & 
  \frac{1}{\sqrt{2}}(\tau_{2}-\nu_{2})&=-\nu_{6}.
\end{align*}
There are four linearly independent solutions to this set of equations, given by 
\begin{align*}
  \text{\emph{Solution 1:}} && 
    \tau_{1} &= 1 & \tau_{3} &=-1 & 
    \tau_{5} &= -\frac{1}{\sqrt{2}} & \tau_{6} &= -\frac{1}{\sqrt{2}} &&
    \text{all other coefficients }0 \\
  \text{\emph{Solution 2:}} && 
    \nu_{1} &= 1 & \tau_{4} &= -1 &
    \tau_{5} &= \frac{1}{\sqrt{2}} & \tau_{6} &= -\frac{1}{\sqrt{2}} &&
    \text{all other coefficients }0 \\
  \text{\emph{Solution 3:}} && 
    \nu_{2} &= 1 & \nu_{4} &= -1 &
    \nu_{5} &= -\frac{1}{\sqrt{2}} & \nu_{6} &= \frac{1}{\sqrt{2}} &&
    \text{all other coefficients }0 \\
  \text{\emph{Solution 4:}} && 
    \tau_{2} &= 1 & \nu_{3} &= -1 &
    \nu_{5} &= -\frac{1}{\sqrt{2}} & \nu_{6} &= -\frac{1}{\sqrt{2}} &&
    \text{all other coefficients }0.
\end{align*}
For each of these solutions, and for each $a\in\{0,1\}$, we find a single-particle state with energy $e_1$. This result is summarized in the following Lemma.

\begin{lemma}
$G_{W}$ is an $e_{1}$-gate graph. A basis for the eigenspace
of $A(G_{W})$ with eigenvalue $e_1$ is 
\begin{align}
|\chi_{1,a}\rangle & =\frac{1}{\sqrt{3}}|\psi_{0,a}^{1}\rangle-\frac{1}{\sqrt{3}}|\psi_{0,a}^{3}\rangle-\frac{1}{\sqrt{6}}|\psi_{0,a}^{5}\rangle-\frac{1}{\sqrt{6}}|\psi_{0,a}^{6}\rangle\label{eq:chi_alpha}\\
|\chi_{2,a}\rangle & =\frac{1}{\sqrt{3}}|\psi_{1,a}^{1}\rangle-\frac{1}{\sqrt{3}}|\psi_{0,a}^{4}\rangle+\frac{1}{\sqrt{6}}|\psi_{0,a}^{5}\rangle-\frac{1}{\sqrt{6}}|\psi_{0,a}^{6}\rangle\label{eq:chi_beta}\\
|\chi_{3,a}\rangle & =\frac{1}{\sqrt{3}}|\psi_{1,a}^{2}\rangle-\frac{1}{\sqrt{3}}|\psi_{1,a}^{4}\rangle-\frac{1}{\sqrt{6}}|\psi_{1,a}^{5}\rangle+\frac{1}{\sqrt{6}}|\psi_{1,a}^{6}\rangle\label{eq:chi_gamma}\\
|\chi_{4,a}\rangle & =\frac{1}{\sqrt{3}}|\psi_{0,a}^{2}\rangle-\frac{1}{\sqrt{3}}|\psi_{1,a}^{3}\rangle-\frac{1}{\sqrt{6}}|\psi_{1,a}^{5}\rangle-\frac{1}{\sqrt{6}}|\psi_{1,a}^{6}\rangle\label{eq:chi_delta}
\end{align}
where $a\in\{0,1\}$. 
\end{lemma}

In \fig{W_gadget} we have used a shorthand $\alpha,\beta,\gamma,\delta$ to identify four nodes of the move-together gadget; these are the nodes with labels $(q,z,t)=(1,0,1),(1,1,1),(2,1,1),(2,0,1)$, respectively. We view $\alpha$ and $\gamma$ as ``input'' nodes and $\beta$ and $\delta$ as ``output'' nodes for the move-together gadget. It is natural to associate each single-particle state $|\chi_{i,a}\rangle$ with one of these four nodes.  We also associate the set of 8 vertices represented by the node with the corresponding node, e.g.,
\[
  S_{\alpha}=\left\{ (1,0,1,j)\colon j\in\{0,\ldots,7\}\right\} .
\]
Looking at equation \eq{chi_alpha} (and perhaps referring back to equation \eq{psi0m}) we see that $|\chi_{1,a}\rangle$ has support on vertices in $S_{\alpha}$ but not on vertices in $S_{\beta}$, $S_{\gamma}$, or $S_{\delta}$. Looking at the picture on the right-hand side of the equality sign in \fig{W_gadget}, we think of $|\chi_{1,a}\rangle$ as localized at the node $\alpha$, with no support on the other three nodes. The states $|\chi_{2,a}\rangle,|\chi_{3,a}\rangle,|\chi_{4,a}\rangle$ are similarly localized at nodes $\beta,\gamma,\delta$. We view $|\chi_{1,a}\rangle$ and $|\chi_{3,a}\rangle$ as input states and $|\chi_{2,a}\rangle$ and $|\chi_{4,a}\rangle$ as output states.

Now we turn our attention to the two-particle frustration-free states of the move-together gadget, i.e., the states $|\Phi\rangle\in\mathcal{Z}_{2}(G_{W})$ in the nullspace of $H(G_W,2)$. Using \lem{FF_characterization} we can write
\begin{equation}
|\Phi\rangle=\sum_{a,b \in \{0,1\},\,I,J \in [4]}C_{(I,a),(J,b)}|\chi_{I,a}\rangle|\chi_{J,b}\rangle\label{eq:chi_superposition}
\end{equation}
where the coefficients are symmetric, i.e.,
\begin{equation}
C_{(I,a),(J,b)}=C_{(J,b),(I,a)},\label{eq:symmetric_coefs}
\end{equation}
and where 
\begin{equation}
\langle\psi_{z,a}^{q}|\langle\psi_{x,b}^{q}|\Phi\rangle=0\label{eq:frustration_free}
\end{equation}
for all $z,a,x,b\in\{0,1\}$ and $q\in[6].$

The move-together gadget is designed so that each solution $|\Phi\rangle$ to these equations  is a superposition of a term where both particles are in input states and a term where both particles are in output states. The particles move from input nodes to output nodes together. We now solve equations \eq{chi_superposition}--\eq{frustration_free} and prove the following.

\begin{lemma}
\label{lem:Wgadget_lemma}
A basis for the nullspace of $H(G_{W},2)$ is 
\begin{equation}
|\Phi_{a,b}\rangle=\Sym\left(\frac{1}{\sqrt{2}}|\chi_{1,a}\rangle|\chi_{3,b}\rangle+\frac{1}{\sqrt{2}}|\chi_{2,a}\rangle|\chi_{4,b}\rangle\right),\quad a,b\in\{0,1\}.\label{eq:phi_a1_a2}
\end{equation}
There are no $N$-particle frustration-free states on $G_{W}$ for $N\geq3$, i.e.,
\[
\lambda_{N}^{1}(G_{W})>0\quad\text{for }N\geq3.
\]
\end{lemma}

\begin{proof}
The states $|\Phi_{a,b}\rangle$ manifestly satisfy equations \eq{chi_superposition} and \eq{symmetric_coefs}, and one can directly verify that they also satisfy \eq{frustration_free} (the nontrivial cases to check are $q=5$ and $q=6$). 

To complete the proof that \eq{phi_a1_a2} is a basis for the nullspace of $H(G_W,2)$, we verify that any state satisfying these conditions must be a linear combination of these four states. Applying equation \eq{frustration_free} gives
\begin{align*}
\langle\psi_{0,a}^{1}|\langle\psi_{0,b}^{1}|\Phi\rangle & =\frac{1}{3} C_{(1,a),(1,b)}=0 &
\langle\psi_{1,a}^{1}|\langle\psi_{1,b}^{1}|\Phi\rangle & =\frac{1}{3} C_{(2,a),(2,b)}=0\\
\langle\psi_{1,a}^{2}|\langle\psi_{1,b}^{2}|\Phi\rangle & =\frac{1}{3} C_{(3,a),(3,b)}=0 &
\langle\psi_{0,a}^{2}|\langle\psi_{0,b}^{2}|\Phi\rangle & =\frac{1}{3} C_{(4,a),(4,b)}=0\\
\langle\psi_{0,a}^{1}|\langle\psi_{1,b}^{1}|\Phi\rangle & =\frac{1}{3} C_{(1,a),(2,b)}=0 &
\langle\psi_{0,a}^{2}|\langle\psi_{1,b}^{2}|\Phi\rangle & =\frac{1}{3} C_{(4,a),(3,b)}=0\\
\langle\psi_{0,a}^{3}|\langle\psi_{1,b}^{3}|\Phi\rangle & =\frac{1}{3} C_{(1,a),(4,b)}=0 &
\langle\psi_{0,a}^{4}|\langle\psi_{1,b}^{4}|\Phi\rangle & =\frac{1}{3} C_{(2,a),(3,b)}=0
\end{align*}
for all $a,b\in \{0,1\}$. Using the fact that all of these coefficients are zero, and using equation \eq{symmetric_coefs}, we get 
\[
|\Phi\rangle=\sum_{a,b\in\{0,1\}}\left(C_{(1,a),(3,b)}\left(|\chi_{1,a}\rangle|\chi_{3,b}\rangle+|\chi_{3,b}\rangle|\chi_{1,a}\rangle\right)+C_{(2,a),(4,b)}\left(|\chi_{2,a}\rangle|\chi_{4,b}\rangle+|\chi_{4,b}\rangle|\chi_{2,a}\rangle\right)\right).
\]
Finally, applying equation \eq{frustration_free} again gives
\[
\langle\psi_{0,a}^{6}|\langle\psi_{1,b}^{6}|\Phi\rangle=\frac{1}{6}C_{(2,a),(4,b)}-\frac{1}{6}C_{(1,a),(3,b)}=0.
\]
Hence
\[
|\Phi\rangle=\sum_{a,b\in\{0,1\}}C_{(1,a),(3,b)}\left(|\chi_{1,a}\rangle|\chi_{3,b}\rangle+|\chi_{3,b}\rangle|\chi_{1,a}\rangle+|\chi_{2,a}\rangle|\chi_{4,b}\rangle+|\chi_{4,b}\rangle|\chi_{2,a}\rangle\right),
\]
which is a superposition of the states $|\Phi_{a,b}\rangle.$ 

Finally, we prove that there are no frustration-free ground states of the Bose-Hubbard model on $G_{W}$ with more than two particles. By \lem{increase_part_number}, it suffices to prove that there are no frustration-free three-particle states.

Suppose (for a contradiction) that $|\Gamma\rangle\in\mathcal{Z}_{3}(G_{W})$ is a normalized three-particle frustration-free state. Write 
\[
|\Gamma\rangle=\sum D_{(i,a),(j,b),(k,c)}|\chi_{i,a}\rangle|\chi_{j,b}\rangle|\chi_{k,c}\rangle.
\]
Note that each reduced density matrix of $|\Gamma\rangle$ on two of the three subsystems must have all of its support on two-particle frustration-free states (see the remark following \lem{FF_characterization}), i.e., on the states $|\Phi_{a,b}\rangle$. Using this fact for the subsystem consisting of the first two particles, we see in particular that
\begin{equation}
(i,j)\notin\{(1,3),(3,1),(2,4),(4,2)\}\quad\Longrightarrow\quad D_{(i,a),(j,b),(k,c)}=0\label{eq:ij_constraint1}
\end{equation}
(since $|\Phi_{a_1,a_2}\rangle$ only has support on vectors $|\chi_{i,a}\rangle|\chi_{j,b}\rangle$ with $i,j\in \{(1,3),(3,1),(2,4),(4,2)\}$).

Using this fact for subsystems consisting of particles $2,3$ and $1,3$, respectively, gives 
\begin{align}
(j,k)\notin\{(1,3),(3,1),(2,4),(4,2)\}\quad\Longrightarrow\quad D_{(i,a),(j,b),(k,c)} & =0\label{eq:ij_constraint2}\\
(i,k)\notin\{(1,3),(3,1),(2,4),(4,2)\}\quad\Longrightarrow\quad D_{(i,a),(j,b),(k,c)} & =0.\label{eq:ij_constraint3}
\end{align}
Putting together equations \eq{ij_constraint1}, \eq{ij_constraint2}, and \eq{ij_constraint3}, we see that $|\Gamma\rangle=0$. This is a contradiction, so no three-particle frustration-free states exist.
\end{proof}

Next we show how the move-together gadget can be used to build gadgets that implement two-qubit gates.

\subsection{Gadgets for two-qubit gates}
\label{sec:Gadgets-for-two-qubit}

In this Section we define a gate graph for each of the two-qubit unitaries
\[
  \{\CNOT_{12}, \CNOT_{21}, \CNOT_{12}\left(H\otimes\id\right),
    \CNOT_{12}\left(HT\otimes\id\right)\}.
\]
Here $\CNOT_{12}$ is the standard controlled-not gate with the second qubit as a target, whereas $\CNOT_{21}$ has the first qubit as target.

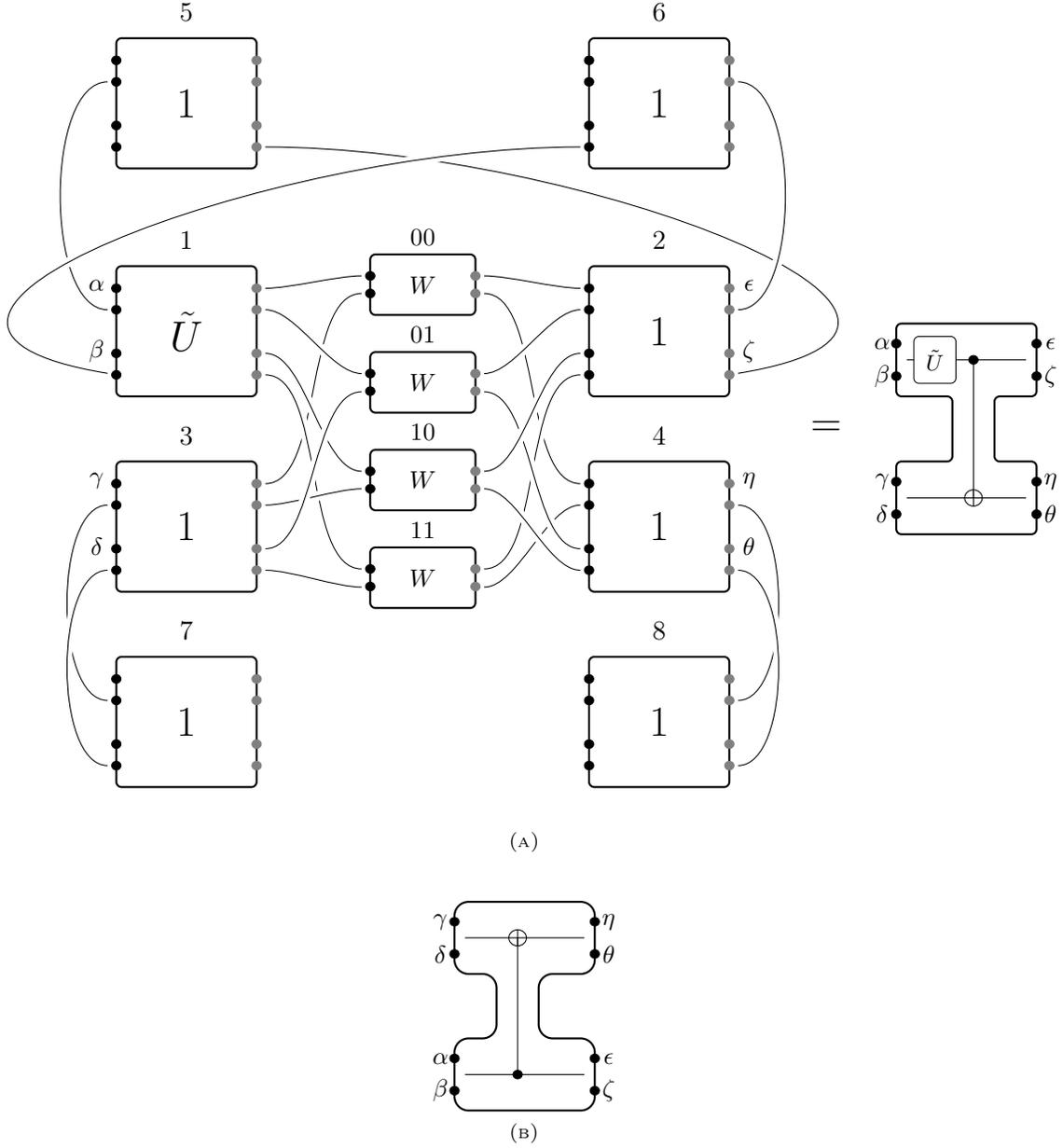
\begin{figure}
\centering 
\subfloat[][]{
\begin{tikzpicture}[yscale=.929] \label{fig:GVucnot}

\path[use as bounding box](-5.5,-3) rectangle (10,10.5);
     
\foreach \xshift / \yshift /\xscale / \lab / \unitary in{3.12/0.5/1/4/1,3.12/3.5/1/2/1,-3.62/0.5/1/3/1,-3.62/3.5/1/1/\tilde U, -3.62/-2.5/1/7/1,3.12/-2.5/1/8/1, 3.12/7/1/6/1, -3.62/7/1/5/1}
{ 
\begin{scope}[shift={(\xshift,\yshift)},xscale=\xscale]
  \draw[rounded corners=0.75mm,thick] (0,0) rectangle (2 cm, 2cm);
  \node at (1,1) {\huge$\unitary$};   
\node at (1,2.4) {\large\lab};
  \foreach \y in {.33,.66,1.33,1.66}
	{   
		\foreach \x /\color in {0/black,2/gray}
			{    
				\draw[fill=\color,draw=\color] (\x cm, \y cm) circle (.66mm);  
			}
	} 
\end{scope}
}
  
\foreach \yshift/\from in {0.25/11,1.75/10,3.25/01,4.75/00}
{ 
\begin{scope}[yshift=\yshift cm] 
 \draw[rounded corners=0.75mm,thick] (0,0) rectangle (1.5cm,.93cm); 
 \node at (.75,.465) {$W$};   
\node at (.75,1.2) {$\from$};
  \foreach \x/\color in {0/black,1.5/gray}
	{   
		\foreach \y in {0.33,.6}
		{   
			 \draw[fill=\color,draw=\color] (\x,\y) circle (.66mm); 
		 }
	}
\end{scope}
}
  
\foreach \l/\r in {   5.35/5.16,5.08/2.16,   .85/3.83,.57/1.83,   3.85/4.83,3.57/1.16,   2.35/4.16,2.08/.83} 
{   
\node (a) at (1.5,\l) {}; 
 \node (b) at (3.12,\r)   {}; 
 \draw[looseness=.66,line width=4pt,color=white] (a) to [out=0,in=180] (b);
  \draw[looseness=.66] (a) to [out=0,in=180] (b); 
}
  
\foreach \l/\r in {5.16/5.35,2.16/5.08,   3.83/.85,.83/.57,   4.83/3.85,1.83/2.08,   4.16/2.35,1.16/3.57}
{   
\node (a) at (-1.618,\l) {}; 
 \node (b) at (0,\r)   {};
  \draw[looseness=.66,line width=4pt,color=white] (a) to [out=0,in=180] (b);
  \draw[looseness=.66] (a) to [out=0,in=180] (b); 
}

\node (c) at (-3.62,1.83){};
\node (d) at (-3.62,-1.17){};

 \draw[looseness=.66,line width=4pt,color=white] (c) to [out=180,in=180] (d);
 \draw[looseness=.66] (c) to [out=180,in=180] (d); 

\node (c) at (-3.62,0.83){};
\node (d) at (-3.62,-2.17){};

 \draw[looseness=.66,line width=4pt,color=white] (c) to [out=180,in=180] (d);
 \draw[looseness=.66] (c) to [out=180,in=180] (d);

\node (c) at (5.12,1.83){};
\node (d) at (5.12,-1.17){};

 \draw[looseness=.66,line width=4pt,color=white] (c) to [out=0,in=0] (d);
 \draw[looseness=.66] (c) to [out=0,in=0] (d); 

\node (c) at (5.12,0.83){};
\node (d) at (5.12,-2.17){};

 \draw[looseness=.66,line width=4pt,color=white] (c) to [out=0,in=0] (d);
 \draw[looseness=.66] (c) to [out=0,in=0] (d);

\node (c) at (-3.62,8.33){};
\node (d) at (-3.62,4.83){};

 \draw[looseness=.66,line width=4pt,color=white] (c) to [out=180,in=180] (d);
 \draw[looseness=.66] (c) to [out=180,in=180] (d); 

\node (c) at (-1.62,7.33){};
\node (d) at (5.12,3.83){};

 \draw[looseness=1.5,line width=4pt,color=white] (c) to [out=0,in=10] (d);
 \draw[looseness=1.5] (c) to [out=0,in=10] (d);

\node (c) at (5.12,8.33){};
\node (d) at (5.12,4.83){};

 \draw[looseness=.66,line width=4pt,color=white] (c) to [out=0,in=0] (d);
 \draw[looseness=.66] (c) to [out=0,in=0] (d); 

\node (c) at (3.12,7.33){};
\node (d) at (-3.62,3.83){};

 \draw[looseness=1.5,line width=4pt,color=white] (c) to [out=180,in=170] (d);
 \draw[looseness=1.5] (c) to [out=180,in=170] (d);

 \node at (-3.9,5.2) {$\alpha$};
 \node at (-3.9,4.2) {$\beta$};
\node at (-3.9,2.2) {$\gamma$};
\node at (-3.9,1.2) {$\delta$};

 \node at (5.4,5.2) {$\epsilon$};
 \node at (5.4,4.2) {$\zeta$};
\node at (5.4,2.2) {$\eta$};
\node at (5.4,1.2) {$\theta$};

\setcounter{mycount}{`a} \begin{scope}[xshift = 7.5cm,yshift = 3cm]
  \node at (-1,0) {\huge $=$}; 
 \begin{scope}[xscale=-1,xshift=-2cm]   
\draw[rounded corners = .75mm,thick] (0cm, -1.618 cm) -- (2 cm, -1.618 cm) -- (2cm, -.5 cm) -- (1.2cm, -.5 cm) -- (1.2 cm,.5cm) -- (2cm, .5cm) -- (2cm, 1.618cm) -- (0cm, 1.618cm) -- (0cm, .5cm) -- (.6cm, .5cm) -- (.6cm, -.5cm) -- (0cm, -.5cm) -- cycle; 
    \draw (.15cm, 1.06cm) -- (1.85cm, 1.06cm); 
 \draw (.15cm, -1.06cm) -- (1.85cm, -1.06cm);
  \draw (.9cm,1.06cm) -- (.9cm,-1.06cm);   
\draw[fill=black] (.9,1.06cm) circle (.66mm); 
 \draw (.9,-1.06cm) circle (1.25mm); 
 \draw (.9,-1.19cm) -- (.9,-.93cm);   
  \draw[rounded corners=.75mm,fill=white] (1.15cm, .7cm) rectangle (1.75cm, 1.42cm); 
 \node at (1.45cm, 1.06cm) {\small $\tilde U$}; 
 \end{scope}

\node at ( -0.2,1.31) {$\alpha$};
 \draw[fill=black] (0,1.31) circle (.66mm);   

\node at ( -0.2,0.81) {$\beta$};
 \draw[fill=black] (0,0.81) circle (.66mm);   

\node at ( -0.2,-0.81) {$\gamma$};
 \draw[fill=black] (0,-0.81) circle (.66mm);   

\node at ( -0.2,-1.31) {$\delta$};
 \draw[fill=black] (0,-1.31) circle (.66mm);   

\node at (2.2,1.31) {$\epsilon$};
 \draw[fill=black] (2,1.31) circle (.66mm);   

\node at (2.2,0.81) {$\zeta$};
 \draw[fill=black] (2,0.81) circle (.66mm);   

\node at ( 2.2,-0.81) {$\eta$};
 \draw[fill=black] (2,-0.81) circle (.66mm);   

\node at ( 2.2,-1.31) {$\theta$};
 \draw[fill=black] (2,-1.31) circle (.66mm);

\end{scope} 
\end{tikzpicture}}
\vspace{0.63cm}
\subfloat[][]{\begin{tikzpicture}[yscale=.92] \label{fig:GVcnot}

  \setcounter{mycount}{`a} \begin{scope}[xshift = 7.2 cm,yshift = 3cm]   
     
\draw[rounded corners = 2mm,thick] (0cm, -1.618 cm) -- (2 cm, -1.618 cm) -- (2cm, -.5 cm) -- (1.2cm, -.5 cm) -- (1.2 cm,.5cm) -- (2cm, .5cm) -- (2cm, 1.618cm) -- (0cm, 1.618cm) -- (0cm, .5cm) -- (.6cm, .5cm) -- (.6cm, -.5cm) -- (0cm, -.5cm) -- cycle;      \draw (.15cm, 1.06cm) -- (1.85cm, 1.06cm);   \draw (.15cm, -1.06cm) -- (1.85cm, -1.06cm);   \draw (.9cm,1.06cm) -- (.9cm,-1.06cm);
 \draw[fill=black] (.9,-1.06cm) circle (.66mm);   \draw (.9,1.06cm) circle (1.25mm);   \draw (.9,1.19cm) -- (.9,.93cm);

\node at ( -0.2,-0.81) {$\alpha$};
 \draw[fill=black] (0,-0.81) circle (.66mm);   

\node at ( -0.2,-1.31) {$\beta$};
 \draw[fill=black] (0,-1.31) circle (.66mm);   

\node at ( -0.2,1.31) {$\gamma$};
 \draw[fill=black] (0,1.31) circle (.66mm);   

\node at ( -0.2,0.81) {$\delta$};
 \draw[fill=black] (0,0.81) circle (.66mm);   

\node at ( 2.2,-0.81) {$\epsilon$};
 \draw[fill=black] (2,-0.81) circle (.66mm);   

\node at ( 2.2,-1.31) {$\zeta$};
 \draw[fill=black] (2,-1.31) circle (.66mm);   

\node at (2.2,1.31) {$\eta$};
 \draw[fill=black] (2,1.31) circle (.66mm);   

\node at (2.2,0.81) {$\theta$};
 \draw[fill=black] (2,0.81) circle (.66mm);   

\end{scope}
\end{tikzpicture}}

\caption{\subfig{GVucnot} Gadget for the
 two-qubit unitary $U=(\tilde U\otimes\id)\CNOT_{12}$ with $\tilde U\in \{1,H,HT\}$.
\subfig{GVcnot} For the $U=\CNOT_{21}$ gate (first qubit is the target), we use the same gate graph as in \subfig{GVucnot} with $\tilde U=1$; we represent it schematically as shown.}
\end{figure}

We define the gate graphs by exhibiting their gate diagrams. For the three cases
\[
  U=\CNOT_{12}(\tilde U\otimes\id)
\]
with $\tilde U\in\{\id,H,HT\}$, we associate $U$ with the gate diagram shown in \fig{GVucnot}. In the Figure we also indicate a shorthand used to represent this gate diagram. As one might expect, for the case $U=\CNOT_{21}$, we use the same gate diagram as for $U=\CNOT_{12}$; however, we use the slightly different shorthand shown in \fig{GVcnot}.

Roughly speaking, the two-qubit gate gadgets work as follows. In \fig{GVucnot} there are four move-together gadgets, one for each two-qubit basis state $|00\rangle, |01\rangle, |10\rangle, |11\rangle$. These enforce the constraint that two particles must move through the graph together. The connections between the four diagram elements labeled $1,2,3,4$ and the move-together gadgets ensure that certain frustration-free two-particle states encode two-qubit computations, while the connections between diagram elements $1,2,3,4$ and $5,6,7,8$ ensure that there are no additional frustration-free two-particle states (i.e., states that do not encode computations).

To describe the frustration-free states of the gate graph depicted in \fig{GVucnot}, first recall the definition of the states $|\chi_{1,a}\rangle, |\chi_{2,a}\rangle, |\chi_{3,a}\rangle, |\chi_{4,a}\rangle$ from equations \eq{chi_alpha}--\eq{chi_delta}. For each of the move-together gadgets $xy\in\{00,01,10,11\}$ in \fig{GVucnot}, write 
\[
|\chi_{L,a}^{xy}\rangle
\]
for the state $|\chi_{L,a}\rangle$ with support (only) on the gadget labeled $xy$. Write 
\[
U(a)=\begin{cases}
U & \text{ if }a=0\\
U^{*} & \text{ if }a=1
\end{cases}
\]
and similarly for $\tilde U$ (we use this notation throughout the paper to indicate a unitary or its elementwise complex conjugate).

In \app{graph_gadgets}
we prove the following Lemma, which shows that $G_{U}$ is an $e_{1}$-gate
graph and characterizes its frustration-free states.

\begin{restatable}{lemma}{Twoqub}\label{lem:2qub_gate}
Let $U=\CNOT_{12}(\tilde U\otimes\id)$ where $\tilde U\in\{\id,H,HT\}$. The corresponding gate graph $G_U$ is defined by its gate diagram shown in \fig{GVucnot}. The adjacency matrix $A(G_U)$ has ground energy $e_{1}$; a basis for the corresponding eigenspace is
\begin{align}
|\rho_{z,a}^{1,U}\rangle & =\frac{1}{\sqrt{8}}|\psi_{z,a}^{1}\rangle-\frac{1}{\sqrt{8}}|\psi_{z,a}^{5+z}\rangle-\sqrt{\frac{3}{8}}\sum_{x,y=0}^{1}\tilde U(a)_{yz}|\chi_{1,a}^{yx}\rangle
 & |\rho_{z,a}^{2,U}\rangle
 & =\frac{1}{\sqrt{8}}|\psi_{z,a}^{2}\rangle-\frac{1}{\sqrt{8}}|\psi_{z,a}^{6-z}\rangle-\sqrt{\frac{3}{8}}\sum_{x=0}^{1}|\chi_{2,a}^{zx}\rangle\label{eq:rho1_1}\\
|\rho_{z,a}^{3,U}\rangle & =\frac{1}{\sqrt{8}}|\psi_{z,a}^{3}\rangle-\frac{1}{\sqrt{8}}|\psi_{z,a}^{7}\rangle-\sqrt{\frac{3}{8}}\sum_{x=0}^{1}|\chi_{3,a}^{xz}\rangle 
& |\rho_{z,a}^{4,U}\rangle & =\frac{1}{\sqrt{8}}|\psi_{z,a}^{4}\rangle-\frac{1}{\sqrt{8}}|\psi_{z,a}^{8}\rangle-\sqrt{\frac{3}{8}}\sum_{x=0}^{1}|\chi_{4,a}^{x\left(z\oplus x\right)}\rangle\label{eq:rho2_1}
\end{align}
where $z,a\in\{0,1\}$. A basis for the nullspace of $H(G_U,2)$ is
\begin{equation}
\Sym(|T_{z_{1},a,z_{2},b}^U\rangle),\quad z_{1},z_{2},a,b\in\{0,1\}\label{eq:twopartstate_1}
\end{equation}
where 
\begin{equation}
|T_{z_{1},a,z_{2},b}^U\rangle=\frac{1}{\sqrt{2}}|\rho_{z_{1},a}^{1,U}\rangle|\rho_{z_{2},b}^{3,U}\rangle+\frac{1}{\sqrt{2}}\sum_{x_1,x_2=0}^{1}U(a)_{x_{1}x_{2},z_{1}z_{2}}|\rho_{x_{1},a}^{2,U}\rangle|\rho_{x_{2},b}^{4,U}\rangle\label{eq:twopartstate_2}
\end{equation}
for $z_{1},z_{2},a,b\in\{0,1\}$. There are no $N$-particle frustration-free
states on $G_U$ for $N\geq3$, i.e., 
\[
\lambda_{N}^{1}(G_U)>0\quad\text{for }N\geq3.
\]
\end{restatable}

We view the nodes labeled $\alpha,\beta,\gamma,\delta$ in \fig{GVucnot} as ``input'' nodes and those labeled $\epsilon, \zeta,\eta,\theta$ as ``output nodes''. Each of the states $|\rho_{x,y}^{i,U}\rangle$ is associated with one of the nodes, depending on the values of $i\in\{1,2,3,4\}$ and $x\in\{0,1\}$. For example, the states $|\rho_{0,0}^{1,U}\rangle$ and $|\rho_{0,1}^{1,U}\rangle$ are associated with input node $\alpha$ since they both have nonzero amplitude on vertices of the gate graph that are associated with $\alpha$ (and zero amplitude on vertices associated with other labeled nodes).

The two-particle state $\Sym(|T_{z_{1},a,z_{2},b}^{U}\rangle)$ is a superposition of a term
\[
\Sym\bigg(\frac{1}{\sqrt{2}}|\rho_{z_{1},a}^{1,U}\rangle|\rho_{z_{2},b}^{3,U}\rangle\bigg)
\]
with both particles located on vertices corresponding to input nodes and a term 
\[
\Sym\Bigg(\frac{1}{\sqrt{2}}\sum_{x_{1},x_{2}\in\{0,1\}}U(a)_{x_{1}x_{2},z_{1}z_{2}}|\rho_{x_{1},a}^{2,U}\rangle|\rho_{x_{2},b}^{4,U}\rangle\Bigg)
\]
with both particles on vertices corresponding to output nodes. The two-qubit gate $U(a)$ is applied as the particles move from input nodes to output nodes.

\subsection{The boundary gadget}
\label{sec:Other-gate-graph}

\begin{figure}
\centering
\begin{tikzpicture}[yscale=0.92] 

\path[use as bounding box](-5.5,-3) rectangle (10,10.5);

\foreach \xshift / \yshift /\xscale / \lab / \unitary in{3.12/0.5/1/4/1,3.12/3.5/1/2/1,-3.62/0.5/1/3/1,-3.62/3.5/1/1/1, -3.62/-2.5/1/7/1,3.12/-2.5/1/8/1, 3.12/7/1/6/1, -3.62/7/1/5/1}
{ 
\begin{scope}[shift={(\xshift,\yshift)},xscale=\xscale]
  \draw[rounded corners=0.75mm,thick] (0,0) rectangle (2 cm, 2cm);
  \node at (1,1) {\huge$\unitary$};   
\node at (1,2.4) {\large\lab};
  \foreach \y in {.33,.66,1.33,1.66}
	{   
		\foreach \x /\color in {0/black,2/gray}
			{    
				\draw[fill=\color,draw=\color] (\x cm, \y cm) circle (.66mm);  
			}
	} 
\end{scope}
}

\foreach \yshift/\from in {0.25/11,1.75/10,3.25/01,4.75/00}
{ 
\begin{scope}[yshift=\yshift cm] 
 \draw[rounded corners=0.75mm,thick] (0,0) rectangle (1.5cm,.93cm); 
 \node at (.75,.465) {$W$};   
\node at (.75,1.2) {$\from$};
  \foreach \x/\color in {0/black,1.5/gray}
	{   
		\foreach \y in {0.33,.6}
		{   
			 \draw[fill=\color,draw=\color] (\x,\y) circle (.66mm); 
		 }
	}
\end{scope}
}

\foreach \l/\r in {   5.35/5.16,5.08/2.16,   .85/3.83,.57/1.83,   3.85/4.83,3.57/1.16,   2.35/4.16,2.08/.83} 
{   
\node (a) at (1.5,\l) {}; 
 \node (b) at (3.12,\r)   {}; 
 \draw[looseness=.66,line width=4pt,color=white] (a) to [out=0,in=180] (b);
  \draw[looseness=.66] (a) to [out=0,in=180] (b); 
}

\foreach \l/\r in {5.16/5.35,2.16/5.08,   3.83/.85,.83/.57,   4.83/3.85,1.83/2.08,   4.16/2.35,1.16/3.57}
{   
\node (a) at (-1.618,\l) {}; 
 \node (b) at (0,\r)   {};
  \draw[looseness=.66,line width=4pt,color=white] (a) to [out=0,in=180] (b);
  \draw[looseness=.66] (a) to [out=0,in=180] (b); 
}

\node (c) at (-3.62,1.83){};
\node (d) at (-3.62,-1.17){};

 \draw[looseness=.66,line width=4pt,color=white] (c) to [out=180,in=180] (d);
 \draw[looseness=.66] (c) to [out=180,in=180] (d); 

\node (c) at (-3.62,0.83){};
\node (d) at (-3.62,-2.17){};

 \draw[looseness=.66,line width=4pt,color=white] (c) to [out=180,in=180] (d);
 \draw[looseness=.66] (c) to [out=180,in=180] (d);

\node (c) at (5.12,1.83){};
\node (d) at (5.12,-1.17){};

 \draw[looseness=.66,line width=4pt,color=white] (c) to [out=0,in=0] (d);
 \draw[looseness=.66] (c) to [out=0,in=0] (d); 

\node (c) at (5.12,0.83){};
\node (d) at (5.12,-2.17){};

 \draw[looseness=.66,line width=4pt,color=white] (c) to [out=0,in=0] (d);
 \draw[looseness=.66] (c) to [out=0,in=0] (d);

\node (c) at (-3.62,8.33){};
\node (d) at (-3.62,4.83){};

 \draw[looseness=.66,line width=4pt,color=white] (c) to [out=180,in=180] (d);
 \draw[looseness=.66] (c) to [out=180,in=180] (d); 

\node (c) at (-1.62,7.33){};
\node (d) at (5.12,3.83){};

 \draw[looseness=1.5,line width=4pt,color=white] (c) to [out=0,in=10] (d);
 \draw[looseness=1.5] (c) to [out=0,in=10] (d);

\node (c) at (5.12,8.33){};
\node (d) at (5.12,4.83){};

 \draw[looseness=.66,line width=4pt,color=white] (c) to [out=0,in=0] (d);
 \draw[looseness=.66] (c) to [out=0,in=0] (d); 

\node (c) at (3.12,7.33){};
\node (d) at (-3.62,3.83){};

 \draw[looseness=1.5,line width=4pt,color=white] (c) to [out=180,in=170] (d);
 \draw[looseness=1.5] (c) to [out=180,in=170] (d);

\draw[looseness=150] (-3.7,5.19) to [out=150,in=210] (-3.7,5.18) ;   
\draw[looseness=150] (-3.7,4.19) to [out=150,in=210] (-3.7,4.18) ;   

\draw[looseness=150] (-3.7,2.19) to [out=150,in=210] (-3.7,2.18) ;   
\draw[looseness=150] (-3.7,1.19) to [out=150,in=210] (-3.7,1.18) ;   

\draw[looseness=150] (5.2,5.19) to [out=30,in=-30] (5.2,5.18) ;   
\draw[looseness=150] (5.2,4.19) to [out=30,in=-30] (5.2,4.18) ;   

\node at (5.4,2.2) {$\alpha$};
\node at (5.4,1.2) {$\beta$};
\node at (5.4,-0.8) {$\gamma$};
\node at (5.4,-1.8) {$\delta$};

  \node at (6.5cm,3cm) {\huge $=$};
\begin{scope}[xshift = 7.5cm,yshift = 3.81cm]       
\draw[rounded corners = 2mm,thick] (0,0) -- (.6,0) -- (.6,-.5) -- (1.4,-.5) -- (1.4,-1.618) -- (0,-1.618) -- cycle;                    \draw[fill=black] (1,-.5) circle (.66mm);             
\draw[fill=black] (1.4,-.81) circle (.66mm);       
\draw[fill=black] (1.4,-1.41) circle (.66mm);       
\draw[fill=black] (1,-1.618) circle (.66mm);              
\node at (1,-.25) {$\alpha$};       
\node at (1.6,-.76) {$\gamma$};       
\node at (1.6,-1.46) {$\beta$};       
\node at (1,-1.868) {$\delta$};              
\node at (.7,-1.06) {\Large Bnd};          
\end{scope}
\end{tikzpicture}
\caption{The gate diagram for the boundary gadget is obtained from \fig{GVucnot} by setting $\tilde U=1$ and adding 6 self-loops.}\label{fig:GVbdy}
\end{figure}
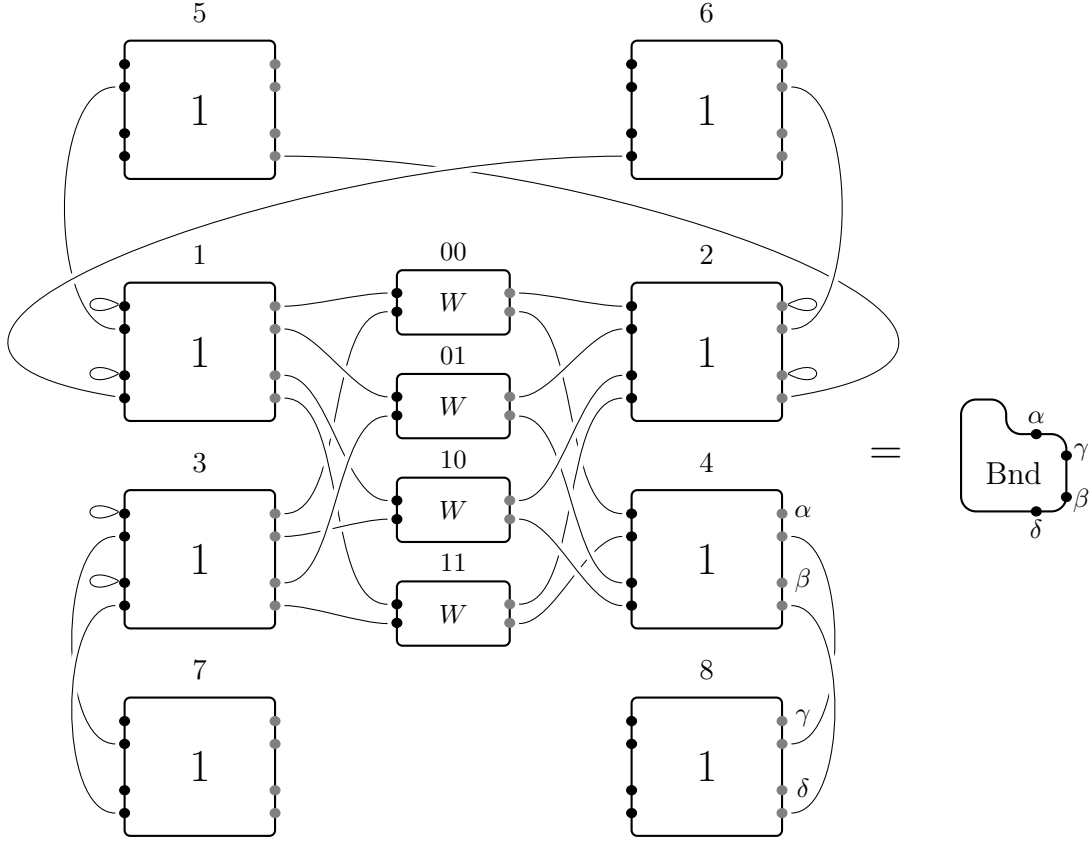

The \emph{boundary gadget} is shown in \fig{GVbdy}. This gate diagram is obtained from \fig{GVucnot} (with $\tilde U=\id$) by adding self-loops. The adjacency matrix is
\[
  A(G_{\bnd})=A(G_{\CNOT_{12}})+h_{\mathcal{S}}
\]
where 
\[
  h_{\mathcal{S}}
  =\sum_{z=0}^{1}(
    |1,z,1\rangle\langle1,z,1|\otimes\id_{j}
   +|2,z,5\rangle\langle2,z,5|\otimes\id_{j}
   +|3,z,1\rangle\langle3,z,1|\otimes\id_{j}
  ).
\]
The single-particle ground states (with energy $e_{1}$) are superpositions of the states $|\rho_{z,a}^{i,U}\rangle$ from \lem{2qub_gate} that are in the nullspace of $h_{\mathcal{S}}$. Note that 
\[
\langle\rho_{x,b}^{j,U}|h_{\mathcal{S}}|\rho_{z,a}^{i,U}\rangle=\delta_{a,b}\delta_{x,z}\left(\delta_{i,1}\delta_{j,1}+\delta_{i,2}\delta_{j,2}+\delta_{i,3}\delta_{j,3}\right)\frac{1}{8}\cdot\frac{1}{8}
\]
(one factor of $\frac{1}{8}$ comes from the normalization in equations \eq{rho1_1}--\eq{rho2_1} and the other factor comes from the normalization in equation \eq{psi0m}), so the only single-particle ground states are 
\[
|\rho_{z,a}^{\bnd}\rangle = |\rho_{z,a}^{4,U}\rangle
\]
with $z,a\in\{0,1\}$. Thus there are no two- (or more) particle frustration-free states, because no superposition of the states \eq{twopartstate_1} lies in the subspace 
\[
\spn\{ \Sym(|\rho_{z,a}^{4,U}\rangle|\rho_{x,b}^{4,U}\rangle)\colon z,a,x,b\in\{0,1\}\} 
\]
of states with single-particle reduced density matrices in the ground space of $A(G_{\bnd})$.  We summarize these results as follows.

\begin{lemma}\label{lem:boundary_lemma}
The smallest eigenvalue of $A(G_{\bnd})$ is $e_{1}$, with corresponding eigenvectors 
\begin{equation}
|\rho_{z,a}^{\bnd}\rangle=\frac{1}{\sqrt{8}}|\psi_{z,a}^{4}\rangle-\frac{1}{\sqrt{8}}|\psi_{z,a}^{8}\rangle-\sqrt{\frac{3}{8}}\sum_{x=0,1}|\chi_{4,a}^{x\left(z\oplus x\right)}\rangle\label{eq:rho_bnd}
\end{equation}
for $z,a \in \{0,1\}$. There are no frustration-free states with two or more particles, i.e., $\lambda_{N}^{1}(G_{\bnd})>0$ for $N\geq2$.
\end{lemma}

\section{Bose-Hubbard Hamiltonian is QMA-hard}
\label{sec:From-circuits-to}

Since Frustration-Free Bose-Hubbard Hamiltonian is a special case of Bose-Hubbard Hamiltonian, to prove the latter is QMA-hard it suffices to prove the former is QMA-hard. To achieve this, we efficiently map any sufficiently large instance of any problem in QMA to an equivalent instance of Frustration-Free Bose-Hubbard Hamiltonian.

For any instance $X$ of a problem in QMA there is a verification circuit $\mathcal{C}_{X}$. We require the verification circuit to be of a certain form described in \sec{Universality}. Since the class of circuits we consider is universal, this choice is without loss of generality. We also assume without loss of generality \cite{KSV02,MW05} that $\mathcal{C}_{X}$ satisfies a stronger version of \defn{QMA} where the completeness threshold $\frac{2}{3}$ is replaced by $1-\frac{1}{2^{|X|}}$, i.e.,
\begin{itemize}
\item If $X\in L_{\mathrm{yes}}$ then there exists a state $|\psi_{\wit}\rangle$ with $\AP(\mathcal{C}_{X},|\psi_{\wit}\rangle)\geq1-\frac{1}{2^{|X|}}$.
\item If $X\in L_{\mathrm{no}}$ then $\AP(\mathcal{C}_{X},|\phi\rangle)\leq\frac{1}{3}$ for all $|\phi\rangle$.
\end{itemize}
In this Section we exhibit an efficiently computable mapping from the $n$-qubit, $M$-gate verification circuit $\mathcal{C}_{X}$ to an $e_{1}$-gate graph $\gx$ with $R=32(M+2n-2)$ diagram elements and an occupancy constraints graph $\gxoc$ on $R$ vertices. We consider the Hamiltonian $H(\gx,\gxoc,n)$ and its smallest eigenvalue $\lambda_n^1(\gx,\gxoc)$. In \sec{Proof-of-Theorem} we prove the following.

\begin{theorem}
\label{thm:main_thm_with_occ_constraints}
If there exists a state $|\psi_{\wit}\rangle$ with $\AP(\mathcal{C}_{X},|\psi_{\wit}\rangle)\geq1-\frac{1}{2^{|X|}}$, then
\begin{equation}
\lambda_{n}^{1}(\gx,\gxoc)\leq\frac{1}{2^{|X|}} \label{eq:cond1}
\end{equation}
On the other hand, if $\AP(\mathcal{C}_{X},|\phi\rangle)\leq\frac{1}{3}$ for all $|\phi\rangle$, then
\begin{equation}
\lambda_{n}^{1}(\gx,\gxoc)\geq\frac{\mathcal{K}}{n^{4}M^{4}}\label{eq:cond2}
\end{equation}
where $\mathcal{K}\in (0,1]$ is an absolute constant.
\end{theorem}
The number $\mathcal{K}$ appearing in this Theorem can in principle be computed (see \sec{Proof-of-Theorem}) but we will not need to know its value.

Note that if $X$ is a yes instance then \eq{cond1} is satisfied and if $X$ is a no instance then \eq{cond2} is satisfied. By applying the \ocl (\lem{oc}), with gate graph $G_X$ and occupancy constraints graph $\gxoc$, and using the fact that $R=32\left(M+2n-2\right)\leq 64(M+n)$, we get an efficiently computable $e_{1}$-gate graph $G_{X}^{\square}$ with at most $7 \cdot 64^2 (n+M)^2$ diagram elements such that the following holds.

\begin{theorem}
\label{thm:QMA-hardness}
If $X$ is a yes instance, then
\[
\lambda_{n}^{1}(G_{X}^{\square})\leq\frac{1}{2^{|X|}}.
\]
On the other hand, if X is a no instance, then
\[
\lambda_{n}^{1}(G_{X}^{\square})\geq\frac{\mathcal{K}}{n^{4}M^{4}832^9\left(M+n\right)^{9}}
\] 
where $\mathcal{K}\in(0,1]$ is the absolute constant from \thm{main_thm_with_occ_constraints}.
\end{theorem}

This Theorem is sufficient to prove that Frustration-Free Bose-Hubbard Hamiltonian is QMA-hard. To see this, first let $Q=\lceil{\frac{1}{\mathcal{K}}}\rceil$ and define the precision parameter
\[
\epsilon=\frac{1}{2Q}\frac{1}{n^{4}M^{4}832^9\left(M+n\right)^{9}}.
\]
Note that $\frac{1}{\epsilon}$ is at least four times the number of vertices in the graph (i.e., at least $4 \cdot 128 \cdot 7 \cdot 64^2 (n+M)^2$, since each diagram element corresponds to $128$ vertices); this condition  is required in the definition of Frustration-Free Bose-Hubbard Hamiltonian. We now show that solving the instance of Frustration-Free Bose-Hubbard Hamiltonian with graph $G_{X}^{\square}$, number of particles $n$, and precision parameter $\epsilon$ is sufficient to solve the original instance $X$ (for $|X|$ at least some constant).  

If $X$ is a no instance then the Theorem states that $\lambda_{n}^{1}(G_{X}^{\square})\geq 2\epsilon\geq \epsilon+\epsilon^3$, so the corresponding instance of Frustration-Free Bose-Hubbard Hamiltonian is also a no instance.  Furthermore, since $M$ and $n$ are upper bounded by a polynomial function of the instance size $|X|$, there exists a constant instance size above which 
\[
\frac{1}{2^{|X|}}\leq \epsilon^3,
\]
so when $X$ is a sufficiently large yes instance, the Theorem says $\lambda_{n}^{1}(G_{X}^{\square})\leq \epsilon^3$, i.e., the corresponding instance of Frustration-Free Bose-Hubbard Hamiltonian is also a yes instance. This proves that Frustration-Free Bose-Hubbard Hamiltonian is QMA-hard, which implies that Bose-Hubbard Hamiltonian is QMA-hard.

The remainder of this section describes the ingredients used to prove \thm{main_thm_with_occ_constraints}. In \sec{Universality} we describe the verification circuit $\mathcal{C}_{X}$, in \sec{The-gate-graph} we define the gate graph $\gx$, and in \sec{The-occupancy-constraints} we define the occupancy constraints graph $\gxoc$. We prove \thm{main_thm_with_occ_constraints} in \sec{Proof-of-Theorem}.

\subsection{The verification circuit}
\label{sec:Universality}

We take the verification circuit $\mathcal{C}_{X}$ to be from the following universal circuit family. Write $n$ for the number of qubits and $M$ for the number of gates in the circuit
\[
U_{\mathcal{C}_{X}}=U_{M}U_{M-1}\ldots U_{1}.
\]
The qubits labeled $1,\ldots,n_{\inn}$ hold the input state. The remaining $n-n_{\inn}$ ancilla qubits are each initialized to $|0\rangle$.  We require each qubit to be involved in at least one gate. The qubit labeled $2$ is the output qubit that contains the result of the computation after the circuit is applied.
The qubit labeled $1$ mediates two-qubit gates: each gate $U_j$ is a two-qubit gate acting nontrivially on the qubit labeled $1$ and another qubit $s(j)\in\{2,\ldots,n\}$. Each unitary $U_j$ is chosen from the set
\begin{equation}
\{\CNOT_{1s(j)},\CNOT_{s(j)1},\CNOT_{1s(j)}\left(H\otimes\id\right),\CNOT_{1s(j)}\left(HT\otimes\id\right)\},\label{eq:gate_set_Vr}
\end{equation}
but no two consecutive gates ($U_j$ and $U_{j+1}$) act between the same two qubits, i.e., 
\begin{equation}
s(j)\neq s(j+1)\qquad j\in[M-1].\label{eq:condition_ij}
\end{equation}
As discussed in \sec{Definitions-and-Results}, the acceptance probability for the circuit acting on input state $|\psi_{\inn}\rangle\in(\C^{2})^{\otimes n_{\inn}}$ is the probability that a final measurement of the output qubit in the computational basis gives the value $1$: 
\[
\AP\left(\mathcal{C}_{X},|\psi_{\inn}\rangle\right)=\left\Vert |1\rangle\langle1|^{(2)}U_{\mathcal{C}_{X}}|\psi_{\inn}\rangle|0\rangle^{\otimes n-n_{\inn}}\right\Vert ^{2}.
\]

We now establish that circuits of this form are universal. We show that any quantum circuit (with $n\geq 4$ qubits) expressed using the universal gate set
\begin{equation}
\{\CNOT,H,HT\}\label{eq:standard_gate_ste}
\end{equation}
can be efficiently rewritten in the prescribed form without increasing the number of qubits and with at most a constant factor increase in the number of gates.

First we map a circuit from the gate set \eq{standard_gate_ste} to the gate set \eq{gate_set_Vr} (without necessarily satisfying condition \eq{condition_ij}). A $\SWAP$ gate between qubits $1$ and $k$ can be performed using the identity
\begin{align*}
  \SWAP_{1k} & = \CNOT_{1k}\CNOT_{k1}\CNOT_{1k}.
\end{align*}
To perform a $\CNOT_{ik}$ gate between two qubits $i,k$ (neither of which is qubit $1$), we swap qubits $1$ and $i$, apply $\CNOT_{1k}$, and then swap back. Similarly, we can apply a single-qubit gate $U\in\{H,HT\}$ to some qubit $k\neq1$ using the sequence of gates
\[
  \SWAP_{1k} \CNOT_{12} \left(\CNOT_{12} U\otimes\id\right) \SWAP_{1k}.
\]
Applying these replacement rules, we obtain a circuit over the gate set \eq{gate_set_Vr}. However, the resulting circuit will not in general satisfy equation \eq{condition_ij}. To enforce this condition, we insert a sequence of four gates equal to the identity, namely
\[
  \id=\CNOT_{1a}\CNOT_{1b}\CNOT_{1a}\CNOT_{1b},
\]
between any two consecutive gates $U_j$ and $U_{j+1}$ with $s(j)=s(j+1)$, where $a\neq b \neq s(j)$. For example, 
\[
\CNOT_{15}\CNOT_{51}\quad\longrightarrow\quad\CNOT_{15}\CNOT_{12}\CNOT_{13}\CNOT_{12}\CNOT_{13}\CNOT_{51}.
\]
Thus we map a circuit over the gate set \eq{standard_gate_ste} into the prescribed form. Note that $n\geq 4$ is needed to ensure that any quantum circuit can be efficiently rewritten so that $s(j) \neq s(j+1)$.  (There do exist circuits of the desired of the form with $n=3$, such as in the example shown in \fig{step-by-step}.)

\subsection{The gate graph}
\label{sec:The-gate-graph}

For any $n$-qubit, $M$-gate verification circuit $\mathcal{C}_{X}$ of the form described above, we associate a gate graph $\gx$. The gate diagram for $\gx$ is built using the gadgets described in \sec{Gadgets}; specifically, we use $M$ two-qubit gadgets and $2(n-1)$ boundary gadgets. Since each two-qubit gadget and each boundary gadget contains $32$ diagram elements, the total number of diagram elements in $\gx$ is $R=32(M+2n-2)$.

We now present the construction of the gate diagram for $\gx$.  We also describe some gate graphs obtained as intermediate steps that are used in our analysis in \sec{Proof-of-Theorem}. The reader may find this description easier to follow by looking ahead to \fig{step-by-step}, which illustrates the construction for a specific $3$-qubit circuit.

\begin{enumerate}
\item \textbf{Draw a grid} with columns labeled $j=0,1,\ldots,M+1$ and rows labeled $i=1,\ldots,n$ (this grid is only used to help describe the diagram).
\item \textbf{Place gadgets in the grid to mimic the quantum circuit.}
For each $j=1,\ldots,M$, place a gadget for the two-qubit gate $U_j$ between rows $1$ and $s(j)$ in the $j$th column. Place boundary gadgets in rows $i=2,\ldots,n$ of column $0$ and in the same rows of column $M+1$. Write $G_{1}$ for the gate graph associated with the resulting diagram.
\item \textbf{Connect the nodes within each row.}
First add edges connecting the nodes in rows $i=2,\ldots,n$; call the resulting gate graph $G_{2}$. Then add edges connecting the nodes in row $1$; call the resulting gate graph $G_{3}$.
\item \textbf{Add self-loops to the boundary gadgets.}
In this step we add self-loops to enforce initialization of ancillas (at the beginning) and the proper output of the circuit (at the end). For each row $k=n_{\inn}+1,\ldots,n$, add a self-loop to node $\delta$ (as shown in  \fig{GVbdy}) of the corresponding boundary gadget in column $r=0$, giving the gate diagram for $G_{4}$. Finally, add a self-loop to node $\alpha$ of the boundary gadget (as in \fig{GVbdy}) in row $2$ and column $M+1$, giving the gate diagram for $\gx$.
\end{enumerate}

\fig{step-by-step} illustrates the step-by-step construction of $G_X$ using a simple $3$-qubit circuit with four gates 
\[
\CNOT_{12}\left(\CNOT_{13}HT\otimes\id\right)\CNOT_{21}\CNOT_{13}.
\]
In this example, two of the qubits are input qubits (so $n_{\inn}=2$), while the third qubit is an ancilla initialized to $|0\rangle$. Following the convention described in \sec{Universality}, we take qubit $2$ to be the output qubit. (In this example the circuit is not meant to compute anything interesting; its only purpose is to illustrate our method of constructing a gate graph).

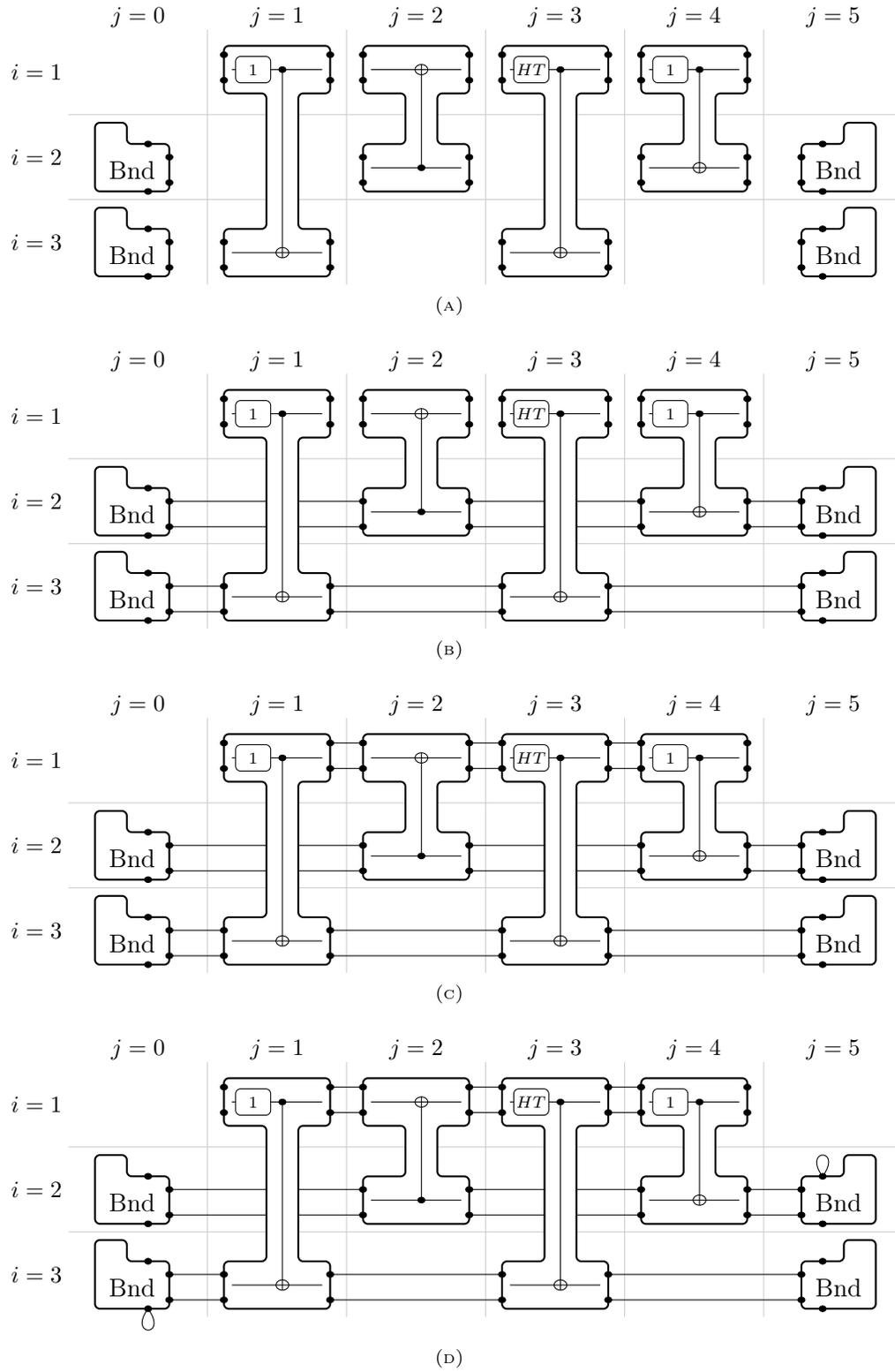
\begin{figure}
\centering \subfloat[][]{ \label{fig:G1}
\begin{tikzpicture}[scale=.8,yscale=.8]

\foreach \y in {2,4}
{   
\draw[draw=black!20] (0,\y cm) -- (15.7 cm, \y cm); 
} 
\foreach \x in {2.618,5.236,7.854,10.472,13.09}
{   
\draw[draw=black!20] (\x cm, 0cm) -- (\x cm, 6cm); 
}
\setcounter{mycount}{1}; 
\foreach \y in {5,3,1}
{   
\node at (-.6,\y) {$i = \arabic{mycount}$\addtocounter{mycount}{1}}; 
}
\foreach \x in {0, ...,5}
{   
\node at ( 2.618*\x + 1.31,6.3) {$j = \x$}; 
}

\foreach \xscope /\scale in {0.5/1,15.2/-1}
{ 
	\foreach \yscope in {1.81,3.81}
	{    
		\begin{scope}[yshift = \yscope cm,xshift = \xscope cm,xscale=\scale]     
		\draw[rounded corners = .75mm,thick] (0,0) -- (.6,0) -- (.6,-.5) -- (1.4,-.5) -- (1.4,-1.618) -- (0,-1.618) -- cycle;                  \draw[fill=black] (1,-.5) circle (.66mm);             
		\draw[fill=black] (1.4,-.81) circle (.66mm);       
		\draw[fill=black] (1.4,-1.41) circle (.66mm);       
		\draw[fill=black] (1,-1.618) circle (.66mm);
        \node at (.7,-1.11) {\large Bnd};   
		\end{scope}
	}
}  
   
\setcounter{mycount}{1}   \foreach \ybottom/ \control /\unitary in {.19/1/1,2.19/0/1,.19/1/HT,2.19/1/1}
{   
\begin{scope}[xshift = \value{mycount}*2.618 cm + .31cm]     
\begin{scope}[xshift=2cm,xscale=-1]     
\draw[rounded corners = .75mm,thick,fill=white] (0cm,\ybottom cm) -- (2cm, \ybottom cm)-- (2cm,\ybottom cm + 1.118cm) -- (1.2cm, \ybottom cm + 1.118cm) -- (1.2cm, 4.5cm)-- (2cm, 4.5cm) -- (2cm, 5.62 cm) -- (0cm, 5.62cm) -- (0cm, 4.5cm) -- (.6cm, 4.5cm)            -- (.6cm, \ybottom cm + 1.118cm) -- (0cm, \ybottom cm + 1.118cm) -- cycle; 

\draw (.15cm, 5.06cm) -- (1.85cm, 5.06cm);     
\draw (.15cm, \ybottom cm + .56 cm) -- (1.85cm, \ybottom cm + .56 cm);     
\draw (.9cm, 5.06cm) -- (.9cm,\ybottom cm + .56 cm);          
\foreach \x in {0,2}
{     
	\foreach \y in {5.41,4.81,\ybottom + .81,\ybottom + .21}
	{       
		\draw[fill=black] (\x,\y) circle (.66mm);     
	}
}          
\if\control1
{       
\draw[rounded corners=.75mm,fill=white] (1.12cm,4.75cm) rectangle (1.78cm, 5.37cm);       
\node at (1.45,5.06) {\scriptsize $\unitary$};       
\draw[fill=black] (.9, 5.06) circle (.66mm);       
\begin{scope}[yshift = \ybottom cm]         
\draw[fill=white] (.9,.56) circle (1.2mm);         
\draw (.78,.56) -- (1.02,.56);         
\draw (.9,.44) -- (.9,.68);       
\end{scope}     
}
\else
{       
\draw[fill=black] (.9,\ybottom+.56) circle (.66mm);       
\draw[fill=white] (.9,5.06) circle (1.2mm);       
\draw (.78,5.06) -- (1.02,5.06);       
\draw (.9,4.94) -- (.9,5.18);     
}\fi     
\end{scope}   
\end{scope}   
\addtocounter{mycount}{1}; 
}
\end{tikzpicture}}

\subfloat[][]{ \label{fig:G2}
\begin{tikzpicture}[scale=.8,yscale=.8]
  
\foreach \y in {2,4}
{   
\draw[draw=black!20] (0,\y cm) -- (15.7 cm, \y cm); 
} 
\foreach \x in {2.618,5.236,7.854,10.472,13.09}
{   
\draw[draw=black!20] (\x cm, 0cm) -- (\x cm, 6cm); 
}
\setcounter{mycount}{1}; 
\foreach \y in {5,3,1}
{   
\node at (-.6,\y) {$i = \arabic{mycount}$\addtocounter{mycount}{1}}; 
}
\foreach \x in {0, ...,5}
{   
\node at ( 2.618*\x + 1.31,6.3) {$j = \x$}; 
}
  
\foreach \y in {.4,1,2.4,3}
{   
\draw (1.9,\y) -- (13.8,\y); 
}

\foreach \xscope /\scale in {0.5/1,15.2/-1}
{ 
	\foreach \yscope in {1.81,3.81}
	{    
		\begin{scope}[yshift = \yscope cm,xshift = \xscope cm,xscale=\scale]     
		\draw[rounded corners = .75mm,thick] (0,0) -- (.6,0) -- (.6,-.5) -- (1.4,-.5) -- (1.4,-1.618) -- (0,-1.618) -- cycle; 
        \draw[fill=black] (1,-.5) circle (.66mm);             
		\draw[fill=black] (1.4,-.81) circle (.66mm);       
		\draw[fill=black] (1.4,-1.41) circle (.66mm);       
		\draw[fill=black] (1,-1.618) circle (.66mm);
        \node at (.7,-1.11) {\large Bnd};   
		\end{scope}
	}
}     

\setcounter{mycount}{1}   \foreach \ybottom/ \control /\unitary in {.19/1/1,2.19/0/1,.19/1/HT,2.19/1/1}
{   
\begin{scope}[xshift = \value{mycount}*2.618 cm + .31cm]     
\begin{scope}[xshift=2cm,xscale=-1]     
\draw[rounded corners = .75mm,thick,fill=white] (0cm,\ybottom cm) -- (2cm, \ybottom cm)-- (2cm,\ybottom cm + 1.118cm) -- (1.2cm, \ybottom cm + 1.118cm) -- (1.2cm, 4.5cm)-- (2cm, 4.5cm) -- (2cm, 5.62 cm) -- (0cm, 5.62cm) -- (0cm, 4.5cm) -- (.6cm, 4.5cm)            -- (.6cm, \ybottom cm + 1.118cm) -- (0cm, \ybottom cm + 1.118cm) -- cycle;

 \draw (.15cm, 5.06cm) -- (1.85cm, 5.06cm);     
\draw (.15cm, \ybottom cm + .56 cm) -- (1.85cm, \ybottom cm + .56 cm);     
\draw (.9cm, 5.06cm) -- (.9cm,\ybottom cm + .56 cm);          
\foreach \x in {0,2}
{     
	\foreach \y in {5.41,4.81,\ybottom + .81,\ybottom + .21}
	{       
		\draw[fill=black] (\x,\y) circle (.66mm);     
	}
}          
\if\control1
{       
\draw[rounded corners=.75mm,fill=white] (1.12cm,4.75cm) rectangle (1.78cm, 5.37cm);       
\node at (1.45,5.06) {\scriptsize $\unitary$};       
\draw[fill=black] (.9, 5.06) circle (.66mm);       
\begin{scope}[yshift = \ybottom cm]         
\draw[fill=white] (.9,.56) circle (1.2mm);         
\draw (.78,.56) -- (1.02,.56);         
\draw (.9,.44) -- (.9,.68);       
\end{scope}     
}
\else
{       
\draw[fill=black] (.9,\ybottom+.56) circle (.66mm);       
\draw[fill=white] (.9,5.06) circle (1.2mm);       
\draw (.78,5.06) -- (1.02,5.06);       
\draw (.9,4.94) -- (.9,5.18);     
}\fi     
\end{scope}   
\end{scope}   
\addtocounter{mycount}{1}; 
}
\end{tikzpicture}}

\subfloat[][]{ \label{fig:G3}
\begin{tikzpicture}[scale=.8,yscale=.8]

\foreach \y in {2,4}
{   
\draw[draw=black!20] (0,\y cm) -- (15.7 cm, \y cm); 
} 
\foreach \x in {2.618,5.236,7.854,10.472,13.09}
{   
\draw[draw=black!20] (\x cm, 0cm) -- (\x cm, 6cm); 
} 
\setcounter{mycount}{1}; 
\foreach \y in {5,3,1}
{   
\node at (-.6,\y) {$i = \arabic{mycount}$\addtocounter{mycount}{1}}; 
}
\foreach \x in {0, ...,5}
{   
\node at ( 2.618*\x + 1.31,6.3) {$j = \x$}; 
}
  
\foreach \y in {.4,1,2.4,3}
{   
\draw (1.9,\y) -- (13.8,\y); 
}
  
\foreach \y in {5.41,4.81}
{   
\draw (4.93,\y) -- (10.78,\y); 
}
  
\foreach \xscope /\scale in {0.5/1,15.2/-1}
{ 
	\foreach \yscope in {1.81,3.81}
	{    
		\begin{scope}[yshift = \yscope cm,xshift = \xscope cm,xscale=\scale]     
		\draw[rounded corners = .75mm,thick] (0,0) -- (.6,0) -- (.6,-.5) --(1.4,-.5) -- (1.4,-1.618) -- (0,-1.618) -- cycle;
        \draw[fill=black] (1,-.5) circle (.66mm);    
		\draw[fill=black] (1.4,-.81) circle (.66mm);       
		\draw[fill=black] (1.4,-1.41) circle (.66mm);       
		\draw[fill=black] (1,-1.618) circle (.66mm);
        \node at (.7,-1.11) {\large Bnd};   
		\end{scope}
	}
}    

 \setcounter{mycount}{1}   \foreach \ybottom/ \control /\unitary in {.19/1/1,2.19/0/1,.19/1/HT,2.19/1/1}
{   
\begin{scope}[xshift = \value{mycount}*2.618 cm + .31cm]     
\begin{scope}[xshift=2cm,xscale=-1]     
\draw[rounded corners = .75mm,thick,fill=white] (0cm,\ybottom cm) -- (2cm, \ybottom cm)-- (2cm,\ybottom cm + 1.118cm) -- (1.2cm, \ybottom cm + 1.118cm) -- (1.2cm, 4.5cm)-- (2cm, 4.5cm) -- (2cm, 5.62 cm) -- (0cm, 5.62cm) -- (0cm, 4.5cm) -- (.6cm, 4.5cm)            -- (.6cm, \ybottom cm + 1.118cm) -- (0cm, \ybottom cm + 1.118cm) -- cycle;
           
 \draw (.15cm, 5.06cm) -- (1.85cm, 5.06cm);     
\draw (.15cm, \ybottom cm + .56 cm) -- (1.85cm, \ybottom cm + .56 cm);     
\draw (.9cm, 5.06cm) -- (.9cm,\ybottom cm + .56 cm);          
\foreach \x in {0,2}
{     
	\foreach \y in {5.41,4.81,\ybottom + .81,\ybottom + .21}
	{       
		\draw[fill=black] (\x,\y) circle (.66mm);     
	}
}          
\if\control1
{       
\draw[rounded corners=.75mm,fill=white] (1.12cm,4.75cm) rectangle (1.78cm, 5.37cm);       
\node at (1.45,5.06) {\scriptsize $\unitary$};       
\draw[fill=black] (.9, 5.06) circle (.66mm);       
\begin{scope}[yshift = \ybottom cm]         
\draw[fill=white] (.9,.56) circle (1.2mm);         
\draw (.78,.56) -- (1.02,.56);         
\draw (.9,.44) -- (.9,.68);       
\end{scope}     
}
\else
{       
\draw[fill=black] (.9,\ybottom+.56) circle (.66mm);       
\draw[fill=white] (.9,5.06) circle (1.2mm);       
\draw (.78,5.06) -- (1.02,5.06);    
\draw (.9,4.94) -- (.9,5.18);     
}\fi     
\end{scope}   
\end{scope}   
\addtocounter{mycount}{1}; 
}
\end{tikzpicture}}

\subfloat[][]{ \label{fig:GX}
\begin{tikzpicture}[scale=.8,yscale=.8]

\foreach \y in {2,4}
{   
\draw[draw=black!20] (0,\y cm) -- (15.7 cm, \y cm); 
} 
\foreach \x in {2.618,5.236,7.854,10.472,13.09}
{   
\draw[draw=black!20] (\x cm, 0cm) -- (\x cm, 6cm); 
}
\setcounter{mycount}{1}; 
\foreach \y in {5,3,1}
{   
\node at (-.6,\y) {$i = \arabic{mycount}$\addtocounter{mycount}{1}}; 
}
\foreach \x in {0, ...,5}
{   
\node at ( 2.618*\x + 1.31,6.3) {$j = \x$}; 
}
  
\foreach \y in {.4,1,2.4,3}
{   
\draw (1.9,\y) -- (13.8,\y); 
}
  
\foreach \y in {5.41,4.81}
{   
\draw (4.93,\y) -- (10.78,\y); 
}

\foreach \xscope /\scale in {0.5/1,15.2/-1}
{ 
	\foreach \yscope in {1.81,3.81}
	{    
		\begin{scope}[yshift = \yscope cm,xshift = \xscope cm,xscale=\scale]     
		\draw[rounded corners = .75mm,thick] (0,0) -- (.6,0) -- (.6,-.5) --(1.4,-.5) -- (1.4,-1.618) -- (0,-1.618) -- cycle; 
        \draw[fill=black] (1,-.5) circle (.66mm);             
		\draw[fill=black] (1.4,-.81) circle (.66mm);       
		\draw[fill=black] (1.4,-1.41) circle (.66mm);       
		\draw[fill=black] (1,-1.618) circle (.66mm);
        \node at (.7,-1.11) {\large Bnd};   
		\end{scope}
	}
}   
  
 \setcounter{mycount}{1}   \foreach \ybottom/ \control /\unitary in {.19/1/1,2.19/0/1,.19/1/HT,2.19/1/1}
{   
\begin{scope}[xshift = \value{mycount}*2.618 cm + .31cm]     
\begin{scope}[xshift=2cm,xscale=-1]     
\draw[rounded corners = .75mm,thick,fill=white] (0cm,\ybottom cm) -- (2cm, \ybottom cm)-- (2cm,\ybottom cm + 1.118cm) -- (1.2cm, \ybottom cm + 1.118cm) -- (1.2cm, 4.5cm)-- (2cm, 4.5cm) -- (2cm, 5.62 cm) -- (0cm, 5.62cm) -- (0cm, 4.5cm) -- (.6cm, 4.5cm)-- (.6cm, \ybottom cm + 1.118cm) -- (0cm, \ybottom cm + 1.118cm) -- cycle;
 
\draw (.15cm, 5.06cm) -- (1.85cm, 5.06cm);     
\draw (.15cm, \ybottom cm + .56 cm) -- (1.85cm, \ybottom cm + .56 cm);     
\draw (.9cm, 5.06cm) -- (.9cm,\ybottom cm + .56 cm);          
\foreach \x in {0,2}
{     
	\foreach \y in {5.41,4.81,\ybottom + .81,\ybottom + .21}
	{       
		\draw[fill=black] (\x,\y) circle (.66mm);     
	}
}          
\if\control1
{       
\draw[rounded corners=.75mm,fill=white] (1.12cm,4.75cm) rectangle (1.78cm, 5.37cm);       
\node at (1.45,5.06) {\scriptsize $\unitary$};       
\draw[fill=black] (.9, 5.06) circle (.66mm);       
\begin{scope}[yshift = \ybottom cm]         
\draw[fill=white] (.9,.56) circle (1.2mm);         
\draw (.78,.56) -- (1.02,.56);         
\draw (.9,.44) -- (.9,.68);       
\end{scope}     
}
\else
{       
\draw[fill=black] (.9,\ybottom+.56) circle (.66mm);       
\draw[fill=white] (.9,5.06) circle (1.2mm);       
\draw (.78,5.06) -- (1.02,5.06);       
\draw (.9,4.94) -- (.9,5.18);     
}\fi     
\end{scope}   
\end{scope}   
\addtocounter{mycount}{1}; 
}

\draw[looseness=200] (1.5,.19) to [out=-60,in=-120] (1.5,.2);   
\draw[looseness=200] (14.2,3.31) to [out=60,in=120] (14.2,3.30);
\end{tikzpicture}}

\caption{Step-by-step construction of the gate diagram for $\gx$ for the three-qubit example circuit described in the text. 
\subfig{G1} The gate diagram for $G_{1}$.
\subfig{G2} Add edges in all rows except the first to obtain the gate diagram for $G_{2}$.
\subfig{G3} Add edges in the first row to obtain the gate diagram for $G_{3}$.
\subfig{GX} Add self-loops to the boundary gadgets to obtain the gate diagram for $\gx$ (the diagram for $G_{4}$ in this case differs from \subfig{GX} by removing the self-loop in column $5$; this diagram is not shown).
\label{fig:step-by-step}
}
\end{figure}

We made some choices in designing this circuit-to-gate graph mapping that may seem arbitrary (e.g., we chose to place boundary gadgets in each row except the first). We have tried to achieve a balance between simplicity of description and ease of analysis, but we expect that other choices could be made to work.

\subsubsection{Notation for $\gx$}

We now introduce some notation that allows us to easily refer to a subset $\mathcal{L}$ of the diagram elements in the gate diagram for $\gx$.

Recall from \sec{Gadgets} that each two-qubit gate gadget and each boundary gadget is composed of $32$ diagram elements. This can be seen by looking at \fig{GVucnot} and \fig{GVbdy} and noting (from \fig{W_gadget}) that each move-together gadget comprises 6 diagram elements.

For each of the two-qubit gate gadgets in the gate diagram for $\gx$, we focus our attention on the four diagram elements labeled $1$--$4$ in \fig{GVucnot}. In total there are $4M$ such diagram elements in the gate diagram for $\gx$: in each column $j\in\{1,\ldots,M\}$ there are two in row $1$ and two in row $s(j)$. When $U_j\in\{\CNOT_{1s(j)},\CNOT_{1s(j)}\left(H\otimes\id\right),\CNOT_{1s(j)}\left(HT\otimes\id\right)\}$ the diagram elements labeled $1,2$ are in row $1$ and those labeled $3,4$ are in row $s(j)$; when $U_j=\CNOT_{s(j) 1}$ those labeled $1,2$ are in row $s(j)$ and those labeled $3,4$ are in row $1$. We denote these diagram elements by triples $(i,j,d)$. Here $i$ and $j$ indicate (respectively) the row and column of the grid in which the diagram element is found, and $d$ indicates whether it is the leftmost ($d=0$) or rightmost ($d=1$) diagram element in this row and column. We define 
\begin{equation}
\mathcal{L}_{\mathrm{gates}}=\left\{ \left(i,j,d\right)\colon i\in\{1,s(j)\},\; j\in[M],\; d\in\{0,1\}\right\} \label{eq:L_gates}
\end{equation}
to be the set of all such diagram elements.

For example, in \fig{step-by-step} the first gate is 
\[
  U_1=\CNOT_{13},
\]
so the gadget from \fig{GVucnot} (with $\tilde U=1$) appears between rows $1$ and $3$ in the first column. The diagram elements labeled $1,2,3,4$ from \fig{GVucnot} are denoted by $(1,1,0), (1,1,1), (3,1,0), (3,1,1)$, respectively. The second gate in \fig{step-by-step} is $U_2=\CNOT_{21}$, so the gadget from \fig{GVcnot} (with $\tilde U=1$) appears between rows $2$ and $1$; in this case the diagram elements labeled $1,2,3,4$ in \fig{GVucnot} are denoted by $(2,2,0),(2,2,1),(1,2,0),(1,2,1)$, respectively.

We also define notation for the boundary gadgets in $\gx$. For each boundary gadget. we focus on a single diagram element, labeled $4$ in \fig{GVbdy}. For the left hand-side and right-hand side boundary gadgets, respectively, we denote these diagram elements as
\begin{align}
 & \mathcal{L}_{\inn}=\left\{ (i,0,1)\colon i\in\{2,\ldots,n\}\right\} \label{eq:L_in}\\
 & \mathcal{L}_{\out}=\left\{ (i,M+1,0)\colon i\in\{2,\ldots,n\}\right\} .\label{eq:L_out}
\end{align}
 
\begin{definition}
\label{defn:scriptL}Let $\mathcal{L}$ be the set of diagram
elements
\[
\mathcal{L}=\mathcal{L}_{\inn}\cup\mathcal{L}_{\mathrm{gates}}\cup\mathcal{L}_{\out}
\]
where $\mathcal{L}_{\inn}$, $\mathcal{L}_{\mathrm{gates}}$, and $\mathcal{L}_{\out}$ are given by equations \eq{L_in}, \eq{L_gates}, and \eq{L_out}, respectively.
\end{definition}

Finally, it is convenient to define a function $F$ that describes horizontal movement within the rows of the gate diagram for $\gx$. The function $F$ takes as input a two-qubit gate $j\in[M]$, a qubit $i\in\{2,\ldots,n\}$, and a single bit and outputs a diagram element from the set $\mathcal{L}$. If the bit is $0$ then $F$ outputs the diagram element in row $i$ that appears in a column $0\leq k<j$ with $k$ maximal (i.e., the closest diagram element in row $i$ to the left of column $j$):
\begin{equation}
F(i,j,0)=\begin{cases}
(i,k,1) & \text{ where }1\leq k<j\text{ is the largest }k\text{ such that }s(k)=i\text{, if it exists}\\
(i,0,1) & \text{ otherwise.}
\end{cases}\label{eq:F_bit0}
\end{equation}
On the other hand, if the bit is $1$, then $F$ outputs the diagram element in row $i$ that appears in a column $j<k\leq M+1$ with $k$ minimal (i.e., the closest diagram element in row $i$ to the right of column $j$). 
\begin{equation}
F(i,j,1)=\begin{cases}
(i,k,0) & \text{ where }j<k\leq M\text{ is the smallest }k\text{ such that }s(k)=j\text{, if it exists}\\
(i,M+1,0) & \text{ otherwise}.
\end{cases}\label{eq:F_bit1}
\end{equation}

\subsection{The occupancy constraints graph} \label{sec:The-occupancy-constraints}

In this Section we define an occupancy constraints graph $\gxoc$. Along with $\gx$ and the number of particles $n$, this determines a subspace $\mathcal{I}(\gx,\gxoc,n)\subset\mathcal{Z}_{n}(\gx)$ through equation \eq{occup_space_defn}. We will see in \sec{Proof-of-Theorem} how low-energy states of the Bose-Hubbard model that live entirely within this subspace encode computations corresponding to the quantum circuit $\mathcal{C}_{X}$. This fact is used in the proof of \thm{main_thm_with_occ_constraints}, which shows that the smallest eigenvalue $\lambda_{n}^{1}(\gx,\gxoc)$ of
\[
H(\gx,\gxoc,n)=H(\gx,n)\big|_{\mathcal{I}(\gx,\gxoc,n)}
\]
is related to the maximum acceptance probability of the circuit.

We encode quantum data in the locations of $n$ particles in the graph $\gx$ as follows. Each particle encodes one qubit and is located in one row of the graph $\gx$. Since all two-qubit gates in $\mathcal{C}_{X}$ involve the first qubit, the location of the particle in the first row determines how far along the computation has proceeded. We design the occupancy constraints graph to ensure that low-energy states of $H(\gx,\gxoc,n)$ have exactly one particle in each row (since there are $n$ particles and $n$ rows), and so that the particles in rows $2,\ldots,n$ are not too far behind or ahead of the particle in the first row. To avoid confusion, we emphasize that not \emph{all} states in the subspace $\mathcal{I}(\gx,\gxoc,n)$ have the desired properties---for example, there are states in this subspace with more than one particle in a given row. We see in the next Section that states with low energy for $H(\gx,n)$ that also satisfy the occupancy constraints (i.e., low-energy states of $H(\gx,\gxoc,n)$) have the desired properties.

We now define $\gxoc$, which is a simple graph with a vertex for each diagram element in $\gx$. Each edge in $\gxoc$ places a constraint on the locations of particles in $\gx$. The graph $\gxoc$ only has edges between diagram elements in the set $\mathcal{L}$ from \defn{scriptL}; we define the edge set $E(\gxoc)$ by specifying pairs of diagram elements $L_{1},L_{2}\in\mathcal{L}$. We also indicate (in bold) the reason for choosing the constraints.
\begin{enumerate}
\item \textbf{No two particles in the same row.} For each $i\in[n]$ we add constraints between diagram elements $(i,j,c)\in\mathcal{L}$ and $(i,k,d)\in\mathcal{L}$ in row $i$ but in different columns, i.e.,
\begin{equation}
\{\left(i,j,c\right),\left(i,k,d\right)\}\in E(\gxoc)\text{ whenever }j\neq k.\label{eq:occ_constraints_type1}
\end{equation}
\item \textbf{Synchronization with the particle in the first row.} For each $j\in[M]$ we add constraints between row $1$ and row $s(j)$:
\[
\{(1,j,c),(s(j),k,d)\}\in E(\gxoc)\text{ whenever }k\neq j\text{ and }(s(j),k,d)\neq F(s(j),j,c).
\]
For each $j\in[M]$ we also add constraints between row $1$ and rows $i\in[n]\setminus\{1,s(j)\}$:
\[
\{(1,j,c),(i,k,d)\}\in E(\gxoc)\text{ whenever }(i,k,d)\notin\{F(i,j,0),F(i,j,1)\}.
\]
\end{enumerate}

\section{\texorpdfstring{Proof of \thm{main_thm_with_occ_constraints}}{Proof of Theorem~\ref{thm:main_thm_with_occ_constraints}}}
\label{sec:Proof-of-Theorem}

We write $\gamma(H)$ for the smallest nonzero eigenvalue of a positive semidefinite matrix $H$.

\subsection{Strategy and outline of the proof}

\thm{main_thm_with_occ_constraints} bounds the smallest eigenvalue $\lambda_n^{1}(\gx,\gxoc)$ of $H(\gx,\gxoc,n)$. To prove the Theorem, we investigate a sequence of Hamiltonians starting with $H(G_{1},n)$ and $H(G_{1},\gxoc,n)$ and then work our way up to the Hamiltonian $H(\gx,\gxoc,n)$ by adding positive semidefinite terms.

For each Hamiltonian we consider, we solve for the nullspace and the smallest nonzero eigenvalue. To go from one Hamiltonian to the next, we use the following ``\npl,'' which was used (implicitly) in reference \cite{MLM99}. The Lemma bounds the smallest nonzero eigenvalue $\gamma(H_A+H_B)$ of a sum of positive semidefinite Hamiltonians $H_A$ and $H_B$ using knowledge of the smallest nonzero eigenvalue $\gamma(H_A)$ of $H_A$ and the smallest nonzero eigenvalue $\gamma(H_B|_S)$ of the restriction of $H_B$ to the nullspace $S$ of $H_A$.

\begin{restatable}[\textbf{Nullspace Projection Lemma }\cite{MLM99}]{lemma}{NPL}
\label{lem:npl}Let $H_A, H_B\geq0$. Suppose the nullspace $S$ of $H_A$ is non-empty and 
\[
\gamma(H_{B}|_{S})\geq c > 0\qquad\text{and}\qquad\gamma(H_{A})\geq d >0.
\]
Then 
\begin{equation}
\gamma(H_{A}+H_{B})\geq\frac{cd}{c+d+\left\Vert H_{B}\right\Vert }.\label{eq:lemma_lower_bnd}
\end{equation}
\end{restatable}

We prove the Lemma in \sec{Proof-of-Lemma_MLM}. When we apply this Lemma, we are usually interested in an asymptotic limit where $c,d\ll\left\Vert H_{B}\right\Vert $ and the right-hand side of \eq{lemma_lower_bnd} is $\Omega(\frac{cd}{\Vert H_{B}\Vert })$.

Our proof strategy, using repeated applications of the \npl, is analogous to that of reference \cite{KKR04}, where the so-called Projection Lemma was used similarly. Our technique has the advantage of not requiring the terms we add to our Hamiltonian to have ``unphysical'' problem-size dependent coefficients (it also has this advantage over the method of perturbative gadgets \cite{KKR04,JF08}). This allows us to prove results about the ``physically realistic'' Bose-Hubbard Hamiltonian. A similar technique based on Kitaev's Geometric Lemma was used in reference \cite{GN13} (however, that method is slightly more computation intensive, requiring a lower bound on $\gamma(H_{B})$ as well as bounds on $\gamma(H_{A})$ and $\gamma(H_{B}|_{S})$).

\subsubsection{Adjacency matrices of the gate graphs}
\label{sec:The-adjacency-matrix}

We begin by discussing the graphs
\[
G_{1},G_{2},G_{3},G_{4},\gx
\]
(as defined in \sec{From-circuits-to}; see \fig{step-by-step}) in more detail and deriving some properties of their adjacency matrices.

The graph $G_{1}$ has a component for each of the two-qubit gates $j \in [M]$, for each of the boundary gadgets $i=2,\ldots,n$ in column $0$, and for each of the boundary gadgets $i=2,\ldots,n$ in column $M+1$. In other words 
\begin{equation}
G_{1}=\underbrace{\left(\bigcup_{i=2}^{n}G_{\bnd}\right)}_{\text{left boundary}}\cup\underbrace{\left(\bigcup_{j=1}^{M}G_{U_j}\right)}_{\text{two-qubit gates}}\cup\underbrace{\left(\bigcup_{i=2}^{n}G_{\bnd}\right)}_{\text{right boundary}}.\label{eq:G_alpha}
\end{equation}
We use our knowledge of the adjacency matrices of the components $G_\bnd$ and $G_{U_j}$ to understand the ground space of $A(G_1)$.  Recall (from \sec{Gadgets}) that the smallest eigenvalue of $A(G_{U_j})$ is 
\[
e_{1}=-1-3\sqrt{2}
\]
(with degeneracy $16)$ which is also the smallest eigenvalue of $A(G_{\bnd})$ (with degeneracy $4$). For each diagram element $L\in\mathcal{L}$ and pair of bits $z,a\in\{0,1\}$ there is an eigenstate $|\rho_{z,a}^{L}\rangle$ of $A(G_{1})$ with this minimal eigenvalue $e_{1}$. In total we get sixteen eigenstates
\[
|\rho_{z,a}^{(1,j,0)}\rangle,|\rho_{z,a}^{(1,j,1)}\rangle,|\rho_{z,a}^{(s(j),j,0)}\rangle,|\rho_{z,a}^{(s(j),j,1)}\rangle,\quad z,a\in\{0,1\}
\]
for each two-qubit gate $j\in[M]$, four eigenstates 
\[
|\rho_{z,a}^{(i,0,1)}\rangle,\quad z,a\in\{0,1\}
\]
for each boundary gadget $i\in\{2,\ldots,n\}$ in column $0$, and four eigenstates 
\[
|\rho_{z,a}^{(i,M+1,0)}\rangle,\quad z,a\in\{0,1\}
\]
for each boundary gadget $i\in\{2,\ldots,n\}$ in column $M+1$. The set 
\[
\left\{ |\rho_{z,a}^{L}\rangle\colon z,a\in\{0,1\},\; L\in\mathcal{L}\right\} 
\]
is an orthonormal basis for the ground space of $A(G_{1})$. 

We write the adjacency matrices of $G_{2}$, $G_{3}$, $G_{4}$, and $\gx$ as 
\begin{align*}
  A(G_{2}) &= A(G_{1})+h_{1} & 
  A(G_{4}) &= A(G_{3})+\sum_{i=n_{\inn}+1}^{n}h_{\inn,i} \\
  A(G_{3}) &= A(G_{2})+h_{2} &
  A(\gx)   &= A(G_{4})+h_{\out}.
\end{align*}
From step 3 of the construction of the gate diagram in \sec{The-gate-graph}, we see that $h_{1}$ and $h_{2}$ are both sums of terms of the form 
\[
  \left(|q,z,t\rangle+|q^{\prime},z,t^{\prime}\rangle\right)
  \left(\langle q,z,t|+\langle q^{\prime},z,t^{\prime}|\right)\otimes\id_{j},
\]
where $h_1$ contains a term for each edge in rows $2,\ldots, n$ and $h_2$ contains a term for each of the $2(M-1)$ edges in the first row. The operators
\begin{align}
  h_{\mathrm{in},i} &= |(i,0,1),1,7\rangle\langle(i,0,1),1,7|\otimes\id &
  h_{\out} &= |(2,M+1,0),0,5\rangle\langle(2,M+1,0),0,5|\otimes\id
\label{eq:hin_i}
\end{align}
correspond to the self-loops added in the gate diagram in step 4 of \sec{The-gate-graph}. 

We prove that $G_1$, $G_2$, $G_3$, $G_4$, and $\gx$ are $e_1$-gate graphs.

\begin{lemma}
\label{lem:The-smallest-eigenvalues}The smallest eigenvalues of $G_{1},G_{2},G_{3},G_{4}$
and $\gx$ are 
\[
\mu(G_{1})=\mu(G_{2})=\mu(G_{3})=\mu(G_{4})=\mu(\gx)=e_{1}.
\]
\end{lemma}

\begin{proof}
We showed in the above discussion that $\mu(G_{1})=e_{1}$. The adjacency matrices of $G_{2}$, $G_{3}$, $G_{4}$, and $\gx$ are obtained from that of $G_{1}$ by adding positive semidefinite terms ($h_{1}$, $h_{2}$, $h_{\inn,i}$, and $h_{\out}$ are all positive semidefinite). It therefore suffices to exhibit an eigenstate $|\varrho\rangle$ of $A(G_{1})$ with
\[
  h_{1}|\varrho\rangle
  =h_{2}|\varrho\rangle
  =h_{\inn,i}|\varrho\rangle
  =h_{\out}|\varrho\rangle
  =0
\]
(for each $i\in\{n_{\inn}+1,\ldots,n\}$). There are many states $|\varrho\rangle$ satisfying these conditions; one example is
\[
|\varrho\rangle = |\rho_{0,0}^{(1,1,0)}\rangle
\]
which is supported on vertices where $h_{1}$, $h_{2}$, $h_{\inn,i}$, and $h_{\out}$ have no support.
\end{proof}

\subsubsection{Building up the Hamiltonian}
\label{sec:Building-up-the}

We now outline the sequence of Hamiltonians considered in the following Sections and describe the relationships between them. As a first step, in \sec{Configurations-and-the} we exhibit a basis $\mathcal{B}_{n}$ for the nullspace of $H(G_{1},n)$ and we prove that its smallest nonzero eigenvalue is lower bounded by a positive constant. We then discuss the restriction
\begin{equation}
  H(G_{1},\gxoc,n)
  = H(G_{1},n)\big|_{\mathcal{I}\left(G_{1},\gxoc,n\right)}
\label{eq:restriction}
\end{equation}
in \sec{Legal-configurations-and}, where we prove that a subset $\mathcal{B}_{\legal}\subset\mathcal{B}_{n}$ is a basis for the nullspace of \eq{restriction}, and that its smallest nonzero eigenvalue is also lower bounded by a positive constant.

For the remainder of the proof we use the \npl (\lem{npl}) four times, using the decompositions 
\begin{align}
H(G_{2},\gxoc,n) & =H(G_{1},\gxoc,n)+H_{1}\big|_{\mathcal{I}(G_{2},\gxoc,n)}\label{eq:H_G2}\\
H(G_{3},\gxoc,n) & =H(G_{2},\gxoc,n)+H_{2}\big|_{\mathcal{I}(G_{3},\gxoc,n)}\label{eq:H_G3}\\
H(G_{4},\gxoc,n) & =H(G_{3},\gxoc,n)+\sum_{i=n_{\inn}+1}^{n}H_{\inn,i}\big|_{\mathcal{I}(G_{4},\gxoc,n)}\label{eq:H_G4}\\
H(\gx,\gxoc,n) & =H(G_{4},\gxoc,n)+H_{\out}\big|_{\mathcal{I}(\gx,\gxoc,n)}\label{eq:H_GC}
\end{align}
where 
\begin{equation*}
H_{1}=\sum_{w=1}^{n}h_{1}^{(w)} \qquad H_{\inn,i}=\sum_{w=1}^{n}h_{\inn,i}^{(w)} \qquad 
H_{2}=\sum_{w=1}^{n}h_{2}^{(w)} \qquad  H_{\out}=\sum_{w=1}^{n}h_{\out}^{(w)}
\end{equation*}
are all positive semidefinite, with $h_{1},h_{2},h_{\inn,i},h_{\out}$ as defined in \sec{The-adjacency-matrix}. Note that in writing equations \eq{H_G2}, \eq{H_G3}, \eq{H_G4}, and \eq{H_GC}, we have used the fact (from \lem{The-smallest-eigenvalues}) that the adjacency matrices of the graphs we consider all have the same smallest eigenvalue $e_{1}$. Also note that
\[ \mathcal{\mathcal{I}}\left(G_{i},\gxoc,n\right)=\mathcal{\mathcal{I}}\left(\gx,\gxoc,n\right) \] for $i\in[4]$ since the gate diagrams for each of the graphs $G_{1},G_{2},G_{3},G_{4}$ and $\gx$ have the same set of diagram elements.

Let $S_{k}$ be the nullspace of $H(G_{k},\gxoc,n)$ for $k=1,2,3,4$. Since these positive semidefinite Hamiltonians are related by adding positive semidefinite terms, their nullspaces satisfy 
\[
S_{4}\subseteq S_{3}\subseteq S_{2}\subseteq S_{1}\subseteq\mathcal{\mathcal{I}}\left(\gx,\gxoc,n\right).
\]
We solve for $S_1=\spn(\mathcal{B}_{\legal})$ in \sec{Legal-configurations-and} and we characterize the spaces $S_2,S_3$, and $S_4$ in \sec{Frustration-free-states} in the course of applying our strategy. 

For example, to use the \npl to lower bound the smallest nonzero eigenvalue of $H(G_{2},\gxoc,n)$, we consider the restriction 
\begin{equation}
  \Big(H_{1}\big|_{\mathcal{I}(G_{2},\gxoc,n)}\Big)\Big|_{S_{1}}
  =H_{1}\big|_{S_{1}}.\label{eq:H1_restriction}
\end{equation}
We also solve for $S_2$, which is equal to the nullspace of \eq{H1_restriction}. To obtain the corresponding lower bounds on the smallest nonzero eigenvalues of $H(G_{k},\gxoc,n)$ for $k=2,3,4$ and $H(\gx,\gxoc,n)$, we consider restrictions 
\[
  H_{2}\big|_{S_{2}},
  \sum_{i=n_{\inn}+1}^{n}H_{\inn,i}\big|_{S_{3}},\quad\text{and}\quad
  H_{\out}\big|_{S_{4}}.
\]
Analyzing these restrictions involves extensive computation of matrix elements. To simplify and organize these computations, we first compute the restrictions of each of these operators to the space $S_{1}$. We present the results of this computation in \sec{Matrix-elements-in}; details of the calculation can be found in \sec{matrix_els_details}. In \sec{Frustration-free-states} we proceed with the remaining computations and apply the \npl three times using equations \eq{H_G2}, \eq{H_G3}, and \eq{H_G4}. Finally, in \sec{Completeness-and-Soundness} we apply the Lemma again using equation \eq{H_GC} and we prove \thm{main_thm_with_occ_constraints}.

\subsection{Configurations}
\label{sec:Configurations-and-the}

In this Section we use \lem{BH_disconnected_graphs} to solve for the nullspace of $H(G_{1},n)$, i.e., the $n$-particle frustration-free states on $G_{1}$. \lem{BH_disconnected_graphs} describes how frustration-free states for $G_1$ are built out of frustration-free states for its components.

To see how this works, consider the example from \fig{G1}. In this example, with $n=3$, we construct a basis for the nullspace of $H(G_{1},3)$ by considering two types of eigenstates. First, there are frustration-free states
\begin{equation}
\Sym(|\rho_{z_{1},a_{1}}^{L_{1}}\rangle|\rho_{z_{2},a_{2}}^{L_{2}}\rangle|\rho_{z_{3},a_{3}}^{L_{3}}\rangle)
\label{eq:twopart_states1}
\end{equation}
where $L_{k}=(i_{k},j_{k},d_{k})\in\mathcal{L}$ belong to different components of $G_{1}$. That is to say, $j_{w}\neq j_{t}$ unless $j_{w}=j_{t}\in\{0,5\}$, in which case $i_{w}\neq i_{t}$ (in this case the particles are located either at the left or right boundary, in different rows of $G_{1}$). There are also frustration-free states where two of the three particles are located in the same two-qubit gadget $J\in[M]$ and one of the particles is located in a diagram element $L_{1}$ from a different component of the graph. These states have the form
\begin{equation}
\Sym(|T_{z_{1},a_{1},z_{2},a_{2}}^{J}\rangle|\rho_{z_{3},a_{3}}^{L_{1}}\rangle)
\label{eq:twopart_states2}
\end{equation}
where 
\begin{equation}
|T_{z_{1},a_{1},z_{2},a_{2}}^{J}\rangle=\frac{1}{\sqrt{2}}|\rho_{z_{1},a_{1}}^{(1,J,0)}\rangle|\rho_{z_{2},a_{2}}^{(s(J),J,0)}\rangle+\frac{1}{\sqrt{2}}\sum_{x_{1},x_{2}\in\{0,1\}}U_{J}(a_{1})_{x_{1}x_{2},z_{1}z_{2}}|\rho_{x_{1},a_{1}}^{(1,J,1)}\rangle|\rho_{x_{2},a_{2}}^{(s(J),J,1)}\rangle\label{eq:T_state}
\end{equation}
and $L_{1}=(i,j,k)\in\mathcal{L}$ satisfies $j\neq J$ (using \lem{2qub_gate} and also the fact that the controlled-not gate is real for the case $U_J = \CNOT_{s(J)1}$).

Each of the states \eq{twopart_states1} and \eq{twopart_states2} is specified by $6$ ``data'' bits $z_{1},z_{2},z_{3},a_{1},a_{2},a_{3}\in\{0,1\}$ and a ``configuration'' indicating where the particles are located in the graph. The configuration is specified either by three diagram elements $L_{1},L_{2},L_{3}\in\mathcal{L}$ from different components of $G_{1}$ or by a two-qubit gate $J\in [M]$ along with a diagram element $L_{1}\in\mathcal{L}$ from a different component of the graph.

We now define the notion of a configuration for general $n$. Informally, we can think of an $n$-particle configuration as a way of placing $n$ particles in the graph $G_{1}$ subject to the following restrictions. We first place each of the $n$ particles in a component of the graph, with the restriction that no boundary gadget may contain more than one particle and no two-qubit gadget may contain more than two particles. For each particle on its own in a component (i.e., in a component with no other particles), we assign one of the diagram elements $L\in\mathcal{L}$ associated to that component. We therefore specify a configuration by a set of two-qubit gadgets $J_{1},\ldots,J_{Y}$ that contain two particles, along with a set of diagram elements $L_{k}\in\mathcal{L}$ that give the locations of the remaining $n-2Y$ particles. We choose to order the $J$s and the $L$s so that each configuration is specified by a unique tuple $(J_{1},\ldots,J_{Y},L_{1},\ldots,L_{n-2Y})$. For concreteness, we use the lexicographic order on diagram elements in the set $\mathcal{L}$: $L_{A}=(i_{A},j_{A},d_{A})$ and $L_{B}=(i_{B},j_{B},d_{B})$ satisfy $L_{A}<L_{B}$ iff either $i_{A}<i_{B}$, or $i_{A}=i_{B}$ and $j_{A}<j_{B}$, or $(i_{A},j_{A})=(i_{B},j_{B})$ and $d_{A}<d_{B}$.

\begin{definition}
[\textbf{Configuration}]\label{defn:configuration}An $n$-particle configuration on the gate graph $G_{1}$ is a tuple 
\[
(J_{1},\ldots,J_{Y},L_{1},\ldots,L_{n-2Y})
\]
with $Y\in\{0,\ldots,\left\lfloor \frac{n}{2}\right\rfloor \}$, ordered integers
\[
1\leq J_{1}<J_{2}<\cdots<J_{Y}\leq M,
\]
and lexicographically ordered diagram elements 
\[
L_{1}<L_{2}<\cdots<L_{n-2Y},\qquad L_{k}=\left(i_{k},j_{k},d_{k}\right)\in\mathcal{L}.
\]
We further require that each $L_{k}$ is from a different component of $G_{1}$, i.e., 
\[
j_{w}=j_{t}\quad\Longrightarrow\quad j_{w}\in\{0,M+1\}\text{ and }i_{w}\neq i_{t},
\]
and we require that $j_{u}\neq J_{v}$ for all $u\in[n-2Y]$ and $v\in[Y]$.
\end{definition}

\newcommand{\configex}[1]{
\raisebox{-0.5\height}{
\begin{tikzpicture}[scale=.4,yscale=.8, every node/.style={scale=0.4}]

\foreach \y in {2,4}{
  \draw[draw=black!20] (0,\y cm) -- (15.7 cm, \y cm);
}
\foreach \x in {2.618,5.236,7.854,10.472,13.09}{
  \draw[draw=black!20] (\x cm, 0cm) -- (\x cm, 6cm);
}
\setcounter{mycount}{1};
\foreach \y in {5,3,1}{
  \node at (-.6,\y) {$i = \arabic{mycount}$\addtocounter{mycount}{1}};
}
\foreach \x in {0, ...,5}{
  \node at ( 2.618*\x + 1.31,6.3) {$j = \x$};
}

\foreach \xscope /\scale in {0.5/1,15.2/-1}{ 
  \foreach \yscope in {1.81,3.81}{
    \begin{scope}[yshift = \yscope cm,xshift = \xscope cm,xscale=\scale]
      \draw[rounded corners = .75mm] (0,0) -- (.6,0) -- (.6,-.5) --
        (1.4,-.5) -- (1.4,-1.618) -- (0,-1.618) -- cycle;
         \end{scope}
  }
}

\setcounter{mycount}{1}
\foreach \ybottom/ \control /\unitary in {.19/1/1,2.19/0/1,.19/1/HT,2.19/1/1}{
  \begin{scope}[xshift = \value{mycount}*2.618 cm + .31cm]
    \begin{scope}[xshift=2cm,xscale=-1]
      \draw[rounded corners = .75mm,fill=white] 
        (0cm,\ybottom cm) -- (2cm, \ybottom cm) -- (2cm,\ybottom cm + 1.118cm) 
          -- (1.2cm, \ybottom cm + 1.118cm) -- (1.2cm, 4.5cm) -- (2cm, 4.5cm) 
          -- (2cm, 5.62 cm) -- (0cm, 5.62cm) -- (0cm, 4.5cm) -- (.6cm, 4.5cm)
          -- (.6cm, \ybottom cm + 1.118cm) -- (0cm, \ybottom cm + 1.118cm)
          -- cycle; 
    \end{scope}
  \end{scope}
  \addtocounter{mycount}{1};
}

#1

\end{tikzpicture}
}}

\begin{figure}
\subfloat[][]{\label{fig:configA}\configex{
\node at (2.618+.31+1.7,5.06) {\huge $\times$};
\node at (2*2.618+.31+.3,2.75) {\huge $\times$};
\node at (3*2.618+.31+.3,.75) {\huge $\times$};
}}
\subfloat[][]{\label{fig:configB}\configex{
\node at (2.618+.31+1.7,.75) {\huge $\times$};
\foreach \yoffset in {.2,-.2}{
  \node at (2*2.618+.31+1.1,4+\yoffset) {\large $\times$};
}
\draw[line width=.1mm] (2*2.618+.31+1.1,4) ellipse (.25 cm and .45 cm);
}}

\subfloat[][]{\label{fig:configC}\configex{
\node at (2.618+.31+1.7,5.06) {\huge $\times$};
\node at (.5+1.1,2.75) {\huge $\times$};
\node at (3*2.618+.31+.3,.75) {\huge $\times$};
}}
\subfloat[][]{\label{fig:configD}\configex{
\node at (.5+1.1,2.75) {\huge $\times$};
\foreach \yoffset in {.2,-.2}{
  \node at (3*2.618+.31+1.1,3+\yoffset) {\large $\times$};
}
\draw[line width=.1mm] (3*2.618+.31+1.1,3) ellipse (.25 cm and .45 cm);
}}

\subfloat[][]{\label{fig:configE}\configex{
\node at (3*2.618+.31+.3,5.06) {\huge $\times$};
\node at (2*2.618+.31+1.7,2.75) {\huge $\times$};
\node at (4*2.618+.31+.3,2.75) {\huge $\times$};
}}
\subfloat[][]{\label{fig:configF}\configex{
\node at (2.618+.31+1.7,5.06) {\huge $\times$};
\node at (2*2.618+.31+.3,2.75) {\huge $\times$};
\node at (5*2.618+.718+.3,.75) {\huge $\times$};
}}
\caption{Diagrammatic depictions of configurations for the example where $G_1$ is the gate graph from \fig{G1}. The Figures show the locations of each of the three particles in the gate graph.  The symbol \protect\tikz[scale=.4, yscale=.8, every node/.style={scale=0.4}]{
\protect \node at (0,0) {\protect \huge $\times$};} indicates a single-particle state and the symbol \protect\tikz[scale=.4, yscale=.8, every node/.style={scale=0.4}]{
\protect \node at (0,.2) {\protect \large $\times$};
\protect \node at (0,-.2) {\protect \large $\times$};
\protect \draw (0,0) ellipse (.25 cm and .45 cm);}
indicates a two-particle state.
\subfig{configA} $((1,1,1),(2,2,0),(3,3,0))$.
\subfig{configB} $(2,(3,1,1))$.
\subfig{configC} $((1,1,1),(2,0,1),(3,3,0)$.
\subfig{configD} $(3,(2,0,1))$.
\subfig{configE} $((1,3,0),(2,2,1),(2,4,0))$.
\subfig{configF} $((1,1,1),(2,2,0),(3,5,0))$.
\label{fig:configs}}
\end{figure}

In \fig{configs} we give some examples of configurations (for the example from \fig{G1} with $n=3$) and we introduce a diagrammatic notation for them.

For any configuration and $n$-bit strings $\vec{z}$ and $\vec{a}$, there is a state in the nullspace of $H(G_{1},n)$, given by
\begin{equation}
\Sym(|T_{z_{1},a_{1},z_{2},a_{2}}^{J_{1}}\rangle\ldots|T_{z_{2Y-1},a_{2Y-1},z_{2Y},a_{2Y}}^{J_{Y}}\rangle|\rho_{z_{2Y+1},a_{2Y+1}}^{L_{1}}\rangle\ldots|\rho_{z_{n},a_{n}}^{L_{n-2Y}}\rangle).\label{eq:config_data_state}
\end{equation}
The ordering in the definition of a configuration ensures that each distinct choice of configuration and $n$-bit strings $\vec{z},\vec{a}$ gives a different state.

\begin{definition}
Let $\mathcal{B}_{n}$ be the set of all states of the form \eq{config_data_state}, where $(J_{1},\ldots,J_{Y},L_{1},\ldots,L_{n-2Y})$ is a configuration and $\vec{z},\vec{a}\in\{0,1\}^{n}$.
\end{definition}

\begin{lemma}
\label{lem:gs_g_alpha}The set $\mathcal{B}_{n}$ is an orthonormal
basis for the nullspace of $H(G_{1},n)$. Furthermore,
\begin{equation}
  \gamma(H(G_{1},n)) \geq \mathcal{K}_0
\label{eq:first_lowerbnd}
\end{equation}
where $\mathcal{K}_0\in (0,1]$ is an absolute constant.
\end{lemma}

\begin{proof}
Each component of $G_{1}$ is either a two-qubit gadget or a boundary gadget (see equation \eq{G_alpha}). The single-particle states of $A(G_{1})$ with energy $e_{1}$ are the states $|\rho_{z,a}^{L}\rangle$ for $L\in\mathcal{L}$ and $z,a\in\{0,1\}$, as discussed in \sec{The-adjacency-matrix}. Each of these states has support on only one component of $G_{1}$. In addition, $G_{1}$ has a two-particle frustration-free state for each two-qubit gadget $J\in[M]$ and bits $z,a,x,b$, namely $\Sym(|T_{z,a,x,b}^{J}\rangle)$. Furthermore, no component of $G_{1}$ has any three- (or more) particle frustration-free states. Using these facts and applying \lem{BH_disconnected_graphs}, we see that $\mathcal{B}_{n}$ spans the nullspace of $H(G_{1},n)$.

\lem{BH_disconnected_graphs} also expresses each eigenvalue of $H(G_{1},n)$ as a sum of eigenvalues for its components. We use this fact to obtain the desired lower bound on the smallest nonzero eigenvalue. Our analysis proceeds on a case-by-case basis, depending on the occupation numbers for each component of $G_{1}$ (the values $N_{1},\ldots,N_{k}$ in \lem{BH_disconnected_graphs}).

First consider any set of occupation numbers where some two-qubit gate gadget $J\in[M]$ contains 3 or more particles. By \lem{increase_part_number} and \lem{BH_disconnected_graphs}, any such eigenvalue is at least $\lambda_{3}^{1}(G_{U_J})$, which is a positive constant by \lem{2qub_gate}. Next consider a case where some boundary gadget contains more than one particle. The corresponding eigenvalues are similarly lower bounded by $\lambda_{2}^{1}(G_{\bnd})$, which is also a positive constant by \lem{boundary_lemma}. Finally, consider a set of occupation numbers where each two-qubit gadget contains at most two particles and each boundary gadget contains at most one particle. The smallest eigenvalue with such a set of occupation numbers is zero. The smallest nonzero eigenvalue is either
\[
\gamma(H(G_{U_J},1)),\,\gamma(H(G_{U_J},2))\text{ for some $J\in[M]$, or }\gamma(H(G_{\bnd},1)).
\]
However, these quantities are at least some positive constant since $H(G_{U_J},1)$, $H(G_{U_J},2)$, and $H(G_{\bnd},1)$ are nonzero constant-sized positive semidefinite matrices.

Now combining the lower bounds discussed above and using the fact that, for each $J\in [M]$, the two-qubit gate $U_J$ is chosen from a fixed finite gate set (given in \eq{gate_set_Vr}), we see that $\gamma (H(G_1,n))$ is lower bounded by the positive constant 
\begin{equation}
  \min\{\lambda_{3}^{1}(G_U),
        \lambda_{2}^{1}(G_{\bnd}),
        \gamma(H(G_U,1)),
        \gamma(H(G_U,2)),
        \gamma(H(G_{\bnd},1))\colon
        \text{$U$ is from the gate set \eq{gate_set_Vr}}\}.
\label{eq:minterms}
\end{equation}
The condition $\mathcal{K}_0\leq1$ can be ensured by setting $\mathcal{K}_0$ to be the minimum of $1$ and \eq{minterms}.
\end{proof}

Note that the constant $\mathcal{K}_0$ can in principle be computed using \eq{minterms}: each quantity on the right-hand side can be evaluated by diagonalizing a specific finite-dimensional matrix.

\subsection{Legal configurations}
\label{sec:Legal-configurations-and}

In this section we define a subset of the $n$-particle configurations that we call legal configurations, and we prove that the subset of the basis vectors in $\mathcal{B}_{n}$ that have legal configurations spans the nullspace of $H(G_{1},\gxoc,n).$

We begin by specifying the set of legal configurations. Every legal configuration 
\[
(J_{1},\ldots,J_{Y},L_{1},\ldots,L_{n-2Y})
\]
has $Y\in\{0,1\}$. The legal configurations with $Y=0$ are 
\begin{equation}
((1,j,d_{1}),F(2,j,d_{2}),F(3,j,d_{3}),\ldots,F(n,j,d_{n}))\label{eq:config_c01}
\end{equation}
where $j\in[M]$ and where $\vec{d}=\left(d_{1},\ldots,d_{n}\right)$ satisfies $d_{i}\in\{0,1\}$ and $d_{1}=d_{s(j)}$. (Recall that the function $F$, defined in equations \eq{F_bit0} and \eq{F_bit1}, describes horizontal movement of particles.) The legal configurations with $Y=1$ are 
\begin{equation}
(j,F(2,j,d_{2}),\ldots,F(s(j)-1,j,d_{s(j)-1}),F(s(j)+1,j,d_{s(j)+1}),\ldots,F(n,j,d_{n}))\label{eq:config_2}
\end{equation}
where $j\in\{1,\ldots,M\}$ and $d_{i}\in\{0,1\}$ for $i\in[n]\setminus\{1,s(j)\}$. Although the values $d_{1}$ and $d_{s(j)}$ are not used in equation \eq{config_2}, we choose to set them to 
\[
d_{1}=d_{s(j)}=2
\]
for any legal configuration with $Y=1$. In this way we identify the set of legal configurations with the set of pairs $j,\vec{d}$ with $j\in[M]$ and 
\[
\vec{d}=(d_{1},d_{2},d_{3},\ldots,d_{n})
\]
satisfying 
\begin{align*}
d_{1} & =d_{s(j)}\in\{0,1,2\} \quad \text{and} \quad d_{i}\in\{0,1\} \text{ for } i\notin\{1,s(j)\}.
\end{align*}
The legal configuration is given by equation \eq{config_c01} if $d_{1}=d_{s(j)}\in\{0,1\}$ and equation \eq{config_2} if $d_{1}=d_{s(j)}=2$. 

The examples in \figs{configA}, \subfig{configB}, and \subfig{configC} show legal configurations whereas the examples in \figs{configD}, \subfig{configE}, and \subfig{configF} are illegal.  The legal examples correspond to $j=1$, $\vec d=(1,1,1)$; $j=2$, $\vec d=(2,2,0)$; and $j=1$, $\vec d=(1,0,1)$, respectively.  We now explain why the other examples are illegal. Looking at \eq{config_2}, we see that the configuration $(3,(2,0,1))$ depicted in \fig{configD} is illegal since $F(2,3,0)=(2,2,1) \neq (2,0,1)$ and $F(2,3,1)=(2,4,0) \neq (2,0,1)$. The configuration in \fig{configE} is illegal since there are two particles in the same row. Looking at equation \eq{config_c01}, we see that the configuration $((1,1,1),(2,2,0),(3,5,0))$ in \fig{configF} is illegal since $(3,5,0)\notin \{F(3,1,0),F(3,1,1)\} = \{(3,0,1),(3,3,0)\}$.

We now identify the subset of basis vectors $\mathcal{B}_{\legal}\subset\mathcal{B}_{n}$ that have legal configurations. We write each such basis vector as
\begin{equation}
|j,\vec{d},\vec{z},\vec{a}\rangle=\begin{cases}
\Sym\Big(|\rho_{z_{1},a_{1}}^{(1,j,d_{1})}\rangle\displaystyle\bigotimes_{i=2}^{n}|\rho_{z_{i},a_{i}}^{F(i,j,d_{i})}\rangle\Big) & d_{1}=d_{s(j)}\in\{0,1\}\\
\Sym\bigg(|T_{z_{1},a_{1},z_{s(j)},a_{s(j)}}^{j}\rangle\displaystyle\bigotimes_{\substack{i=2\\i\neq s(j)}}^n |\rho_{z_{i},a_{i}}^{F(i,j,d_{i})}\rangle\bigg) & d_{1}=d_{s(j)}=2
\end{cases}\label{eq:legal_states}
\end{equation}
where $j,\vec{d}$ specifies the legal configuration and $\vec{z,}\vec{a}\in\{0,1\}^{n}$. (Note that the bits in $\vec{z}$ and $\vec{a}$ are ordered slightly differently than in equation \eq{config_data_state}; here the labeling reflects the indices of the encoded qubits). 

\begin{definition}
Let 
\[
\mathcal{B}_{\legal}=\big\{ |j,\vec{d},\vec{z},\vec{a}\rangle\colon j\in [M],\; d_{1}=d_{s(j)}\in\{0,1,2\}\text{ and }\; d_{i}\in\{0,1\}\text{ for }i\notin\{1,s(j)\},\;\vec{z},\vec{a}\in\{0,1\}^{n}\big\} 
\]
 and $\mathcal{B}_{\illegal}=\mathcal{B}_{n}\setminus\mathcal{B}_{\illegal}.$ 
\end{definition}

The basis $\mathcal{B}_{n}=\mathcal{B}_{\legal}\cup\mathcal{B}_{\illegal}$ is convenient when considering the restriction to the subspace $\mathcal{I}(G_{1},\gxoc,n)$. Letting $\Pi_{0}$ be the projector onto $\mathcal{I}(G_{1},\gxoc,n)$, the following Lemma (proven in \sec{Proof-of-Lemma Pi0_restriction}) shows that the restriction 
\begin{equation}
\Pi_{0}\big|_{\spn(\mathcal{B}_{n})}\label{eq:restrictionBn}
\end{equation}
is diagonal in the basis $\mathcal{B}_n$. The Lemma also bounds the diagonal entries.

\begin{restatable}{lemma}{restrictionlemma}
\label{lem:Pi_0_restriction}
Let $\Pi_{0}$ be the projector onto $\mathcal{I}(G_{1},\gxoc,n)$. For any $|j,\vec{d},\vec{z},\vec{a}\rangle\in\mathcal{B}_{\legal}$, we have
\begin{align}
\Pi_{0}|j,\vec{d},\vec{z},\vec{a}\rangle & =|j,\vec{d},\vec{z},\vec{a}\rangle\label{eq:subspace_eqn0}.
\end{align}
Furthermore, for any two distinct basis vectors $|\phi\rangle,|\psi\rangle\in\mathcal{B}_{\illegal}$, we have 
\begin{align}
\langle\phi|\Pi_{0}|\phi\rangle & \leq\frac{255}{256}\label{eq:subspace_eqn1}\\
\langle\phi|\Pi_{0}|\psi\rangle & =0.\label{eq:subspace_eqn2}
\end{align}
\end{restatable}

We use this Lemma to characterize the nullspace of $H(G_{1},\gxoc,n)$ and bound its smallest nonzero eigenvalue.

\begin{lemma}
\label{lem:H_Galpha_Gtilde}The nullspace $S_1$ of $H(G_{1},\gxoc,n)$ is spanned by the orthonormal basis $\mathcal{B}_{\legal}$. Its smallest nonzero eigenvalue is
\begin{equation}
  \gamma(H(G_{1},\gxoc,n)) \geq \frac{\mathcal{K}_0}{256} \label{eq:G_alpha_lowerbnd}
\end{equation}
where $\mathcal{K}_0\in (0,1]$ is the absolute constant from \lem{gs_g_alpha}.
\end{lemma}

\begin{proof}
Recall from \sec{The-occupancy-constraints} that 
\[
H(G_{1},\gxoc,n)=H(G_{1},n)|_{\mathcal{I}(G_{1},\gxoc,n)}.
\]
Its nullspace is the space of states $|\kappa\rangle$ satisfying
\[
\Pi_{0}|\kappa\rangle=|\kappa\rangle\quad\text{and}\quad H(G_{1},n)|\kappa\rangle=0
\]
(recall that $\Pi_{0}$ is the projector onto $\mathcal{I}(G_{1},\gxoc,n)$, the states satisfying the occupancy constraints). Since $\mathcal{B}_{n}$ is a basis for the nullspace of $H(G_{1},n)$, to solve for the nullspace of $H(G_{1},\gxoc,n)$ we consider the restriction \eq{restrictionBn} and solve for the eigenspace with eigenvalue $1$. This calculation is simple because \eq{restrictionBn} is diagonal in the basis $\mathcal{B}_{n}$, according to \lem{Pi_0_restriction}. We see immediately from the Lemma that $\mathcal{\mathcal{B}_{\legal}}$ spans the nullspace of $H(G_{1},\gxoc,n)$; we now show that \lem{Pi_0_restriction} also implies the lower bound \eq{G_alpha_lowerbnd}. Note that
\[
\gamma(H(G_{1},\gxoc,n))=\gamma(\Pi_{0}H(G_{1},n)\Pi_{0}).
\]
Let $\Pi_{\legal}$ and $\Pi_{\illegal}$ project onto the spaces spanned by $\mathcal{B}_{\legal}$ and $\mathcal{B}_{\illegal}$ respectively, so $\Pi_{\legal}+\Pi_{\illegal}$ projects onto the nullspace of $H(G_{1},n)$. The operator inequality
\[
H(G_{1},n)\geq\gamma(H(G_{1},n))\cdot\left(1-\Pi_{\legal}-\Pi_{\illegal}\right)
\]
implies 
\[
\Pi_{0}H(G_{1},n)\Pi_{0}\geq\gamma(H(G_{1},n))\cdot\Pi_{0}(1-\Pi_{\legal}-\Pi_{\illegal})\Pi_{0}.
\]
Since the operators on both sides of this inequality are positive semidefinite and have the same nullspace, their smallest nonzero eigenvalues are bounded as 
\[
\gamma(\Pi_{0}H(G_{1},n)\Pi_{0})\geq\gamma(H(G_{1},n))\cdot\gamma(\Pi_{0}(1-\Pi_{\legal}-\Pi_{\illegal})\Pi_{0}).
\]
Hence 
\begin{equation}
\gamma(H(G_{1},\gxoc,n))=\gamma(\Pi_{0}H(G_{1},n)\Pi_{0}) \geq \mathcal{K}_0\cdot\gamma(\Pi_{0}(1-\Pi_{\legal}-\Pi_{\illegal})\Pi_{0})
\label{eq:gamma_bnd1}
\end{equation}
where we used \lem{gs_g_alpha}. From equations \eq{subspace_eqn1} and \eq{subspace_eqn2} we see that 
\begin{equation}
\Pi_{0}|g\rangle=|g\rangle\quad\text{and}\quad\Pi_{\illegal}|f\rangle=|f\rangle\qquad\Longrightarrow\qquad\langle f|g\rangle\langle g|f\rangle\leq\frac{255}{256}.\label{eq:fg_eqn}
\end{equation}
The nullspace of 
\begin{equation}
\Pi_{0}\left(1-\Pi_{\legal}-\Pi_{\illegal}\right)\Pi_{0}\label{eq:projector_conjugated}
\end{equation}
is spanned by 
\[
\mathcal{B}_{\legal}\cup\left\{ |\tau\rangle\colon\Pi_{0}|\tau\rangle=0\right\} .
\]
To see this, note that \eq{projector_conjugated} commutes with $\Pi_{0}$, and the space of $+1$ eigenvectors of $\Pi_{0}$ that are annihilated by \eq{projector_conjugated} is spanned by $\mathcal{B}_{\legal}$ (by \lem{Pi_0_restriction}). Any eigenvector $|g_{1}\rangle$ corresponding to the smallest nonzero eigenvalue of this operator therefore satisfies $\Pi_{0}|g_{1}\rangle=|g_{1}\rangle$ and $\Pi_{\legal}|g_{1}\rangle=0$, so 
\begin{align*}
\gamma(\Pi_{0}(1-\Pi_{\legal}-\Pi_{\illegal})\Pi_{0}) 
&= 1-\langle g_{1}|\Pi_{\mathcal{\illegal}}|g_{1}\rangle\geq\frac{1}{256}
\end{align*}
using equation \eq{fg_eqn}. Plugging this into equation \eq{gamma_bnd1} gives the lower bound \eq{G_alpha_lowerbnd}. 
\end{proof}

\subsection{Matrix elements between states with legal configurations}
\label{sec:Matrix-elements-in}

We now consider 
\begin{equation}
H_{1}|_{S_{1}},H_{2}|_{S_{1}},H_{\inn,i}|_{S_{1}},H_{\out}|_{S_{1}}\label{eq:ops_restriction_S1}
\end{equation}
where these operators are defined in \sec{Building-up-the} and 
\[
S_{1}=\spn(\mathcal{B}_{\legal})
\]
is the nullspace of $H(G_1,\gxoc,n)$.

We specify the operators \eq{ops_restriction_S1} by their matrix elements in an orthonormal basis for $S_{1}$. Although the basis $\mathcal{B}_{\legal}$ was convenient in \sec{Legal-configurations-and}, here we use a different basis in which the matrix elements of $H_{1}$ and $H_{2}$ are simpler. We define
\begin{equation}
  |j,\vec{d},\In(\vec{z}),\vec{a}\rangle
  =\sum_{\vec{x}\in\{0,1\}^{n}} \big(\langle\vec{x} 
   |\bar{U}_{j,d_1}(a_{1}) 
   |\vec{z}\rangle\big)|j,\vec{d},\vec{x},\vec{a}\rangle 
  \label{eq:phi_init_basis}
\end{equation}
where
\begin{equation}
  \bar{U}_{j,d_1}(a_1) = \begin{cases}
    U_{j-1}(a_1) U_{j-2}(a_1) \ldots U_1(a_1) & \text{if $d_1 \in \{0,2\}$} \\
    U_j(a_1) U_{j-1}(a_1) \ldots U_1(a_1) & \text{if $d_1=1$}.
  \end{cases}
  \label{eq:ubar}
\end{equation}
In each of these states the quantum data (represented by the $\vec{x}$ register on the right-hand side) encodes the computation in which the unitary $\bar{U}_{j,d_1}(a_1)$ is applied to the initial $n$-qubit state $|\vec{z}\rangle$ (the notation $\In(\vec{z})$ indicates that $\vec{z}$ is the input). The vector $\vec{a}$ is only relevant insofar as its first bit $a_{1}$ determines whether or not each two-qubit unitary is complex conjugated; the other bits of $\vec{a}$ go along for the ride. Letting $\vec{z},\vec{a}\in\{0,1\}^{n}$ , $j\in[M]$, and \[ \vec{d}=(d_{1},\ldots,d_{n})\quad\text{with}\quad d_{1}=d_{s(j)}\in\{0,1,2\}\quad\text{and}\quad d_{i}\in\{0,1\},\; i\notin\{1,s(j)\}, \] we see that the states \eq{phi_init_basis} form an orthonormal basis for $S_{1}$. In \sec{matrix_els_details} we compute the matrix elements of the operators \eq{ops_restriction_S1} in this basis, which are reproduced below.

Roughly speaking, the nonzero off-diagonal matrix elements of the operator $H_{1}$ in the basis \eq{phi_init_basis} occur between states $|j,\vec{d},\In(\vec{z}),\vec{a}\rangle$ and $|j,\vec{c},\In(\vec{z}),\vec{a}\rangle$ where the legal configurations $j,\vec{d}$ and $j,\vec{c}$ are related by horizontal motion of a particle in one of the rows $i\in\{2,\ldots,n\}$. 

\begin{mdframed}[frametitle=Matrix elements of $H_{1}$]
\begin{equation}
\langle k,\vec{c},\In(\vec{x}),\vec{b}|H_{1}|j,\vec{d},\In(\vec{z}),\vec{a}\rangle=\delta_{k,j}\delta_{\vec{a},\vec{b}}\delta_{\vec{z},\vec{x}}\cdot\begin{cases}
\frac{n-1}{64} & \vec{c}=\vec{d}\\
\frac{1}{64}\displaystyle\prod_{\substack{r=1\\r\neq i}}^n
\delta_{d_{r},c_{r}} & d_{i}\neq c_{i}\;\text{for some}\; i\in[n]\setminus\{1,s(j)\}\\
\frac{1}{64\sqrt{2}} \displaystyle\prod_{\substack{r=2\\r\neq s(j)}}^n \delta_{d_{r},c_{r}} & (c_{1},d_{1})\in\{(2,0),(0,2),(1,2),(2,1)\}\\
0 & \text{otherwise.}
\end{cases}\label{eq:H1_goodbasis-1}
\end{equation}
\end{mdframed}

\newcommand{\configtrans}[2]{
\configex{#1}
$~\longrightarrow~$
\configex{#2}
}
\begin{figure}
\subfloat[][$\langle j,\vec{c},\In(\vec{z}),\vec{a}|H_1|j,\vec{d},\In(\vec{z}),\vec{a}\rangle
    =\frac{1}{64}$; here $j=2$ and $\vec{d}=(2,2,0) \rightarrow \vec{c}=(2,2,1)$.]{\label{fig:H1_InbasisA}
\configtrans{
\node at (2.618+.31+1.7,.75) {\huge $\times$};
\foreach \yoffset in {.2,-.2}{
  \node at (2*2.618+.31+1.1,4+\yoffset) {\large $\times$};
}
\draw[line width=.1mm] (2*2.618+.31+1.1,4) ellipse (.25 cm and .45 cm);
}{
\foreach \yoffset in {.2,-.2}{
  \node at (2*2.618+.31+1.1,4+\yoffset) {\large $\times$};
}
\draw[line width=.1mm] (2*2.618+.31+1.1,4) ellipse (.25 cm and .45 cm);
\node at (3*2.618+.31+.3,.75) {\huge $\times$};
}}

\subfloat[][$\langle j,\vec{c},\In(\vec{z}),\vec{a}|H_1|j,\vec{d},\In(\vec{z}),\vec{a}\rangle
   =\frac{1}{64\sqrt{2}}$; here $j=3$ and $\vec{d}=(0,0,0) \rightarrow \vec{c}=(2,0,2)$. ]{\label{fig:H1_InbasisB}
\configtrans{
\node at (2.618+.31+1.7,.75) {\huge $\times$};
\node at (2*2.618+.31+1.7,2.75) {\huge $\times$};
\node at (3*2.618+.31+.3,5.06) {\huge $\times$};
}{
\node at (2*2.618+.31+1.7,2.75) {\huge $\times$};
\foreach \yoffset in {.2,-.2}{
  \node at (3*2.618+.31+1.1,3+\yoffset) {\large $\times$};
}
\draw[line width=.1mm] (3*2.618+.31+1.1,3) ellipse (.25 cm and .45 cm);
}}

\subfloat[][$\langle j,\vec{c},\In(\vec{z}),\vec{a}|H_1|j,\vec{d},\In(\vec{z}),\vec{a}\rangle
  =\frac{1}{64\sqrt{2}}$; here $j=1$ and $\vec{d}=(1,1,1) \rightarrow \vec{c}=(2,1,2)$.]{\label{fig:H1_InbasisC}
\configtrans{
\node at (2.618+.31+1.7,5.06) {\huge $\times$};
\node at (2*2.618+.31+.3,2.75) {\huge $\times$};
\node at (3*2.618+.31+.3,.75) {\huge $\times$};
}{
\foreach \yoffset in {.2,-.2}{
  \node at (2.618+.31+1.1,3+\yoffset) {\large $\times$};
}
\draw[line width=.1mm] (2.618+.31+1.1,3) ellipse (.25 cm and .45 cm);
\node at (2*2.618+.31+.3,2.75) {\huge $\times$};
}}

\caption{
Examples of matrix elements of $H_1$ in the basis \eq{phi_init_basis} of $S_1$. The relevant matrix elements (as indicated above) are computed using \subfig{H1_InbasisA} the second case and \subfig{H1_InbasisB}, \subfig{H1_InbasisC} the third case in equation \eq{H1_goodbasis-1}. \label{fig:H1_Inbasis}
}
\end{figure}

From this expression we see that $H_{1}|_{S_1}$ is block diagonal in the basis \eq{phi_init_basis}, with a block for each $\vec{z},\vec{a}\in\{0,1\}^{n}$ and $j\in[M]$. Moreover, the submatrix for each block is the same. In \fig{H1_Inbasis} we illustrate some of the off-diagonal matrix elements of $H_{1}|_{S_{1}}$ for the example from \fig{step-by-step}.

Next, we present the matrix elements of $H_2$.

\begin{mdframed}[frametitle=Matrix elements of $H_{2}$]
\begin{align}
\langle k,\vec{c},\In(\vec{x}),\vec{b}|H_{2}|j,\vec{d},\In(\vec{z}),\vec{a}\rangle & =f_{\mathrm{diag}(\vec{d},j)\cdot\delta_{j,k}\delta_{\vec{a},\vec{b}}\delta_{\vec{z},\vec{x}}\delta_{\vec{c},\vec{d}}}\label{eq:H2_formula_with_f}\\
 & \quad +\big(f_{\textrm{off-diag}}(\vec{c},\vec{d},j)\cdot\delta_{k,j-1}+f_{\textrm{off-diag}}(\vec{d},\vec{c},k)\cdot\delta_{k-1,j}\big)\delta_{\vec{a},\vec{b}}\delta_{\vec{z},\vec{x}}\nonumber 
\end{align}
where
\begin{equation}
f_{\mathrm{diag}}(\vec{d},j)=\begin{cases}
0 & d_{1}=0\text{ and }j=1,\text{ or }d_{1}=1\text{ and }j=M\\
\frac{1}{128} & d_{1}=2\text{ and }j\in\{1,M\}\\
\frac{1}{64} & \text{otherwise}
\end{cases}\label{eq:H2_diag_goodbasis-1}
\end{equation}
and 
\begin{equation}
f_{\textrm{off-diag}}(\vec{c},\vec{d},j)=\left(
\prod_{\substack{r=2\\r\notin\{s(j-1),s(j)\}}}^n \mkern-36mu\delta_{d_{r},c_{r}}\right)\cdot\begin{cases}
\frac{1}{64\sqrt{2}} & (c_{1},c_{s(j)},d_{1},d_{s(j-1)})\in\{(2,0,0,0),(1,1,2,1)\}\\
\frac{1}{64} & (c_{1},c_{s(j)},d_{1},d_{s(j-1)})=(1,0,0,1)\\
\frac{1}{128} & (c_{1},c_{s(j)},d_{1},d_{s(j-1)})=(2,1,2,0)\\
0 & \text{otherwise.}
\end{cases}\label{eq:H2_offdiag_goodbasis-1}
\end{equation}
\end{mdframed}

\begin{figure}
\subfloat[][$\langle j-1,\vec{c},\In(\vec{z}),\vec{a}|H_2|j,\vec{d},\In(\vec{z}),\vec{a}\rangle
  =\frac{1}{64\sqrt{2}}$; here $j=4$ and $\vec{d}=(0,0,0)\rightarrow\vec{c}=(2,0,2)$]{\label{fig:H2_InbasisA}
\configtrans{
\node at (2*2.618+.31+1.7,2.75) {\huge $\times$};
\node at (3*2.618+.31+1.7,.75) {\huge $\times$};
\node at (4*2.618+.31+.3,5.06) {\huge $\times$};
}{
\node at (2*2.618+.31+1.7,2.75) {\huge $\times$};
\foreach \yoffset in {.2,-.2}{
  \node at (3*2.618+.31+1.1,3+\yoffset) {\large $\times$};
}
\draw[line width=.1mm] (3*2.618+.31+1.1,3) ellipse (.25 cm and .45 cm);
}}

\subfloat[][
  $\langle j-1,\vec{c},\In(\vec{z}),\vec{a}|H_2|j,\vec{d},\In(\vec{z}),\vec{a}\rangle
  =\frac{1}{64 \sqrt{2}}$; here $j=3$ and $\vec{d}=(2,1,2)\rightarrow\vec{c}=(1,1,1)$]{\label{fig:H2_InbasisB}
\configtrans{
\foreach \yoffset in {.2,-.2}{
  \node at (3*2.618+.31+1.1,3+\yoffset) {\large $\times$};
}
\draw[line width=.1mm] (3*2.618+.31+1.1,3) ellipse (.25 cm and .45 cm);
\node at (4*2.618+.31+.3,2.75) {\huge $\times$};
}{
\node at (2*2.618+.31+1.7,5.06) {\huge $\times$};
\node at (3*2.618+.31+.3,.75) {\huge $\times$};
\node at (4*2.618+.31+.3,2.75) {\huge $\times$};
}}

\subfloat[][
  $\langle j-1,\vec{c},\In(\vec{z}),\vec{a}|H_2|j,\vec{d},\In(\vec{z}),\vec{a}\rangle
  =\frac{1}{64}$; here $j=3$ and $\vec{d}=(0,1,0)\rightarrow\vec{c}=(1,1,0)$]{
\label{fig:H2_InbasisC}
\configtrans{
\node at (2.618+.31+1.7,.75) {\huge $\times$};
\node at (3*2.618+.31+.3,5.06) {\huge $\times$};
\node at (4*2.618+.31+.3,2.75) {\huge $\times$};
}{
\node at (2.618+.31+1.7,.75) {\huge $\times$};
\node at (2*2.618+.31+1.7,5.06) {\huge $\times$};
\node at (4*2.618+.31+.3,2.75) {\huge $\times$};
}}

\subfloat[][
  $\langle j-1,\vec{c},\In(\vec{z}),\vec{a}|H_2|j,\vec{d},\In(\vec{z}),\vec{a}\rangle
  =\frac{1}{128}$; here $j=2$ and $\vec{d}=(2,2,0)\rightarrow\vec{c}=(2,1,2)$]{
\label{fig:H2_InbasisD}
\configtrans{
\node at (2.618+.31+1.7,.75) {\huge $\times$};
\foreach \yoffset in {.2,-.2}{
  \node at (2*2.618+.31+1.1,4+\yoffset) {\large $\times$};
}
\draw[line width=.1mm] (2*2.618+.31+1.1,4) ellipse (.25 cm and .45 cm);
}{
\foreach \yoffset in {.2,-.2}{
  \node at (2.618+.31+1.1,4+\yoffset) {\large $\times$};
}
\draw[line width=.1mm] (2.618+.31+1.1,4) ellipse (.25 cm and .45 cm);
\node at (2*2.618+.31+.3,2.75) {\huge $\times$};
}}
\caption{
Examples of matrix elements of $H_2$ in the basis \eq{phi_init_basis} of $S_1$. The relevant matrix elements (as indicated above) are computed using the \subfig{H2_InbasisA}, \subfig{H2_InbasisB} first, \subfig{H2_InbasisC} second, and \subfig{H2_InbasisD} third cases in equation \eq{H2_offdiag_goodbasis-1}. \label{fig:H2_Inbasis}}
\end{figure}

This shows that $H_{2}|_{S_{1}}$ is block diagonal in the basis \eq{phi_init_basis}, with a block for
each $\vec{z},\vec{a}\in\{0,1\}^{n}$. Also note that (in contrast
with $H_{1}$) $H_{2}$ connects states with different values of $j$.
In \fig{H2_Inbasis} we illustrate some of the off-diagonal
matrix elements of $H_{2}|_{S_{1}}$, for the example from \fig{step-by-step}.

Finally, we present the matrix elements of $H_{\inn,i}$ (for
$i\in\{n_{\inn}+1,\ldots,n\}$) and $H_{\out}$:

\begin{mdframed}[frametitle=Matrix elements of $H_{\inn,i}$]
For each ancilla qubit $i\in\{n_{\inn}+1,\ldots,n\}$, define  $j_{\min,i}=\min\left\{ j\in[M]\colon s(j)=i\right\}$ to be the index of the first gate in the circuit that involves this qubit (recall from \sec{From-circuits-to} that we consider circuits where each ancilla qubit is involved in at least one gate). The operator $H_{\inn,i}$ is diagonal in the basis \eq{phi_init_basis}, with entries
\begin{equation}
\langle j,\vec{d},\In(\vec{z}),\vec{a}|H_{\inn,i}|j,\vec{d},\In(\vec{z}),\vec{a}\rangle=\begin{cases}
\frac{1}{64} & j\leq j_{\min,i},\text{ }z_{i}=1,\text{ and }d_{i}=0\\
0 & \text{otherwise.}
\end{cases}\label{eq:Hin_mat_els}
\end{equation}
\end{mdframed}

\begin{mdframed}[frametitle=Matrix elements of $H_{\out}$]
Let $j_{\max}=\max\{j\in[M]\colon s(j)=2\}$ be the index of the last gate $U_{j_{\max}}$ in the circuit that acts between qubits $1$ and $2$ (the output qubit). Then
\begin{equation}
\langle k,\vec{c},\In(\vec{x}),\vec{b}|H_{\out}|j,\vec{d},\In(\vec{z}),\vec{a}\rangle=\delta_{j,k}\delta_{\vec{c},\vec{d}}\delta_{\vec{a},\vec{b}}\begin{cases}
\langle\vec{x}|U_{\mathcal{C}_X}^{\dagger}(a_{1})|0\rangle\langle0|_{2}U_{\mathcal{C}_X}(a_{1})|\vec{z}\rangle\frac{1}{64} & j\geq j_{\max}\text{ and }d_{2}=1\\
0 & \text{otherwise}.
\end{cases}\label{eq:hout_matels-1}
\end{equation}
\end{mdframed}

\subsection{From $H(G_{2},\gxoc,n)$ to $H(G_{4},\gxoc,n)$}
\label{sec:Frustration-free-states}

Define the $(n-2)$-dimensional hypercubes
\[
\mathcal{D}_{k}^{j}=\left\{ (d_{1},\ldots,d_{n})\colon d_{1}=d_{s(j)}=k,\, d_{i}\in\{0,1\}\text{ for }i\in[n]\setminus\{1,s(j)\}\right\} 
\]
for $j\in\{1,\ldots,M\}$ and $k\in\{0,1,2\}$, and the superpositions
\[
|{\Cube_{k}(j,\vec{z},\vec{a})}\rangle=\frac{1}{\sqrt{2^{n-2}}}\sum_{\vec{d}\in\mathcal{D}_{k}^{j}}\left(-1\right)^{\sum_{i=1}^{n}d_{i}}|j,\vec{d},\In(\vec{z}),\vec{a}\rangle
\]
for $k\in\{0,1,2\}$, $j\in[M]$, and $\vec{z},\vec{a}\in\{0,1\}^{n}.$
For each $j\in[M]$ and $\vec{z},\vec{a}\in\{0,1\}^{n}$, let
\begin{equation}
|C(j,\vec{z},\vec{a})\rangle=\frac{1}{2}|{\Cube_{0}(j,\vec{z},\vec{a})}\rangle+\frac{1}{2}|{\Cube_{1}(j,\vec{z},\vec{a})}\rangle-\frac{1}{\sqrt{2}}|{\Cube_{2}(j,\vec{z},\vec{a})}\rangle.\label{eq:cubestates-1}
\end{equation}
We prove

\begin{lemma}
\label{lem:beta_bound}The Hamiltonian $H(G_{2},\gxoc,n)$ has nullspace $S_2$ spanned by the states
\[
|C(j,\vec{z},\vec{a})\rangle
\]
for $j\in[M]$ and $\vec{z},\vec{a}\in\{0,1\}^{n}$. Its smallest nonzero eigenvalue is 
\[
  \gamma(H(G_{2},\gxoc,n)) \geq \frac{\mathcal{K}_0}{35000n}
\]
where $\mathcal{K}_0\in (0,1]$ is the absolute constant from \lem{gs_g_alpha}.
\end{lemma}

\begin{proof}
Recall from the previous section that $H_{1}|_{S_1}$ is block diagonal in the basis \eq{phi_init_basis}, with a block for each $j\in[M]$ and $\vec{z},\vec{a}\in\{0,1\}^{n}$. That is to say, $\langle k,\vec{c},\In(\vec{x}),\vec{b}|H_{1}|j,\vec{d},\In(\vec{z}),\vec{a}\rangle$ is zero unless $\vec{a}=\vec{b}$, $k=j$, and $\vec{z}=\vec{x}$. Equation \eq{H1_goodbasis-1} gives the nonzero matrix elements within a given block, which we use to compute the frustration-free ground states of $H_{1}|_{S_{1}}$.

Looking at equation \eq{H1_goodbasis-1}, we see that the matrix for each block can be written as a sum of $n$ commuting matrices: $\frac{n-1}{64}$ times the identity matrix (case 1 in equation \eq{H1_goodbasis-1}), $n-2$ terms that each flip a single bit $i\notin\{1,s(j)\}$ of $\vec{d}$ (case 2), and a term that changes the value of the ``special'' components $d_{1}=d_{s(j)}\in\{0,1,2\}$ (case 3). Thus
\[
\langle j,\vec{c},\In(\vec{z}),\vec{a}|H_{1}|j,\vec{d},\In(\vec{z}),\vec{a}\rangle=\langle j,\vec{c},\In(\vec{z}),\vec{a}|\frac{1}{64}(n-1)+\frac{1}{64}\sum_{i\in[n]\setminus\{1,s(j)\}}H_{\mathrm{flip},i}+\frac{1}{64}H_{\mathrm{special},j}|j,\vec{d},\In(\vec{z}),\vec{a}\rangle
\]
where
\begin{align*}
\langle j,\vec{c},\In(\vec{z}),\vec{a}|H_{\mathrm{flip},i}|j,\vec{d},\In(\vec{z}),\vec{a}\rangle &= \delta_{c_{i},d_{i}\oplus1} \prod_{r\in[n]\setminus\{i\}}\delta_{c_{r},d_{r}}
\end{align*}
and 
\[
\langle j,\vec{c},\In(\vec{z}),\vec{a}|H_{\mathrm{special},j}|j,\vec{d},\In(\vec{z}),\vec{a}\rangle=\begin{cases}
\frac{1}{\sqrt{2}} & \begin{aligned}
  &(c_{1},d_{1})\in\{(2,0),(0,2),(1,2),(2,1)\}\\ 
  &\text{and}\; d_{r}=c_{r}\;\text{for}\; r\in [n]\setminus\{1,s(j)\}
\end{aligned}\\
0 & \text{otherwise.}
\end{cases}
\]
Note that these $n$ matrices are mutually commuting, each eigenvalue of $H_{\mathrm{flip},i}$ is $\pm1$, and each eigenvalue of $H_{\mathrm{special},j}$ is equal to one of the eigenvalues of the matrix 
\[
\frac{1}{\sqrt{2}}\begin{pmatrix}
0 & 0 & 1\\
0 & 0 & 1\\
1 & 1 & 0
\end{pmatrix},
\]
which are $\{-1,0,1\}$. Thus we see that the eigenvalues of $H_1|_{S_{1}}$ within a block for some $j\in [M]$ and $\vec{z},\vec{a}\in\{0,1\}^{n}$ are
\begin{equation}
  \frac{1}{64}\bigg(n-1+\sum_{i\notin\{1,s(j)\}}y_{i}+w\bigg)
  \label{eq:block_eig}
\end{equation}
where $y_{i}\in\pm1$ for each $i\in[n]\setminus\{1,s(j)\}$ and $w\in\{-1,0,1\}$. In particular, the smallest eigenvalue within the block is zero (corresponding to $y_i=w=-1$). The corresponding eigenspace is spanned by the simultaneous $-1$ eigenvectors of each $H_{\mathrm{flip},i}$ for $i\in [n]\setminus \{1,s(j)\}$ and $H_{\mathrm{special},j}$. The space of simultaneous $-1$ eigenvectors of $H_{\mathrm{flip},i}$ for $i\in [n]\setminus\{1,s(j)\}$ within the block is spanned by $\{|{\Cube_{0}(j,\vec{z},\vec{a})}\rangle, |{\Cube_{1}(j,\vec{z},\vec{a})}\rangle, |{\Cube_{2}(j,\vec{z},\vec{a})}\rangle\}$. The state $|C(j,\vec{z},\vec{a})\rangle$  is the unique superposition of these states that is a $-1$ eigenvector of $H_{\mathrm{special},j}$. Hence, for each block we obtain a unique state $|C(j,\vec{z},\vec{a})\rangle$ in the space $S_2$. Ranging over all blocks $j\in [M]$ and $\vec{z},\vec{a}\in \{0,1\}^n$, we get the basis described in the Lemma.

The smallest nonzero eigenvalue within each block is $\frac{1}{64}$ (corresponding to $y_{i}=-1$ and $w=0$ in equation \eq{block_eig}), so 
\begin{equation}
\gamma(H_{1}|_{S_{1}})=\frac{1}{64}.\label{eq:lowerbnd_H1S_alph}
\end{equation}

To get the stated lower bound, we use \lem{npl} with $H(G_{2},\gxoc,n) = H_{A}+H_{B}$ where 
\[
H_{A}=H(G_{1},\gxoc,n)\qquad H_{B}=H_{1}|_{\mathcal{I}(G_{2},\gxoc,n)}
\]
(as in equation \eq{H_G2}), along with the bounds 
\[
\gamma(H_{A}) \geq \frac{\mathcal{K}_0}{256}\qquad\gamma(H_{B}|_{S_{1}})=\gamma(H_{1}|_{S_{1}})=\frac{1}{64}\qquad\left\Vert H_{B}\right\Vert \leq\left\Vert H_{1}\right\Vert \leq n\left\Vert h_{1}\right\Vert =2n
\]
from \lem{H_Galpha_Gtilde}, equations \eq{H1_restriction} and \eq{lowerbnd_H1S_alph}, and the fact that $\left\Vert h_{1}\right\Vert =2$ from \eq{h_E_bound}. This gives
\[
  \gamma(H(G_2,\gxoc,n))
  \geq \frac{\mathcal{K}_0}{64\mathcal{K}_0+256+2n\cdot64\cdot256}
  \geq \frac{\mathcal{K}_0}{35000n}
\]
where we used the facts that $\mathcal{K}_0\leq 1$ and $n \ge 1$.
\end{proof}
For each $\vec{z},\vec{a}\in\{0,1\}^{n}$ define the uniform superposition
\[
  |\mathcal{H}(\vec{z},\vec{a})\rangle
  =\frac{1}{\sqrt{M}}\sum_{j=1}^{M}|C\left(j,\vec{z},\vec{a}\right)\rangle
\]
that encodes (somewhat elaborately) the history of the computation that consists of applying either $U_{\mathcal{C}_X}$ or $U_{\mathcal{C}_X}^{*}$ to the state $|\vec{z}\rangle$. The first bit of $\vec{a}$ determines whether the circuit $\mathcal{C}_{X}$ or its complex conjugate is applied. 

\begin{lemma}
\label{lem:HG3}
The Hamiltonian $H(G_{3},\gxoc,n)$ has nullspace $S_3$ spanned by the states
\[
  |\mathcal{H}(\vec{z},\vec{a})\rangle
\]
for $\vec{z},\vec{a}\in\{0,1\}^{n}$. Its smallest nonzero eigenvalue is 
\[
  \gamma(H(G_{3},\gxoc,n)) \geq \frac{\mathcal{K}_0}{10^7 n^{2}M^{2}}
\]
where $\mathcal{K}_0\in (0,1]$ is the absolute constant from \lem{gs_g_alpha}.
\end{lemma}

\begin{proof}
Recall that
\[
H(G_{3},\gxoc,n)=H(G_{2},\gxoc,n)+H_{2}|_{\mathcal{I}(G_{3},\gxoc,n)}
\]
with both terms on the right-hand side positive semidefinite. To solve for the nullspace of $H(G_{3},\gxoc,n)$, it suffices to restrict our attention to the space
\begin{equation}
  S_{2}=\spn\{ |C(j,\vec{z},\vec{a})\rangle\colon
    j\in[M],\;\vec{z},\vec{a}\in\{0,1\}^{n}\}
  \label{eq:S_beta_basis}
\end{equation}
of states in the nullspace of $H(G_{2},\gxoc,n)$. We begin by computing the matrix elements of $H_{2}$ in the basis for $S_{2}$ given above. We use equations \eq{H2_formula_with_f} and \eq{cubestates-1} to compute the diagonal matrix elements: 
\begin{align}
\langle C\left(j,\vec{z},\vec{a}\right)|H_{2}|C\left(j,\vec{z},\vec{a}\right)\rangle= & \frac{1}{4}\langle{\Cube_{0}(j,\vec{z},\vec{a})}|H_{2}|{\Cube_{0}(j,\vec{z},\vec{a})}\rangle+\frac{1}{4}\langle{\Cube_{1}(j,\vec{z},\vec{a})}\rangle|H_{2}|{\Cube_{1}(j,\vec{z},\vec{a})}\rangle\\
 & +\frac{1}{2}\langle{\Cube_{2}(j,\vec{z},\vec{a})}|H_{2}|{\Cube_{2}(j,\vec{z},\vec{a})}\rangle\\
= & \begin{cases}
0+\frac{1}{256}+\frac{1}{256} & j=1\\
\frac{1}{256}+\frac{1}{256}+\frac{1}{128} & j\in\{2,\ldots,M-1\}\\
\frac{1}{256}+0+\frac{1}{256} & j=M
\end{cases}\\
= & \begin{cases}
\frac{1}{128} & j\in\{1,M\} \\
\frac{1}{64} & j\in\{2,\ldots,M-1\}.
\end{cases}\label{eq:diag_C_mat_els}
\end{align}
In the second line we used equation \eq{H2_diag_goodbasis-1}. Looking at equation \eq{H2_formula_with_f}, we see that the only nonzero off-diagonal matrix elements of $H_{2}$ in this basis are of the form 
\[
\langle C(j-1,\vec{z},\vec{a})|H_{2}|C(j,\vec{z},\vec{a})\rangle\quad
\text{or}\quad\langle C(j,\vec{z},\vec{a})|H_{2}|C(j-1,\vec{z},\vec{a})\rangle
=\langle C(j-1,\vec{z},\vec{a})|H_{2}|C(j,\vec{z},\vec{a})\rangle^{*}
\]
for $j\in\{2,\ldots,M\}$, $\vec{z},\vec{a}\in\{0,1\}^{n}$. To compute these matrix elements we first use equation \eq{H2_offdiag_goodbasis-1} to evaluate 
\[
\langle{\Cube_{w}(j-1,\vec{z},\vec{a})}|H_{2}|{\Cube_{v}(j,\vec{z},\vec{a})}\rangle
\]
for $v,w\in\{0,1,2\}$ and $j\in\{2,\ldots,M\}$. For example, using the second case of equation \eq{H2_offdiag_goodbasis-1}, we get
\begin{align*}
\langle{\Cube_{1}\left(j-1,\vec{z},\vec{a}\right)}|H_{2}|{\Cube_{0}\left(j,\vec{z},\vec{a}\right)}\rangle & = 
 \frac{1}{2^{n-2}}\sum_{\vec{d}\in\mathcal{D}_{0}^{j}}\; \sum_{\vec{c}\in\mathcal{D}_{1}^{j-1}} \!\!(-1)^{\sum_{i\in[n]} (c_i +d_i)}\langle j-1,\vec{c},\text{In}(\vec{z}),\vec{a}|H_2|j,\vec{d},\text{In}(\vec{z}),\vec{a}\rangle \\
& = \frac{1}{2^{n-2}}\sum_{\vec{d}\in\mathcal{D}_{0}^{j}:d_{s(j-1)}=1}\!\!(-1)\cdot\frac{1}{64}=-\frac{1}{128}.
\end{align*}
To go from the first to the second line we used the fact that, for each $\vec{d}\in\mathcal{D}_{0}^{j}$ with $d_{s(j-1)}=1$, there is one $\vec{c}\in\mathcal{D}_{1}^{j-1}$ for which $\langle j-1,\vec{c},\text{In}(\vec{z}),\vec{a}|H_2|j,\vec{d},\text{In}(\vec{z}),\vec{a}\rangle=\frac{1}{64}$ (with all other such matrix elements equal to zero). This $\vec{c}$ satisfies $c_1=c_{s(j-1)}=1$ and $c_{s(j)}=0$, with all other bits equal to those of $\vec{d}$, so 
\[
(-1)^{\sum_{i=1}^{n}\left(c_{i}+d_{i}\right)}=(-1)^{c_{1}+c_{s(j)}+c_{s(j-1)}+d_{1}+d_{s(j)}+d_{s(j-1)}}=-1
\]
for each nonzero term in the sum.

We perform a similar calculation using cases 1, 3, and 4 in equation \eq{H2_offdiag_goodbasis-1} to obtain
\[
\langle{\Cube_{w}(j-1,\vec{z},\vec{a})}|H_{2}|{\Cube_{v}(j,\vec{z},\vec{a})}\rangle=\begin{cases}
-\frac{1}{128} & (w,v)=(1,0)\\
\frac{1}{128\sqrt{2}} & (w,v)\in\{(2,0),(1,2)\}\\
-\frac{1}{256} & (w,v)=(2,2)\\
0 & \text{otherwise.}
\end{cases}
\]
Hence 
\begin{align*}
 & \langle C\left(j-1,\vec{z},\vec{a}\right)|H_{2}|C\left(j,\vec{z},\vec{a}\right)\rangle\\
 &\quad =\frac{1}{4}\langle{\Cube_{1}\left(j-1,\vec{z},\vec{a}\right)}|H_{2}|{\Cube_{0}\left(j,\vec{z},\vec{a}\right)}\rangle-\frac{1}{2\sqrt{2}}\langle{\Cube_{2}\left(j-1,\vec{z},\vec{a}\right)}|H_{2}|{\Cube_{0}\left(j,\vec{z},\vec{a}\right)}\rangle\\
 &\qquad +\frac{1}{2}\langle{\Cube_{2}\left(j-1,\vec{z},\vec{a}\right)}|H_{2}|{\Cube_{2}\left(j,\vec{z},\vec{a}\right)}\rangle-\frac{1}{2\sqrt{2}}\langle{\Cube_{1}\left(j-1,\vec{z},\vec{a}\right)}|H_{2}|{\Cube_{2}\left(j,\vec{z},\vec{a}\right)}\rangle\\
 &\quad =-\frac{1}{128}.
\end{align*}
Combining this with equation \eq{diag_C_mat_els}, we see that $H_{2}|_{S_{2}}$ is block diagonal in the basis \eq{S_beta_basis}, with a block for each pair of $n$-bit strings $\vec{z},\vec{a}\in\{0,1\}^{n}$. Each of the $2^{2n}$ blocks is equal to the $M\times M$ matrix
\[
  \frac{1}{128}
  \begin{pmatrix}
    1 & -1 & 0 & 0 & \cdots & 0 \mystrut\\
    -1 & 2 & -1 & 0 & \cdots & 0 \mystrut\\
    0 & -1 & 2 & -1 & \ddots & \vdots \\
    0 & 0 & -1 & \ddots & \ddots & 0 \\
    \vdots & \vdots & \ddots & \ddots & 2 & -1 \\
    0 & 0 & \cdots & 0 & -1 & 1 \mystrut
  \end{pmatrix}.
\]
This matrix is just $\frac{1}{128}$ times the Laplacian of a path of length $M$, whose spectrum is well known. In particular, it has a unique eigenvector with eigenvalue zero (the all-ones vector) and its eigenvalue gap is $2(1-\cos (\frac{\pi}{M}))\geq \frac{4}{M^2}$. Thus for each of the $2^{2n}$ blocks there is an eigenvector of $H_{2}|_{S_{2}}$ with eigenvalue $0$, equal to the uniform superposition $|\mathcal{H}(\vec{z},\vec{a})\rangle$ over the $M$ states in the block. Furthermore, the smallest nonzero eigenvalue within each block is at least $\frac{1}{32M^{2}}$. Hence
\begin{equation}
  \gamma(H_{2}|_{S_{2}}) \geq \frac{1}{32M^{2}}.\label{eq:gamma_HB}
\end{equation}
To get the stated lower bound on $\gamma(H(G_{3},\gxoc,n))$, we apply \lem{npl} with 
\begin{align*}
  H_{A} &= H(G_{2},\gxoc,n) \qquad H_{B}=H_{2}|_{\mathcal{I}(G_{3},\gxoc,n)}
\end{align*}
and 
\begin{equation}
  \gamma(H_{A}) \geq \frac{\mathcal{K}_0}{35000 n} \qquad
  \gamma(H_{B}|_{S_{2}}) = \gamma(H_{2}|_{S_{2}}) 
  \geq \frac{1}{32 M^{2}} \qquad
  \|H_{B}\| \leq \|H_{2}\| \leq n \|h_{2}\| = 2n
\label{eq:gamma_HA}
\end{equation}
from \lem{beta_bound}, equation \eq{gamma_HB}, and the fact that $\left\Vert h_{2}\right\Vert =2$ from \eq{h_E_bound}. This gives 
\begin{align*}
  \gamma(H(G_3,\gxoc,n)) 
  &\geq \frac{\mathcal{K}_0}{32M^2\mathcal{K}_0+35000n +2n(35000n)(32M^2)} \\
  &\geq \frac{\mathcal{K}_0}{M^2 n^2(32+35000+70000\cdot32)}
   \geq \frac{\mathcal{K}_0}{10^7 M^2 n^2}. \qedhere
\end{align*}
\end{proof}

\begin{lemma}
\label{lem:G_IV}The nullspace $S_4$ of $H(G_{4},\gxoc,n)$ is spanned by
the states 
\begin{equation}
|\mathcal{H}\left(\vec{z},\vec{a}\right)\rangle\quad\text{where}\quad\vec{z}=z_{1}z_{2}\ldots z_{n_{\inn}}\underbrace{00\ldots0}_{n-n_{\inn}}\label{eq:history_states_zero_init}
\end{equation}
for $\vec{a}\in\{0,1\}^{n}$ and $z_{1},\ldots,z_{n_{\inn}}\in\{0,1\}$.
Its smallest nonzero eigenvalue satisfies
\[
\gamma(H(G_{4},\gxoc,n)) \geq \frac{\mathcal{K}_0}{10^{10} M^3n^3}
\]
where $\mathcal{K}_0\in (0,1]$ is the absolute constant from \lem{gs_g_alpha}.
\end{lemma}

\begin{proof}
Using equation \eq{Hin_mat_els}, we find
\begin{align*}
\langle C(k,\vec{x},\vec{b})|H_{\inn,i}|C(j,\vec{z},\vec{a})\rangle
&=\delta_{k,j} \delta_{\vec{x},\vec{z}}\delta_{\vec{a},\vec{b}}\bigg(\frac{1}{4}\langle {\Cube_0(j,\vec{z},\vec{a})}|H_{\inn,i}|{\Cube_0(j,\vec{z},\vec{a})}\rangle\\
&\quad+\frac{1}{4}\langle {\Cube_1(j,\vec{z},\vec{a})}|H_{\inn,i}| {\Cube_1(j,\vec{z},\vec{a})}\rangle \\
&\quad+\frac{1}{2}\langle {\Cube_2(j,\vec{z},\vec{a})}|H_{\inn,i}| {\Cube_2(j,\vec{z},\vec{a})}\rangle\bigg)\\
&=\delta_{k,j} \delta_{\vec{x},\vec{z}}\delta_{\vec{a},\vec{b}}\left(\frac{1}{64}\delta_{z_i,1}\right)\begin{cases}
\frac{1}{4}\cdot\frac{1}{2}+\frac{1}{4}\cdot\frac{1}{2}+\frac{1}{2}\cdot\frac{1}{2} & j<j_{\min,i}\\
\frac{1}{4}+0+0 & j=j_{\min,i}\\
0+0+0 & j>j_{\min,i}
\end{cases}
\end{align*}
for each $i\in\{n_{\inn}+1,\ldots,n\}$. Hence 
\[
\langle\mathcal{H}(\vec{x},\vec{b})|\sum_{i=n_{\inn}+1}^{n}H_{\inn,i}|\mathcal{H}(\vec{z},\vec{a})\rangle=\frac{1}{M}\delta_{\vec{x},\vec{z}}\delta_{\vec{a},\vec{b}}\sum_{i=n_{\inn}+1}^{n}\frac{1}{64}\left(\frac{j_{\min,i}-1}{2}+\frac{1}{4}\right)\delta_{z_{i},1}.
\]
Therefore
\[
\sum_{i=n_{\inn}+1}^{n}H_{\inn,i}\big|_{S_{3}}
\]
is diagonal in the basis $\{|\mathcal{H}(\vec{z},\vec{a})\rangle\colon\vec{z},\vec{a}\in\{0,1\}^{n}\}$. The zero eigenvectors are given by equation \eq{history_states_zero_init}, and the smallest nonzero eigenvalue is
\begin{equation}
\gamma\left(\sum_{i=n_{\inn}+1}^{n}H_{\inn,i}\big|_{S_{3}}\right)\geq\frac{1}{256M}\label{eq:lbnd_Hinj}
\end{equation}
since $j_{\min,i}\geq1$. To get the stated lower bound we now apply \lem{npl} with 
\begin{align*}
H_{A} & =H(G_{3},\gxoc,n)\qquad H_{B}=\sum_{i=n_{\inn}+1}^{n}H_{\inn,i}\big|_{\mathcal{I}(G_{4},\gxoc,n)}
\end{align*}
and
\[
  \gamma(H_{A}) \geq \frac{\mathcal{K}_0}{10^7 M^{2}n^{2}} \qquad
  \gamma(H_{B}|_{S_{3}})\geq\frac{1}{256M} \qquad 
  \left\Vert H_{B}\right\Vert 
    \leq n\left\Vert \sum_{i=n_{\inn}+1}^{n}h_{\inn,i}\right\Vert = n
\]
where we used \lem{HG3}, equation \eq{lbnd_Hinj}, and the fact that $\left\Vert \sum_{i=n_{\inn}+1}^{n}h_{\inn,i}\right\Vert =1$ (from equation \eq{h_S_bound}. This gives
\begin{align*}
\gamma \left(H(G_4,\gxoc,n)\right)
&\geq \frac{\mathcal{K}_0}{256M\mathcal{K}_0+10^7n^2M^2+n(256M)(10^7n^2M^2)} \\
&\geq \frac{\mathcal{K}_0}{(M^3n^3)(256+10^7+256\cdot 10^7)}
 \geq \frac{\mathcal{K}_0}{10^{10} M^3n^3}. \qedhere
\end{align*}
\end{proof}

\subsection{Completeness and soundness}
\label{sec:Completeness-and-Soundness}

We now finish the proof of \thm{main_thm_with_occ_constraints}. Using equation \eq{hout_matels-1} we get
\begin{align*}
\langle C(k,\vec{x},\vec{b})|H_{\out}|C(j,\vec{z},\vec{a})\rangle &=\delta_{k,j}\delta_{\vec{a},\vec{b}}\bigg(\frac{1}{4}\langle \Cube_0(j,\vec{x},\vec{a})|H_{\out}| \Cube_0(j,\vec{z},\vec{a})\rangle\\
&\quad +\frac{1}{4}\langle \Cube_1(j,\vec{x},\vec{a})|H_{\out}| \Cube_1(j,\vec{z},\vec{a})\rangle \\
&\quad+\frac{1}{2}\langle \Cube_2(j,\vec{x},\vec{a})|H_{\out}| \Cube_2(j,\vec{z},\vec{a})\rangle\bigg)\\
&= \frac{\delta_{k,j}\delta_{\vec{a},\vec{b}}}{64}\langle\vec{x}|U_{\mathcal{C}_X}^{\dagger}(a_{1})\left(|0\rangle\langle0|_{2}\right)U_{\mathcal{C}_X}(a_{1})|\vec{z}\rangle \begin{cases}
\frac{1}{4}\cdot\frac{1}{2}+\frac{1}{4}\cdot\frac{1}{2}+\frac{1}{2}\cdot\frac{1}{2} & j>j_{\max}\\
0+\frac{1}{4}+0 & j=j_{\max}\\
0+0+0 & j<j_{\max}
\end{cases}
\end{align*}
and
\begin{equation}
\langle\mathcal{H}(\vec{x},\vec{b})|H_{\out}|\mathcal{H}(\vec{z},\vec{a})\rangle
=\delta_{\vec{a},\vec{b}} \frac{1}{128M}(M-j_{\max}+\tfrac{1}{2})\langle\vec{x}|U_{\mathcal{C}}^{\dagger}(a_{1})\left(|0\rangle\langle0|_{2}\right)U_{\mathcal{C}}(a_{1})|\vec{z}\rangle.\label{eq:history_hout_matels}
\end{equation}
For any $n_{\inn}$-qubit state $|\phi\rangle$, define 
\begin{equation}
|\mathcal{\widehat{H}}(\phi,\vec{a})\rangle=\sum_{\vec{z}\in\{0,1\}^{n}}\left(\langle\vec{z}|\phi\rangle|0\rangle^{\otimes n-n_{\inn}}\right)|\mathcal{H}(\vec{z},\vec{a})\rangle\label{eq:H_hat}.
\end{equation}
Note (from \lem{G_IV}) that $|\mathcal{\widehat{H}}(\phi,\vec{a})\rangle$ is in the nullspace of $H(G_{4},\gxoc,n)$.

\subsubsection{Completeness}

Suppose there exists an $n_{\inn}$-qubit state $|\psi_{\wit}\rangle$ such that $\AP(\mathcal{C}_{X},|\psi_{\wit}\rangle)\geq1-\frac{1}{2^{|X|}}$, i.e., 
\begin{equation}
\langle\psi_{\wit}|\langle0|^{n-n_{\inn}}U_{\mathcal{C}_{X}}^{\dagger}|0\rangle\langle0|_{2}U_{\mathcal{C}_{X}}|\psi_{\wit}\rangle|0\rangle^{n-n_{\inn}}\leq\frac{1}{2^{|X|}}.\label{eq:completeness_case}
\end{equation}
Then (letting $\vec{0}$ denote the all-zeros vector) 
\begin{align*}
\langle\widehat{\mathcal{H}}(\psi_{\wit},\vec{0})|H(\gx,\gxoc,n)|\widehat{\mathcal{H}}(\psi_{\wit},\vec{0})\rangle & = \langle\widehat{\mathcal{H}}(\psi_{\wit},\vec{0})|H(G_{4},\gxoc,n)+H_{\out}|\widehat{\mathcal{H}}(\psi_{\wit},\vec{0})\rangle\\
 & = \langle\widehat{\mathcal{H}}(\psi_{\wit},\vec{0})|H_{\out}|\widehat{\mathcal{H}}(\psi_{\wit},\vec{0})\rangle\\
 & = \frac{1}{128M}(M-j_{\max}+\tfrac{1}{2})\langle\psi_{\wit}|\langle0|^{n-n_{\inn}}U_{\mathcal{C}_{X}}^{\dagger}|0\rangle\langle0|_{2}U_{\mathcal{C}_{X}}|\psi_{\wit}\rangle|0\rangle^{n-n_{\inn}}\\
 & \leq \frac{1}{2^{|X|}}
\end{align*}
(using equations \eq{history_hout_matels} and \eq{H_hat} to go from the second to the third line, and equation \eq{completeness_case} in the last line). Hence $\lambda_{n}^{1}(\gx,\gxoc)\leq\frac{1}{2^{|X|}}$, establishing equation \eq{cond1}.

\subsubsection{Soundness}

Now suppose $\AP(\mathcal{C}_{X},|\phi\rangle)\leq\frac{1}{3}$ for all normalized $n_{\inn}$-qubit states $|\phi\rangle$, i.e.,
\begin{equation}
\langle\phi|\langle0|^{n-n_{\inn}}U_{\mathcal{C}_{X}}^{\dagger}|0\rangle\langle0|_{2}U_{\mathcal{C}_{X}}|\phi\rangle|0\rangle^{n-n_{\inn}}\geq\frac{2}{3}\quad\text{for all normalized }|\phi\rangle\in(\C^{2})^{\otimes n_{\inn}}.\label{eq:soundness_a0}
\end{equation}
Complex-conjugating this equation gives
\[
\langle\phi|^{*}\langle0|^{n-n_{\inn}}U_{\mathcal{C}_{X}}^{\dagger*}|0\rangle\langle0|_{2}U_{\mathcal{C}_{X}}^{*}|\phi\rangle^{*}|0\rangle^{n-n_{\inn}}\geq\frac{2}{3}\quad\text{for all normalized }|\phi\rangle\in(\C^{2})^{\otimes n_{\inn}},
\]
or equivalently (replacing $|\phi\rangle$ with its complex conjugate),
\begin{equation}
\langle\phi|\langle0|^{n-n_{\inn}}U_{\mathcal{C}_{X}}^{\dagger*}|0\rangle\langle0|_{2}U_{\mathcal{C}_{X}}^{*}|\phi\rangle|0\rangle^{n-n_{\inn}}\geq\frac{2}{3}\quad\text{for all normalized }|\phi\rangle\in(\C^{2})^{\otimes n_{\inn}}.\label{eq:soundness_a1}
\end{equation}
Recall that $S_{4}$ is the nullspace of $H(G_{4},\gxoc,n)$ and consider the restriction 
\begin{equation}
H_{\out}\big|_{S_{4}}.\label{eq:res_S4}
\end{equation}
We now show that the smallest eigenvalue of \eq{res_S4} is strictly positive. This implies that the nullspace of $H(\gx,\gxoc,n)$ is empty, which can be seen from \eq{H_GC} since both terms in this equation are positive semidefinite and $S_4$ is the nullspace of the first term. We then use \npl to lower bound the smallest eigenvalue $\lambda_n^{1}(\gx,\gxoc)$ of $H(\gx,\gxoc,n)$.

By \lem{G_IV}, the states 
\[
|\widehat{\mathcal{H}}(\vec{z},\vec{a})\rangle = |\mathcal{H}(z_1 z_2\ldots z_{n_\inn}\underbrace{00\ldots0}_{n-n_\inn},\vec{a})\rangle,
\quad\vec{a}\in\{0,1\}^{n},\,\vec{z}\in\{0,1\}^{n_{\inn}}
\]
are a basis for $S_{4}$ and in this basis $H_{\out}$ is block diagonal with a block for each $\vec{a}\in\{0,1\}^{n}$, as can be seen using equation \eq{history_hout_matels}: 
\[
\langle\widehat{\mathcal{H}}(\vec{x},\vec{b})|H_{\out}|\widehat{\mathcal{H}}(\vec{z},\vec{a})\rangle=\delta_{\vec{a},\vec{b}} \frac{1}{128M}(M-j_{\max}+\tfrac{1}{2})\langle\vec{x}|\langle0|^{n-n_{\inn}}U_{\mathcal{C}_{X}}^{\dagger}(a_{1})|0\rangle\langle0|_{2}U_{\mathcal{C}_{X}}(a_{1})|\vec{z}\rangle|0\rangle^{n-n_{\inn}}.
\]
From equation \eq{H_hat} we can see that any normalized state that lives in the block corresponding to some fixed $\vec{a}\in\{0,1\}^{n}$ can be written as $|\widehat{\mathcal{H}}(\phi,\vec{a})\rangle$ for some normalized $n_{\inn}$-qubit state $|\phi\rangle$. The smallest eigenvalue of $H_{\out}|_{S_{4}}$ is therefore
\[
\langle\widehat{\mathcal{H}}(\theta,\vec{\alpha})|H_{\out}|\widehat{\mathcal{H}}(\theta,\vec{\alpha})\rangle
\]
for some normalized $n_{\inn}$-qubit state $|\theta\rangle$ and some $\vec{\alpha}\in\{0,1\}^{n}$. Now 
\begin{align}
\langle\widehat{\mathcal{H}}(\theta,\vec{\alpha})|H_{\out}|\widehat{\mathcal{H}}(\theta,\vec{\alpha})\rangle & = \frac{1}{128M}(M-j_{\max}+\tfrac{1}{2})\langle\theta|\langle0|^{n-n_{\inn}}U_{\mathcal{C}_{X}}^{\dagger}(\alpha_{1})|0\rangle\langle0|_{2}U_{\mathcal{C}_{X}}(\alpha_{1})|\theta\rangle|0\rangle^{n-n_{\inn}}\nonumber \\
 & \geq \frac{1}{256M}\cdot\frac{2}{3}\label{eq:h_phi_a_expectation}
\end{align}
using equation \eq{soundness_a0} if $\alpha_{1}=0$ and equation \eq{soundness_a1} if $\alpha_{1}=1$. Since \eq{h_phi_a_expectation} is a lower bound on the smallest eigenvalue within each block, the nullspace of $H_{\out}|_{S_{4}}$ is empty and 
\begin{equation}
\gamma(H_{\out}|_{S_{4}})\geq\frac{1}{384M}.\label{eq:bound_lowesteigvalue}
\end{equation}
As noted above, the fact that this matrix has strictly positive eigenvalues implies that so does
$H(\gx,\gxoc,n)$, i.e.,
\[
\lambda_{n}^{1}\left(\gx,\gxoc\right)=\gamma(H(\gx,\gxoc,n)).
\]
To lower bound this quantity we apply \lem{npl} with
\begin{align*}
H_{A} &= H(G_{4},\gxoc,n)\qquad H_{B}=H_{\out}\big|_{\mathcal{I}(\gx,\gxoc,n)}
\end{align*}
and we use the bound \eq{bound_lowesteigvalue} as well as
\[
\gamma(H_{A}) \geq \frac{\mathcal{K}_0}{10^{10} n^{3}M^{3}}
\]
(from \lem{G_IV}) and $\left\Vert H_{B}\right\Vert \leq\left\Vert H_{\out}\right\Vert \leq n\left\Vert h_{\out}\right\Vert=n$ (using equation \eq{hin_i}). Applying the Lemma gives
\begin{align*}
\lambda_{n}^{1}\left(\gx,\gxoc\right)&=\gamma(H(\gx,\gxoc,n))\\
&\geq \frac{\mathcal{K}_0}{10^{10}n^3M^3+384M\mathcal{K}_0+n(10^{10}n^3M^3)(384M)}\\
&\geq \frac{\mathcal{K}_0}{n^4M^4(10^{10}+384+10^{10} \cdot 384)}\\
&\geq \frac{\mathcal{K}_0}{10^{13}n^4M^4}.
\end{align*}
Now choosing $\mathcal{K}$ (the constant in the statement of \thm{main_thm_with_occ_constraints}) to be equal to $\frac{\mathcal{K}_0}{10^{13}}$ (recall $\mathcal{K}_0\in (0,1]$ is the absolute constant from \lem{gs_g_alpha}) proves equation \eq{cond2}. This completes the proof of \thm{main_thm_with_occ_constraints}.

\section*{Acknowledgments}

This work was supported in part by NSERC; the Ontario Ministry of Research and Innovation; the Ontario Ministry of Training, Colleges, and Universities; and the US ARO.


\providecommand{\bysame}{\leavevmode\hbox to3em{\hrulefill}\thinspace}

\appendix

\section{Complexity of computing the smallest eigenvalue of a graph}
\label{app:complexity_smallest_graph_eig}

How hard is it to compute the smallest eigenvalue of a symmetric $0$-$1$ matrix (i.e., the adjacency matrix of a graph)? If the matrix is given explicitly then this can be done in time polynomial in the input size. Instead, we consider adjacency matrices $A$ specified by a (classical) circuit that takes as input the index $r$ of a row and outputs the indices $\{j\colon A_{rj}=1\}$ of the nonzero entries in that row. Note that the circuit describing a $D\times D$ matrix $A$ must have size $\Omega(\log(D))$ since it takes as input a row index $r\in[D]$. Efficiently row-computable matrices are those described by circuits of size polynomial in $\log D$ \cite{AT03}. We consider the following promise problem.

\begin{mdframed}
\begin{problem}
[\textbf{Minimum Graph Eigenvalue}] We are given a $D\times D$ symmetric $0$-$1$ matrix $A$, specified by a classical deterministic circuit that takes as input a row $r\in[D]$ and computes the locations of the nonzero entries in that row. We are also given a real number $a$ and a precision parameter $\epsilon=\frac{1}{N}$, where $N\in\mathbb{N}$ is specified in unary. We are promised that either the smallest eigenvalue of $A$ is at most $a$ (yes instance), or else it is at least $a+\epsilon$ (no instance), and asked to decide which is the case.
\end{problem}
\end{mdframed}

The input size of this problem is the size of the circuit that describes
$A$, plus the description size of $a$, plus $N$ bits. The interesting
set of instances are those for which $A$ is efficiently row-computable,
so that the input size is $\poly(\log D,\frac{1}{\epsilon})$.

\begin{theorem}
Minimum Graph Eigenvalue is QMA-complete. 
\end{theorem}

The fact that Minimum Graph Eigenvalue is contained in QMA follows from standard techniques (applying phase estimation to the adjacency matrix) since a Hamiltonian that is a symmetric $0$-$1$ matrix can be simulated efficiently as a function of the size of the circuit computing its rows \cite{AT03}.

To prove that Minimum Graph Eigenvalue is QMA-hard, we show that an instance of any problem in QMA can be converted (in deterministic polynomial time on a classical computer) into an equivalent instance of Minimum Graph Eigenvalue.

Recall that an instance $x$ of a problem in QMA has a verification circuit $\mathcal{C}_{x}$ with $n$ qubits and $M$ gates, both upper bounded by a polynomial function of $|x|$. Here we assume that the verification circuit $\mathcal{C}_{x}$ satisfies a modified version of \defn{QMA} where the completeness parameter $\frac{2}{3}$ is replaced by $1-\frac{1}{2^{|x|}}$ (modifying the definition in this way still gives the same class QMA \cite{KSV02,MW05}). In \sec{circuit-to-graph} we map $\mathcal{C}_{x}$ to a row-sparse and efficiently row-computable symmetric $0$-$1$ matrix $A_{x}$. The mapping we exhibit has the following two properties (shown in \sec{mappingyes} and \sec{mappingno}, respectively):
\begin{enumerate}
\item If $x$ is a yes instance then the smallest eigenvalue of $A_{x}$
is at most $-1-3\sqrt{2}+\frac{1}{2M}\frac{1}{2^{|x|}}$.
\item If $x$ is a no instance then the smallest eigenvalue of $A_{x}$
is at least $-1-3\sqrt{2}+\frac{1}{120M^{4}n}$. 
\end{enumerate}
Note that if $|x|$ is larger than some constant, then $\frac{1}{2M}\frac{1}{2^{|x|}}\leq\frac{1}{240M^{4}n}$, since $M$ and $n$ are both upper bounded by a polynomial in $|x|$. We assume without loss of generality that this holds. For any such instance $x$, we associate an instance of Minimum Graph Eigenvalue defined by the matrix $A_{x}$ along with the number $a=-1-3\sqrt{2}+\frac{1}{240M^{4}n}$ and precision parameter
\[
\epsilon=\frac{1}{240M^{4}n}.
\]
Using the two claimed properties of the circuit-to-graph mapping, we see that if $x$ is a yes instance then the smallest eigenvalue of $A_{x}$ is at most $a$, and if it is a no instance then this eigenvalue is at least $a+\epsilon$. This shows that Minimum Graph Eigenvalue is QMA-hard.

\subsection{Circuit-to-graph mapping}
\label{sec:circuit-to-graph}

The circuit-to-graph mapping we use generalizes the one used in \sec{Encoding-a-Computation}; we repeat the details here so this Section can be read independently. We suppose $\mathcal{C}_{x}$ implements a unitary 
\begin{equation}
U_{\mathcal{C}_{x}}=U_{M}\ldots U_{2}U_{1}\label{eq:single_qubit_circuit}
\end{equation}
 where each $U_{i}$ is from the universal gate set 
\[
\mathcal{G}=\{H,HT,\left(HT\right)^{\dagger},\left(H\otimes\id\right)\CNOT\}
\]
with 
\[
H=\frac{1}{\sqrt{2}}\begin{pmatrix}
1 & 1\\
1 & -1
\end{pmatrix} \quad T=\begin{pmatrix}
1 & 0\\
0 & e^{i\frac{\pi}{4}}
\end{pmatrix}\quad\CNOT=\begin{pmatrix}
1 & 0 & 0 & 0\\
0 & 1 & 0 & 0\\
0 & 0 & 0 & 1\\
0 & 0 & 1 & 0
\end{pmatrix}.
\]
The verification circuit $\mathcal{C}_{x}$ has an $n_{\inn}$-qubit input register and $n-n_{\inn}$ ancilla qubits initialized to $|0\rangle$ at the beginning of the computation. One of these $n$ qubits serves as an output qubit.

It will be convenient to consider 
\[
  U_{\mathcal{C}_{x}}^{\dagger}U_{\mathcal{C}_{x}}=W_{2M}\ldots W_{2}W_{1}
\]
where
\[
W_{t}=\begin{cases}
U_{t} & 1\leq t\leq M\\
U_{2M+1-t}^{\dagger} & M+1 \le t \le 2M.
\end{cases}
\]
As in \sec{Encoding-a-Computation} we start with a version of the Feynman-Kitaev Hamiltonian \cite{Fey85,KSV02}
\begin{equation}
H_{x}=-\sqrt{2}\sum_{t=1}^{2M}\left(W_{t}^{\dagger}\otimes|t\rangle\langle t+1|+W_{t}\otimes|t+1\rangle\langle t|\right)\label{eq:H_x}
\end{equation}
acting on the Hilbert space $\mathcal{H}_{\text{comp}} \otimes \mathcal{H}_{\clock}$ where $\mathcal{H}_{\text{comp}} = \left(\C^{2}\right)^{\otimes n}$ is an $n$-qubit computational register and $\mathcal{H}_{\clock}=\C^{2M}$ is a $2M$-level register with periodic boundary conditions (i.e., we let $|2M+1\rangle=|1\rangle$). Note that 
\begin{equation}
V^{\dagger}H_{x}V=-\sqrt{2}\sum_{t=1}^{2M}\left(\id\otimes|t\rangle\langle t+1|+\id\otimes|t+1\rangle\langle t|\right)\label{eq:V_H_V}
\end{equation}
where 
\[
V=\sum_{t=1}^{2M}\bigg(\prod_{j=t-1}^{1}W_{j}\bigg)\otimes|t\rangle\langle t|
\]
and $W_{0}=1$. Since $V$ is unitary, the eigenvalues of $H_{x}$ are the same as the eigenvalues of \eq{V_H_V}, namely 
\[
-2\sqrt{2}\cos\left(\frac{\pi \ell}{M}\right)
\]
for $\ell=0,\ldots,2M-1$. The ground energy of \eq{V_H_V} is $-2\sqrt{2}$ and its ground space is spanned by 
\[
  |\phi\rangle\frac{1}{\sqrt{2M}}\sum_{j=1}^{2M}|t\rangle,
  \quad|\phi\rangle\in\Lambda
\]
where $\Lambda$ is any orthonormal basis for $\mathcal{H}_{\mathrm{comp}}$. A basis for the ground space of $H_{x}$ is therefore 
\[
V\bigg(|\phi\rangle\frac{1}{\sqrt{2M}}\sum_{j=1}^{2M}|t\rangle\bigg)=\frac{1}{\sqrt{2M}}\sum_{t=1}^{2M}W_{t-1}W_{t-2}\ldots W_{1}|\phi\rangle|t\rangle
\]
 where $|\phi\rangle\in\Lambda$. The first excited energy of $H_{x}$ is 
\[
\eta=-2\sqrt{2}\cos\left(\frac{\pi}{M}\right)
\]
and the gap between ground and first excited energies is lower bounded as 
\begin{equation}
\eta+2\sqrt{2}\geq\sqrt{2}\frac{\pi^{2}}{M^{2}}\label{eq:bound_on_eta}
\end{equation}
(using the fact that $1-\cos(x)\leq\frac{x^{2}}{2}$).

The universal set $\mathcal{G}$ is chosen so that each gate has nonzero entries that are integer powers of $\omega=e^{i\frac{\pi}{4}}$.  Correspondingly, the nonzero standard basis matrix elements of $H_{x}$ are also integer powers of $\omega=e^{i\frac{\pi}{4}}$. We consider the $8\times8$ shift operator
\[
S=\sum_{j=0}^{7}|j+1\text{ mod 8}\rangle\langle j|
\]
and note that $\omega$ is an eigenvalue of $S$ with eigenvector
\[
|\omega\rangle=\frac{1}{\sqrt{8}}\sum_{j=0}^{7}\omega^{-j}|j\rangle.
\]
We modify $H_{x}$ as follows. For each operator $-\sqrt{2}H$, $-\sqrt{2}HT$, $-\sqrt{2}(HT)^{\dagger}$, or $-\sqrt{2}\left(H\otimes\id\right)\CNOT$ appearing in equation \eq{H_x}, define another operator that acts on $\C^{2}\otimes\C^{8}$ or $\C^{4}\otimes\C^{8}$ (as appropriate) by replacing nonzero matrix elements with powers of the operator $S$:
\begin{align*}
\omega^{k} &\mapsto S^{k}.
\end{align*}
Matrix elements that are zero are mapped to the $8\times8$ all-zeroes matrix. Write $B(W)$ for the operators obtained by making this replacement, e.g., 
\begin{align*}
-\sqrt{2}HT & = \begin{pmatrix}
\omega^{4} & \omega^{5}\\
\omega^{4} & \omega
\end{pmatrix} \mapsto B(HT)=\begin{pmatrix}
S^{4} & S^{5}\\
S^{4} & S
\end{pmatrix}.
\end{align*}
Adjoining an 8-level ancilla as a third register and making this replacement in equation \eq{single_qubit_ham} gives 
\begin{align}
H_{\mathcal{\prop}} & =\sum_{t=1}^{2M}\left(B(W_{t})^{\dagger}_{13}\otimes|t\rangle\langle t+1|_2+B(W_{t})_{13}\otimes|t+1\rangle\langle t|_2\right)\label{eq:Hprop-1}
\end{align}
which is a symmetric $0$-1 matrix (the subscripts indicate which registers the operators act on). Note that $H_{\prop}$ commutes with $S$ (acting on the $8$-level ancilla) and therefore is block diagonal with eight sectors. In the sector where $S$ has eigenvalue $\omega$, $H_{\prop}$ is identical to the Hamiltonian $H_x$ that we started with (see equation \eq{H_x}). There is also a sector (where $S$ has eigenvalue $\omega^*$) where the Hamiltonian is the complex conjugate of $H_x$. We will add a term to $H_{\prop}$ that introduces an energy penalty for states in any of the other six sectors, ensuring that none of these states lie in the ground space.

To see what kind of energy penalty is needed, we lower bound the eigenvalues of $H_{\prop}$. Note that for each $W\in\mathcal{G}$, $B(W)$ contains at most 2 ones in each row or column. Looking at equation \eq{Hprop-1} and using this fact, we see that each row and each column of $H_{\prop}$ contains at most four ones (with the remaining entries all zero). Therefore $\|H_{\prop}\| \leq4$, so every eigenvalue of $H_{\prop}$ is at least $-4$.

The matrix $A_{x}$ associated with the circuit $\mathcal{C}_{x}$ acts on the Hilbert space 
\[
\mathcal{H}_{\text{comp}}\otimes\mathcal{H}_{\clock}\otimes\mathcal{H}_{\anc}
\]
where $\mathcal{H}_{\anc}=\C^{8}$ holds the $8$-level ancilla.  We define 
\begin{equation}
A_{x}=H_{\prop}+H_{\text{penalty}}+H_{\text{input}}+H_{\text{output}}\label{eq:A_x_eqn}
\end{equation}
 where 
\[
H_{\text{penalty}}=\id\otimes\id\otimes\left(S^{3}+S^{4}+S^{5}\right)
\]
is the penalty ensuring that the ancilla register holds either
$|\omega\rangle$ or $|\omega^*\rangle$ and the terms
\begin{align*}
H_{\text{input}} & =\sum_{j=n_{\inn}+1}^{n}|1\rangle\langle1|_{j}\otimes|1\rangle\langle1|\otimes\id\\
H_{\text{output}} & =|0\rangle\langle0|_{\out}\otimes|M+1\rangle\langle M+1|\otimes\id
\end{align*}
ensure that the ancilla qubits are initialized in the state $|0\rangle$ when $t=1$ and that the output qubit is in the state $|1\rangle\langle1|$ when the circuit $\mathcal{C}_{x}$ has been applied (i.e., at time $t=M+1$). Observe that $A_{x}$ is a symmetric $0$-$1$ matrix. 

Now consider the ground space of the first two terms $H_{\prop}+H_{\text{penalty}}$ in \eq{A_x_eqn}. Note that $[H_{\prop},H_{\text{penalty}}]=0$, so these operators can be simultaneously diagonalized. Furthermore, $H_{\mathrm{penalty}}$ has smallest eigenvalue $-1-\sqrt{2}$, with eigenspace spanned by $|\omega\rangle$ and $|\omega^*\rangle$. One can also easily confirm that the first excited energy of $H_{\mathrm{penalty}}$ is $-1$.

The ground space of $H_{\prop}+H_{\text{penalty}}$ lives in the sector where $H_{\text{penalty}}$ has minimal eigenvalue $-1-\sqrt{2}$. To see this, note that within this sector $H_{\prop}$ has the same eigenvalues as $H_{x}$, and therefore has lowest eigenvalue $-2\sqrt{2}$. The minimum eigenvalue $e_{1}$ of $H_{\prop}+H_{\text{penalty}}$ in this sector is 
\begin{equation}
e_{1}=-2\sqrt{2}+\left(-1-\sqrt{2}\right)=-1-3\sqrt{2}=-5.24\ldots,\label{eq:e_1_definition-1}
\end{equation}
whereas in any other sector $H_{\text{penalty}}$ has eigenvalue at least $-1$ and (using the fact that $H_{\prop}\geq-4$) the minimum eigenvalue of $H_{\prop}+H_{\text{penalty}}$ is at least $-5.$ Thus, an orthonormal basis for the ground space of $H_{\prop}+H_{\mathrm{penalty}}$ is furnished by the states 
\begin{align}
 & \frac{1}{\sqrt{2M}}\sum_{t=1}^{2M}W_{t-1}W_{t-2}\ldots W_{1}|\phi\rangle|t\rangle|\omega\rangle\label{eq:gs_Hprop_Hpen1}\\
 & \frac{1}{\sqrt{2M}}\sum_{t=1}^{2M}(W_{t-1}W_{t-2}\ldots W_{1})^{*}|\phi^{*}\rangle|t\rangle|\omega^{*}\rangle\label{eq:gs_Hprop_Hpen2}
\end{align}
where $|\phi\rangle$ ranges over the basis $\Lambda$ for $\mathcal{H}_{\text{comp}}$ and $*$ denotes (elementwise) complex conjugation.

\subsection{Upper bound on the smallest eigenvalue for yes instances}
\label{sec:mappingyes}

Suppose $x$ is a yes instance; then there exists some $n_{\inn}$-qubit state $|\psi_{\inn}\rangle$ satisfying $\AP\left(\mathcal{C}_{x},|\psi_{\inn}\rangle\right)\geq1-\frac{1}{2^{|x|}}$. Let
\[
|\wit\rangle=\frac{1}{\sqrt{2M}}\sum_{t=1}^{2M}W_{t-1}W_{t-2}\ldots W_{1}\left(|\psi_{\inn}\rangle|0\rangle^{\otimes n-n_{\inn}}\right)|t\rangle|\omega\rangle
\]
and note that this state is in the $e_{1}$-energy ground space of
$H_{\prop}+H_{\text{penalty}}$ (since it has the form \eq{gs_Hprop_Hpen1}).
One can also directly verify that $|\wit\rangle$ has zero energy for $H_{\text{input}}$. Thus
\begin{align*}
\langle\wit|A_{x}|\wit\rangle & =e_{1}+\langle\wit|H_{\text{output}}|\wit\rangle\\
 & =e_{1}+\frac{1}{2M}\langle\psi_{\inn}|\langle0|^{\otimes n-n_{\inn}}U_{\mathcal{C}_{x}}^{\dagger}|0\rangle\langle0|_{\out}U_{\mathcal{C}_{x}}|\psi_{\inn}\rangle|0\rangle^{\otimes n-n_{\inn}}\\
 & =e_{1}+\frac{1}{2M}\left(1-\AP(\mathcal{C}_{x},|\psi_{\inn}\rangle)\right)\\
 & \leq e_{1}+\frac{1}{2M}\frac{1}{2^{|x|}}.
\end{align*}

\subsection{Lower bound on the smallest eigenvalue for no instances}
\label{sec:mappingno}

Now suppose $x$ is a no instance. Then the verification circuit $\mathcal{C}_{x}$ has acceptance probability $\AP\left(\mathcal{C}_{x},|\psi\rangle\right)\leq\frac{1}{3}$ for all $n_{\inn}$-qubit input states $|\psi\rangle$.

We backtrack slightly to obtain bounds on the eigenvalue gaps of the Hamiltonians $H_{\prop}+H_{\text{penalty}}$ and $H_{\prop}+H_{\text{penalty}}+H_{\text{input}}$. We begin by showing that the energy gap of $H_{\prop}+H_{\text{penalty}}$ is at least an inverse polynomial function of $M$. Subtracting a constant equal to the ground energy times the identity matrix sets the smallest eigenvalue to zero, and the smallest nonzero eigenvalue satisfies 
\begin{equation}
\gamma(H_{\prop}+H_{\text{penalty}}-e_{1}\cdot\id)\geq\min\left\{ \sqrt{2}\frac{\pi^{2}}{M^{2}},-5-e_{1}\right\} \geq\frac{1}{5M^{2}}.\label{eq:gamma_Ha}
\end{equation}
since $-5-e_{1}\approx0.24\ldots>\frac{1}{5}$. The first inequality above follows from the fact that every eigenvalue of $H_{\prop}$ in the range $[e_{1},-5)$ is also an eigenvalue of $H_{x}$ (as discussed above) and the bound \eq{bound_on_eta} on the energy gap of $H_{x}$.

Now use the \npl (\lem{npl}) with 
\[
H_{A}=H_{\prop}+H_{\text{penalty}}-e_{1}\cdot\id\qquad H_{B}=H_{\text{input}}.
\]
Note that $H_{A}$ and $H_{B}$ are positive semidefinite. Let $S_{A}$ be the ground space of $H_{A}$ and consider the restriction $H_{B}|_{S_{A}}$. Here it is convenient to use the basis for $S_{A}$ given by \eq{gs_Hprop_Hpen1} and \eq{gs_Hprop_Hpen2} with $|\phi\rangle$ ranging over the computational basis states of $n$ qubits. In this basis, $H_{B}|_{S_{A}}$ is diagonal with all diagonal entries equal to $\frac{1}{2M}$ times an integer, so $\gamma(H_{B}|_{S_{A}})\geq\frac{1}{2M}$. We also have $\gamma(H_{A})\geq\frac{1}{5M^{2}}$ from equation \eq{gamma_Ha}, and clearly $\left\Vert H_{B}\right\Vert \leq n$. Thus \lem{npl} gives
\begin{equation}
\gamma(H_{\prop}+H_{\text{penalty}}+H_{\text{input}}-e_{1}\cdot\id)\geq\frac{\left(\frac{1}{5M^{2}}\right)\left(\frac{1}{2M}\right)}{\frac{1}{5M^{2}}+\frac{1}{2M}+n}\geq\frac{1}{10M^{3}\left(1+n\right)}\geq\frac{1}{20M^{3}n}.\label{eq:lowerbound_uptoinput}
\end{equation}

Now consider adding the final term $H_{\mathrm{output}}$. We use \lem{npl} again, now setting 
\[
H_{A}=H_{\prop}+H_{\mathrm{penalty}}+H_{\mathrm{input}}-e_{1}\cdot\id\qquad H_{B}=H_{\mathrm{output}}.
\]
Let $S_{A}$ be the ground space of $H_{A}$. Note that it is spanned by states of the form \eq{gs_Hprop_Hpen1} and \eq{gs_Hprop_Hpen2} where $|\phi\rangle=|\psi\rangle|0\rangle^{\otimes n-n_{\inn}}$ and $|\psi\rangle$ ranges over any orthonormal basis of the $n_{\inn}$-qubit input register. The restriction $H_{B}|_{S_{A}}$ is block diagonal, with one block for states of the form 
\begin{equation}
\frac{1}{\sqrt{2M}}\sum_{t=1}^{2M}W_{t-1}W_{t-2}\ldots W_{1}\left(|\psi\rangle|0\rangle^{\otimes n-n_{\inn}}\right)|t\rangle|\omega\rangle\label{eq:block1}
\end{equation}
and another block for states of the form 
\begin{equation}
\frac{1}{\sqrt{2M}}\sum_{t=1}^{2M}\left(W_{t-1}W_{t-2}\ldots W_{1}\right)^{*}\left(|\psi\rangle^{*}|0\rangle^{\otimes n-n_{\inn}}\right)|t\rangle|\omega^{*}\rangle.\label{eq:block2}
\end{equation}
We now show that the minimum eigenvalue of $H_{B}|_{S_{A}}$ is nonzero, and we lower bound it. We consider the two blocks separately. By linearity, every state in the first block can be written in the form \eq{block1} for some state $|\psi\rangle$. Thus the minimum eigenvalue within this block is the minimum expectation of $H_{\text{output}}$ in a state \eq{block1}, where the minimum is taken over all $n_{\inn}$-qubit states $|\psi\rangle$. This is equal to 
\[
\min_{|\psi\rangle}\frac{1}{2M}\left(1-\AP(\mathcal{C}_{x},|\psi\rangle)\right)\geq\frac{1}{3M}
\]
where we used the fact that $\AP\left(\mathcal{C}_{x},|\psi\rangle\right)\leq\frac{1}{3}$ for all $|\psi\rangle$. Likewise, every state in the second block can be written as \eq{block2} for some state $|\psi\rangle$, and the minimum eigenvalue within this block is 
\[
\min_{|\psi\rangle}\frac{1}{2M}\left(1-\AP(\mathcal{C}_{x},|\psi\rangle)^{*}\right)\geq\frac{1}{3M}
\]
(since $\AP(\mathcal{C}_{x},|\psi\rangle)^{*}=\AP(\mathcal{C}_{x},|\psi\rangle)\leq\frac{1}{3}$). Thus we see that $H_{B}|_{S_{A}}$ has an empty nullspace, so its smallest eigenvalue is equal to its smallest nonzero eigenvalue, namely 
\[
\gamma(H_{B}|_{S_{A}})\geq\frac{1}{3M}.
\]
Now applying \lem{npl} using this bound, the fact that $\left\Vert H_{B}\right\Vert =1$, and the fact that $\gamma(H_{A})\geq\frac{1}{20M^{3}n}$ (from equation \eq{lowerbound_uptoinput}), we get 
\[
\gamma(A_{x}-e_{1}\cdot\id)\geq\frac{\frac{1}{60M^{4}n}}{\frac{1}{20M^{3}n}+\frac{1}{3M}+1}\geq\frac{1}{120M^{4}n}.
\]
Since $H_{B}|_{S_{A}}$ has an empty nullspace, $A_{x}-e_{1}\cdot\id$ has an empty nullspace, and this is a lower bound on its smallest eigenvalue.

\section{XY Hamiltonian is QMA-complete}
\label{app:XY}

In this Appendix we prove \thm{XY}, showing that XY Hamiltonian is QMA-complete.

\begin{proof}
An instance of XY Hamiltonian can be verified by the standard QMA verification protocol for the Local Hamiltonian problem \cite{KSV02} with one slight modification: before running the protocol Arthur measures the magnetization of the witness and rejects unless it is equal to $N$.  Thus the problem is contained in QMA.
	
To prove QMA-hardness, we show that the solution (yes or no) of an instance of Frustration-Free Bose-Hubbard Hamiltonian with input $G$, $N$, $\epsilon$ is equal to the solution of the instance of XY Hamiltonian with the same graph $G$ and integer $N$, with precision parameter $\frac{\epsilon}{4}$ and $c=N\mu(G)+\frac{\epsilon}{4}$. 

We separately consider yes instances and no instances of Frustration-Free Bose-Hubbard Hamiltonian and show that the corresponding instance of XY Hamiltonian has the same solution in both cases.

\subsubsection*{Case 1: no instances}

First consider a no instance of Frustration-Free Bose-Hubbard Hamiltonian, for which $\lambda_N^1 (G)\geq \epsilon+\epsilon^3$.
We have
\begin{align}
\lambda_N^1 (G) & = \min_{\substack{|\phi\rangle\in \mathcal{Z}_N(G)\\ \langle \phi|\phi\rangle=1}} \langle \phi| H_G^N-N\mu(G)|\phi\rangle 
\label{eq:lambda_theta1}\\
& \leq \min_{\substack{|\phi\rangle\in \mathcal{W}_N(G)\\ \langle \phi|\phi\rangle=1}} \langle \phi| H_G^N-N\mu(G)|\phi\rangle
\label{eq:lambda_theta2}\\ & =\min_{\substack{|\phi\rangle\in \mathrm{Wt}_N(G)\\ \langle \phi|\phi\rangle=1}} \langle \phi| O_G-N\mu(G)|\phi\rangle
\label{eq:lambda_theta3} \\
&= \theta_N(G)-N\mu(G)
\label{eq:lambda_theta4}
\end{align}
where in the inequality we used the fact that $\mathcal{W}_N(G)\subset \mathcal{Z}_N(G)$. Hence 
\[
\theta_N(G)\geq N\mu(G)+\lambda_N^1 (G) \geq N\mu(G)+\epsilon+\epsilon^3\geq N\mu(G)+\frac{\epsilon}{2},
\]
so the corresponding instance of XY Hamiltonian is a no instance.

\subsubsection*{Case 2: yes instances}

Now consider a yes instance of Frustration-Free Bose-Hubbard Hamiltonian,  so $0\leq \lambda_N^1 (G)\leq \epsilon^3$. 

We consider the case $\lambda_N^1 (G)=0$ separately from the case where it is strictly positive. If $\lambda_N^1 (G)=0$ then any state $|\psi\rangle$ in the ground space of $H_G^N$ satisfies 
\[
\langle \phi|\sum_{w=1}^N \left(A(G)-\mu(G)\right)^{(w)}+\sum_{k\in V} \widehat{n}_k(\widehat{n}_k-1)|\phi\rangle=0.
\]
Since both terms are positive semidefinite, the state $|\phi\rangle$ has zero energy for each of them. In particular, it has zero energy for the second term, or equivalently, $|\phi\rangle\in \mathcal{W}_N(G)$. Therefore
\begin{align*}
\lambda_N^1 (G) = \min_{\substack{|\phi\rangle\in \mathcal{W}_N(G)\\ \langle \phi|\phi\rangle=1}} \langle \phi| H_G^N-N\mu(G)|\phi\rangle =\min_{\substack{|\phi\rangle\in \mathrm{Wt}_N(G)\\ \langle \phi|\phi\rangle=1}} \langle \phi| O_G-N\mu(G)|\phi\rangle= \theta_N(G)-N\mu(G),
\end{align*}
so $\theta_N(G)= N\mu(G)$, and the corresponding instance of XY Hamiltonian is a yes instance.

Finally, suppose $0 < \lambda_N^1 (G)\leq \epsilon^3$. Then $\lambda_N^1(G)$ is also the smallest \emph{nonzero} eigenvalue of $H(G,N)$, which we denote by $\gamma(H(G,N))$. (Here and throughout this paper we write $\gamma(M)$ for the smallest nonzero eigenvalue of a positive semidefinite matrix $M$.) Note that $ \lambda_N^1 (G)>0$ also implies (by the inequalities \eq{lambda_theta1}--\eq{lambda_theta4}) that $\theta_N(G)-N\mu(G)>0$, so
\[
\theta_N(G)-N\mu(G)=\gamma\left((O_G -N\mu(G))\big|_{\text{Wt}_N(G)}\right).
\]
To upper bound $\theta_N(G)$ we use the \npl (\lem{npl}). We apply this Lemma using the decomposition $H(G,N)=H_A+H_B$ where
\[
H_A=\sum_{k\in V} \widehat{n}_k(\widehat{n}_k-1)\big|_{\mathcal{Z}_N(G)} \qquad H_B=\sum_{w=1}^N \left(A(G)-\mu(G)\right)^{(w)}\big|_{\mathcal{Z}_N(G)}.
\]
Note that $H_A$ and $H_B$ are both positive semidefinite, and that the nullspace $S$ of $H_A$ is equal the space $\mathcal{W}_N(G)$ of hard-core bosons. To apply the Lemma we compute bounds on $\gamma(H_A)$, $\|H_B\|$, and $\gamma(H_B|_S)$. We use the bounds $\gamma(H_A)=2$ (since the operators $\{\widehat{n}_k\colon k\in V\}$ commute and have nonnegative integer eigenvalues),
\[
\|H_B\|\leq N\|A(G)-\mu(G)\|\leq N(\|A(G)\|+\mu(G)) \leq 2N\|A(G)\|\leq 2KN\leq 2K^2
\]
(where we used the fact that $\|A(G)\|$ is at most the maximum degree of $G$, which is at most the number of vertices $K$), and
\begin{align*}
\gamma (H_B|_S)& =\gamma \left(\sum_{w=1}^N \left(A(G)-\mu(G)\right)^{(w)}\big|_{\mathcal{W}_N(G)}\right)\\
& =\gamma\left((O_G -N\mu(G))\big|_{\text{Wt}_N(G)}\right)\\
&= \theta_N(G)-N\mu(G).
\end{align*}
Now applying the Lemma, we get 
\[
\lambda_N^1 (G)=\gamma(H(G,N))\geq \frac{2(\theta_N(G)-N\mu(G))}{2+(\theta_N(G)-N\mu(G))+2K^2}.
\]
Rearranging this inequality gives
\[
\theta_N(G)-N\mu(G) \leq \lambda_N^1(G)\frac{2(K^2+1)}{2-\lambda_N^1(G)}\leq 4K^2 \lambda_N^1(G)\leq 4K^2 \epsilon^3 
\]
where in going from the second to the third inequality we used the fact that $1\leq K^2$ in the numerator and $\lambda_N^1(G)\leq \epsilon^3<1$ in the denominator.  Now using the fact (from the definition of Frustration-Free Bose-Hubbard Hamiltonian) that $\epsilon\leq \frac{1}{4K}$, we get 
\[
\theta_N(G)\leq N\mu(G)+\frac{\epsilon}{4}, 
\]
i.e., the corresponding instance of XY Hamiltonian is a yes instance.
\end{proof}

\section{Proof of the Occupancy Constraints Lemma}
\label{app:Occupancy-Constraints-Lemma}

In this Appendix we prove the \ocl.

\occup*

\subsection{Definitions and notation}
\label{sec:Definitions-and-Notation_G_square}

In this Section we establish notation and we describe how the gate graph $G^{\square}$ is constructed from $G$ and $\goc$. We also define two related gate graphs $G^{\triangle}$ and $G^{\diamondsuit}$ that we use in our analysis.

Let us first fix notation for the gate graph $G$ and the occupancy constraints graph $\goc$. Write the adjacency matrix of $G$ as (see equation \eq{adj_gate_graph}) 
\[
A(G)=\sum_{q=1}^{R}|q\rangle\langle q|\otimes A(g_{0})+h_{\mathcal{E}^{G}}+h_{\mathcal{S}^{G}}
\]
where $h_{\mathcal{E}^{G}}$ and $h_{\mathcal{S}^{G}}$ are determined (through equations \eq{h_edges} and \eq{h_loops}) by the sets $\mathcal{E}^{G}$ and $\mathcal{S}^{G}$ of edges and self-loops in the gate diagram for $G$, and where $g_0$ is the 128-vertex graph from \fig{g_0}. Recall that the occupancy constraints graph $\goc$ is a simple graph with vertices labeled $q\in[R]$, one for each diagram element in $G$. We write $E(\goc)\subseteq\binom{[R]}{2} = \{\{x,y\}\colon x,y\in[R],\, x\neq y\}$ for the edge set of $\goc$. 

\subsubsection*{Definition of $G^{\square}$}

To ensure that the ground space has the appropriate form, the construction of $G^{\square}$ is slightly different depending on whether $R$ is even or odd. The following description handles both cases.

\begin{enumerate}
\item Replace each diagram element $q\in[R]$ in the gate diagram for $G$ as shown in \fig{replace_gate_diagram}, with diagram elements labeled $q_{\inn},q_{\out}$ and $d(q,s)$ where $q,s\in[R]$ and $q\neq s$ if $R$ is even. Each node $(q,z,t)$ in the gate diagram for $G$ is mapped to a new node $\new(q,z,t)$ as shown by the black and grey arrows, i.e., 
\begin{equation}
\new(q,z,t)=\begin{cases}
(q_{\mathrm{in}},z,t) & \text{ if }(q,z,t)\text{ is an input node}\\
(q_{\mathrm{out}},z,t) & \text{ if }(q,z,t)\text{ is an output node}.
\end{cases}\label{eq:node_mapping_G_G_square}
\end{equation}
Edges and self-loops in the gate diagram for $G$ are replaced by edges and self-loops between the corresponding nodes in the modified diagram.
\item For each edge $\{q_{1},q_{2}\}\in E(\goc)$ in the occupancy constraints graph we add four diagram elements of the type shown in \fig{diagram_element1} (i.e., diagram elements corresponding to the identity). We refer to these diagram elements by labels $e_{ij}(q_{1},q_{2})$ with $i,j\in\{0,1\}$. For these diagram elements the labeling function is symmetric, i.e., $e_{ij}(q_{1},q_{2})=e_{ji}(q_{2},q_{1})$ whenever $\{q_{1},q_{2}\}\in E(\goc)$.
\item For each non-edge $\{q_{1},q_{2}\}\notin E(\goc)$ with $q_{1},q_{2}\in[R]$ and $q_{1}\neq q_{2}$ we add $8$ diagram elements of the type shown in \fig{diagram_element1}. We refer to these diagram elements as $e_{ij}(q_{1},q_{2})$ and $e_{ij}(q_{2},q_{1})$ with $i,j\in\{0,1\}$; when $\{q_1,q_2\}\notin E(\goc)$ the labeling function is not symmetric, i.e., $e_{ij}(q_{1},q_{2})\neq e_{ji}(q_{2},q_{1})$. If $R$ is odd we also add $4R$ diagram elements labeled $e_{ij}(q,q)$ with $i,j\in\{0,1\}$ and $q\in[R]$.
\item Finally, we add edges and self-loops to the gate diagram as shown in \fig{add_edges}. This gives the gate diagram for $G^{\square}$.
\end{enumerate}

\begin{figure}
\centering
\begin{tikzpicture}[scale=0.5]
\begin{scope}[xshift=-10 cm]       
\draw[rounded corners=2mm,thick] (0,0) rectangle (2.43cm,1.5 cm);  
  \foreach \x /\color in {0/black,2.43/gray}
{     \foreach \y in {1.2,.95,.55,.3}
{       \draw[fill=\color,draw=\color] (\x cm, \y cm) circle (.66mm);
    }} 
\node at (1.22cm, .75cm) {\large{$\tilde U$}};   
 \node at (1.22cm, 1.85cm){\large{$q$}};

\node at (6cm, 0.75cm) {\Large{$\longrightarrow$}};
\draw [->, thick, color=black] (0,2.5) --(0,1.6);
\draw [->, thick, color=gray] (2.43,2.5) --(2.43,1.6);
\end{scope}
\draw [decorate,decoration={brace,amplitude=10pt}] (3.43,2.5) -- (14.43,2.5) node [black,midway,yshift=17]  {$d(q,q)$ is omitted if $R$ is even};
\draw [->, thick, color=black] (0,2.5) --(0,1.6);
\draw [->, thick, color=gray] (17.86,2.5) --(17.86,1.6);

 \foreach \y in {0.925,.55}{  
  \draw[thick] (2.43,\y) -- (3.43,\y);
	\draw[thick] (5.86,\y) -- (6.86,\y);	
	\draw[thick] (14.43,\y) -- (15.43,\y);
	\draw[thick] (9.29,\y) -- (9.79,\y);
	\draw[thick] (11.5,\y) -- (12,\y);
	\node at (10.6,\y){\large{$\ldots$}};
}

 \foreach \offset/\unitary/\label in {0/1/q_{\mathrm{in}},3.43/1/{d(q,1)},6.86/1/{d(q,2)},12/1/{d(q,R)},15.43/\tilde U/q_{\mathrm{out}}}
{   
\begin{scope}[xshift=\offset cm]       
\draw[rounded corners=2mm,thick] (0,0) rectangle (2.43cm,1.5 cm);  
  \foreach \x /\color in {0/black,2.43/gray}
{     \foreach \y in {1.2,.95,.55,.3}
{       \draw[fill=\color,draw=\color] (\x cm, \y cm) circle (.66mm);
    }} 
  \node at (1.22cm, .75cm) {\large{$\unitary$}};   
 \node at (1.22cm, 2cm){\large{$\label$}};   
\end{scope}}   
\end{tikzpicture}

\caption{The first step in constructing the gate diagram of $G^{\square}$
from that of $G$ is to replace each diagram element as shown. The
four input nodes (black arrow) and four output nodes (grey arrow)
on the left-hand side are identified with nodes on the right-hand
side as shown.\label{fig:replace_gate_diagram}}
\end{figure}
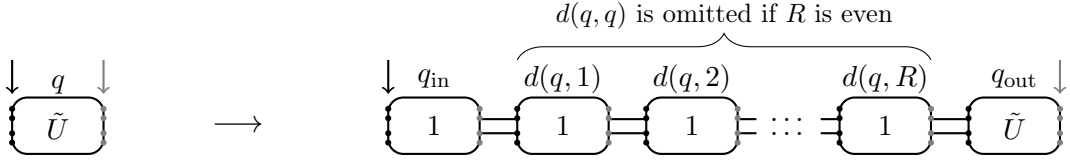

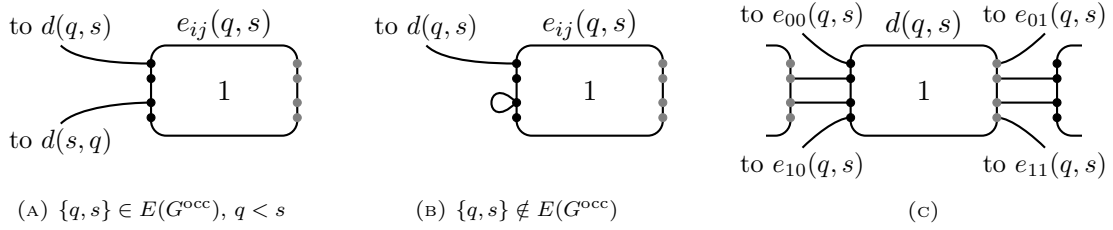
\begin{figure}
\centering
\subfloat[b][$\{q,s\}\in E(\goc),\,q<s$]{
\label{fig:add_edgesA}
\begin{tikzpicture}[baseline=(X.base),scale=0.8]
\node at (0,0) (X){};
\begin{scope}[xshift=0cm]       
\draw[rounded corners=2mm,thick] (0,0) rectangle (2.43cm,1.5 cm);  
  \foreach \x /\color in {0/black,2.43/gray}
{     \foreach \y in {1.2,.95,.55,.3}
{       \draw[fill=\color,draw=\color] (\x cm, \y cm) circle (.66mm);
    }} 
\node at (1.22cm, .75cm) {\large{$1$}};   
 \node at (1.22cm, 1.85cm){\large{$e_{ij}(q,s)$}};
\draw[thick,looseness=.66] (0cm,1.2cm) to [out=180,in=-45] (-1.5cm,1.5cm);
\node at (-1.5,1.8){to $d(q,s)$};
\draw[thick,looseness=.66] (0cm,0.55cm) to [out=180,in=45] (-1.5cm,0.2cm);
\node at (-1.5,-0.1){to $d(s,q)$};
\end{scope}
\node at (0,-0.7){};
\end{tikzpicture}
}
\hspace{.5cm}
\subfloat[b][$\{q,s\}\notin E(\goc)$]{
\label{fig:add_edgesB}
\begin{tikzpicture}[baseline=(X.base),scale=0.8]
\node at (0,0) (X){};

\begin{scope}[xshift=0cm]       
\draw[rounded corners=2mm,thick] (0,0) rectangle (2.43cm,1.5 cm);  
  \foreach \x /\color in {0/black,2.43/gray}
{     \foreach \y in {1.2,.95,.55,.3}
{       \draw[fill=\color,draw=\color] (\x cm, \y cm) circle (.66mm);
    }} 
\node at (1.22cm, .75cm) {\large{$1$}};   
 \node at (1.22cm, 1.85cm){\large{$e_{ij}(q,s)$}};
\draw[thick,looseness=.66] (0cm,1.2cm) to [out=180,in=-45] (-1.5cm,1.5cm);
\node at (-1.5,1.8){to $d(q,s)$};
\draw[thick,looseness = 200] (0,0.55) to [out = 135, in = 225] (-.01,0.55);
\end{scope}
\node at (0,-0.7){};
\end{tikzpicture}
}
\hspace{.5cm}
\subfloat[b][]{
\label{fig:add_edgesC}
\begin{tikzpicture}[baseline=(X.base),scale=0.8]
\node at (0,0) (X){};
\draw[thick,looseness=.66] (0cm,1.2cm) to [out=180,in=-45] (-0.8cm,1.7cm);
\node at (-0.8cm,2cm){to $e_{00}(q,s)$};
\draw[thick,looseness=.66] (0cm,0.3cm) to [out=180,in=45] (-0.8cm,-0.2cm);
\node at (-0.8cm,-0.5cm){to $e_{10}(q,s)$};

\draw[thick,looseness=.66] (2.43cm,1.2cm) to [out=0,in=225] (3.23cm,1.7cm);
\node at (3.23cm,2cm){to $e_{01}(q,s)$};
\draw[thick,looseness=.66] (2.43cm,0.3cm) to [out=0,in=135] (3.23cm,-0.2cm);
\node at (3.23cm,-0.5cm){to $e_{11}(q,s)$};
\draw[thick] (-1,0.95)--(0,0.95);
\draw[thick] (-1,0.55)--(0,0.55);
\draw[thick] (2.43,0.95)--(3.43,0.95);
\draw[thick] (2.43,0.55)--(3.43,0.55);
\node at (0,-0.7){};
\begin{scope}
\clip (-1.4,-0.1) rectangle (3.83,2.1);
     
\draw[rounded corners=2mm,thick] (0,0) rectangle (2.43cm,1.5 cm);  
  \foreach \x /\color in {0/black,2.43/gray}
{     \foreach \y in {1.2,.95,.55,.3}
{       \draw[fill=\color,draw=\color] (\x cm, \y cm) circle (.66mm);
    }} 
\node at (1.22cm, .75cm) {\large{$1$}};   
 \node at (1.22cm, 1.85cm){\large{$d(q,s)$}};
\begin{scope}[xshift=3.43cm]
\draw[rounded corners=2mm,thick] (0,0) rectangle (2.43cm,1.5 cm);  
  \foreach \x /\color in {0/black,2.43/gray}
{     \foreach \y in {1.2,.95,.55,.3}
{       \draw[fill=\color,draw=\color] (\x cm, \y cm) circle (.66mm);
    }} 
\node at (1.22cm, .75cm) {\large{$1$}};   
 \node at (1.22cm, 1.85cm){\large{$d(q,\widehat{q})$}};
\end{scope}
\begin{scope}[xshift=-3.43cm]
\draw[rounded corners=2mm,thick] (0,0) rectangle (2.43cm,1.5 cm);  
  \foreach \x /\color in {0/black,2.43/gray}
{     \foreach \y in {1.2,.95,.55,.3}
{       \draw[fill=\color,draw=\color] (\x cm, \y cm) circle (.66mm);
    }} 
\node at (1.22cm, .75cm) {\large{$1$}};   
\node at (1.22cm, 1.85cm){\large{$d(q,s)$}};
\end{scope}
\end{scope}
\end{tikzpicture}
}
\caption{Edges and self-loops added in step 4 of the construction of the gate diagram of $G^{\square}$. When $\{q,s\}\in E(\goc)$ with $q<s$, we add two outgoing edges to $e_{ij}(q,s)$ as shown in \subfig{add_edgesA}. Note that if $q>s$ and $\{q,s\}\in E(\goc)$ then $e_{ij}(q,s)=e_{ji}(s,q)$. When $\{q,s\}\notin E(\goc)$ we add a self-loop and a single outgoing edge from $e_{ij}(q,s)$ as shown in \subfig{add_edgesB}. Each diagram element $d(q,s)$ has eight outgoing edges (four of which are added in step 4), as shown in \subfig{add_edgesC}.\label{fig:add_edges}}
\end{figure}

The set of diagram elements in the gate graph for $G^{\square}$ is indexed by
\begin{equation}
L^{\square}=Q_{\inn}\cup D\cup E_{\text{edges}}\cup E_{\text{non-edges}}\cup Q_{\out}\label{eq:L_square}
\end{equation}
where
\begin{align}
Q_{\inn} & =\left\{ q_{\mathrm{in}}:q\in[R]\right\} \label{eq:Q_in}\\
D & =\left\{ d(q,s):\, q,s\in[R]\text{ and }q\neq s\text{ if }R\text{ is even}\right\} \label{eq:defn_of D}\\
E_{\text{edges}} & =\left\{ e_{ij}(q,s):\, i,j\in\{0,1\},\,\{q,s\}\in E(\goc)\text{ and }q<s\right\} \nonumber \\
E_{\text{non-edges}} & =\left\{ e_{ij}(q,s):\, i,j\in\{0,1\},\,\{q,s\}\notin E(\goc)\text{ and }q\neq s\text{ if }R\text{ is even}\right\} \nonumber \\
Q_{\out} & =\left\{ q_{\mathrm{out}}:q\in[R]\right\} .\label{eq:Q_out}
\end{align}
The total number of diagram elements in $G^{\square}$ is 
\begin{align*}
|L^{\square}| & =|Q_{\inn}|+|D|+|E_{\text{edges}}|+|E_{\text{non-edges}}|+|Q_{\out}|\\
 & =\begin{cases}
R+R^{2}+4|E(\goc)|+4\left(R^{2}-2|E(\goc)|\right)+R & R\text{ odd}\\
R+R\left(R-1\right)+4|E(\goc)|+4\left(R(R-1)-2|E(\goc)|\right)+R & R\text{ even}
\end{cases}\\
 & =\begin{cases}
5R^{2}+2R-4|E(\goc)| & R\text{ odd}\\
5R^{2}-3R-4|E(\goc)| & R\text{ even}.
\end{cases}
\end{align*}
In both cases this is upper bounded by $7R^2$ as claimed in the statement of the Lemma. Write 
\begin{equation}
A(G^{\square})=\sum_{l\in L^{\square}}|l\rangle\langle l|\otimes A(g_{0})+h_{\mathcal{S}^{\square}}+h_{\mathcal{E}^{\square}}\label{eq:A_G_squAre}
\end{equation}
where $\mathcal{S}^{\square}$ and $\mathcal{E}^{\square}$ are the sets of self-loops and edges in the gate diagram for $G^{\square}$. 

We now focus on the input nodes of diagram elements in $Q_{\inn}$ and the output nodes of the diagram elements in $Q_{\out}$. These are the nodes indicated by the black and grey arrows in \fig{replace_gate_diagram}. Write $\mathcal{E}^{0}\subset\mathcal{E}^{\square}$ and $\mathcal{S}^{0}\subset\mathcal{S^{\square}}$ for the sets of edges and self-loops that are incident on these nodes in the gate diagram for $G^{\square}$. Note that the sets $\mathcal{E}^{0}$ and $\mathcal{S}^{0}$ are in one-to-one correspondence with (respectively) the sets $\mathcal{E}^{G}$ and $\mathcal{S}^{G}$ of edges and self-loops in the gate diagram for $G$. The other edges and self-loops in $G^{\square}$ do not depend on the sets of edges and self-loops in $G$. Writing
\[
\mathcal{S}^{\triangle}=\mathcal{S}^{\square}\setminus S^{0}\qquad\mathcal{E}^{\triangle}=\mathcal{E}^{\square}\setminus\mathcal{E}^{0},
\]
we have 
\begin{equation}
h_{\mathcal{S}^{\square}}=h_{\mathcal{S}^{0}}+h_{\mathcal{S}^{\triangle}}\qquad h_{\mathcal{E}^{\square}}=h_{\mathcal{E}^{0}}+h_{\mathcal{E}^{\triangle}}.\label{eq:h_se_square}
\end{equation}

\subsubsection*{Definition of $G^{\triangle}$ }

The gate diagram for $G^{\triangle}$ is obtained from that of $G^{\square}$ by removing all edges and self-loops attached to the input nodes of the diagram elements in $Q_{\inn}$ and the output nodes of the diagram elements in $Q_{\out}$. Its adjacency matrix is
\begin{equation}
A(G^{\triangle})=\sum_{l\in L^{\square}}|l\rangle\langle l|\otimes A(g_{0})+h_{\mathcal{S}^{\triangle}}+h_{\mathcal{E}^{\triangle}}.\label{eq:A_g_triangle}
\end{equation}
Note that $G^{\triangle}=G^{\square}$ whenever the gate diagram for $G$ contains no edges or self-loops. 

\subsubsection*{Definition of $G^{\diamondsuit}$}

We also define a gate graph $G^{\diamondsuit}$ with gate diagram obtained from that of $G^{\triangle}$ by removing all edges (but leaving the self-loops). Note that $G^{\diamondsuit}$ has a component for each diagram element $l\in L^{\square}$. The components corresponding to diagram elements without a self-loop (those with $l\in L^{\square}\setminus E_{\text{non-edges}}$) have adjacency matrix $A(g_{0})$; those with a self-loop ($l\in E_{\text{non-edges}}$) have adjacency matrix $A(g_{0})+|1,1\rangle\langle1,1|\otimes\id$, so
\begin{align}
A(G^{\diamondsuit}) & =\sum_{l\in L^{\square}}|l\rangle\langle l|\otimes A(g_{0})+h_{\mathcal{S}^{\triangle}}\label{eq:A_G_diamond_defn_line1}\\
 & =\sum_{l\in L^{\square}\setminus E_{\text{non-edges}}}|l\rangle\langle l|\otimes A(g_{0})+\sum_{l\in E_{\text{non-edges}}}|l\rangle\langle l|\otimes\left(A(g_{0})+|1,1\rangle\langle1,1|\otimes\id\right).\label{eq:A_G_diamond_defn}
\end{align}

\subsubsection*{Example}

We provide an example of this construction in \fig{example_G_Gtilde} (which shows a gate graph and an occupancy constraints graph) and \fig{big_example_G_square} (which describes the derived gate graphs $G^{\square}$, $G^{\triangle}$, and $G^{\diamondsuit}$).

\begin{figure}
\subfloat[][]{
\label{fig:example_G_GtildeA}
\begin{tikzpicture}[scale=0.5]

\draw[thick,looseness=1.3] (0,-2.2cm) to [out=180,in=180] (0,-3.8cm);
\draw[thick,looseness = 200] (2.43,-4.7) to [out = 45, in = -45] (2.42,-4.7);  
 \foreach \offset/\unitary/\label in {0/H/1,-2.5/HT/2,-5/1/3}
{   
\begin{scope}[yshift=\offset cm]       
\draw[rounded corners=2mm,thick] (0,0) rectangle (2.43cm,1.5 cm);  
  \foreach \x /\color in {0/black,2.43/gray}
{     \foreach \y in {1.2,.95,.55,.3}
{       \draw[fill=\color,draw=\color] (\x cm, \y cm) circle (.66mm);
    }} 
  \node at (1.22cm, .75cm) {\large{$\unitary$}};   
 \node at (1.22cm, 2cm){\large{$\label$}};   
\end{scope}}  
\end{tikzpicture}
}
\hspace{4cm}
\subfloat[][]{
\label{fig:example_G_GtildeB}
\begin{tikzpicture}[scale=0.5]

 \draw[fill=black,draw=black] (1,0) circle (.66mm);
\node at (1.5,0) {1};
 \draw[fill=black,draw=black] (1,-2.5) circle (.66mm);
\node at (1.5,-2.5) {2};
 \draw[fill=black,draw=black] (1,-5) circle (.66mm);
\node at (1.5,-5) {3};
\draw[thick] (1,0) -- (1,-2.5);
\end{tikzpicture}
}

\caption{An example \subfig{example_G_GtildeA} Gate diagram for a gate graph $G$ and \subfig{example_G_GtildeB} Occupancy constraints graph $\goc$. In the text we describe how these two ingredients are mapped to a gate graph $G^{\square}$; the gate diagram for $G^{\square}$ is shown in \fig{big_example_G_square}.\label{fig:example_G_Gtilde}}
\end{figure}
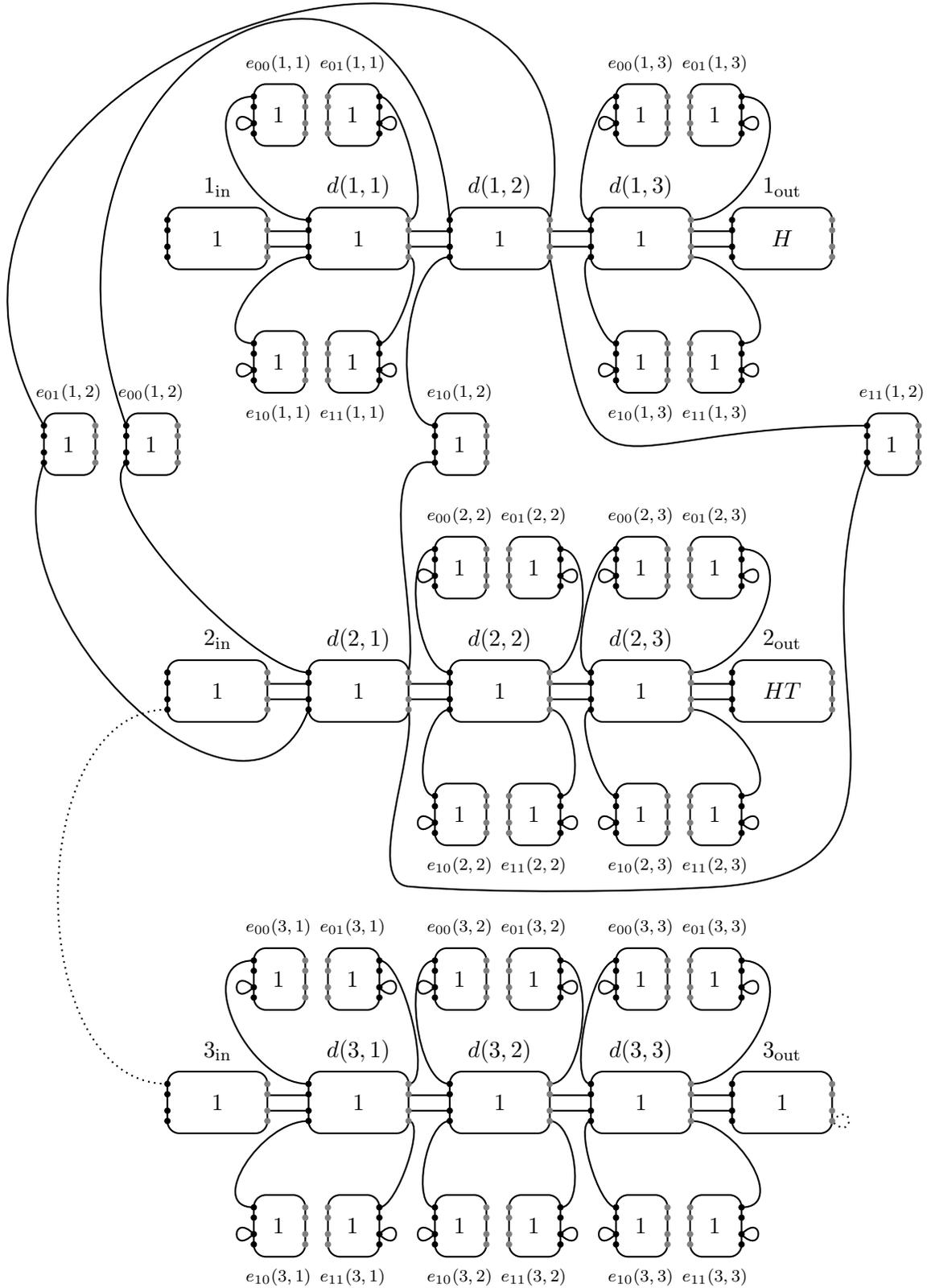
\begin{figure}

\begin{tikzpicture}[scale=0.68]
\path[use as bounding box](-4,-25) rectangle (18.5,7);
\begin{scope}[yshift=-5cm]
\foreach \offset/\unitary/\label in
{-1/1/{e_{00}(1,2)}}
{   
\begin{scope}[xshift=\offset cm]       
\draw[rounded corners=2mm,thick] (0,0) rectangle (1.25cm,1.5 cm);  
  \foreach \x /\color in {0/black,1.25/gray}
{     \foreach \y in {1.2,.95,.55,.3}
{       \draw[fill=\color,draw=\color] (\x cm, \y cm) circle (.66mm);
    }} 
  \node at (0.6cm, .75cm) {\large{$\unitary$}};   
 \node at (0.6cm, 2cm){\footnotesize{$\label$}};
\end{scope}}
\foreach \offset/\unitary/\label in
{-3/1/{e_{01}(1,2)}}
{   
\begin{scope}[xshift=\offset cm]       
\draw[rounded corners=2mm,thick] (0,0) rectangle (1.25cm,1.5 cm);  
  \foreach \x /\color in {0/black,1.25/gray}
{     \foreach \y in {1.2,.95,.55,.3}
{       \draw[fill=\color,draw=\color] (\x cm, \y cm) circle (.66mm);
    }} 
  \node at (0.6cm, .75cm) {\large{$\unitary$}};   
 \node at (0.6cm, 2cm){\footnotesize{$\label$}}; 
\end{scope}}
\begin{scope}
\foreach \offset/\unitary/\label in
{6.5/1/{e_{10}(1,2)}}
{   
\begin{scope}[xshift=\offset cm]       
\draw[rounded corners=2mm,thick] (0,0) rectangle (1.25cm,1.5 cm);  
  \foreach \x /\color in {0/black,1.25/gray}
{     \foreach \y in {1.2,.95,.55,.3}
{       \draw[fill=\color,draw=\color] (\x cm, \y cm) circle (.66mm);
    }} 
  \node at (0.6cm, .75cm) {\large{$\unitary$}};   
 \node at (0.6cm, 2cm){\footnotesize{$\label$}};
\end{scope}}
\foreach \offset/\unitary/\label in
{17/1/{e_{11}(1,2)}}
{   
\begin{scope}[xshift=\offset cm]       
\draw[rounded corners=2mm,thick] (0,0) rectangle (1.25cm,1.5 cm);  
  \foreach \x /\color in {0/black,1.25/gray}
{     \foreach \y in {1.2,.95,.55,.3}
{       \draw[fill=\color,draw=\color] (\x cm, \y cm) circle (.66mm);
    }} 
  \node at (0.6cm, .75cm) {\large{$\unitary$}};   
 \node at (0.6cm, 2cm){\footnotesize{$\label$}}; 
\end{scope}}
\end{scope}
\draw[thick,looseness = 2.6] (6.86,6.2) to [out = 95, in = 110] (-1,1.2); 
\draw[thick,looseness = 0.5] (3.43,-4.8) to [out = 180, in = 225] (-1,0.3); 
\draw[thick,looseness = 1.2] (3.43,-5.7) to [out = 250, in = 250] (-3,0.3); 
\draw[thick,looseness = 0.7] (6.86,5.3) to [out = 180, in = 180] (6.5,1.2); 
\draw[thick,looseness = 0.8] (5.86,-4.8) to [out = 80, in = 170] (6.5,0.3); 

\draw[thick,looseness = 2] (9.29,5.3) to [out =-80, in = 180] (17,1.2);  
\draw[thick,looseness = 1.3] (13.5,-10) to [out =5, in = 250] (17,0.3);
\draw[thick,looseness = 0.66] (5.86,-10) to [out =-5, in = 184] (13.5,-10);
\draw[thick,looseness = 0.66] (5.86,-5.7) to [out =-80, in = 175] (5.86,-10);
\draw[thick,looseness = 2] (9.29,6.2) to [out =80, in = 120] (-3,1.2); 
\end{scope}

\begin{scope}
\begin{scope}[yshift=3cm]

\foreach \offset/\unitary/\label in
{2.1/1/{e_{00}(1,1)},10.9/1/{e_{00}(1,3)}}
{   
\begin{scope}[xshift=\offset cm]       
\draw[rounded corners=2mm,thick] (0,0) rectangle (1.25cm,1.5 cm);  
  \foreach \x /\color in {0/black,1.25/gray}
{     \foreach \y in {1.2,.95,.55,.3}
{       \draw[fill=\color,draw=\color] (\x cm, \y cm) circle (.66mm);
    }} 
  \node at (0.6cm, .75cm) {\large{$\unitary$}};   
 \node at (0.6cm, 2cm){\footnotesize{$\label$}};
\draw[thick,looseness = 200] (0,0.55) to [out = 135, in = 225] (-.01,0.55);  
\end{scope}}  
\end{scope}
\begin{scope}[yshift=3cm]

\foreach \offset/\unitary/\label in
{3.9/1/{e_{01}(1,1)},12.7/1/{e_{01}(1,3)}}
{   
\begin{scope}[xshift=\offset cm]       
\draw[rounded corners=2mm,thick] (0,0) rectangle (1.25cm,1.5 cm);  
  \foreach \x /\color in {1.25/black,0/gray}
{     \foreach \y in {1.2,.95,.55,.3}
{       \draw[fill=\color,draw=\color] (\x cm, \y cm) circle (.66mm);
    }} 
  \node at (0.6cm, .75cm) {\large{$\unitary$}};   
 \node at (0.6cm, 2cm){\footnotesize{$\label$}};
\draw[thick,looseness = 200] (1.25,0.55) to [out = 45, in = -45] (1.24,0.55);  
\end{scope}}  
\draw[thick,looseness = 1.2] (3.43,-1.8) to [out = 180, in = 180] (2.1,1.2); 
\draw[thick,looseness = 0.5] (5.86,-1.8) to [out = 0, in = 0] (5.15,1.2); 
\draw[thick,looseness = 0.5] (10.29,-1.8) to [out = 180, in = 180] (10.9,1.2); 
\draw[thick,looseness = 1.2] (12.72,-1.8) to [out = 0, in = 0] (13.95,1.2); 
\end{scope}
\begin{scope}[yshift=-3cm]

\foreach \offset/\unitary/\label in
{2.1/1/{e_{10}(1,1)},10.9/1/{e_{10}(1,3)}}
{   
\begin{scope}[xshift=\offset cm]       
\draw[rounded corners=2mm,thick] (0,0) rectangle (1.25cm,1.5 cm);  
  \foreach \x /\color in {0/black,1.25/gray}
{     \foreach \y in {1.2,.95,.55,.3}
{       \draw[fill=\color,draw=\color] (\x cm, \y cm) circle (.66mm);
    }} 
\node at (0.6cm, .75cm) {\large{$\unitary$}};   
 \node at (0.6cm, -0.5cm){\footnotesize{$\label$}};   
\draw[thick,looseness = 200] (0,0.55) to [out = 135, in = 225] (-.01,0.55);  
\end{scope}}  
\end{scope}
\begin{scope}[yshift=-3cm]
\foreach \offset/\unitary/\label in
{3.9/1/{e_{11}(1,1)},12.7/1/{e_{11}(1,3)}}
{   
\begin{scope}[xshift=\offset cm]       
\draw[rounded corners=2mm,thick] (0,0) rectangle (1.25cm,1.5 cm);  
  \foreach \x /\color in {1.25/black,0/gray}
{     \foreach \y in {1.2,.95,.55,.3}
{       \draw[fill=\color,draw=\color] (\x cm, \y cm) circle (.66mm);
    }} 
  \node at (0.6cm, .75cm) {\large{$\unitary$}};   
 \node at (0.6cm, -0.5cm){\footnotesize{$\label$}};   
\draw[thick,looseness = 200] (1.25,0.55) to [out = 45, in = -45] (1.24,0.55);  
\end{scope}}  
\draw[thick,looseness = 1.2] (3.43,3.3) to [out = 180, in = 180] (2.1,1.2); 
\draw[thick,looseness = 0.5] (5.86,3.3) to [out = 0, in = 0] (5.15,1.2); 
\draw[thick,looseness = 0.5] (10.29,3.3) to [out = 180, in = 180] (10.9,1.2); 
\draw[thick,looseness = 1.2] (12.72,3.3) to [out = 0, in = 0] (13.95,1.2); 
\end{scope}
 
 \foreach \y in {0.925,.55}{  
  \draw[thick] (2.43,\y) -- (3.43,\y);
	\draw[thick] (5.86,\y) -- (6.86,\y);	
	\draw[thick] (9.29,\y) -- (10.29,\y);
	\draw[thick] (12.72,\y) -- (13.72,\y);
}
  
 \foreach \offset/\unitary/\label in {0/1/1_{\mathrm{in}},3.43/1/{d(1,1)},6.86/1/{d(1,2)},10.29/1/{d(1,3)},13.72/H/1_{\mathrm{out}}}
{   
\begin{scope}[xshift=\offset cm]       
\draw[rounded corners=2mm,thick] (0,0) rectangle (2.43cm,1.5 cm);  
  \foreach \x /\color in {0/black,2.43/gray}
{     \foreach \y in {1.2,.95,.55,.3}
{       \draw[fill=\color,draw=\color] (\x cm, \y cm) circle (.66mm);
    }} 
  \node at (1.22cm, .75cm) {\large{$\unitary$}};   
 \node at (1.22cm, 2cm){\large{$\label$}};   
\end{scope}}  
\end{scope}

\begin{scope}[yshift=-11cm]
\begin{scope}[yshift=3cm]

\foreach \offset/\unitary/\label in
{6.5/1/{e_{00}(2,2)},10.9/1/{e_{00}(2,3)}}
{   
\begin{scope}[xshift=\offset cm]       
\draw[thick,looseness = 200] (0,0.55) to [out = 135, in = 225] (-.01,0.55);  
\draw[rounded corners=2mm,thick] (0,0) rectangle (1.25cm,1.5 cm);  
  \foreach \x /\color in {0/black,1.25/gray}
{     \foreach \y in {1.2,.95,.55,.3}
{       \draw[fill=\color,draw=\color] (\x cm, \y cm) circle (.66mm);
    }} 
  \node at (0.6cm, .75cm) {\large{$\unitary$}};   
 \node at (0.6cm, 2cm){\footnotesize{$\label$}};
\end{scope}}  
\end{scope}
\begin{scope}[yshift=3cm]

\draw[thick,looseness = 0.7] (6.86,-1.8) to [out = 180, in = 180] (6.5,1.2); 
\draw[thick,looseness = 0.75] (9.29,-1.8) to [out = 0, in = 0] (9.55,1.2); 
\draw[thick,looseness = 0.5] (10.29,-1.8) to [out = 180, in = 180] (10.9,1.2); 
\draw[thick,looseness = 1.2] (12.72,-1.8) to [out = 0, in = 0] (13.95,1.2); 
\foreach \offset/\unitary/\label in
{8.3/1/{e_{01}(2,2)},12.7/1/{e_{01}(2,3)}}
{   
\begin{scope}[xshift=\offset cm]       
\draw[thick,looseness = 200] (1.25,0.55) to [out = 45, in = -45] (1.24,0.55);  
\draw[rounded corners=2mm,thick] (0,0) rectangle (1.25cm,1.5 cm);  
  \foreach \x /\color in {1.25/black,0/gray}
{     \foreach \y in {1.2,.95,.55,.3}
{       \draw[fill=\color,draw=\color] (\x cm, \y cm) circle (.66mm);
    }} 
  \node at (0.6cm, .75cm) {\large{$\unitary$}};   
 \node at (0.6cm, 2cm){\footnotesize{$\label$}};
\end{scope}}  
\end{scope}
\begin{scope}[yshift=-3cm]

\foreach \offset/\unitary/\label in
{6.5/1/{e_{10}(2,2)},10.9/1/{e_{10}(2,3)}}
{   
\begin{scope}[xshift=\offset cm]       
\draw[thick,looseness = 200] (0,0.55) to [out = 135, in = 225] (-.01,0.55);  
\draw[rounded corners=2mm,thick] (0,0) rectangle (1.25cm,1.5 cm);  
  \foreach \x /\color in {0/black,1.25/gray}
{     \foreach \y in {1.2,.95,.55,.3}
{       \draw[fill=\color,draw=\color] (\x cm, \y cm) circle (.66mm);
    }} 
\node at (0.6cm, .75cm) {\large{$\unitary$}};   
 \node at (0.6cm, -0.5cm){\footnotesize{$\label$}};   
\end{scope}}  
\end{scope}
\begin{scope}[yshift=-3cm]

\foreach \offset/\unitary/\label in
{8.3/1/{e_{11}(2,2)},12.7/1/{e_{11}(2,3)}}
{
\begin{scope}[xshift=\offset cm]       
\draw[rounded corners=2mm,thick] (0,0) rectangle (1.25cm,1.5 cm);  
  \foreach \x /\color in {1.25/black,0/gray}
{     \foreach \y in {1.2,.95,.55,.3}
{       \draw[fill=\color,draw=\color] (\x cm, \y cm) circle (.66mm);
    }} 
  \node at (0.6cm, .75cm) {\large{$\unitary$}};   
 \node at (0.6cm, -0.5cm){\footnotesize{$\label$}};   
\draw[thick,looseness = 200] (1.25,0.55) to [out = 45, in = -45] (1.24,0.55);  
\end{scope}}   
\draw[thick,looseness = 0.7] (6.86,3.3) to [out = 180, in = 180] (6.5,1.2); 
\draw[thick,looseness = 0.75] (9.29,3.3) to [out = 0, in = 0] (9.55,1.2); 
\draw[thick,looseness = 0.5] (10.29,3.3) to [out = 180, in = 180] (10.9,1.2); 
\draw[thick,looseness = 1.2] (12.72,3.3) to [out = 0, in = 0] (13.95,1.2); 
\end{scope}

 \foreach \y in {0.925,.55}{  
  \draw[thick] (2.43,\y) -- (3.43,\y);
	\draw[thick] (5.86,\y) -- (6.86,\y);	
	\draw[thick] (9.29,\y) -- (10.29,\y);
	\draw[thick] (12.72,\y) -- (13.72,\y);
}

 \foreach \offset/\unitary/\label in {0/1/2_{\mathrm{in}},3.43/1/{d(2,1)},6.86/1/{d(2,2)},10.29/1/{d(2,3)},13.72/HT/2_{\mathrm{out}}}
{   
\begin{scope}[xshift=\offset cm]       
\draw[rounded corners=2mm,thick] (0,0) rectangle (2.43cm,1.5 cm);  
  \foreach \x /\color in {0/black,2.43/gray}
{     \foreach \y in {1.2,.95,.55,.3}
{       \draw[fill=\color,draw=\color] (\x cm, \y cm) circle (.66mm);
    }} 
  \node at (1.22cm, .75cm) {\large{$\unitary$}};   
 \node at (1.22cm, 2cm){\large{$\label$}};   
\end{scope}}
\end{scope}

\begin{scope}[yshift=-21cm]

\draw[thick,dotted,looseness = 200] (16.15,0.3) to [out = 45, in = -45] (16.14,0.3); 
\draw[thick,dotted,looseness = 1] (0,1.2) to [out =180, in = 180] (0,10.3);

\begin{scope}[yshift=3cm]
\foreach \offset/\unitary/\label in
{2.1/1/{e_{00}(3,1)},6.5/1/{e_{00}(3,2)},10.9/1/{e_{00}(3,3)}}
{   
\begin{scope}[xshift=\offset cm]       
\draw[rounded corners=2mm,thick] (0,0) rectangle (1.25cm,1.5 cm);  
  \foreach \x /\color in {0/black,1.25/gray}
{     \foreach \y in {1.2,.95,.55,.3}
{       \draw[fill=\color,draw=\color] (\x cm, \y cm) circle (.66mm);
    }} 
  \node at (0.6cm, .75cm) {\large{$\unitary$}};   
 \node at (0.6cm, 2cm){\footnotesize{$\label$}};
\draw[thick,looseness = 200] (0,0.55) to [out = 135, in = 225] (-.01,0.55);  
\end{scope}}  
\end{scope}
\begin{scope}[yshift=3cm]

\foreach \offset/\unitary/\label in
{3.9/1/{e_{01}(3,1)},8.3/1/{e_{01}(3,2)},12.7/1/{e_{01}(3,3)}}
{   
\begin{scope}[xshift=\offset cm]       
\draw[rounded corners=2mm,thick] (0,0) rectangle (1.25cm,1.5 cm);  
  \foreach \x /\color in {1.25/black,0/gray}
{     \foreach \y in {1.2,.95,.55,.3}
{       \draw[fill=\color,draw=\color] (\x cm, \y cm) circle (.66mm);
    }} 
  \node at (0.6cm, .75cm) {\large{$\unitary$}};   
 \node at (0.6cm, 2cm){\footnotesize{$\label$}};
\draw[thick,looseness = 200] (1.25,0.55) to [out = 45, in = -45] (1.24,0.55);  
\end{scope}}  
\draw[thick,looseness = 1.2] (3.43,-1.8) to [out = 180, in = 180] (2.1,1.2); 
\draw[thick,looseness = 0.5] (5.86,-1.8) to [out = 0, in = 0] (5.15,1.2); 
\draw[thick,looseness = 0.7] (6.86,-1.8) to [out = 180, in = 180] (6.5,1.2); 
\draw[thick,looseness = 0.75] (9.29,-1.8) to [out = 0, in = 0] (9.55,1.2); 
\draw[thick,looseness = 0.5] (10.29,-1.8) to [out = 180, in = 180] (10.9,1.2); 
\draw[thick,looseness = 1.2] (12.72,-1.8) to [out = 0, in = 0] (13.95,1.2); 
\end{scope}
\begin{scope}[yshift=-3cm]
\foreach \offset/\unitary/\label in
{2.1/1/{e_{10}(3,1)},6.5/1/{e_{10}(3,2)},10.9/1/{e_{10}(3,3)}}
{   
\begin{scope}[xshift=\offset cm]       
\draw[rounded corners=2mm,thick] (0,0) rectangle (1.25cm,1.5 cm);  
  \foreach \x /\color in {0/black,1.25/gray}
{     \foreach \y in {1.2,.95,.55,.3}
{       \draw[fill=\color,draw=\color] (\x cm, \y cm) circle (.66mm);
    }} 
\node at (0.6cm, .75cm) {\large{$\unitary$}};   
 \node at (0.6cm, -0.5cm){\footnotesize{$\label$}};   
\draw[thick,looseness = 200] (0,0.55) to [out = 135, in = 225] (-.01,0.55);  
\end{scope}}  
\end{scope}
\begin{scope}[yshift=-3cm]
\foreach \offset/\unitary/\label in
{3.9/1/{e_{11}(3,1)},8.3/1/{e_{11}(3,2)},12.7/1/{e_{11}(3,3)}}
{   
\begin{scope}[xshift=\offset cm]       
\draw[rounded corners=2mm,thick] (0,0) rectangle (1.25cm,1.5 cm);  
  \foreach \x /\color in {1.25/black,0/gray}
{     \foreach \y in {1.2,.95,.55,.3}
{       \draw[fill=\color,draw=\color] (\x cm, \y cm) circle (.66mm);
    }} 
  \node at (0.6cm, .75cm) {\large{$\unitary$}};   
 \node at (0.6cm, -0.5cm){\footnotesize{$\label$}};   
\draw[thick,looseness = 200] (1.25,0.55) to [out = 45, in = -45] (1.24,0.55);  
\end{scope}}  
\draw[thick,looseness = 1.2] (3.43,3.3) to [out = 180, in = 180] (2.1,1.2); 
\draw[thick,looseness = 0.5] (5.86,3.3) to [out = 0, in = 0] (5.15,1.2); 
\draw[thick,looseness = 0.7] (6.86,3.3) to [out = 180, in = 180] (6.5,1.2); 
\draw[thick,looseness = 0.75] (9.29,3.3) to [out = 0, in = 0] (9.55,1.2); 
\draw[thick,looseness = 0.5] (10.29,3.3) to [out = 180, in = 180] (10.9,1.2); 
\draw[thick,looseness = 1.2] (12.72,3.3) to [out = 0, in = 0] (13.95,1.2); 
\end{scope}
  
 \foreach \y in {0.925,.55}{  
  \draw[thick] (2.43,\y) -- (3.43,\y);
	\draw[thick] (5.86,\y) -- (6.86,\y);	
	\draw[thick] (9.29,\y) -- (10.29,\y);
	\draw[thick] (12.72,\y) -- (13.72,\y);
}

 \foreach \offset/\unitary/\label in {0/1/3_{\mathrm{in}},3.43/1/{d(3,1)},6.86/1/{d(3,2)},10.29/1/{d(3,3)},13.72/1/3_{\mathrm{out}}}
{   
\begin{scope}[xshift=\offset cm]       
\draw[rounded corners=2mm,thick] (0,0) rectangle (2.43cm,1.5 cm);  
  \foreach \x /\color in {0/black,2.43/gray}
{     \foreach \y in {1.2,.95,.55,.3}
{       \draw[fill=\color,draw=\color] (\x cm, \y cm) circle (.66mm);
    }} 
  \node at (1.22cm, .75cm) {\large{$\unitary$}};   
 \node at (1.22cm, 2cm){\large{$\label$}};   
\end{scope}}  
\end{scope}
\end{tikzpicture}

\caption{The gate diagram for $G^{\triangle}$ (only solid lines) and $G^{\square}$ (including dotted lines) derived from the example gate graph $G$ and occupancy constraints graph $\goc$ from \fig{example_G_Gtilde}. The gate diagram for $G^{\diamondsuit}$ is obtained from that of $G^{\triangle}$ by removing all edges (but leaving the self-loops).\label{fig:big_example_G_square}}
\end{figure}

\subsection{The gate graph $G^{\diamondsuit}$}
\label{sec:The-gate-graph_G_DIAMOND}

We now solve for the $e_{1}$-energy ground states of the adjacency matrix $A(G^{\diamondsuit})$. Write $g_{1}$ for the graph with adjacency matrix
\[
A(g_{1})=A(g_{0})+|1,1\rangle\langle1,1|\otimes\id
\]
(i.e., $g_0$ with $8$ self-loops added), so (recalling equation \eq{A_G_diamond_defn}) each component of $G^{\diamondsuit}$ is either $g_{0}$ or $g_{1}$. Recall from \sec{Encoding-a-Computation} that $A(g_{0})$ has four orthonormal $e_{1}$-energy ground states $|\psi_{z,a}\rangle$ with $z,a\in\{0,1\}$. It is also not hard to verify that the $e_{1}$-energy ground space of $A(g_{1})$ is spanned by two of these states $|\psi_{0,a}\rangle$ for $a\in\{0,1\}$. Now letting $|\psi_{z,a}^{l}\rangle=|l\rangle|\psi_{z,a}\rangle$, we choose a basis $\mathcal{W}$ for the $e_{1}$-energy ground space of $A(G^{\diamondsuit})$ where each basis vector is supported on one of the components:
\begin{equation}
\mathcal{W}=\big\{ |\psi_{z,a}^{l}\rangle:\, z,a\in\{0,1\},\, l\in L^{\square}\setminus E_{\text{non-edges}}\big\} \cup \big\{ |\psi_{0,a}^{l}\rangle:\, a\in\{0,1\},\, l\in E_{\text{non-edges}}\big\} .\label{eq:definition_of_D}
\end{equation}
The eigenvalue gap of $A(G^{\diamondsuit})$ is equal to that of either $A(g_{0})$ or $A(g_{1})$. Since $g_{0}$ and $g_{1}$ are specific $128$-vertex graphs we can calculate their eigenvalue gaps using a computer; we get $\gamma(A(g_{0})-e_{1})=0.7785\ldots$ and $\gamma(A(g_{1})-e_{1})=0.0832\ldots$. Hence
\begin{align}
\gamma(A(G^{\diamondsuit})-e_{1}) & \geq 0.0832\ldots > \frac{1}{13}.\label{eq:one_thirteenth_bound}
\end{align}

The ground space of $A(G^{\diamondsuit})$ has dimension 
\begin{align}
|\mathcal{W}|=4\big|L^{\square}\big|  -2\left|E_{\text{non-edges}}\right|
&=\begin{cases}
4\left(5R^{2}+2R-4|E(\goc)|\right)-2\left(4R^{2}-8|E(\goc)|\right) & R\text{ odd}\\
4\left(5R^{2}-3R-4|E(\goc)|\right)-2\left(4R(R-1)-8|E(\goc)|\right) & R\text{ even}
\end{cases}\nonumber \\
&= \begin{cases}
12R^{2}+8R & R\text{ odd}\\
12R^{2}-4R & R\text{ even}.
\end{cases}\label{eq:num_basis_D}
\end{align}

We now consider the $N$-particle Hamiltonian $H(G^{\diamondsuit},N)$ and characterize its nullspace.

\begin{lemma}
\label{lem:The-nullspace-of_Hdiamond}The nullspace of $H(G^{\diamondsuit},N)$
is 
\[
\mathcal{I}_{\diamondsuit}=\spn\{ \Sym(|\psi_{z_{1},a_{1}}^{q_{1}}\rangle|\psi_{z_{2},a_{2}}^{q_{2}}\rangle\ldots|\psi_{z_{N},a_{N}}^{q_{N}}\rangle):|\psi_{z_{i},a_{i}}^{q_{i}}\rangle\in\mathcal{W}\text{ and }q_{i}\neq q_{j}\text{ for all distinct }i,j\in[N]\} 
\]
where $\mathcal{W}$ is given in equation \eq{definition_of_D}. The smallest nonzero eigenvalue satisfies $\gamma(H(G^{\diamondsuit},N)) > \frac{1}{300}$.
\end{lemma}

\begin{proof}
For the first part of the proof we use the fact that the basis vectors $|\psi_{z,a}^{l}\rangle\in\mathcal{W}$ span the $e_{1}$-eigenspace of the component $G_{l}^{\diamondsuit}$ of $G^{\diamondsuit}$ corresponding to the diagram element $l\in L^{\square},$ i.e., the nullspace of $H(G_{l}^{\diamondsuit},1)$. Furthermore, no component of $G^{\diamondsuit}$ supports a two-particle frustration-free state, i.e., $\lambda_{2}^{1}(g_{0})>0$ and $\lambda_{2}^{1}(g_{1})>0$ (by \lem{2particle}). Now applying \lem{BH_disconnected_graphs} we see that $\mathcal{I}_{\diamondsuit}$ is the nullspace of $H(G^{\diamondsuit},N)$. We also see that the smallest nonzero eigenvalue $\gamma(H(G^{\diamondsuit},N))$ is either $\lambda_{2}^{1}(g_{0})$, $\lambda_{2}^{1}(g_{1})$, $\gamma(H(g_{0},1))$, or $\gamma(H(g_{1},1))$. These constants can be calculated numerically using a computer; they are $\lambda_{2}^{1}(g_{0})=0.0035\ldots$, $\lambda_{2}^{1}(g_{1})=0.0185\ldots$, $\gamma(H(g_{0},1))=0.7785\ldots$, and $\gamma(H(g_{1},1))=0.0832\ldots$. Hence
\[
\gamma(H(G^{\diamondsuit},N))\geq\min\{ \lambda_{2}^{1}(g_{0}),\lambda_{2}^{1}(g_{1}),\gamma(H(g_{0},1)),\gamma(H(g_{1},1))\} > \frac{1}{300}. \qedhere
\]
\end{proof}

\subsection{The adjacency matrix of the gate graph $G^{\triangle}$}
\label{sec:The-gate-graph_G_triangle}

We begin by solving for the $e_{1}$-energy ground space of the adjacency matrix $A(G^{\triangle})$. From equations \eq{A_g_triangle} and \eq{A_G_diamond_defn_line1} we have
\begin{equation}
A(G^{\triangle})=A(G^{\diamondsuit})+h_{\mathcal{E}^{\triangle}}.\label{eq:A_g_diamond_triangle}
\end{equation}
Recall the $e_{1}$-energy ground space of $A(G^{\diamondsuit})$ is spanned by $\mathcal{W}$ from equation \eq{definition_of_D}. Since $h_{\mathcal{E}^{\triangle}}\geq0$ it follows that $A(G^{\triangle})\geq e_{1}$. To solve for the $e_{1}$-energy groundpsace of $A(G^{\triangle})$ we construct superpositions of vectors from $\mathcal{W}$ that are in the nullspace of $h_{\mathcal{E}^{\triangle}}$. To this end we consider the restriction
\begin{equation}
h_{\mathcal{E}^{\triangle}}\big|_{\spn\left(\mathcal{W}\right)}.\label{eq:restriction_h_triangle}
\end{equation}
We now show that it is block diagonal in the basis $\mathcal{W}$ and we compute its matrix elements.

First recall from equation \eq{h_edges} that
\begin{equation}
h_{\mathcal{E}^{\triangle}}=\sum_{\left\{ (l,z,t),(l^{\prime},z^{\prime},t^{\prime})\right\} \in\mathcal{E}^{\triangle}}\left(|l,z,t\rangle+|l^{\prime},z^{\prime},t^{\prime}\rangle\right)\left(\langle l,z,t|+\langle l^{\prime},z^{\prime},t^{\prime}|\right)\otimes\id.\label{eq:equation_for_h_epsilon_triangle}
\end{equation}
The edges $\left\{ (l,z,t),(l^{\prime},z^{\prime},t^{\prime})\right\} \in\mathcal{E}^{\triangle}$ can be read off from \fig{replace_gate_diagram} and \fig{add_edges}, respectively (referring back to \fig{diagram_elements} for our convention regarding the labeling of nodes on a diagram element). The edges from \fig{replace_gate_diagram} are
\begin{equation}
\left\{ (q_{\mathrm{in}},z,t),(d(q,1),z,t^{\prime})\right\} ,\left\{ (d(q,2),z,t),(d(q,3),z,t^{\prime})\right\} ,\ldots,\left\{ (d(q,R),z,t),(q_{\mathrm{out}},z,t^{\prime})\right\} \label{eq:epsilon_triangle_set1}
\end{equation}
with $q\in[R]$ and $\left(z,t,t^{\prime}\right)=(0,7,3)\text{ or }(1,5,1)$, and where $d(q,q)$ does not appear if $R$ is even (i.e., $d(q,q-1)$ is followed by $d(q,q+1)$. The edges from \fig{add_edges} are 
\begin{align}
\left\{ \left(d(q,s),0,1\right),(e_{00}(q,s),\alpha(q,s),1)\right\}  & ,\left\{ \left(d(q,s),1,3\right),(e_{10}(q,s),\alpha(q,s),1)\right\}, \label{eq:epsilon_triangle_set2}\\
\left\{ \left(d(q,s),0,5\right),(e_{01}(q,s),\alpha(q,s),1)\right\}  & ,\{\left(d(q,s),1,7\right),(e_{11}(q,s),\alpha(q,s),1)\}\nonumber 
\end{align}
with $q,s\in[R]$ and $q\neq s$ if $R$ is even, and where 
\[
\alpha(q,s)=\begin{cases}
1 & q>s\text{ and }\{q,s\}\in E(\goc)\\
0 & \text{otherwise}.
\end{cases}
\]
The set $\mathcal{E}^{\triangle}$ consists of all edges \eq{epsilon_triangle_set1} and \eq{epsilon_triangle_set2}.

We claim that \eq{restriction_h_triangle} is block diagonal with a block $\mathcal{W}_{(z,a,q)}\subseteq\mathcal{W}$ of size
\[
\left|\mathcal{W}_{(z,a,q)}\right|=\begin{cases}
3R+2 & R\text{ odd}\\
3R-1 & R\text{ even}
\end{cases}
\]
for each for each triple $(z,a,q)$ with $z,a\in\{0,1\}$ and $q\in[R]$. Using equation \eq{num_basis_D} we confirm that $|\mathcal{W}|=4R\left|\mathcal{W}_{(z,a,q)}\right|$, so this accounts for all basis vectors in $\mathcal{W}$. The subset of basis vectors for a given block is
\begin{align}
\mathcal{W}_{(z,a,q)} 
&=\big\{ |\psi_{z,a}^{q_{\inn}}\rangle,|\psi_{z,a}^{q_{\out}}\rangle\big\} \cup\big\{ |\psi_{z,a}^{d(q,s)}\rangle\colon s\in[R],\, s\neq q\text{ if }R\text{ even}\big\} \nonumber \\
 &\quad \cup \big\{ |\psi_{\alpha(q,s),a}^{e_{zx}(q,s)}\rangle\colon x\in\{0,1\},\, s\in[R],\, s\neq q\text{ if }R\text{ even}\big\} .\label{eq:subset_W}
\end{align}
Using equation \eq{equation_for_h_epsilon_triangle} and the description of $\mathcal{E}^{\triangle}$ from \eq{epsilon_triangle_set1} and \eq{epsilon_triangle_set2}, one can check by direct inspection that \eq{restriction_h_triangle} only has nonzero matrix elements between basis vectors in $\mathcal{W}$ from the same block. We also compute the matrix elements between vectors from the same block. For example, if $R$ is odd or if $R$ is even and $q\neq1$, there are edges $\left\{ (q_{\inn},0,7),(d(q,1),0,3)\right\} ,\left\{ (q_{\inn},1,5),(d(q,1),1,1)\right\} \in\mathcal{E^{\triangle}}$. Using the fact that $|\psi_{z,a}^{l}\rangle=|l\rangle|\psi_{z,a}\rangle$ where $|\psi_{z,a}\rangle$ is given by \eq{psi0m} and \eq{psi1m}, we compute the relevant matrix elements:
\begin{align*}
&\langle\psi_{z,a}^{q_{\inn}}|h_{\mathcal{E}^{\triangle}}|\psi_{z,a}^{d(q,1)}\rangle \\ 
&\quad=\langle\psi_{z,a}^{q_{\inn}}|\Bigg(\sum_{(z',t,t^{\prime})\in\{(0,7,3),(1,5,1)\}}\left(|q_{\inn},z',t\rangle+|d(q,1),z',t^{\prime}\rangle\right)\left(\langle q_{\inn},z',t|+\langle d(q,1),z',t^{\prime}|\right)\otimes\id\Bigg)|\psi_{z,a}^{d(q,1)}\rangle\\
&\quad=\sum_{(z',t,t^{\prime})\in\{(0,7,3),(1,5,1)\}}\langle\psi_{z,a}|\left(|z',t\rangle\langle z',t^{\prime}|\otimes\id\right)|\psi_{z,a}\rangle
=\frac{1}{8}.
\end{align*}
Continuing in this manner, we compute the principal submatrix of \eq{restriction_h_triangle} corresponding to the set $\mathcal{W}_{(z,a,q)}$. This matrix is shown in \fig{mat_els_for_a_blockA}. In the Figure each vertex is associated with a state in the block and the weight on a given edge is the matrix element between the two states associated with vertices joined by that edge. The diagonal matrix elements are described by the weights on the self-loops. The matrix described by \fig{mat_els_for_a_blockA} is the same for each block.

\begin{figure}
\centering
\subfloat[][The matrix $h_{\mathcal{E}}^{\triangle}|_{\spn(\mathcal{W})}$ is block diagonal in the basis $\mathcal{W}$, with a block $\mathcal{W}_{(z,a,q)}$ for each $z,a \in \{0,1\}$ and $q\in\{1,\ldots, R\}$. The states involved in a given block and the matrix elements between them are depicted.]{\label{fig:mat_els_for_a_blockA}
\begin{tikzpicture}[scale=1.6]
\foreach \x in {1,2,4}
{
  \draw[fill=black,draw=black] (\x cm, 0) circle (.5mm) ;

\draw[looseness = 200] (\x,0) to [out = 80, in = 10] (\x-.01,0);
}
\draw[fill=black,draw=black] (0 cm, 0) circle (.5mm) ;
\draw[fill=black,draw=black] (5 cm, 0) circle (.5mm) ;
\draw[looseness = 200] (5,0) to [out = 45, in = -45] (4.99,0);
\draw[looseness = 200] (0,0) to [out = 135, in = 225] (-.01,0);

\draw (0,0)--(1,0)--(2,0)--(2.5,0);
\draw (3.5,0)--(4,0)--(5,0);
\node at (3,0) {\Large{$\ldots$}};
\foreach \x in {1,2,4}
{
  \draw[fill=black,draw=black] (\x cm, 1.5) circle (.5mm) ;
	\draw[fill=black,draw=black] (\x cm, -1.5) circle (.5mm) ;
	\draw (\x,1.5)--(\x,0)--(\x,-1.5);

\draw[looseness = 200] (\x,1.5) to [out = 45, in = 135] (\x-.01,1.5);
\draw[looseness = 200] (\x,-1.5) to [out = -45, in = 225] (\x-.01,-1.5);
}

\node at (1,2.3) {$|\psi_{\alpha(q,1),a}^{e_{z0}(q,1)}\rangle$};
\node at (2,2.3) {$|\psi_{\alpha(q,2),a}^{e_{z0}(q,2)}\rangle$};
\node at (4,2.3) {$|\psi_{\alpha(q,R),a}^{e_{z0}(q,R)}\rangle$};

\node at (-0.8,0){$|\psi_{z,a}^{q_{\inn}}\rangle$};
\node at (1.4,-0.3){$|\psi_{z,a}^{d(q,1)}\rangle$};
\node at (2.4,-0.3){$|\psi_{z,a}^{d(q,2)}\rangle$};
\node at (4.5,-0.3){$|\psi_{z,a}^{d(q,R)}\rangle$};
\node at (5.8,0){$|\psi_{z,a}^{q_{\out}}\rangle$};

\node at (1,-2.3) {$|\psi_{\alpha(q,1),a}^{e_{z1}(q,1)}\rangle$};
\node at (2,-2.3) {$|\psi_{\alpha(q,2),a}^{e_{z1}(q,2)}\rangle$};
\node at (4,-2.3) {$|\psi_{\alpha(q,R),a}^{e_{z1}(q,R)}\rangle$};

\node at (-0.3,0.3){$\nicefrac{1}{8}$};

\node at (0.5,0.15){$\nicefrac{1}{8}$};
\node at (1.6,0.15){$\nicefrac{1}{8}$};

\node at (1.3,0.5){$\nicefrac{1}{2}$};
\node at (2.3,0.5){$\nicefrac{1}{2}$};
\node at (4.3,0.5){$\nicefrac{1}{2}$};

\node at (0.75,2){$\nicefrac{1}{8}$};
\node at (0.85,.85){$\nicefrac{1}{8}$};
\node at (0.85,-.85){$\nicefrac{1}{8}$};
\node at (0.75,-2){$\nicefrac{1}{8}$};

\node at (1.75,2){$\nicefrac{1}{8}$};
\node at (1.85,.85){$\nicefrac{1}{8}$};
\node at (1.85,-.85){$\nicefrac{1}{8}$};
\node at (1.75,-2){$\nicefrac{1}{8}$};

\node at (3.75,2){$\nicefrac{1}{8}$};
\node at (3.85,.85){$\nicefrac{1}{8}$};
\node at (3.85,-.85){$\nicefrac{1}{8}$};
\node at (3.75,-2){$\nicefrac{1}{8}$};

\node at (4.6,0.15){$\nicefrac{1}{8}$};
\node at (5.3,0.3){$\nicefrac{1}{8}$};

\end{tikzpicture}
}
\\
\subfloat[][After multiplying some of the basis vectors by $-1$, the matrix depicted in \subfig{mat_els_for_a_blockA} is transformed into $1/8$ times the Laplacian matrix of this graph.]{\label{fig:mat_els_for_a_blockB}
\begin{tikzpicture}[scale=1.6]
\foreach \x in {0,1,2,4,5}
{
  \draw[fill=black,draw=black] (\x cm, 0) circle (.5mm) ;
}
\draw (0,0)--(1,0)--(2,0)--(2.5,0);
\draw (3.5,0)--(4,0)--(5,0);
\node at (3,0) {\Large{$\ldots$}};
\foreach \x in {1,2,4}
{
  \draw[fill=black,draw=black] (\x cm, 1) circle (.5mm) ;
	\draw[fill=black,draw=black] (\x cm, -1) circle (.5mm) ;
	\draw (\x,1)--(\x,0)--(\x,-1);
}
\draw [decorate,decoration={brace,amplitude=10pt}] (1,1.25) -- (4,1.25) node [black,midway,yshift=17]  {$R$ (for $R$ odd) or $R-1$ (for $R$ even)};
\end{tikzpicture}
}
\caption{\label{fig:mat_els_for_a_block}}
\end{figure}

For each triple $(z,a,q)$ with $z,a\in\{0,1\}$ and $q\in[R]$, define
\begin{equation}
|\phi_{z,a}^{q}\rangle=\begin{cases}
\frac{1}{\sqrt{3R+2}}\left(|\psi_{z,a}^{q_{\mathrm{in}}}\rangle+\sum_{j\in[R]}\left(-1\right)^{j}\left(|\psi_{z,a}^{d(q,j)}\rangle-|\psi_{\alpha(q,j),a}^{e_{z0}(q,j)}\rangle-|\psi_{\alpha(q,j),a}^{e_{z1}(q,j)}\rangle\right)+|\psi_{z,a}^{q_{\mathrm{out}}}\rangle\right) & R\text{ odd}\\
\frac{1}{\sqrt{3R-1}}\left(|\psi_{z,a}^{q_{\mathrm{in}}}\rangle+\left(\sum_{j<q}-\sum_{j>q}\right)\left(-1\right)^{j}\left(|\psi_{z,a}^{d(q,j)}\rangle-|\psi_{\alpha(q,j),a}^{e_{z0}(q,j)}\rangle-|\psi_{\alpha(q,j),a}^{e_{z1}(q,j)}\rangle\right)+|\psi_{z,a}^{q_{\mathrm{out}}}\rangle\right) & R\text{ even.}
\end{cases}\label{eq:phi_z_a_q}
\end{equation}
Next we show that these states span the ground space of $A(G^{\triangle})$.  The choice to omit $d(q,q)$ for $R$ even ensures that $|\psi^{q_{\inn}}_{z,a}\rangle$ and $|\psi^{q_{\out}}_{z,a}\rangle$ have the same sign in these ground states.

\begin{lemma}
\label{lem:A(G_triangle)}
An orthonormal basis for the $e_{1}$-energy ground space of $A(G^{\triangle})$ is given by the states
\[
\left\{ |\phi_{z,a}^{q}\rangle:\, z,a\in\{0,1\},\, q\in[R]\right\} 
\]
defined by equation \eq{phi_z_a_q}. The eigenvalue gap is bounded as 
\begin{equation}
  \gamma(A(G^{\triangle})-e_{1}) > \frac{1}{(30R)^{2}}.
\label{eq:G^triangle_lower_bnd}
\end{equation}
\end{lemma}

\begin{proof}
The $e_1$-energy ground space of $A(G^{\triangle})$ is equal to the nullspace of \eq{restriction_h_triangle}. Since this operator is block diagonal in the basis $\mathcal{W}$, we can solve for the eigenvectors in the nullspace of each block. Thus, to prove the first part of the Lemma, we analyze the $|\mathcal{W}_{(z,a,q)}|\times|\mathcal{W}_{(z,a,q)}|$ matrix described by \fig{mat_els_for_a_blockA} and show that \eq{phi_z_a_q} is the unique vector in its nullspace. We first rewrite it in a slightly different basis obtained by multiplying some of the basis vectors by a phase of $-1$. Specifically, we use the basis
\[
\left\{
   |\psi_{z,a}^{q_{\mathrm{in}}}\rangle,
  -|\psi_{z,a}^{d(q,1)}\rangle,
   |\psi_{\alpha(q,1),a}^{e_{z0}(q,1)}\rangle,
   |\psi_{\alpha(q,1),a}^{e_{z1}(q,1)}\rangle,
   |\psi_{z,a}^{d(q,2)}\rangle,
  -|\psi_{\alpha(q,2),a}^{e_{z0}(q,2)}\rangle,
  -|\psi_{\alpha(q,2),a}^{e_{z1}(q,2)}\rangle,
   \ldots,
   |\psi_{z,a}^{q_{\mathrm{out}}}\rangle
\right\} 
\]
where the state associated with each vertex on one side of a bipartition of the graph is multiplied by $-1$; these are the phases appearing in equation \eq{phi_z_a_q}. Changing to this basis replaces the weight $\frac{1}{8}$ on each edge in \fig{mat_els_for_a_blockA} by $-\frac{1}{8}$ and does not change the weights on the self-loops. The resulting matrix is $\frac{1}{8}L_{0}$, where $L_{0}$ is the Laplacian matrix of the graph shown in \fig{mat_els_for_a_blockB}. Now we use the fact that the Laplacian of any connected graph has smallest eigenvalue zero, with a unique eigenvector equal to the all-ones vector. Hence for each block we get an eigenvector in the nullspace of \eq{restriction_h_triangle}) given by \eq{phi_z_a_q}. Ranging over all $z,a\in\{0,1\}$ and $q\in[R]$ gives the claimed basis for the $e_1$-energy ground space of $A(G^{\triangle})$.

To prove the lower bound, we use the \npl (\lem{npl}) with 
\[
H_{A}=A(G^{\diamondsuit})-e_{1}\qquad H_{B}=h_{\mathcal{E}^{\triangle}}
\]
and where $S=\spn(\mathcal{W})$ is the nullspace of $H_{A}$ as shown in \sec{The-gate-graph_G_DIAMOND}. Since it is block diagonal in the basis $\mathcal{W}$, the smallest nonzero eigenvalue of \eq{restriction_h_triangle} is equal to the smallest nonzero eigenvalue of one of its blocks. The matrix for each block is $\frac{1}{8}L_{0}$. Thus we can lower bound the smallest nonzero eigenvalue of $H_B|_S$ using standard bounds on the smallest nonzero eigenvalue of the Laplacian $L$ of a graph $G$. In particular, Theorem 4.2 of reference \cite{Moh91} shows that 
\begin{equation*}
\gamma(L) \ge \frac{4}{|V(G)| \diam(G)} \ge \frac{4}{|V(G)|^2}
\end{equation*}
(where $\diam(G)$ is the diameter of $G$). Since the size of the graph in \fig{mat_els_for_a_blockB} is either $3R-1$ or $3R+2$, we have
\begin{equation*}
\gamma(H_{B}|_{S})=\frac{1}{8}\gamma(L_{0}) \geq \frac{1}{8}\frac{4}{\left(3R+2\right)^{2}}\geq\frac{1}{32R^{2}}
\end{equation*}
since $R\geq2$. Using this bound and the fact that $\gamma(H_{A})>\frac{1}{13}$ (from equation \eq{one_thirteenth_bound}) and $\left\Vert H_{B}\right\Vert =2$ (from equation \eq{h_E_bound}) and plugging into \lem{npl} gives
\[
\gamma(A(G^{\triangle})-e_{1})\geq\frac{\frac{1}{13}\cdot\frac{1}{32R^{2}}}{\frac{1}{13}+\frac{1}{32R^{2}}+2}\geq\frac{1}{(32+13+832)R^{2}}>\frac{1}{(30R)^{2}}. \qedhere
\]
\end{proof}

\subsection{The Hamiltonian $H(G^{\triangle},N)$}

We now consider the $N$-particle Hamiltonian $H(G^{\triangle},N)$ and solve for its nullspace. We use the following fact about the subsets $\mathcal{W}_{(z,a,q)}\subset\mathcal{W}$ defined in equation \eq{subset_W}.

\begin{definition}
\label{defn:overlap_diagram_element}
We say $\mathcal{W}_{(z_{1},a_{1},q_{1})}$ and $\mathcal{W}_{(z_{2},a_{2},q_{2})}$ \emph{overlap on a diagram element} if there exists $l\in L^{\square}$ such that $|\psi_{x_{1},b_{1}}^{l}\rangle\in\mathcal{W}_{(z_{1},a_{1},q_{1})}$ and $|\psi_{x_{2},b_{2}}^{l}\rangle\in\mathcal{W}_{(z_{2},a_{2},q_{2})}$ for some $x_{1},x_{2},b_{1},b_{2}\in\{0,1\}$.
\end{definition}

\begin{fact}
[Key property of $\mathcal{W}_{(z,a,q)}$]\label{fct:block_property}
Sets $\mathcal{W}_{(z_{1},a_{1},q_{1})}$ and $\mathcal{W}_{(z_{2},a_{2},q_{2})}$ overlap on a diagram element if and only if $q_{1}=q_{2}$ or $\{q_{1},q_{2}\}\in E(\goc)$.
\end{fact}
This fact can be confirmed by direct inspection of the sets $\mathcal{W}_{(z,a,q)}$. If $q_{1}=q_{2}=q$ the diagram element $l$ on which they overlap can be chosen to be $l=q_{\mathrm{in}}$; if $q_{1}\neq q_{2}$ and $\{q_{1},q_{2}\}\in E(\goc)$ then $l=e_{z_{1}z_{2}}(q_{1},q_{2})=e_{z_{2}z_{1}}(q_{2},q_{1})$.

Conversely, if $\{q_1,q_2\} \notin E(\goc)$ with $q_1 \ne q_2$, then there is no overlap.

We show that the nullspace $\mathcal{I}_{\triangle}$ of $H(G^{\triangle},N)$ is 
\begin{equation}
\mathcal{I}_{\triangle}=\spn\{ \Sym(|\phi_{z_{1},a_{1}}^{q_{1}}\rangle|\phi_{z_{2},a_{2}}^{q_{2}}\rangle\ldots|\phi_{z_{N},a_{N}}^{q_{N}}\rangle)\colon 
\text{$z_{i},a_{i}\in\{0,1\}$, $q_{i}\in[R]$, $q_{i}\neq q_{j}$, and $\{q_{i},q_{j}\}\notin E(\goc)$}\}.
\label{eq:I_triangle}
\end{equation}
Note that $\mathcal{I}_{\triangle}\subset\mathcal{Z}_{N}(G^{\triangle})$ is very similar to $\mathcal{I}(G,\goc,N)\subset\mathcal{Z}_{N}(G)$ (from equation \eq{occup_space_defn}) but with each single-particle state $|\psi_{z,a}^{q}\rangle\in\mathcal{Z}_{N}(G)$ replaced by $|\phi_{z,a}^{q}\rangle\in\mathcal{Z}_{N}(G^{\triangle})$.

\begin{lemma}
\label{lem:The-nullspace-of_H_triangle}The nullspace of $H(G^{\triangle},N)$ is $\mathcal{I}_{\triangle}$ as defined in equation \eq{I_triangle}. Its smallest nonzero eigenvalue is 
\begin{equation}
\gamma(H(G^{\triangle},N)) > \frac{1}{\left(17R\right)^{7}}.\label{eq:lowerbound_HG_triangle}
\end{equation}
\end{lemma}

In addition to \fct{block_property}, we use the following simple fact in the proof of the Lemma.

\begin{fact}
\label{fct:lin_alg_fact}
Let $|p\rangle=c|\alpha_{0}\rangle+\sqrt{1-c^{2}}|\alpha_{1}\rangle$ with $\langle\alpha_{i}|\alpha_{j}\rangle=\delta_{ij}$ and $c\in[0,1]$. Then
\[
|p\rangle\langle p|=c^{2}|\alpha_{0}\rangle\langle\alpha_{0}|+M
\]
where $\left\Vert M\right\Vert \leq1-\frac{3}{4}c^{4}$.
\end{fact}

To prove this Fact, one can calculate $\left\Vert M\right\Vert =\frac{1}{2}(1-c^{2})+\frac{1}{2}\sqrt{1+2c^{2}-3c^{4}}$ and use the inequality $\sqrt{1+x}\leq1+\frac{x}{2}$ for $x\geq-1$.

\begin{proof}[Proof of \protect\lem{The-nullspace-of_H_triangle}]
Using equation \eq{A_g_diamond_triangle} and the fact that the smallest eigenvalues of $A(G^{\diamondsuit})$ and $A(G^{\triangle})$ are the same (equal to $e_{1}$, from \sec{The-gate-graph_G_DIAMOND} and \lem{A(G_triangle)}), we have
\begin{equation}
H(G^{\triangle},N)=H(G^{\diamondsuit},N)+\sum_{w=1}^{N}h_{\mathcal{E}^{\triangle}}^{(w)}\bigg|_{\mathcal{Z}_{N}(G^{\triangle})}.\label{eq:H_G^triangle,diamond}
\end{equation}
Recall from \lem{The-nullspace-of_Hdiamond} that the nullspace of $H(G^{\diamondsuit},N)$ is $\mathcal{\mathcal{I}_{\diamondsuit}}$. We consider 
\begin{equation}
\sum_{w=1}^{N}h_{\mathcal{E}^{\triangle}}^{(w)}\bigg|_{\mathcal{I}_{\diamondsuit}}.\label{eq:restriction to script R}
\end{equation}
We show that its nullspace is equal to $\mathcal{I}_{\triangle}$ (establishing the first part of the Lemma), and we lower bound its smallest nonzero eigenvalue. Specifically, we prove
\begin{equation}
\gamma\left(\sum_{w=1}^{N}h_{\mathcal{E}^{\triangle}}^{(w)}\bigg|_{\mathcal{I}_{\diamondsuit}}\right)>\frac{1}{(9R)^{6}}.\label{eq:bound_R6}
\end{equation}

Now we prove equation \eq{lowerbound_HG_triangle} using this bound. We apply the \npl (\lem{npl}) with $H_{A}$ and $H_{B}$ given by the first and second terms in equation \eq{H_G^triangle,diamond}; in this case the nullspace of $H_{A}$ is $S=\mathcal{I}_{\diamondsuit}$ (from \lem{The-nullspace-of_Hdiamond}). Now applying \lem{npl} and using the bounds $\gamma(H_{A})>\frac{1}{300}$ (from \lem{The-nullspace-of_Hdiamond}), $\left\Vert H_{B}\right\Vert \leq N\left\Vert h_{\mathcal{E}^{\triangle}}\right\Vert =2N\leq2R$ (from equation \eq{h_E_bound} and the fact that $N\leq R$), and the bound \eq{bound_R6} on $\gamma(H_{B}|_{S})$, we find
\[
\gamma(H(G^{\triangle},N))\geq\frac{\frac{1}{300(9R)^{6}}}{\frac{1}{300}+\frac{1}{(9R)^{6}}+2R}\geq\left(\frac{1}{9^6+300+600\cdot 9^6}\right)\frac{1}{R^{7}}>\frac{1}{\left(17R\right)^{7}}.
\]

To complete the proof we must establish that the nullspace of \eq{restriction to script R} is $\mathcal{I}_{\triangle}$ and prove the lower bound \eq{bound_R6}. To analyze \eq{restriction to script R} we use the fact (established in \sec{The-gate-graph_G_triangle}) that \eq{restriction_h_triangle} is block diagonal with a block $\mathcal{W}_{(z,a,q)}\subset\mathcal{W}$ for each triple $(z,a,q)$ with $z,a\in\{0,1\}$ and $q\in[R]$. The operator \eq{restriction to script R} inherits a block structure from this fact. For any basis vector 
\begin{equation}
\Sym(|\psi_{z_{1},a_{1}}^{q_{1}}\rangle|\psi_{z_{2},a_{2}}^{q_{2}}\rangle\ldots|\psi_{z_{N},a_{N}}^{q_{N}}\rangle)\in\mathcal{I}_{\diamondsuit},
\label{eq:basis_vector_in_scriptF}
\end{equation}
we define a set of occupation numbers
\[
\mathcal{N}=\left\{ N_{(x,b,r)} \colon x,b\in\{0,1\},\, r\in[R]\right\} 
\]
where
\[
  N_{(x,b,r)}
  =|\{j \colon |\psi_{z_{j},a_{j}}^{q_{j}}\rangle\in\mathcal{W}_{(x,b,r)}\}|.
\]
Now observe that \eq{restriction to script R} conserves the set of occupation numbers and is therefore block diagonal with a block for each possible set $\mathcal{N}$. 

For a given block corresponding to a set of occupation numbers $\mathcal{N}$, we write $\mathcal{I}_{\diamondsuit}(\mathcal{N})\subset\mathcal{\mathcal{I}_{\diamondsuit}}$ for the subspace spanned by basis vectors \eq{basis_vector_in_scriptF} in the block. We classify the blocks into three categories depending on $\mathcal{N}$.

\begin{mdframed}[frametitle=Classification of the blocks of \eq{restriction to script R} according to $\mathcal{N}$]
Consider the following two conditions on a set $\mathcal{N}=\{ N_{(x,b,r)}\colon x,b\in\{0,1\},\, r\in[R]\}$
of occupation numbers:
\begin{enumerate}[label=(\alph*)]
\item $N_{(x,b,r)}\in\{0,1\}$ for all $x,b\in\{0,1\}$ and
$r\in[R]$. If this holds, write $\left(y_{i},c_{i},s_{i}\right)$ for the nonzero occupation numbers (with some arbitrary ordering), i.e., $N_{(y_{i},c_{i},s_{i})}=1$ for $i\in[N]$.
\item The sets $\mathcal{W}_{(y_{i},c_{i},s_{i})}$ and $\mathcal{W}_{(y_{j},c_{j},s_{j})}$
do not overlap on a diagram element for all distinct $i,j\in[N]$.
\end{enumerate}
We say a block is of type 1 if $\mathcal{N}$ satisfies (a) and (b). We say it is of type 2 if $\mathcal{N}$ does not satisfy (a). We say it is of type 3 if $\mathcal{N}$ satisfies (a) but does not satisfy (b).
\end{mdframed}

Note that every block is either of type 1, 2, or 3. We consider each type separately. Specifically, we show that each block of type $1$ contains one state in the nullspace of \eq{restriction to script R} and, ranging over all blocks of this type, we obtain a basis for $\mathcal{I}_{\triangle}$. We also show that the smallest nonzero eigenvalue within a block of type 1 is at least $\frac{1}{32R^{2}}$. Finally, we show that blocks of type 2 and 3 do not contain any states in the nullspace of \eq{restriction to script R} and that the smallest eigenvalue within any block of type 2 or 3 is greater than $\frac{1}{(9R)^{6}}$. Hence, the nullspace of \eq{restriction to script R} is $\mathcal{I}_{\triangle}$ and its smallest nonzero eigenvalue is lower bounded as in equation \eq{bound_R6}.

\medskip

\noindent \textbf{Type 1}

\smallskip

\noindent Note (from \defn{overlap_diagram_element}) that (b) implies $q\neq r$ whenever 
\[
|\psi_{x,b}^{q}\rangle\in\mathcal{W}_{(y_{i},c_{i},s_{i})}\text{ and }|\psi_{z,a}^{r}\rangle\in\mathcal{W}_{(y_{j},c_{j},s_{j})}
\]
for distinct $i,j\in[N]$. Hence 
\begin{align*}
\mathcal{\mathcal{I}_{\diamondsuit}}(\mathcal{N}) & =\spn\{ \Sym(|\psi_{z_{1},a_{1}}^{q_{1}}\rangle|\psi_{z_{2},a_{2}}^{q_{2}}\rangle\ldots|\psi_{z_{N},a_{N}}^{q_{N}}\rangle)\colon q_{i}\neq q_{j}\text{ and }|\psi_{z_{j},a_{j}}^{q_{j}}\rangle\in\mathcal{W}_{(y_{j},c_{j},s_{j})}\} \\
 & =\spn\{ \Sym(|\psi_{z_{1},a_{1}}^{q_{1}}\rangle|\psi_{z_{2},a_{2}}^{q_{2}}\rangle\ldots|\psi_{z_{N},a_{N}}^{q_{N}}\rangle)\colon|\psi_{z_{j},a_{j}}^{q_{j}}\rangle\in\mathcal{W}_{(y_{j},c_{j},s_{j})}\}.
\end{align*}
From this we see that 
\[
\dim(\mathcal{I}_{\diamondsuit}(\mathcal{N}))=\prod_{j=1}^{N}\left|{\mathcal{W}_{(y_{j},c_{j},s_{j})}}\right|=\begin{cases}
\left(3R+2\right)^{N} & R\text{ odd}\\
\left(3R-1\right)^{N} & R\text{ even.}
\end{cases}
\]
We now solve for all the eigenstates of \eq{restriction to script R} within the block. 

It is convenient to write an orthonormal basis of eigenvectors of the $|\mathcal{W}_{(z,a,q)}|\times|\mathcal{W}_{(z,a,q)}|$ matrix described by \fig{mat_els_for_a_blockA} as 
\begin{equation}
|\phi_{z,a}^{q}(u)\rangle, \quad u\in[|\mathcal{W}_{(z,a,q)}|]\label{eq:phi_u}
\end{equation}
and their ordered eigenvalues as
\[
\theta_{1}\leq\theta_{2}\leq\ldots\leq\theta_{|\mathcal{W}_{(z,a,q)}|}.
\]
From the proof of \lem{A(G_triangle)}, the eigenvector with smallest eigenvalue $\theta_{1}=0$ is $|\phi_{z,a}^{q}\rangle=|\phi_{z,a}^{q}(1)\rangle$ and $\theta_{2}\geq\frac{1}{32R^{2}}$. For any $u_{1},u_{2},\ldots,u_{N}\in [|\mathcal{W}_{(z,a,q)}|]$, the state 
\[
\Sym(|\phi_{y_{1},c_{1}}^{s_{1}}(u_{1})\rangle|\phi_{y_{2},c_{2}}^{s_{2}}(u_{2})\rangle\ldots|\phi_{y_{N},c_{N}}^{s_{N}}(u_{N})\rangle)
\]
is an eigenvector of \eq{restriction to script R} with eigenvalue $\sum_{j=1}^{N}\theta_{j}.$ Furthermore, states corresponding to different choices of $u_{1},\ldots,u_{N}$ are orthogonal, and ranging over all $\dim(\mathcal{I}_{\diamondsuit}(\mathcal{N}))$ choices we get every eigenvector in the block. The smallest eigenvalue within the block is $N\theta_{1}=0$ and there is a unique vector in the nullspace, given by
\begin{equation}
\Sym(|\phi_{y_{1},c_{1}}^{s_{1}}\rangle|\phi_{y_{2},c_{2}}^{s_{2}}\rangle\ldots|\phi_{y_{N},c_{N}}^{s_{N}}\rangle)\label{eq:vector_type1}
\end{equation}
(recall $|\phi_{z,a}^{q}\rangle=|\phi_{z,a}^{q}(1)\rangle$). The smallest nonzero eigenvalue of \eq{restriction to script R} within the block is $(N-1)\theta_{1}+\theta_{2}=\theta_{2}\geq\frac{1}{32R^{2}}$.

Finally, we show that the collection of states \eq{vector_type1} obtained from all blocks of type 1 spans the space $\mathcal{I}_{\triangle}$. Each block of type 1 corresponds to a set of occupation numbers 
\[
N_{(y_{1},c_{1},s_{1})}=N_{(y_{2},c_{2},s_{2})}=\cdots=N_{(y_{N},c_{N},s_{N})}=1\quad\text{(with all other occupation numbers zero)}
\]
and gives a unique vector \eq{vector_type1} in the nullspace of $H(G^{\triangle},N)$. The sets $\mathcal{W}_{(y_{i},c_{i},s_{i})}$ and $\mathcal{W}_{(y_{j},c_{j},s_{j})}$ do not overlap on a diagram element for all distinct $i,j\in[N]$. Using \fct{block_property} we see this is equivalent to $s_{i}\neq s_{j}$ and $\{s_{i},s_{j}\}\notin E(\goc)$ for distinct $i,j\in[N]$. Hence the set of states \eq{vector_type1} obtained from of all blocks of type 1 is 
\[
\left\{ \Sym(|\phi_{y_{1},c_{1}}^{s_{1}}\rangle|\phi_{y_{2},c_{2}}^{s_{2}}\rangle\ldots|\phi_{y_{N},c_{N}}^{s_{N}}\rangle)\colon y_{i},c_{i}\in\{0,1\},\, s_{i}\in[R],\, s_{i}\neq s_{j},\,\{s_{i},s_{j}\}\notin E(\goc)\right\} 
\]
which spans $\mathcal{I}_{\triangle}$.

\medskip

\noindent \textbf{Type 2}

\smallskip

\noindent If $\mathcal{N}$ is of type 2 then there exist $x,b\in\{0,1\}$ and $r\in[R]$ such that $N_{(x,b,r)}\geq2$. We show there are no eigenvectors in the nullspace of \eq{restriction to script R} within a block of this type and we lower bound the smallest eigenvalue within the block. Specifically, we show
\begin{equation}
\min_{|\kappa\rangle\in\mathcal{I}_{\diamondsuit}(\mathcal{N})}\langle\kappa|\sum_{w=1}^{N}h_{\mathcal{E}^{\triangle}}^{(w)}|\kappa\rangle >\frac{1}{(9R)^6}.\label{eq:boundfortype2}
\end{equation}
First note that all $|\kappa\rangle\in\mathcal{I}_{\diamondsuit}$ satisfy $(A(G^{\diamondsuit})-e_{1})^{(w)}|\kappa\rangle=0$ for each $w\in [N]$, which can be seen using the definition of $\mathcal{I}_{\diamondsuit}$ and the fact that $\mathcal{W}$ spans the nullspace of $A(G^{\diamondsuit})-e_{1}$. Using this fact and equation \eq{A_g_diamond_triangle}, we get
\begin{equation}
\min_{|\kappa\rangle\in\mathcal{I}_{\diamondsuit}(\mathcal{N})}\langle\kappa|\sum_{w=1}^{N}h_{\mathcal{E}^{\triangle}}^{(w)}|\kappa\rangle=\min_{|\kappa\rangle\in\mathcal{I}_{\diamondsuit}(\mathcal{N})}\langle\kappa|\sum_{w=1}^{N}\left(A(G^{\triangle})-e_{1}\right)^{(w)}|\kappa\rangle.\label{eq:min_over_blockFbl}
\end{equation}
Now we use the operator inequality 
\begin{align}
  \sum_{w=1}^{N}\left(A(G^{\triangle})-e_{1}\right)^{(w)}
  & \geq\gamma\left(\sum_{w=1}^{N}\left(A(G^{\triangle})-e_{1}\right)^{(w)}\right)\cdot
    \left(1-\Pi^{\triangle}\right) \nonumber \\
  &= \gamma(A(G^{\triangle})-e_{1})\cdot\left(1-\Pi^{\triangle}\right)
   > \frac{1}{(30R)^{2}}\left(1-\Pi^{\triangle}\right),
\label{eq:op_ineq_A_Gtriangle}
\end{align}
where $\Pi^{\triangle}$ is the projector onto the nullspace of $\sum_{w=1}^{N}\left(A(G^{\triangle})-e_{1}\right)^{(w)}$, and where in the last step we used \lem{A(G_triangle)}. Plugging equation \eq{op_ineq_A_Gtriangle} into equation \eq{min_over_blockFbl} gives
\begin{equation}
\min_{|\kappa\rangle\in\mathcal{I}_{\diamondsuit}(\mathcal{N})}\langle\kappa|\sum_{w=1}^{N}h_{\mathcal{E}^{\triangle}}^{(w)}|\kappa\rangle>\frac{1}{(30R)^{2}}\Big(1-\max_{|\kappa\rangle\in\mathcal{I}_{\diamondsuit}(\mathcal{N})}\langle\kappa|\Pi^{\triangle}|\kappa\rangle\Big).\label{eq:bound1}
\end{equation}

In the following we show that $\langle\kappa|\Pi^{\triangle}|\kappa\rangle=\langle\kappa|\Pi_{\mathcal{N}}^{\triangle}|\kappa\rangle$ for all $|\kappa\rangle\in\mathcal{I}_{\diamondsuit}(\mathcal{N})$, where $\Pi_{\mathcal{N}}^{\triangle}$ is a Hermitian operator with
\begin{equation}
  1-\big\Vert \Pi_{\mathcal{N}}^{\triangle}\big\Vert \ge \frac{3}{4}\left(\frac{1}{4R}\right)^{4} = \frac{3}{1024R^4}.
  \label{eq:bound_norm_pi_triangle_n}
\end{equation}
Plugging this into \eq{bound1} gives 
\[
\min_{|\kappa\rangle\in\mathcal{I}_{\diamondsuit}(\mathcal{N})}\langle\kappa|\sum_{w=1}^{N}h_{\mathcal{E}^{\triangle}}^{(w)}|\kappa\rangle>\frac{3}{(30R)^2 \cdot 1024R^4}>\frac{1}{(9R)^6}.
\]

To complete the proof, we exhibit the operator $\Pi_{\mathcal{N}}^{\triangle}$ and show that its norm is bounded as \eq{bound_norm_pi_triangle_n}. Using \lem{A(G_triangle)} we can write $\Pi^{\triangle}$ explicitly as
\begin{equation}
\Pi^{\triangle}=\sum_{(\vec{z},\vec{a},\vec{q})\in\mathcal{Q}}\mathcal{P}_{(\vec{z},\vec{a},\vec{q})}\label{eq:Pi_N_triangle}
\end{equation}
where 
\begin{align*}
\mathcal{P}_{(\vec{z},\vec{a},\vec{q})} & =|\phi_{z_{1},a_{1}}^{q_{1}}\rangle\langle\phi_{z_{1},a_{1}}^{q_{1}}|\otimes|\phi_{z_{2},a_{2}}^{q_{2}}\rangle\langle\phi_{z_{2},a_{2}}^{q_{2}}|\otimes\cdots\otimes|\phi_{z_{N},a_{N}}^{q_{N}}\rangle\langle\phi_{z_{N},a_{N}}^{q_{N}}|\\
\mathcal{Q} & =\left\{ (z_{1},\ldots z_{N},a_{1},\ldots,a_{N},q_{1},\ldots,q_{N}):\, z_{i},a_{i}\in\{0,1\}\text{ and }q_{i}\in[R]\right\} .
\end{align*}
For each $(\vec{z},\vec{a},\vec{q})\in\mathcal{Q}$ we also define a space 
\[
S_{(\vec{z},\vec{a},\vec{q})}=\spn(\mathcal{W}_{(z_{1},a_{1},q_{1})})\otimes\spn(\mathcal{W}_{(z_{2},a_{2},q_{2})})\otimes\cdots\otimes\spn(\mathcal{W}_{(z_{N},a_{N},q_{N})}).
\]
Note that $\mathcal{P}_{(\vec{z},\vec{a},\vec{q})}$ has all of its support in $S_{(\vec{z},\vec{a},\vec{q})}$, and that
\begin{equation}
S_{(\vec{z},\vec{a},\vec{q})}\perp S_{(\vec{z}^{\prime},\vec{a}^{\prime},\vec{q}^{\prime})}\text{ for distinct }(\vec{z},\vec{a},\vec{q}),(\vec{z}^{\prime},\vec{a}^{\prime},\vec{q}^{\prime})\in\mathcal{Q}.\label{eq:perp_S}
\end{equation}
Therefore $\mathcal{P}_{(\vec{z},\vec{a},\vec{q})}\mathcal{P}_{(\vec{z}^{\prime},\vec{a}^{\prime},\vec{q}^{\prime})}=0$ for distinct $(\vec{z},\vec{a},\vec{q}),(\vec{z}^{\prime},\vec{a}^{\prime},\vec{q}^{\prime})\in\mathcal{Q}$. (Below we use similar reasoning to obtain a less obvious result.) Note that $\mathcal{P}_{(\vec{z},\vec{a},\vec{q})}$ is orthogonal to $\mathcal{I}_{\diamondsuit}(\mathcal{N})$ unless 
\begin{equation}
\left|\left\{ j:(z_{j},a_{j},q_{j})=(w,u,v)\right\} \right|=N_{(w,u,v)}\text{ for all }w,u\in\{0,1\},\, v\in[R].\label{eq:satisfy_occ_numbers}
\end{equation}
We restrict our attention to the projectors that are not orthogonal to $\mathcal{I}_{\diamondsuit}(\mathcal{N})$. Letting $\mathcal{Q}(\mathcal{N})\subset\mathcal{Q}$ be the set of $(\vec{z},\vec{a},\vec{q})$ satisfying equation \eq{satisfy_occ_numbers}, we have
\begin{equation}
\langle\kappa|\sum_{(\vec{z},\vec{a},\vec{q})\in\mathcal{Q}}\mathcal{P}_{(\vec{z},\vec{a},\vec{q})}|\kappa\rangle=\langle\kappa|\sum_{(\vec{z},\vec{a},\vec{q})\in\mathcal{Q}(\mathcal{N})}\mathcal{P}_{(\vec{z},\vec{a},\vec{q})}|\kappa\rangle\quad\text{for all }|\kappa\rangle\in\mathcal{I}_{\diamondsuit}(\mathcal{N}).\label{eq:restrict_attention_mathcalN}
\end{equation}

Since $N_{(x,b,r)}\geq2$, note that in each term $\mathcal{P}_{(\vec{z},\vec{a},\vec{q})}$ with $(\vec{z},\vec{a},\vec{q})\in\mathcal{Q}(\mathcal{N})$, the operator 
\[
|\phi_{x,b}^{r}\rangle\langle\phi_{x,b}^{r}|\otimes|\phi_{x,b}^{r}\rangle\langle\phi_{x,b}^{r}|
\]
appears between two of the $N$ registers (tensored with rank-1 projectors on the other $N-2$ registers). Using equation \eq{phi_z_a_q} we may expand $|\phi_{x,b}^{r}\rangle$ as a sum of states from $\mathcal{W}_{(x,b,r)}$. This gives
\[
|\phi_{x,b}^{r}\rangle|\phi_{x,b}^{r}\rangle=c_{0}|\psi_{x,b}^{r_{\inn}}\rangle|\psi_{x,b}^{r_{\inn}}\rangle+\left(1-c_{0}^{2}\right)^{\frac{1}{2}}|\Phi_{x,b}^{r}\rangle
\]
where $c_{0}$ is either $\frac{1}{3R+2}$ (if $R$ is odd) or $\frac{1}{3R-1}$ (if $R$ is even), and where $|\psi_{x,b}^{r_{\inn}}\rangle|\psi_{x,b}^{r_{\inn}}\rangle$ is orthogonal to $|\Phi_{x,b}^{r}\rangle$. Note that each of the states $|\phi_{x,b}^{r}\rangle|\phi_{x,b}^{r}\rangle$, $|\psi_{x,b}^{r_{\inn}}\rangle|\psi_{x,b}^{r_{\inn}}\rangle$, and $|\Phi_{x,b}^{r}\rangle$ lie in the space 
\begin{equation}
\spn(\mathcal{W}_{(x,b,r)})\otimes\spn(\mathcal{W}_{(x,b,r)}).
\label{eq:tensor_prod_WW}
\end{equation}
Now applying \fct{lin_alg_fact} gives
\begin{equation}
|\phi_{x,b}^{r}\rangle\langle\phi_{x,b}^{r}|\otimes|\phi_{x,b}^{r}\rangle\langle\phi_{x,b}^{r}|=c_{0}^{2}|\psi_{x,b}^{r_{\inn}}\rangle\langle\psi_{x,b}^{r_{\inn}}|\otimes|\psi_{x,b}^{r_{\inn}}\rangle\langle\psi_{x,b}^{r_{\inn}}|+M_{x,b}^{r}\label{eq:expand_phi_phi_proj}
\end{equation}
where $M_{x,b}^{r}$ is a Hermitian operator with all of its support on the space \eq{tensor_prod_WW} and
\begin{equation}
\left\Vert M_{x,b}^{r}\right\Vert \leq1-\frac{3}{4}c_{0}^{4}\leq1-\frac{3}{4}\left(\frac{1}{3R+2}\right)^{4}\leq1-\frac{3}{4}\frac{1}{(4R)^4}
\end{equation}
since $R\geq2$. For each $(\vec{z},\vec{a},\vec{q})\in\mathcal{Q}(\mathcal{N})$ we define $\mathcal{P}_{(\vec{z},\vec{a},\vec{q})}^{M}$ to be the operator obtained from $\mathcal{P}_{(\vec{z},\vec{a},\vec{q})}$ by replacing
\[
|\phi_{x,b}^{r}\rangle\langle\phi_{x,b}^{r}|\otimes|\phi_{x,b}^{r}\rangle\langle\phi_{x,b}^{r}|\mapsto M_{x,b}^{r}
\]
on two of the registers (if $N_{(x,b,r)}>2$ there is more than one way to do this; we fix one choice for each $(\vec{z},\vec{a},\vec{q})\in\mathcal{Q}(\mathcal{N})$). Note that $\mathcal{P}_{(\vec{z},\vec{a},\vec{q})}^{M}$ has all of its support in the space $S_{(\vec{z},\vec{a},\vec{q})}$. Using \eq{perp_S} gives
\[
\mathcal{P}_{(\vec{z},\vec{a},\vec{q})}^{M}\mathcal{P}_{(\vec{z}^{\prime},\vec{a}^{\prime},\vec{q}^{\prime})}^{M}=0\text{ for distinct }(\vec{z},\vec{a},\vec{q}),(\vec{z}^{\prime},\vec{a}^{\prime},\vec{q}^{\prime})\in\mathcal{Q}(\mathcal{N}).
\]
Using equation \eq{expand_phi_phi_proj} and the fact that
\[
\langle\kappa|\Big(|\psi_{x,b}^{r_{\inn}}\rangle\langle\psi_{x,b}^{r_{\inn}}|^{(w_{1})}\Big)\Big(|\psi_{x,b}^{r_{\inn}}\rangle\langle\psi_{x,b}^{r_{\inn}}|^{(w_{2})}\Big)|\kappa\rangle=0\quad\text{for all }|\kappa\rangle\in\mathcal{I}_{\diamondsuit}(\mathcal{N})\text{ and distinct }w_{1},w_{2}\in[N]
\]
(which can be seen from the definition of $\mathcal{I}_\diamondsuit$), we have 
\[
\langle\kappa|\mathcal{P}_{(\vec{z},\vec{a},\vec{q})}|\kappa\rangle=\langle\kappa|\mathcal{P}_{(\vec{z},\vec{a},\vec{q})}^{M}|\kappa\rangle\quad\text{for all }|\kappa\rangle\in\mathcal{I}_{\diamondsuit}(\mathcal{N}).
\]
Hence, letting 
\begin{equation}
\Pi_{\mathcal{N}}^{\triangle}=\sum_{(\vec{z},\vec{a},\vec{q})\in\mathcal{Q}(\mathcal{N})}\mathcal{P}_{(\vec{z},\vec{a},\vec{q})}^{M},\label{eq:pi_triangle_N}
\end{equation}
we have $\langle\kappa|\Pi^{\triangle}|\kappa\rangle=\langle\kappa|\Pi_{\mathcal{N}}^{\triangle}|\kappa\rangle$ for all $|\kappa\rangle\in\mathcal{I}_{\diamondsuit}(\mathcal{N})$. To obtain the desired bound \eq{bound_norm_pi_triangle_n} on the norm of $\Pi_{\mathcal{N}}^{\triangle}$, we use the fact that the norm of a sum of pairwise orthogonal Hermitian operators is upper bounded by the maximum norm of an operator in the sum, so
\begin{equation}
\big\Vert \Pi_{\mathcal{N}}^{\triangle}\big\Vert 
=\Bigg\Vert \sum_{(\vec{z},\vec{a},\vec{q})\in\mathcal{Q}(\mathcal{N})}\mathcal{P}_{(\vec{z},\vec{a},\vec{q})}^{M}\Bigg\Vert 
=\max_{(\vec{z},\vec{a},\vec{q})\in\mathcal{Q}(\mathcal{N})}\big\Vert \mathcal{P}_{(\vec{z},\vec{a},\vec{q})}^{M}\big\Vert 
=\left\Vert M_{x,b}^{r}\right\Vert \leq1-\frac{3}{4}\frac{1}{\left(4R\right)^{4}}.
\label{eq:bound_on_norm_pi_n_triangle}
\end{equation}

\medskip

\noindent \textbf{Type 3}

\smallskip

\noindent If $\mathcal{N}$ is of type 3 then $N_{(x,b,r)}\in\{0,1\}$
for all $x,b\in\{0,1\}$ and $r\in[R]$, and 
\[
N_{(y,c,s)}=N_{(t,d,u)}=1
\]
for some $(y,c,s)\neq(t,d,u)$ with either $u=s$ or $\{u,s\}\in E(\goc)$ (using property (b) and \fct{block_property}). We show there are no eigenvectors in the nullspace of \eq{restriction to script R} within a block of this type and we lower bound the smallest eigenvalue within the block. We establish the same bound \eq{boundfortype2} as for blocks of Type 2.

The proof is very similar to that given above for blocks of Type 2. In fact, the first part of proof is identical, from equation \eq{min_over_blockFbl} up to and including equation \eq{restrict_attention_mathcalN}. That is to say, as in the previous case we have
\begin{equation}
\langle\kappa|\sum_{(\vec{z},\vec{a},\vec{q})\in\mathcal{Q}}\mathcal{P}_{(\vec{z},\vec{a},\vec{q})}|\kappa\rangle=\langle\kappa|\sum_{(\vec{z},\vec{a},\vec{q})\in\mathcal{Q}(\mathcal{N})}\mathcal{P}_{(\vec{z},\vec{a},\vec{q})}|\kappa\rangle\quad\text{for all }|\kappa\rangle\in\mathcal{I}_{\diamondsuit}(\mathcal{N}).\label{eq:restrict_attention_mathcalN-1}
\end{equation}
In this case, since $N_{(y,c,s)}=N_{(t,d,u)}=1$, in each term $\mathcal{P}_{(\vec{z},\vec{a},\vec{q})}$ with $(\vec{z},\vec{a},\vec{q})\in\mathcal{Q}(\mathcal{N})$, the operator 
\[
|\phi_{y,c}^{s}\rangle\langle\phi_{y,c}^{s}|\otimes|\phi_{t,d}^{u}\rangle\langle\phi_{t,d}^{u}|
\]
appears between two of the $N$ registers (tensored with rank 1 projectors on the other $N-2$ registers). Using equation \eq{phi_z_a_q} we may expand $|\phi_{y,c}^{s}\rangle$ and $|\phi_{t,d}^{u}\rangle$ as superpositions (with amplitudes $\pm\frac{1}{\sqrt{3R+2}}$ if $R$ is odd or $\pm\frac{1}{\sqrt{3R-1}}$ if $R$ is even) of the basis states from $\mathcal{W}_{(y,c,s)}$ and $\mathcal{W}_{(t,d,u)}$ respectively. Since $\mathcal{W}_{(y,c,s)}$ and $\mathcal{W}_{(t,d,u)}$ overlap on some diagram element, there exists $l\in L^{\square}$ such that $|\psi_{x_{1},b_{1}}^{l}\rangle\in\mathcal{W}_{(y,c,s)}$ and $|\psi_{x_{2},b_{2}}^{l}\rangle\in\mathcal{W}_{(t,d,u)}$ for some $x_{1},x_{2},b_{1},b_{2}\in\{0,1\}$. Hence
\[
|\phi_{y,c}^{s}\rangle|\phi_{t,d}^{u}\rangle=c_{0}\left(\pm|\psi_{x_{1},b_{1}}^{l}\rangle|\psi_{x_{2},b_{2}}^{l}\rangle\right)+\left(1-c_{0}^{2}\right)^{\frac{1}{2}}|\Phi_{y,c,t,d}^{s,u}\rangle
\]
where $c_{0}$ is either $\frac{1}{3R+2}$ (if $R$ is odd) or $\frac{1}{3R-1}$ (if $R$ is even). Now applying \fct{lin_alg_fact} we get
\begin{equation}
|\phi_{y,c}^{s}\rangle\langle\phi_{y,c}^{s}|\otimes|\phi_{t,d}^{u}\rangle\langle\phi_{t,d}^{u}|=c_{0}^{2}|\psi_{x_{1},b_{1}}^{l}\rangle\langle\psi_{x_{1},b_{1}}^{l}|\otimes|\psi_{x_{2},b_{2}}^{l}\rangle\langle\psi_{x_{2},b_{2}}^{l}|+M_{y,c,t,d}^{s,u}\label{eq:expand_phi_phi_proj-1-1}
\end{equation}
where $\Vert M_{y,c,t,d}^{s,u}\Vert \leq1-\frac{3}{4}\left(\frac{1}{4R}\right)^{4}$. For each $(\vec{z},\vec{a},\vec{q})\in\mathcal{Q}(\mathcal{N})$ we define $\mathcal{P}_{(\vec{z},\vec{a},\vec{q})}^{M}$ to be the operator obtained from $\mathcal{P}_{(\vec{z},\vec{a},\vec{q})}$ by replacing
\[
|\phi_{y,c}^{s}\rangle\langle\phi_{y,c}^{s}|\otimes|\phi_{t,d}^{u}\rangle\langle\phi_{t,d}^{u}|\mapsto M_{y,c,t,d}^{s,u}
\]
on two of the registers and we let $\Pi_{\mathcal{N}}^{\triangle}$ be given by \eq{pi_triangle_N}. Then, as in the previous case, $\langle\kappa|\Pi^{\triangle}|\kappa\rangle=\langle\kappa|\Pi_{\mathcal{N}}^{\triangle}|\kappa\rangle$ for all $|\kappa\rangle\in\mathcal{I}_{\diamondsuit}(\mathcal{N})$ and using the same reasoning as before, we get the bound \eq{bound_norm_pi_triangle_n} on $\Vert \Pi_{\mathcal{N}}^{\triangle}\Vert $. Using these two facts we get the same bound on the smallest eigenvalue within a block of type 3 as the bound we obtained for blocks of type 2:
\begin{align*}
\min_{|\kappa\rangle\in\mathcal{I}_{\diamondsuit}(\mathcal{N})}\langle\kappa|\sum_{w=1}^{N}h_{\mathcal{E}^{\triangle}}^{(w)}|\kappa\rangle &>\frac{1}{(30R)^{2}}\left(1-\big\Vert \Pi_{\mathcal{N}}^{\triangle}\big\Vert \right)>\frac{1}{(9R)^{6}}.\qedhere
\end{align*}
\end{proof}

\subsection{The gate graph $G^{\square}$ }

We now consider the gate graph $G^{\square}$ and prove \lem{oc}. We first show that $G^{\square}$ is an $e_{1}$-gate graph. From equations \eq{A_G_squAre}, \eq{h_se_square}, and \eq{A_g_triangle} we have
\begin{equation}
A(G^{\square})=A(G^{\triangle})+h_{\mathcal{E}^{0}}+h_{\mathcal{S}^{0}}.\label{eq:A_g_square_triangle}
\end{equation}
\lem{A(G_triangle)} characterizes the $e_{1}$-energy ground space of $G^{\triangle}$ and gives an orthonormal basis $\{|\phi_{z,a}^{q}\rangle\colon z,a\in\{0,1\},\, q\in[R]\}$ for it. To solve for the $e_{1}$-energy ground space of $A(G^{\square})$, we solve for superpositions of the states $\{|\phi_{z,a}^{q}\rangle\}$ in the nullspace of $h_{\mathcal{E}^{0}}+h_{\mathcal{S}^{0}}$.

Recall the definition of the sets $\mathcal{E}^{0}$ and $\mathcal{S}^{0}$. From \sec{Definitions-and-Notation_G_square}, each node $(q,z,t)$ in the gate diagram for $G$ is associated with a node $\new(q,z,t)$ in the gate diagram for $G^{\square}$ as described by \eq{node_mapping_G_G_square}. This mapping is depicted in \fig{replace_gate_diagram} by the black and grey arrows. Applying this mapping to each pair of nodes in the edge set $\mathcal{E}^{G}$ and each node in the self-loop set $\mathcal{S}^{G}$ of the gate diagram for $G$, we get the sets $\mathcal{E}^{0}$ and $\mathcal{S}^{0}$. Hence, using equations \eq{h_loops} and \eq{h_edges},
\begin{align}
h_{\mathcal{S}^0} &=\sum_{(q,z,t)\in \mathcal{S}^G} |{\new(q,z,t)}\rangle\langle{\new(q,z,t)}|\otimes \id \label{eq:hS0}\\
h_{\mathcal{E}^0} &=\sum_{\{(q,z,t),(q^\prime,z^\prime,t^\prime)\}\in \mathcal{E}^G} \left(|{\new(q,z,t)}\rangle+|{\new(q^\prime,z^\prime,t^\prime)}\rangle\right)\left(\langle{\new(q,z,t)}|+\langle{\new(q^\prime,z^\prime,t^\prime)}|\right)\otimes \id.\label{eq:hE0}
\end{align}

Using equation \eq{phi_z_a_q}, we see that for all nodes $(q,z,t)$ in the gate diagram for $G$ and for all $j\in\{0,\ldots,7\}$, $x,b\in\{0,1\}$, and $r\in[R]$,
\begin{align}
\langle{\new(q,z,t),j}|\phi^r_{x,b}\rangle & =\sqrt{c_0}
\begin{cases}
\langle q_{\mathrm{in}},z,t,j|\psi^{r_{\mathrm{in}}}_{x,b}\rangle & \text{ if $(q,z,t)$ is an input node}\\
\langle q_{\mathrm{out}},z,t,j|\psi^{r_{\mathrm{out}}}_{x,b}\rangle & \text{ if $(q,z,t)$ is an output node}
\end{cases}\nonumber \\
&=\sqrt{c_0}\delta_{r,q} \langle z,t,j|\psi_{x,b}\rangle
\label{eq:mat_el_newnodes}
\end{align}
where $c_0$ is $\frac{1}{3R+2}$ if $R$ is odd or $\frac{1}{3R-1}$ if $R$ is even, and where $|\psi_{x,b}\rangle$ is defined by equations \eq{psi0m} and \eq{psi1m}. The matrix element on the left-hand side of this equation is evaluated in the Hilbert space $\mathcal{Z}_1 (G^\square)$ where each basis vector corresponds to a vertex of the graph $G^\square$; these vertices are labeled $(l,z,t,j)$ with $l\in L^\square$, $z\in\{0,1\}$, $t\in[8]$, and $j\in\{0,\ldots,7\}$. However, from \eq{mat_el_newnodes} we see that
\begin{equation}
\underbrace{\langle \new(q,z,t),j|\phi^r_{x,b}\rangle}_{\text{in } \mathcal{Z}_1 (G^\square)}=\sqrt{c_0}\underbrace{\langle q,z,t,j|\psi^r_{x,b}\rangle}_{\text{in } \mathcal{Z}_1 (G)} \label{eq:twoHspaces}
\end{equation}
where the right-hand side is evaluated in the Hilbert space $\mathcal{Z}_1 (G)$.

Putting together equations \eq{hS0}, \eq{hE0}, and \eq{twoHspaces} gives
\begin{equation}
\langle\phi_{z,a}^{q}|h_{\mathcal{E}^{0}}+h_{\mathcal{S}^{0}}|\phi_{x,b}^{r}\rangle=\langle\psi_{z,a}^{q}|h_{\mathcal{E}^{G}}+h_{\mathcal{S}^{G}}|\psi_{x,b}^{r}\rangle\cdot\begin{cases}
\frac{1}{3R+2} & R\text{ odd}\\
\frac{1}{3R-1} & R\text{ even}
\end{cases}\label{eq:h_eG_hsG}
\end{equation}
for all $z,a,x,b\in\{0,1\}$ and $q,r\in[R]$. On the left-hand side of this equation, the Hilbert space is $\mathcal{Z}_{1}(G^{\square})$; on the right-hand side it is $\mathcal{Z}_{1}(G)$.

We use equation \eq{h_eG_hsG} to relate the $e_1$-energy ground states of $A(G)$ to those of $A(G^\square)$. Since $G$ is an $e_{1}$-gate graph, there is a state 
\[
|\Gamma\rangle=\sum_{z,a,q}\alpha_{z,a,q}|\psi_{z,a}^{q}\rangle\in\mathcal{Z}_{1}(G)
\]
that satisfies $A(G)|\Gamma\rangle=e_1|\Gamma\rangle$ and hence $h_{\mathcal{E}^{G}}|\Gamma\rangle=h_{\mathcal{S}^{G}}|\Gamma\rangle=0$. Letting
\[
|\Gamma^{\prime}\rangle=\sum_{z,a,q}\alpha_{z,a,q}|\phi_{z,a}^{q}\rangle\in\mathcal{Z}_{1}(G^{\square})
\]
and using equation \eq{h_eG_hsG}, we see that $\langle\Gamma^{\prime}|h_{\mathcal{E}^{0}}+h_{\mathcal{S}^{0}}|\Gamma^{\prime}\rangle=0$ and therefore $\langle\Gamma^{\prime}|A(G^{\square})|\Gamma^{\prime}\rangle=e_{1}$. Hence $G^{\square}$ is an $e_{1}$-gate graph. Moreover, the linear mapping from $\mathcal{Z}_{1}(G)$ to $\mathcal{Z}_{1}(G^\square)$ defined by
\begin{equation}
|\psi_{z,a}^{q}\rangle \mapsto |\phi_{z,a}^{q}\rangle\label{eq:map_1particle}
\end{equation}
maps each $e_{1}$-energy eigenstate of $A(G)$ to an $e_{1}$-energy eigenstate of $A(G^{\square})$.

Now consider the $N$-particle Hamiltonian $H(G^{\square},N)$. Using equation \eq{A_g_square_triangle} and the fact that both $A(G^{\square})$ and $A(G^{\triangle})$ have smallest eigenvalue $e_{1}$, we have
\[
H(G^{\square},N)=H(G^{\triangle},N)+\sum_{w=1}^{N}\left(h_{\mathcal{E}^{0}}+h_{\mathcal{S}^{0}}\right)^{(w)}\bigg|_{\mathcal{Z}_{N}(G^{\square})}.
\]
Recall from \lem{A(G_triangle)} that the nullspace of the first term is $\mathcal{I}_{\triangle}$. The $N$-fold tensor product of the mapping \eq{map_1particle} acts on basis vectors of $\mathcal{I}(G,\goc,N)$ as 
\begin{equation}
\Sym(|\psi_{z_{1},a_{1}}^{q_{1}}\rangle|\psi_{z_{2},a_{2}}^{q_{2}}\rangle\ldots|\psi_{z_{N},a_{N}}^{q_{N}}\rangle)
\mapsto
\Sym(|\phi_{z_{1},a_{1}}^{q_{1}}\rangle|\phi_{z_{2},a_{2}}^{q_{2}}\rangle\ldots|\phi_{z_{N},a_{N}}^{q_{N}}\rangle),
\label{eq:N_particle_mapping}
\end{equation}
where $z_{i},a_{i}\in\{0,1\},\; q_{i}\neq q_{j},\;\text{and }\{q_{i},q_{j}\}\notin E(\goc)$.
Clearly this defines an invertible linear map between the two spaces
$\mathcal{I}(G,\goc,N)$ and $\mathcal{I}_{\triangle}$. Let $|\Theta\rangle\in\mathcal{I}(G,\goc,N)$
and write $|\Theta^{\prime}\rangle\in\mathcal{I}_{\triangle}$ for its image under the map \eq{N_particle_mapping}. Then 
\begin{equation}
\langle\Theta^{\prime}|H(G^{\square},N)|\Theta^{\prime}\rangle=\langle\Theta^{\prime}|\sum_{w=1}^{N}\left(h_{\mathcal{E}^{0}}+h_{\mathcal{S}^{0}}\right)^{(w)}|\Theta^{\prime}\rangle=\langle\Theta|\sum_{w=1}^{N}\left(h_{\mathcal{E}^{G}}+h_{\mathcal{S}^{G}}\right)^{(w)}|\Theta\rangle\cdot\begin{cases}
\frac{1}{3R+2} & R\text{ odd}\\
\frac{1}{3R-1} & R\text{ even}
\end{cases}\label{eq:Theta_Theta_prime_eqn}
\end{equation}
where in the first equality we used the fact that $|\Theta^\prime\rangle$ is in the nullspace $\mathcal{I}_{\triangle}$ of $H(G^\triangle,N)$ and in the second equality we used equation \eq{h_eG_hsG} and the fact that $\langle\phi_{z,a}^{q}|\phi_{x,b}^{r}\rangle=\langle\psi_{z,a}^{q}|\psi_{x,b}^{r}\rangle$. We now complete the proof of \lem{oc} using equation \eq{Theta_Theta_prime_eqn}.

\subsubsection*{Case 1: $\lambda_{N}(G,\goc)\leq a$}

In this case there exists a state $|\Theta\rangle\in\mathcal{I}(G,\goc,N)$ satisfying 
\[
\langle\Theta|\sum_{w=1}^{N}\left(h_{\mathcal{E}^{G}}+h_{\mathcal{S}^{G}}\right)^{(w)}|\Theta\rangle\leq a.
\]
From equation \eq{Theta_Theta_prime_eqn} we see that the state $|\Theta^{\prime}\rangle\in\mathcal{I}_{\triangle}$ satisfies $\langle\Theta^{\prime}|H(G^{\square},N)|\Theta^{\prime}\rangle\leq\frac{a}{3R-1}\leq\frac{a}{R}$. 

\subsubsection*{Case 2: $\lambda_{N}(G,\goc)\geq b$}

In this case 
\begin{align*}
\lambda_{N}(G,\goc)=\min_{|\Theta\rangle\in\mathcal{I}(G,\goc,N)}\langle\Theta|H(G,\goc,N)|\Theta\rangle & = \min_{|\Theta\rangle\in\mathcal{I}(G,\goc,N)}\langle\Theta|\sum_{w=1}^{N}\left(h_{\mathcal{E}^{G}}+h_{\mathcal{S}^{G}}\right)^{(w)}|\Theta\rangle\geq b.
\end{align*}
 Now applying equation \eq{Theta_Theta_prime_eqn} gives
\begin{equation}
\min_{|\Theta^{\prime}\rangle\in\mathcal{I}_{\triangle}}\langle\Theta^{\prime}|H(G^{\square},N)|\Theta^{\prime}\rangle=\min_{|\Theta^{\prime}\rangle\in\mathcal{I}_{\triangle}}\langle\Theta^{\prime}|\sum_{w=1}^{N}\left(h_{\mathcal{E}^{0}}+h_{\mathcal{S}^{0}}\right)^{(w)}|\Theta^{\prime}\rangle\geq\frac{1}{3R+2}b.\label{eq:bound_thetaprime}
\end{equation}
This establishes that the nullspace of $H(G^{\square},N)$ is empty, i.e., $\lambda_{N}^{1}(G^{\square})>0$, so $\lambda_{N}^{1}(G^{\square})=\gamma(H(G^{\square},N))$. We lower bound $\lambda_{N}^{1}(G^{\square})$ using the \npl (\lem{npl}) with
\[
H_{A}=H(G^{\triangle},N)\qquad H_{B}=\sum_{w=1}^{N}\left(h_{\mathcal{E}^{0}}+h_{\mathcal{S}^{0}}\right)^{(w)}\bigg|_{\mathcal{Z}_{N}(G^{\square})}
\]
and where the nullspace of $H_{A}$ is $S=\mathcal{I}_{\triangle}$. We apply \lem{npl} and use the bounds $\gamma(H_{A})>\frac{1}{\left(17R\right)^{7}}$ (from \lem{The-nullspace-of_H_triangle}), $\gamma(H_{B}|_{S})\geq\frac{b}{3R+2}$ (from equation \eq{bound_thetaprime}), and $\left\Vert H_{B}\right\Vert \leq N\left\Vert {h_{\mathcal{E}^{0}}+h_{\mathcal{S}^{0}}}\right\Vert \leq3N\leq3R$ (using equations \eq{h_E_bound} and \eq{h_S_bound} and the fact that $N\leq R$) to find 
\begin{align*}
\lambda_{N}^{1}(G^{\square}) &= \gamma(H(G^{\square},N)) \\
 &\geq \frac{b}{(3R+2)(17R)^7(\frac{1}{(17R)^{7}}+\frac{b}{3R+2}+3R)} \\
 & \geq\frac{b}{R^{9}} \cdot \frac{1}{3+2+b\cdot(17)^{7}+3\cdot\left(3+2\right)(17)^{7}} \\
 & >\frac{b}{(13R)^{9}}
\end{align*}
where in the denominator we used the fact that $b\leq1$.

\section{Analysis of gadgets for two-qubit gates}
\label{app:graph_gadgets}

In this Section we prove \lem{2qub_gate}.

\Twoqub*

\begin{proof}
Recall that the gate graph $G_U$ is specified by its gate diagram, shown in \fig{GVucnot}. The adjacency matrix of the gate graph $G_{U}$ is of the form in equation \eq{adj_gate_graph}. There are 6 diagram elements for each of the move-together gadgets, so there are $32$ diagram elements in total. We will need to refer to those diagram elements labeled $q\in [8]$ in \fig{GVucnot} (i.e., those not contained in the move-together gadgets).

Write 
\[
A(G_{U})=A(G_{U}^{\prime})+h_{\mathcal{E}^{\prime}}
\]
where $G_{U}^{\prime}$ is the gate graph obtained from $G_{U}$ by removing all $24$ edges shown in \fig{GVucnot} ($G_U^{\prime}$ does include the edges within each of the move-together gadgets). Here $h_{\mathcal{E}^{\prime}}$ is given by equation \eq{h_edges} with $\mathcal{E^{\prime}}$ the set of $24$ edges shown in \fig{GVucnot}.

One basis for the $e_1$-energy ground space of $A(G_{U}^{\prime})$ is given by the $64$ states 
\begin{align*}
&|\psi_{z,a}^{q}\rangle,\quad q\in [8],\, z,a\in\{0,1\} \\
&|\chi_{L,a}^{xy}\rangle,\quad x,y,a\in\{0,1\},\, L\in[4].
\end{align*}
However, it is convenient to work with the following slightly different basis for this space:
\begin{align*}
&|\psi_{z,a}^{q}\rangle,\quad q\in[8],\, z,a\in\{0,1\} \\
&\sum_{x\in\{0,1\}}\tilde{U}(a)_{xz}|\chi_{1,a}^{xy}\rangle,\quad y,z,a\in\{0,1\} \\
&|\chi_{L,a}^{xy}\rangle,\quad x,y,a\in\{0,1\},\, L\in \{2,3,4\}.
\end{align*}
Here some of the states are in a superposition corresponding to the output of the single-qubit unitary $\tilde U$.

We are interested in the intersection of the ground space of $A(G_U')$ with the nullspace of $h_{\mathcal{E}^{\prime}}$, so we compute the matrix elements of $h_{\mathcal{E}^{\prime}}$ in the above basis. The resulting $64\times 64$ matrix is block diagonal with sixteen $4\times 4$ blocks. Each block is identical, with entries
\begin{equation}
\begin{pmatrix}
\frac{3}{8} & \frac{1}{8} & \frac{1}{8\sqrt{3}} & \frac{1}{8\sqrt{3}}\\
\frac{1}{8} & \frac{1}{8} & 0 & 0\\
\frac{1}{8\sqrt{3}} & 0 & \frac{1}{24} & 0\\
\frac{1}{8\sqrt{3}} & 0 & 0 &  \frac{1}{24}
\end{pmatrix}.\label{eq:three_by_three_blocks}
\end{equation}
The four states involved in each block are given by (in order from left to right as in the matrix above):
\begin{align*}
 & |\psi_{z,a}^{1}\rangle, |\psi_{z,a}^{5+z}\rangle,\sum_{x\in\{0,1\}}\tilde{U}(a)_{xz}|\chi_{1,a}^{x0}\rangle,\sum_{x\in\{0,1\}}\tilde{U}(a)_{xz}|\chi_{1,a}^{x1}\rangle\\
 & |\psi_{z,a}^{2}\rangle,|\psi_{z,a}^{6-z}\rangle,|\chi_{2,a}^{z0}\rangle,|\chi_{2,a}^{z1}\rangle\\
 & |\psi_{z,a}^{3}\rangle,|\psi_{z,a}^{7}\rangle,|\chi_{3,a}^{0z}\rangle,|\chi_{3,a}^{1z}\rangle\\
 & |\psi_{z,a}^{4}\rangle,|\psi_{z,a}^{8}\rangle,|\chi_{4,a}^{0z}\rangle,|\chi_{4,a}^{1\left(z\oplus1\right)}\rangle.
\end{align*}

The unique zero eigenvector of the matrix \eq{three_by_three_blocks} is 
\[
\frac{1}{\sqrt{8}} \begin{pmatrix}
1\\
-1\\
-\sqrt{3}\\
-\sqrt{3}
\end{pmatrix}.
\]
Constructing this vector within each of the $16$ blocks, we get the states $\{|\rho_{z,a}^{1,U}\rangle,|\rho_{z,a}^{2,U}\rangle,|\rho_{z,a}^{3,U}\rangle,|\rho_{z,a}^{4,U}\rangle\}$. 

Now consider the two-particle sector. Using \lem{FF_characterization} we can write any two-particle frustration-free state as 
\begin{equation}
|\Theta\rangle=\sum_{z,a,x,b\in \{0,1\}}\sum_{I,J\in [4]}B_{\left(z,a,I\right),\left(x,b,J\right)}|\rho_{z,a}^{I,U}\rangle|\rho_{x,b}^{J,U}\rangle
\label{eq:thetastate}
\end{equation}
where 
\begin{equation}
B_{\left(z,a,I\right),\left(x,b,J\right)}=B_{\left(x,b,J\right),\left(z,a,I\right)} \label{eq:sym_B}
\end{equation}
 and 
\begin{equation}
\langle\psi_{x,a}^{q}|\langle\psi_{z,b}^{q}|\Theta\rangle=0\label{eq:ff_condition}
\end{equation}
for all $x,z,a,b\in\{0,1\}$ and $q\in[32]$. To enforce equation \eq{ff_condition} we consider the diagram elements $q\in [8]$ (as labeled in \fig{GVucnot}) separately from the other $24$ diagram elements (those inside the move-together gadgets). 

Using equation \eq{ff_condition} with $q\in \{1,2,3,4,7,8\}$ and $x,z,a,b\in\{0,1\}$ gives
\begin{equation}
B_{\left(x,a,I\right),\left(z,b,I\right)}=0\quad I\in[4], ~ x,z,a,b \in \{0,1\}\label{eq:II_constraint}.
\end{equation}
Using $q=5$, $x=0$, and $z=1$ in equation \eq{ff_condition} gives 
\begin{equation}
  \langle \psi_{0,a}^5|\langle{\psi_{1,b}^5}|\Theta\rangle = \frac{1}{8} B_{(0,a,1),(1,b,2)} = 0,\label{eq:constraint_q5}
\end{equation}
for $a,b \in \{0,1\}$, while $q=6$, $x=0$, and $z=1$ gives
\begin{equation}
  \langle \psi_{0,a}^6|\langle{\psi_{1,b}^6}|\Theta\rangle = \frac{1}{8} B_{(0,a,2),(1,b,1)} = 0.\label{eq:constraint_q6}
\end{equation}
Applying equation \eq{ff_condition} with $q=5$ or $q=6$ and other choices for $x$ and $z$ does not lead to any additional independent constraints on the state $|\Theta\rangle$.

Now consider the constraint \eq{ff_condition} for diagram elements inside the move-together gadgets in \fig{GVucnot}. Let $\Pi_{xy}$ be the projector onto two-particle states where both particles are located at vertices contained within the move-together gadget labeled $xy\in \{00,01,10,11\}$. Using the results of \lem{Wgadget_lemma}, we see that for diagram elements inside the move-together gadgets, \eq{ff_condition} is satisfied if and only if 
\[
\Pi_{xy}|\Theta\rangle\in \mathrm{span}\big\{ \mathrm{Sym}\left(|\chi_{1,a}^{xy}\rangle|\chi_{3,b}^{xy}\rangle+|\chi_{2,a}^{xy}\rangle|\chi_{4,b}^{xy}\rangle\right), \; a,b\in \{0,1\}\big\}.
\]
Since we already know
\[
\Pi_{xy}|\Theta\rangle\in \mathrm{span}\big\{\mathrm{Sym}\left( |\chi_{i,a}^{xy}\rangle|\chi_{j,b}^{xy}\rangle\right), \; i,j\in [4], a,b\in \{0,1\}\big\}
\]
we get
\begin{align}
\langle\chi_{K,a}^{xy}|\langle\chi_{K,b}^{xy}|\Theta\rangle & =0\quad K\in[4]\label{eq:diag_chi_constraint}\\
\langle\chi_{K,a}^{xy}|\langle\chi_{L,b}^{xy}|\Theta\rangle & =0\quad(K,L)\in\{(1,2),(2,3),(3,4),(1,4)\}\label{eq:off_diag_chi_constraint}\\
\big(\langle\chi_{1,a}^{xy}|\langle\chi_{3,b}^{xy}|-\langle\chi_{2,a}^{xy}|\langle\chi_{4,b}^{xy}|\big)|\Theta\rangle & =0\label{eq:sum_chi_constraint}
\end{align}
for all $a,b\in \{0,1\}$. Note that \eq{diag_chi_constraint} is automatically satisfied whenever \eq{II_constraint} holds.  

Applying equation \eq{off_diag_chi_constraint} with $(K,L)=(1,2)$ and $a,b,x,y\in\{0,1\}$, we get 
\begin{equation}
\langle\chi_{1,a}^{xy}|\langle\chi_{2,b}^{xy}|\Theta\rangle=\frac{3}{8}\sum_{z\in \{0,1\}}\tilde U (a)_{xz}B_{\left(z,a,1\right),\left(x,b,2\right)}
= \frac{3}{8} \tilde U (a)_{xx}B_{\left(x,a,1\right),\left(x,b,2\right)}
=0.\label{eq:Uxx_eqn}
\end{equation}
In the second equality we used the fact that $B_{\left(z,a,1\right),\left(x,b,2\right)}$ is zero whenever $z\neq x$ (from equations \eq{sym_B}, \eq{constraint_q5}, and \eq{constraint_q6}). Since $\tilde{U}\in \{1,H,HT\}$ we have $\tilde{U}(a)_{xx} \neq 0$, and it follows that
\begin{equation}
B_{\left(x,a,1\right),\left(x,b,2\right)}=0\label{eq:diag12}
\end{equation}
for all $x,a,b\in \{0,1\}$.

Applying equation \eq{off_diag_chi_constraint} with $(K,L)=(1,4)$ gives
\begin{equation}
\langle\chi_{1,a}^{xy}|\langle\chi_{4,b}^{xy}|\Theta\rangle=\frac{3}{8}\sum_{z\in \{0,1\}} \tilde U (a)_{xz} B_{\left(z,a,1\right),\left(x\oplus y,b,4\right)}=0 \qquad x,y,a,b\in \{0,1\}.
\end{equation}
By taking appropriate combinations of these equations, we have
\begin{equation}
\sum_{x\in \{0,1\}} \tilde U (a)^{\dagger}_{wx} \langle\chi_{1,a}^{x(y\oplus x)}|\langle\chi_{4,b}^{x (y\oplus x)}|\Theta\rangle= B_{\left(w,a,1\right),\left(y,b,4\right)}=0 \qquad w,y,a,b\in \{0,1\}. \label{eq:14constraint}
\end{equation}

Applying equation \eq{off_diag_chi_constraint} with $(K,L)=(2,3)$ and $(K,L)=(3,4)$  gives
\begin{align}
\langle\chi_{2,a}^{xy}|\langle\chi_{3,b}^{xy}|\Theta\rangle& =\frac{3}{8}B_{\left(x,a,2\right),\left(y,b,3\right)}=0 \label{eq:23constraint}\\
\langle\chi_{3,a}^{xy}|\langle\chi_{4,b}^{xy}|\Theta\rangle& =\frac{3}{8}B_{\left(x,a,3\right),\left(x\oplus y,b,4\right)}=0 \label{eq:43constraint}
\end{align}
for all $x,y,a,b\in \{0,1\}$.

Now putting together equations \eq{II_constraint}, \eq{constraint_q5}, \eq{constraint_q6}, \eq{diag12}, \eq{14constraint}, \eq{23constraint}, and \eq{43constraint} (and using the symmetrization \eq{sym_B}), we get
\[
B_{\left(x,a,I\right),\left(z,b,J\right)}=0\quad \text{ for all } x,z,a,b\in \{0,1\},\text{ where }I = J \text{ or } \{I,J\} \in \big\{\{1,2\},\{1,4\},\{2,3\},\{3,4\}\big\},
\]
so
\begin{equation}
|\Theta\rangle=\sum_{z,c,w,d\in \{0,1\} }B_{\left(z,c,1\right),\left(w,d,3\right)}\left(|\rho_{z,c}^{1,U}\rangle|\rho_{w,d}^{3,U}\rangle+|\rho_{w,d}^{3,U}\rangle|\rho_{z,c}^{1,U}\rangle\right)+B_{\left(z,c,2\right),\left(w,d,4\right)}\left(|\rho_{z,c}^{2,U}\rangle|\rho_{w,d}^{4,U}\rangle+|\rho_{w,d}^{4,U}\rangle|\rho_{z,c}^{2,U}\rangle\right).\label{eq:sum4terms}
\end{equation}
Now 
\begin{align*}
\langle\chi_{1,a}^{xy}|\langle\chi_{3,b}^{xy}|\rho_{z,c}^{1,U}\rangle|\rho_{w,d}^{3,U}\rangle & =\frac{3}{8}\delta_{a,c}\delta_{b,d} \tilde{U}(a)_{xz} \delta_{y,w}\\
\langle\chi_{2,a}^{xy}|\langle\chi_{4,b}^{xy}|\rho_{z,c}^{2,U}\rangle|\rho_{w,d}^{4,U}\rangle & =\frac{3}{8}\delta_{a,c}\delta_{b,d} \delta_{x,z} \delta_{y,w\oplus x},
\end{align*}
so enforcing equation \eq{sum_chi_constraint} gives 
\[
\sum_{z\in\{0,1\}}\tilde{U}(a)_{xz}B_{\left(z,a,1\right),\left(y,b,3\right)}=B_{\left(x,a,2\right),\left(x\oplus y,b,4\right)}
\]
for each $x,y,a,b\in\{0,1\}$. In other words 
\[
B_{\left(z,c,2\right),\left(w,d,4\right)} =\sum_{x\in\{0,1\}}\tilde{U}(c)_{zx} B_{\left(x,c,1\right),\left(z\oplus w,d,3\right)} =\sum_{x,y\in\{0,1\}}U(c)_{zw,xy}B_{\left(x,c,1\right),\left(y,d,3\right)}
\]
where we used $U(a)=\mathrm{CNOT}_{12} (\tilde U(a)\otimes 1)$. Plugging this into \eq{sum4terms} gives 
\begin{align*}
|\Theta\rangle&=\sum_{z,c,w,d\in \{0,1\} }\bigg(B_{\left(z,c,1\right),\left(w,d,3\right)}\left(|\rho_{z,c}^{1,U}\rangle|\rho_{w,d}^{3,U}\rangle+|\rho_{w,d}^{3,U}\rangle|\rho_{z,c}^{1,U}\rangle\right)\\
&\quad+\sum_{x,y\in\{0,1\}}U(c)_{zw,xy}B_{\left(x,c,1\right),\left(y,d,3\right)}\left(|\rho_{z,c}^{2,U}\rangle|\rho_{w,d}^{4,U}\rangle+|\rho_{w,d}^{4,U}\rangle|\rho_{z,c}^{2,U}\rangle\right)\bigg)\\
&= \sum_{z,c,w,d\in \{0,1\} }B_{\left(z,c,1\right),\left(w,d,3\right)}\bigg[|\rho_{z,c}^{1,U}\rangle|\rho_{w,d}^{3,U}\rangle+|\rho_{w,d}^{3,U}\rangle|\rho_{z,c}^{1,U}\rangle\\
&\quad+\sum_{x,y\in\{0,1\}}U(c)_{xy,zw}\left(|\rho_{x,c}^{2,U}\rangle|\rho_{y,d}^{4,U}\rangle+|\rho_{y,d}^{4,U}\rangle|\rho_{x,c}^{2,U}\rangle\right)\bigg]\\
&= \sum_{z,c,w,d\in \{0,1\} }2 B_{\left(z,c,1\right),\left(w,d,3\right)}\Sym(|T_{z,c,w,d}\rangle).
\end{align*}
This is the general solution to equations \eq{thetastate}--\eq{ff_condition}, so the space of two-particle frustration-free states for $G_U$ is spanned by the 16 orthonormal states \eq{twopartstate_1}.

Finally, we show that there are no three-particle frustration-free states. By \lem{increase_part_number}, this implies that there are no frustration-free states for more than two particles. Suppose (to reach a contradiction) that $|\Gamma\rangle$ is a normalized three-particle frustration-free state. Write
\[
|\Gamma\rangle=\sum E_{(x,a,q),(y,b,r),(z,c,s)}|\rho_{x,a}^{q}\rangle|\rho_{y,b}^{r}\rangle|\rho_{z,c}^{s}\rangle
\]
and note that each reduced density matrix of $|\Gamma\rangle$ on two of the three subsystems must have all of its support on two-particle frustration-free states (see the remark following \lem{FF_characterization}). Using this fact for each two-particle subsystem we get
\begin{align*}
(q,r)\notin\{(1,3),(3,1),(2,4),(4,2)\}\quad\Longrightarrow\quad E_{(x,a,q),(y,b,r),(z,c,s)} & =0\\
(q,s)\notin\{(1,3),(3,1),(2,4),(4,2)\}\quad\Longrightarrow\quad E_{(x,a,q),(y,b,r),(z,c,s)} & =0\\
(r,s)\notin\{(1,3),(3,1),(2,4),(4,2)\}\quad\Longrightarrow\quad E_{(x,a,q),(y,b,r),(z,c,s)} & =0
\end{align*}
which together imply that $|\Gamma\rangle=0$ (a contradiction). Hence no three-particle frustration-free state exists.
\end{proof}

\section{Technical supporting material}
\label{app:tech_support}
\subsection{Basic properties of the Bose-Hubbard model}\label{sec:basic_properties}

In this short section we prove \lem{increase_part_number} and \lem{BH_disconnected_graphs}.
\incrementN*

\begin{proof}
Let $\widehat{n}_{i}^{N}$ be the number operator \eq{n_hat} defined in the $N$-particle space and let $\widehat{n}_{i}^{N+1}$ be the corresponding operator in the $\left(N+1\right)$-particle space. Note that
\[
\widehat{n}_{i}^{N+1} = \widehat{n}_{i}^{N}\otimes\id+|i\rangle\langle i|^{(N+1)} \geq \widehat{n}_{i}^{N}\otimes\id.
\]
Using this and the fact that $A(G)\geq\mu(G)$, we get 
\[
H_{G}^{N+1}-\left(N+1\right)\mu(G)\geq\left(H_{G}^{N}-N\mu(G)\right)\otimes\id.
\]
Hence 
\begin{align*}
\lambda_{N+1}^{1}(G) & =\min_{|\psi\rangle\in\mathcal{Z}_{N+1}(G)\colon \langle\psi|\psi\rangle=1}\langle\psi|H_{G}^{N+1}-\left(N+1\right)\mu(G)|\psi\rangle\\
 & \geq\min_{|\psi\rangle\in\mathcal{Z}_{N}(G)\otimes\C^{|V|}\colon \langle\psi|\psi\rangle=1}\langle\psi|\left(H_{G}^{N}-N\mu(G)\right)\otimes\id|\psi\rangle\\
 & =\lambda_{N}^{1}(G)
\end{align*}
(using the fact that $\mathcal{Z}_{N+1}(G)\subset\mathcal{Z}_{N}(G)\otimes\C^{|V|}$). 
\end{proof}

\disc*

\begin{proof}
Recall that the action of $H_{G}-N\mu(G)$ on the Hilbert space \eq{occupation_num_states} is the same as the action of $H(G,N)$ on the Hilbert space $\mathcal{Z}_{N}(G)$. States in these Hilbert spaces are identified via the mapping described in equation \eq{occup_num_symmetrized}. It is convenient to prove the Lemma by working with the second-quantized Hamiltonian $H_{G}$. We then translate our results into the first-quantized picture to obtain the stated claims.

For a graph with $k$ components, equation \eq{Bose-Hubbard_Ham} gives 
\begin{equation}
H_{G}=\sum_{i=1}^{k}H_{G_{i}}\label{eq:H_G_disconnected}
\end{equation}
where $[H_{G_{i}},H_{G_{j}}]=0.$ Label each vertex of $G$ by $(a,b)$ where $b\in[k]$ and $a\in[|V_{b}|]$, where $V_{b}$ is the vertex set of the component $G_{b}$. An occupation number basis state \eq{occupation_num_states} can be written 
\begin{equation}
|l_{1,1},\ldots,l_{|V_{1}|,1}\rangle|l_{1,2},\ldots,l_{|V_{2}|,2}\rangle\ldots|l_{1,k},\ldots,l_{|V_{k}|,k}\rangle.\label{eq:prod_basis_occ_num}
\end{equation}
The Hamiltonian $H_{G}-N\mu(G)$ conserves the number of particles $N_{b}$ in each component $b$. Within the sector corresponding to a given set $N_{1},\ldots,N_{k}$ with $\sum_{i \in [k]} N_i=N$, we have
\begin{align*}
& \left(H_{G}-N\mu(G)\right)|l_{1,1},\ldots,l_{|V_{1}|,1}\rangle|l_{1,2},\ldots,l_{|V_{2}|,2}\rangle\ldots|l_{1,k},\ldots,l_{|V_{k}|,k}\rangle\\
&\quad =\big(H_{G_1}-N_1\mu(G_1)|l_{1,1},\ldots,l_{|V_{1}|,1}\rangle\big)|l_{1,2},\ldots,l_{|V_{2}|,2}\rangle\ldots|l_{1,k},\ldots,l_{|V_{k}|,k}\rangle\\
&\qquad+|l_{1,1},\ldots,l_{|V_{1}|,1}\rangle\big(H_{G_2}-N_2\mu(G_2)|l_{1,2},\ldots,l_{|V_{2}|,2}\rangle\big)\ldots|l_{1,k},\ldots,l_{|V_{k}|,k}\rangle + \cdots\\
&\qquad+|l_{1,1},\ldots,l_{|V_{1}|,1}\rangle|l_{1,2},\ldots,l_{|V_{2}|,2}\rangle\ldots\big(H_{G_k}-N_k\mu(G_k)|l_{1,k},\ldots,l_{|V_{k}|,k}\rangle\big),
\end{align*}
where we used the fact that $\mu(G_i)=\mu(G)$ for $i\in[k]$.  From this equation we see that the eigenstates of $H_{G}$ can be obtained as product states with $k$ factors in the basis \eq{prod_basis_occ_num}. In each such product state, the $i$th factor is an eigenstate of $H_{G_{i}}-N_{i}\mu(G_{i})=H_{G_{i}}-N_{i}\mu(G)$ in the $N_{i}$-particle sector, with eigenvalue $\lambda_{N_{i}}^{j_{i}}(G_{i})$. Rewriting this result in the ``first-quantized'' language, we obtain the Lemma. 
\end{proof}
\subsection{\texorpdfstring{Proof of the \npl}{Proof of the Nullspace Projection Lemma}}
\label{sec:Proof-of-Lemma_MLM}

The following proof fills in the details of the argument given in reference \cite{MLM99}.

\NPL*

\begin{proof}
Let $|\psi\rangle$ be a normalized state satisfying 
\[
\langle\psi|H_{A}+H_{B}|\psi\rangle=\gamma(H_{A}+H_{B}).
\]
Let $\Pi_{S}$ be the projector onto the nullspace of $H_{A}$. First suppose that $\Pi_{S}|\psi\rangle=0$, in which case 
\[
\langle\psi|H_{A}+H_{B}|\psi\rangle\geq\langle\psi|H_{A}|\psi\rangle\geq\gamma(H_{A})
\]
and the result follows. On the other hand, if $\Pi_{S}|\psi\rangle\neq0$ then we can write 
\[
|\psi\rangle=\alpha|a\rangle+\beta|a^{\perp}\rangle
\]
with $|\alpha|^{2}+|\beta|^{2}=1$, $\alpha\neq0$, and two normalized states $|a\rangle$ and $|a^{\perp}\rangle$ such that $|a\rangle\in S$ and $|a^{\perp}\rangle\in S^{\perp}$. (If $\beta=0$ then we may choose $|a^{\perp}\rangle$ to be an arbitrary state in $S^{\perp}$ but in the following we fix one specific choice for concreteness.) Note that any state $|\phi\rangle$ in the nullspace of $H_{A}+H_{B}$ satisfies $H_{A}|\phi\rangle=0$ and hence $\langle\phi|a^{\perp}\rangle=0$. Since $\langle\phi|\psi\rangle=0$ and $\alpha\neq0$ we also see that $\langle\phi|a\rangle=0$. Hence any state
\[
|f(q,r)\rangle=q|a\rangle+r|a^{\perp}\rangle
\]
is orthogonal to the nullspace of $H_{A}+H_{B}$, and
\[
\gamma(H_{A}+H_{B})=\min_{|q|^{2}+|r|^{2}=1}\langle f(q,r)|H_{A}+H_{B}|f(q,r)\rangle.
\]
The operator $H_{A}+H_{B}$ acts on the two-dimensional space spanned by $|a\rangle$ and $|a^{\perp}\rangle$ as the matrix 
\[
\begin{pmatrix}
w & v^{*}\\
v & y+z
\end{pmatrix}
\]
where $w=\langle a|H_{B}|a\rangle$, $v=\langle a^{\perp}|H_{B}|a\rangle$, $y=\langle a^{\perp}|H_{A}|a^{\perp}\rangle$, and $z=\langle a^{\perp}|H_{B}|a^{\perp}\rangle$. The smallest eigenvalue of this matrix is 
\[
\gamma(H_{A}+H_{B})=\frac{w+y+z-\sqrt{\left(w+y+z\right)^{2}+4\left(|v|^{2}-wy-wz\right)}}{2}.
\]
Since $H_{B}$ is positive semidefinite, its principal minors are nonnegative, and in particular $wz-|v|^{2}\geq0$. Using this inequality in the above gives
\begin{align}
  \gamma(H_{A}+H_{B}) 
  &\geq \frac{w+y+z}{2} \left(1-\sqrt{1-\frac{4wy}{(w+y+z)^{2}}}\right) 
   \geq \frac{wy}{w+y+z}
   = \frac{1}{\frac{1}{w}+\frac{1}{y}+\frac{z}{wy}}
\label{eq:HAHBbound}
\end{align}
where in the second step we used the fact that $\sqrt{1-x}\leq1-\frac{x}{2}$ for $x\in[0,1]$. Since $|a\rangle$ is orthogonal to the nullspace of $H_{A}+H_{B}$, we have
\[
  w\geq\gamma(H_{B}|_{S})\geq c.
\]
We also have $y\geq\gamma(H_{A})\geq d$ and $z\leq\left\Vert H_{B}\right\Vert$.  Since the right-hand side of \eq{HAHBbound} increases monotonically with $w$ and $y$ and decreases monotonically with $z$, we find 
\[
  \gamma(H_{A}+H_{B})
  \geq \frac{cd}{c+d+\left\Vert H_{B}\right\Vert}
\]
as claimed.
\end{proof}

\subsection{\texorpdfstring{Proof of \lem{Pi_0_restriction}}{Proof of Lemma~\ref{lem:Pi_0_restriction}}}
\label{sec:Proof-of-Lemma Pi0_restriction}

Here we prove \lem{Pi_0_restriction}. We begin with the following Lemma relating the occupancy constraints graph $\gxoc$ to the illegal configurations. The proof of this Lemma uses the definitions of a configuration (\defn{configuration}), the sets of legal and illegal configurations (from \sec{Legal-configurations-and}), and the occupancy constraints graph $\gxoc$ (from \sec{The-occupancy-constraints}).

\begin{lemma}
\label{lem:legalconfig}
For any illegal configuration 
\begin{equation}
(J_{1},\ldots,J_{Y},L_{1},\ldots,L_{n-2Y})\label{eq:config1}
\end{equation}
there exist diagram elements $\{Q_{1},Q_{2}\}\in E(\gxoc)$ satisfying at least one of the following conditions:
\begin{enumerate}[label=(\roman*)]
\item $Q_{1}=(1,J_{k},0)$ and $Q_{2}=(1,J_{l},0)$ for some $k,l\in[Y]$.
\item $Y\in\{0,1\}$, $Q_{1}=L_{s}$, and $Q_{2}=L_{t}$ for some $s,t\in[n-2Y]$.
\item $Y=1$ and $Q_{1}=(i,J_{1},d)$ and $Q_{2}=L_{t}$ for some $i\in\{1,s(J_{1})\}$, $t\in[n-2]$ and $d\in\{0,1\}$.
\end{enumerate}
\end{lemma}

\begin{proof}
We prove the contrapositive: we suppose the configuration \eq{config1} violates each of the conditions (i), (ii), and (iii) for all $\{Q_{1},Q_{2}\}\in E(\gxoc)$ and show it is a legal configuration.

From part (1) of the definition of the occupancy constraints graph in \sec{The-occupancy-constraints}, we see that
\[
  \{(1,j,0),(1,k,0)\}\in E(\gxoc)
\]
for all $j,k\in[M]$ with $j\neq k$. If the configuration \eq{config1} has $Y\geq 2$, then we may choose $Q_1=(1,J_1,0)$ and $Q_2=(1,J_2,0)$ so that $\{Q_1,Q_2\}\in E(\gxoc)$ and condition (i) is satisfied (note $J_1\neq J_2$ follows from the definition of a configuration). Since by assumption, \eq{config1} violates condition (i) for all $\{Q_{1},Q_{2}\}\in E(\gxoc)$, this implies that $Y\in \{0,1\}$. We consider the cases $Y=0$ and $Y=1$ separately.

First suppose $Y=0$, so \eq{config1} is equal to
\[
(L_{1},\ldots,L_{n}).
\]
Since (ii) is violated, $\{L_s,L_t\}\notin E(\gxoc)$ for all $s,t\in [n]$. Using part (1) of the definition of $\gxoc$ and the definition of a configuration, this implies that each diagram element is in a different row, i.e., $L_{i}=(i,j_{i},c_{i})$ for each $i\in[n]$. From part (2) of the definition of $\gxoc$, we see in particular that $\{L_{1},L_{t}\}\notin E(\gxoc)$
for each $t\in\{2,\ldots,n\}$ implies
\[
L_{s(j_{1})}=(s(j_{1}),j_{1},c_{1})\quad\text{and}\quad L_{i}=F(i,j_{1},d_{i})
\]
for $i\in[n]\setminus\{1,s(j_{1})\}$ and bits $d_2,\ldots,d_n\in\{0,1\}$, i.e., the configuration is legal.

Now suppose $Y=1$, so \eq{config1} is
\[
  (J_1,L_{1},\ldots,L_{n-2}).
\]
Since (ii) is violated, each $L_{i}$ for $i\in[n-2]$ is from a different row. Since (iii) is violated, none of these diagram elements are in rows $1$ or $s(J_1)$. Now applying part (2) of the definition of $\gxoc$, we see that the configuration is legal:
\begin{align*}
&(J_1,L_{1},\ldots,L_{n-2}) \\
&\quad=(J_1,F(2,J_1,d_{2}),\ldots,F(s(J_1)-1,J_1,d_{s(J_1)-1}),F(s(J_1)+1,J_1,d_{s(J_1)+1}),\ldots,F(n,J_1,d_{n}))
\end{align*}
where $d_{i}\in\{0,1\}$ for $i\in[n]\setminus\{1,s(J_1)\}$.
\end{proof}

\restrictionlemma*

\begin{proof}
We begin with equation \eq{subspace_eqn0}. Recall from \eq{occup_space_defn} that
\[
\mathcal{I}(G_{1},\gxoc,n)=\spn\{ \Sym(|\psi_{z_{1},a_{1}}^{q_{1}}\rangle|\psi_{z_{2},a_{2}}^{q_{2}}\rangle\ldots|\psi_{z_{n},a_{n}}^{q_{n}}\rangle)\colon z_{i},a_{i}\in\{0,1\},\; q_{i}\neq q_{j},\;\text{and }\{q_{i},q_{j}\}\notin E(\gxoc)\} 
\]
which can alternatively be characterized as the subspace of 
\[
\mathcal{I}(G_{1},n)=\spn\{ \Sym(|\psi_{z_{1},a_{1}}^{q_{1}}\rangle|\psi_{z_{2},a_{2}}^{q_{2}}\rangle\ldots|\psi_{z_{n},a_{n}}^{q_{n}}\rangle)\colon z_{i},a_{i}\in\{0,1\},\; q_{i}\neq q_{j}\} 
\]
consisting of zero eigenvectors of each of the operators 
\begin{equation}
|\psi_{s,t}^{q}\rangle\langle\psi_{s,t}^{q}|\otimes|\psi_{u,v}^{r}\rangle\langle\psi_{u,v}^{r}|\otimes\id^{\otimes n-2},\qquad\{q,r\}\in E(\gxoc),\quad s,t,u,v\in\{0,1\}.\label{eq:annihilate_ops}
\end{equation}
Now using equation \eq{legal_states} and the fact that 
\begin{equation}
  \langle\psi_{x,b}^{\tilde{L}}|\rho_{z,a}^{L}\rangle
  =\frac{1}{\sqrt{8}}\delta_{\tilde{L},L}\delta_{x,z}\delta_{a,b}
  =\frac{1}{\sqrt{8}}\langle\rho_{x,b}^{\tilde{L}}|\rho_{z,a}^{L}\rangle
\label{eq:psi_rho_eqn}
\end{equation}
for all $x,z,a,b\in\{0,1\}$ and $\tilde{L},L\in\mathcal{L}$ (from equations \eq{rho1_1}, \eq{rho2_1}, and \eq{rho_bnd}), we get
\begin{align*}
&\langle j,\vec{d},\vec{z},\vec{a}|\left(|\psi_{s,t}^{q}\rangle\langle\psi_{s,t}^{q}|\otimes|\psi_{u,v}^{r}\rangle\langle\psi_{u,v}^{r}|\otimes\id^{\otimes n-2}\right)|j,\vec{d},\vec{z},\vec{a}\rangle\\
&\quad=\frac{1}{64}\langle j,\vec{d},\vec{z},\vec{a}|\left(|\rho_{s,t}^{q}\rangle\langle\rho_{s,t}^{q}|\otimes|\rho_{u,v}^{r}\rangle\langle\rho_{u,v}^{r}|\otimes\id^{\otimes n-2}\right)|j,\vec{d},\vec{z},\vec{a}\rangle\\
&\quad=0\quad\text{if $\{q,r\}\in E(\gxoc)$}.
\end{align*}
In the last line we used equations \eq{T_state} and \eq{legal_states} and the definition of the occupancy constraints graph $\gxoc$ from \sec{The-occupancy-constraints}. Hence each legal state $|j,\vec{d},\vec{z},\vec{a}\rangle\in\mathcal{I}(G_{1},n)$ is a zero eigenvector of each of the operators \eq{annihilate_ops}, so $|j,\vec{d},\vec{z},\vec{a}\rangle\in\mathcal{I}(G_{1},\gxoc,n)$. This gives equation \eq{subspace_eqn0}.

Now we prove equation \eq{subspace_eqn1}. For each illegal configuration we associate two diagram elements $Q_{1}$ and $Q_{2}$ with $(Q_{1},Q_{2})\in E(\gxoc$) as in \lem{legalconfig} (if there is more than one such pair we fix a specific choice). Likewise for each basis vector $|\phi\rangle\in\mathcal{B}_{\illegal}$ we associate the two diagram elements $Q_{1}$ and $Q_{2}$ corresponding to its (illegal) configuration. Let $P_{\phi}$ be the projector onto the space
\[
\spn\{ \Sym(|\psi_{z_{1},a_{1}}^{Q_{1}}\rangle|\psi_{z_{2},a_{2}}^{Q_{2}}\rangle|\psi_{z_{3},a_{3}}^{q_{3}}\rangle\ldots|\psi_{z_{n},a_{n}}^{q_{n}}\rangle)\colon z_{i},a_{i}\in\{0,1\},\, q_{i}\notin\{Q_{1},Q_{2}\}\} 
\]
where (exactly) one particle is located at $Q_{1}$ and (exactly) one particle is located at $Q_{2}$. We show that 
\begin{equation}
\langle\phi|P_{\phi}|\phi\rangle\geq\frac{1}{256}.\label{eq:q1_q2_eqn}
\end{equation}
Note that $\Pi_{0}P_{\phi}=0$ since $\Pi_0$ projects onto a subspace for which no two (or more) particles are simultaneously located at $Q_1$ and $Q_2$. Therefore 
\[
\langle\phi|\Pi_{0}|\phi\rangle+\langle\phi|P_{\phi}|\phi\rangle\leq1,
\]
and applying \eq{q1_q2_eqn} gives \eq{subspace_eqn0}. Equation \eq{q1_q2_eqn} can be shown by considering cases (i), (ii), and (iii) from \lem{legalconfig}. It is convenient to define 
\begin{align*}
\Pi_{Q_{1}} &= \sum_{x,y\in\{0,1\}}|\psi_{x,y}^{Q_{1}}\rangle\langle\psi_{x,y}^{Q_{1}}| &
\Pi_{Q_{2}} &= \sum_{x,y\in\{0,1\}}|\psi_{x,y}^{Q_{2}}\rangle\langle\psi_{x,y}^{Q_{2}}|.
\end{align*}

In case (i) we have $Q_{1}=(1,J_{k},0)$ and $Q_{2}=(1,J_{l},0)$ for some $k,l \in [Y]$.  Here we consider the case $k=1,l=2$ without loss of generality. Then
\begin{align}
P_{\phi}|\phi\rangle & =\Sym\big(P_{\phi}(|T_{z_{1},a_{1},z_{2},a_{2}}^{J_{1}}\rangle|T_{z_{3},a_{3},z_{4},a_{4}}^{J_{2}}\rangle\ldots|T_{z_{2Y-1},a_{2Y-1},z_{2Y},a_{2Y}}^{J_{Y}}\rangle|\rho_{z_{2Y+1},a_{2Y+1}}^{L_{1}}\rangle\ldots|\rho_{z_{n},a_{n}}^{L_{n-2Y}}\rangle)\big)\nonumber \\
 & =\Sym\big((\Pi_{Q_{1}}\otimes\id)|T_{z_{1},a_{1},z_{2},a_{2}}^{J_{1}}\rangle(\Pi_{Q_{2}}\otimes\id)|T_{z_{3},a_{3},z_{4},a_{4}}^{J_{2}}\rangle\ldots|T_{z_{2Y-1},a_{2Y-1},z_{2Y},a_{2Y}}^{J_{Y}}\rangle|\rho_{z_{2Y+1},a_{2Y+1}}^{L_{1}}\rangle\ldots|\rho_{z_{n},a_{n}}^{L_{n-2Y}}\rangle\big)\label{eq:P_phi_phi}
\end{align}
for some configuration and some $\vec z,\vec a$, where in the first line we used the fact that $P_{\phi}$ commutes with any permutation of the $n$ registers and in the second line we used the fact that
\begin{align*}
\Pi_{Q_{1}}^{(w)}|T_{z_{1},a_{1},z_{2},a_{2}}^{J_{1}}\rangle|T_{z_{3},a_{3},z_{4},a_{4}}^{J_{2}}\rangle\ldots|T_{z_{2Y-1},a_{2Y-1},z_{2Y},a_{2Y}}^{J_{Y}}\rangle|\rho_{z_{2Y+1},a_{2Y+1}}^{L_{1}}\rangle\ldots|\rho_{z_{n},a_{n}}^{L_{n-2Y}}\rangle & =0~\text{unless }w=1\\
\Pi_{Q_{2}}^{(w)}|T_{z_{1},a_{1},z_{2},a_{2}}^{J_{1}}\rangle|T_{z_{3},a_{3},z_{4},a_{4}}^{J_{2}}\rangle\ldots|T_{z_{2Y-1},a_{2Y-1},z_{2Y},a_{2Y}}^{J_{Y}}\rangle|\rho_{z_{2Y+1},a_{2Y+1}}^{L_{1}}\rangle\ldots|\rho_{z_{n},a_{n}}^{L_{n-2Y}}\rangle & =0~\text{unless }w=3.
\end{align*}
We find 
\begin{align}
\langle\phi|P_{\phi}|\phi\rangle & =\langle T_{z_{1},a_{1},z_{2},a_{2}}^{J_{1}}|\left(\Pi_{Q_{1}}\otimes\id\right)|T_{z_{1},a_{1},z_{2},a_{2}}^{J_{1}}\rangle\cdot\langle T_{z_{3},a_{3},z_{4},a_{4}}^{J_{2}}|\left(\Pi_{Q_{2}}\otimes\id\right)|T_{z_{3},a_{3},z_{4},a_{4}}^{J_{2}}\rangle\nonumber \\
 & =\left(\frac{1}{2}\langle\rho_{z_{1},a_{1}}^{(1,J_{1},0)}|\langle\rho_{z_{2},a_{2}}^{(s(J_{1}),J_{1},0)}|\Pi_{Q_{1}}\otimes\id|\rho_{z_{1},a_{1}}^{(1,J_{1},0)}\rangle|\rho_{z_{2},a_{2}}^{(s(J_{1}),J_{1},0)}\rangle\right)^{2} \nonumber \\
 & =\left(\frac{1}{16}\right)^{2}=\frac{1}{256}.\label{eq:lowerboundcase_a}
\end{align}
where in the second line we used the fact that both terms in the product are equal and in the third line we used equation \eq{psi_rho_eqn}. 

In case (ii) we have $Y\in\{0,1\}$ and $Q_{1}=L_{s}$, $Q_{2}=L_{t}$ for some $s,t\in[n-2Y]$. By a similar argument as in \eq{P_phi_phi},
\begin{align}
\langle\phi|P_{\phi}|\phi\rangle & =\langle\rho_{z_{s}a_{s}}^{L_{s}}|\Pi_{Q_{1}}|\rho_{z_{s},a_{s}}^{L_{s}}\rangle\cdot\langle\rho_{z_{t}a_{t}}^{L_{t}}|\Pi_{Q_{2}}|\rho_{z_{t},a_{t}}^{L_{t}}\rangle
=\frac{1}{8}\cdot\frac{1}{8}
=\frac{1}{64}.\label{eq:lowerboundcase_b}
\end{align}

In case (iii) we have $Y=1$, $Q_{1}=(i,J_{1},d)$, and $Q_{2}=L_{t}$ for some $i\in\{1,s(J_{1})\}$, $t\in[n-2]$, and $d\in\{0,1\}$. If $i=1$ then, again by a similar reasoning as in \eq{P_phi_phi}, 
\begin{align}
\langle\phi|P_{\phi}|\phi\rangle & =\langle T_{z_{1},a_{1},z_{2},a_{2}}^{J_{1}}|\Pi_{Q_{1}}\otimes\id|T_{z_{1},a_{1},z_{2},a_{2}}^{J_{1}}\rangle\cdot\langle\rho_{z_{2Y+t}a_{2Y+t}}^{L_{t}}|\Pi_{Q_{2}}|\rho_{z_{2Y+t}a_{2Y+t}}^{L_{t}}\rangle\label{eq:lowerboundcase_c}\\
 & =\frac{1}{16}\cdot\frac{1}{8}=\frac{1}{128}.\label{eq:lowerboundcase_c_2}
\end{align}
If $i=s(J_{1})$ then $\Pi_{Q_{1}}\otimes\id$ should be replaced with $\id\otimes\Pi_{Q_{1}}$ in \eq{lowerboundcase_c} but the lower bound in \eq{lowerboundcase_c_2} is the same.

From equations \eq{lowerboundcase_a}, \eq{lowerboundcase_b}, and \eq{lowerboundcase_c_2}, we see that equation \eq{q1_q2_eqn} holds in cases (i), (ii), and (iii), respectively, thereby establishing \eq{subspace_eqn1}.

Finally, we prove equation \eq{subspace_eqn2}, showing that $\Pi_{0}|_{\spn(\mathcal{B}_{n})}$ is diagonal in the basis $\mathcal{B}_n$. Let 
\begin{align}
|\phi\rangle & =\Sym(|T_{z_{1},a_{1},z_{2},a_{2}}^{J_{1}}\rangle\ldots|T_{z_{2Y-1},a_{2Y}-1,z_{2Y},a_{2Y}}^{J_{Y}}\rangle|\rho_{z_{2Y+1},a_{2Y+1}}^{L_{1}}\rangle\ldots|\rho_{z_{n},a_{n}}^{L_{n-2Y}}\rangle)\label{eq:phi_illegal}\\
|\psi\rangle & =\Sym(|T_{x_{1},b_{1},x_{2},b_{2}}^{\tilde{J}_{1}}\rangle\ldots|T_{x_{2K-1},b_{2K}-1,x_{2K},b_{2K}}^{\tilde{J}_{K}}\rangle|\rho_{x_{2K+1},b_{2K+1}}^{\tilde{L}_{1}}\rangle\ldots|\rho_{x_{n},b_{n}}^{\tilde{L}_{n-2K}}\rangle)
\end{align}
be distinct vectors from $\mathcal{B}_{n}$ (note it is possible that $K=0$ or $Y=0$ or both). Expand each of the $|T\rangle$ states using equation \eq{T_state},
which we can also write as
\[
|T_{z,a,y,b}^{J}\rangle=\frac{1}{\sqrt{2}}\sum_{c=0}^{1}U_{J}(a)^{c}|\rho_{z,a}^{(1,J,c)}\rangle|\rho_{y,b}^{(s(J),J,c)}\rangle
\]
where we use the shorthand
\begin{equation}
U_{J}(a)|\rho_{z,a}^{(1,J,1)}\rangle|\rho_{y,b}^{(s(J),J,1)}\rangle=\sum_{x_{1},x_{2}\in\{0,1\}}U_{J}(a)_{x_{1}x_{2},zy}|\rho_{x_{1},a}^{(1,J,1)}\rangle|\rho_{x_{2},b}^{(s(J),J,1)}\rangle.\label{eq:Vtothec}
\end{equation}
For the state $|\phi\rangle$, this gives the expansion
\[
|\phi\rangle=\left(\frac{1}{\sqrt{2}}\right)^{Y}\!\!\sum_{c_{1},\ldots,c_{Y}\in\{0,1\}}\Sym(|O_{\vec{z},\vec{a}}^{(J_{1},\ldots,J_{Y},L_{1},\ldots,L_{n-2Y}),(c_{1},\ldots,c_{Y})}\rangle)
\]
where 
\begin{equation}
|O_{\vec{z},\vec{a}}^{(J_{1},\ldots,J_{Y},L_{1},\ldots,L_{n-2Y}),(c_{1},\ldots,c_{Y})}\rangle=\left(\bigotimes_{i=1}^{Y}U_{J_{i}}(a_{2i-1})^{c_{i}}|\rho_{z_{2i-1},a_{2i-1}}^{(1,J_{i},c_{i})}\rangle|\rho_{z_{2i},a_{2i}}^{(s(J_{i}),J_{i},c_{i})}\rangle\right)|\rho_{z_{2Y+1},a_{2Y+1}}^{L_{1}}\rangle\ldots|\rho_{z_{n},a_{n}}^{L_{n-2Y}}\rangle.\label{eq:O_state}
\end{equation}

Define the projector 
\[
P_{\mathcal{L}}^{1}=\sum_{z,a\in\{0,1\}}\sum_{L\in\mathcal{L}}|\psi_{z,a}^{L}\rangle\langle\psi_{z,a}^{L}|
\]
which has support only on diagram elements contained in $\mathcal{L}$, and let $P_{\mathcal{L}}^{0}=\id-P_{\mathcal{L}}^{1}$.  Note that for each $L\in\mathcal{L}$ and $z,a\in\{0,1\}$, we can write 
\begin{equation}
|\rho_{z,a}^{L}\rangle=P_{\mathcal{L}}^{1}|\rho_{z,a}^{L}\rangle+P_{\mathcal{L}}^{0}|\rho_{z,a}^{L}\rangle\label{eq:rho_expand}
\end{equation}
where (from equation \eq{psi_rho_eqn})
\[
P_{\mathcal{L}}^{1}|\rho_{z,a}^{L}\rangle=\frac{1}{\sqrt{8}}|\psi_{z,a}^{L}\rangle.
\]
Since the states $|\rho_{z,a}^{L}\rangle$ are orthonormal, and similarly for the states $|\psi_{z,a}^{L}\rangle$, we get 
\begin{equation}
\langle\rho_{x,b}^{\tilde{L}}|P_{\mathcal{L}}^{\alpha}|\rho_{z,a}^{L}\rangle=\begin{cases}
\frac{1}{8}\delta_{z,x}\delta_{a,b}\delta_{L,\tilde{L}} & \alpha=1\\
\frac{7}{8}\delta_{z,x}\delta_{a,b}\delta_{L,\tilde{L}} & \alpha=0.
\end{cases}\label{eq:projector_L_rho}
\end{equation}
Inserting $n$ copies of the identity $P_{\mathcal{L}}^{1}+P_{\mathcal{L}}^{0}=1$ gives
\begin{equation}
|\phi\rangle=\left(\frac{1}{\sqrt{2}}\right)^{Y}\!\!\sum_{c_{1},\ldots,c_{Y}\in\{0,1\}}\sum_{\alpha_{1},\ldots,\alpha_{n}\in\{0,1\}}\Sym(P_{\mathcal{L}}^{\alpha_{1}}\otimes\cdots\otimes P_{\mathcal{L}}^{\alpha_{n}}|O_{\vec{z},\vec{a}}^{(J_{1},\ldots,J_{Y},L_{1},\ldots,L_{n-2Y}),(c_{1},\ldots,c_{Y})}\rangle)\label{eq:phi_sum}
\end{equation}
Likewise for $|\psi\rangle$ we get 
\begin{equation}
|\psi\rangle=\left(\frac{1}{\sqrt{2}}\right)^{K}\!\!\sum_{e_{1},\ldots,e_{K}\in\{0,1\}}\sum_{\beta_{1},\ldots,\beta_{n}\in\{0,1\}}\Sym(P_{\mathcal{L}}^{\beta_{1}}\otimes\cdots\otimes P_{\mathcal{L}}^{\beta_{n}}|O_{\vec{x},\vec{b}}^{(\tilde{J}_{1},\ldots,\tilde{J}_{K},\tilde{L}_{1},\ldots,\tilde{L}_{n-2K}),(e_{1},\ldots,e_{K})}\rangle).\label{eq:psi_sum}
\end{equation}

Using equations \eq{psi_rho_eqn}, \eq{O_state}, and \eq{Vtothec}, we see that the states 
\[
|O_{\vec{z},\vec{a}}^{(J_{1},\ldots,J_{Y},L_{1},\ldots,L_{n-2Y}),(c_{1},\ldots,c_{Y})}\rangle \qquad \text{and} \qquad | O_{\vec{x},\vec{b}}^{(\tilde{J}_{1},\ldots,\tilde{J}_{K},\tilde{L}_{1},\ldots,\tilde{L}_{n-2K}),(e_{1},\ldots,e_{K})}\rangle
\]
are orthogonal for any choice of bit strings $c_{1},\ldots,c_{Y}$ and $e_{1},\ldots,e_{K}$, since $|\phi\rangle \ne |\psi\rangle$ implies that
\[
((J_{1},\ldots,J_{Y},L_{1},\ldots,L_{n-2Y}),\vec{z},\vec{a})\neq ((\tilde{J}_{1},\ldots,\tilde{J}_{K},\tilde{L}_{1},\ldots,\tilde{L}_{n-2K}),\vec{x},\vec{b}).
\]
 
Using equation \eq{projector_L_rho}, we have
\begin{align*}
& \langle O_{\vec{x},\vec{b}}^{(\tilde{J}_{1},\ldots,\tilde{J}_{K},\tilde{L}_{1},\ldots,\tilde{L}_{n-2K}),(e_{1},\ldots,e_{K})}|P_{\mathcal{L}}^{\alpha_{1}}\otimes \cdots\otimes P_{\mathcal{L}}^{\alpha_{n}}|O_{\vec{z},\vec{a}}^{(J_{1},\ldots,J_{Y},L_{1},\ldots,L_{n-2Y}),(c_{1},\ldots,c_{Y})}\rangle \\
&\quad= \left(\frac{1}{8}\right)^{\sum_{i=1}^{n}\alpha_i}\left(\frac{7}{8}\right)^{n-\sum_{i=1}^{n}\alpha_i}\langle O_{\vec{x},\vec{b}}^{(\tilde{J}_{1},\ldots,\tilde{J}_{K},\tilde{L}_{1},\ldots,\tilde{L}_{n-2K}),(e_{1},\ldots,e_{K})}|O_{\vec{z},\vec{a}}^{(J_{1},\ldots,J_{Y},L_{1},\ldots,L_{n-2Y}),(c_{1},\ldots,c_{Y})}\rangle,
\end{align*}
so the states
\begin{equation}
P_{\mathcal{L}}^{\alpha_{1}}\otimes \cdots\otimes P_{\mathcal{L}}^{\alpha_{n}}|O_{\vec{z},\vec{a}}^{(J_{1},\ldots,J_{Y},L_{1},\ldots,L_{n-2Y}),(c_{1},\ldots,c_{Y})}\rangle\label{eq:intermediate_O_states}
\end{equation}
and 
\begin{equation}
P_{\mathcal{L}}^{\beta_{1}}\otimes \cdots\otimes P_{\mathcal{L}}^{\beta_{n}}| O_{\vec{x},\vec{b}}^{(\tilde{J}_{1},\ldots,\tilde{J}_{K},\tilde{L}_{1},\ldots,\tilde{L}_{n-2K}),(e_{1},\ldots,e_{K})}\rangle
\end{equation}
are orthogonal for each choice of bit strings $\alpha_1,\ldots,\alpha_n$, $\beta_1,\ldots,\beta_n$, $c_{1},\ldots,c_{Y}$, and $e_{1},\ldots,e_{K}$. Furthermore, observe that
\begin{align*}
& \Sym(\langle O_{\vec{x},\vec{b}}^{(\tilde{J}_{1},\ldots,\tilde{J}_{K},\tilde{L}_{1},\ldots,\tilde{L}_{n-2K}),(e_{1},\ldots,e_{K})}|P_{\mathcal{L}}^{\beta_{1}}\!\otimes\! \cdots\!\otimes\! P_{\mathcal{L}}^{\beta_{n}})
\Sym(P_{\mathcal{L}}^{\alpha_{1}}\!\otimes\! \cdots\!\otimes\! P_{\mathcal{L}}^{\alpha_{n}}|O_{\vec{z},\vec{a}}^{(J_{1},\ldots,J_{Y},L_{1},\ldots,L_{n-2Y}),(c_{1},\ldots,c_{Y})}\rangle)\\
&\quad=(\langle O_{\vec{x},\vec{b}}^{(\tilde{J}_{1},\ldots,\tilde{J}_{K},\tilde{L}_{1},\ldots,\tilde{L}_{n-2K}),(e_{1},\ldots,e_{K})}|P_{\mathcal{L}}^{\beta_{1}}\otimes \cdots\otimes P_{\mathcal{L}}^{\beta_{n}})(P_{\mathcal{L}}^{\alpha_{1}}\otimes \cdots\otimes P_{\mathcal{L}}^{\alpha_{n}}|O_{\vec{z},\vec{a}}^{(J_{1},\ldots,J_{Y},L_{1},\ldots,L_{n-2Y}),(c_{1},\ldots,c_{Y})}\rangle).
\end{align*}
Thus each symmetrized state in the sum \eq{phi_sum} is orthogonal to each symmetrized state in the sum \eq{psi_sum}.

To complete the proof, we show that 
\[
\Pi_{0}|\phi\rangle
\]
is a superposition of a \emph{subset} of the states in the sum \eq{phi_sum} and hence is orthogonal to $|\psi\rangle$. To see this, first note that $|\phi\rangle\in \mathcal{I}(G_1,n)$ by \lem{FF_characterization} since it is in the nullspace of $H(G_1,n)$ (by \lem{gs_g_alpha}) and $G_1$ is an $e_1$-gate graph (by \lem{The-smallest-eigenvalues}). Now comparing $\mathcal{I}(G_1,n)$ (defined in \eq{Ign}) and $\mathcal{I}(G_1,\gxoc,n)$ (defined in \eq{occup_space_defn}), we see that 
\[
\Pi_0 |\Gamma\rangle=\Pi_0^{\occ} |\Gamma\rangle \text{ for all }|\Gamma\rangle\in \mathcal{I}(G_1,n)
\]
where $\Pi_0^\occ$ projects onto the space
\begin{equation}
\spn\{|\psi_{z_1,a_1}^{q_1}\rangle |\psi_{z_2,a_2}^{q_2}\rangle\ldots|\psi_{z_n,a_n}^{q_n}\rangle\colon z_i,a_i \in\{0,1\},\,	q_i\in [R],\, \{q_i,q_j\}\notin E(\gxoc)\}.\label{eq:hatpi}
\end{equation}
In particular, $\Pi_0 |\phi\rangle=\Pi_0^\occ |\phi\rangle$. We claim that this quantity is a superposition of a subset of the states in the sum \eq{phi_sum}.

The diagram elements $q_1,\ldots,q_n$ appearing in \eq{hatpi} range over the set of all $R$ diagram elements in the gate graph $G_1$; however, recall that $E(\gxoc)$ only contains edges between diagram elements in the subset $\mathcal{L}$ of these diagram elements. Since the $P_{\mathcal{L}}^\alpha$  either project onto this set of diagram elements or onto the complement, each state 
\[
\Sym(P_{\mathcal{L}}^{\alpha_{1}}\otimes \cdots\otimes P_{\mathcal{L}}^{\alpha_{n}}|O_{\vec{z},\vec{a}}^{(J_{1},\ldots,J_{Y},L_{1},\ldots,L_{n-2Y}),(c_{1},\ldots,c_{Y})}\rangle)
\]
is an eigenvector of $\Pi_{0}^{\occ}$. Hence  $\Pi_0 |\phi\rangle$  is a superposition of the terms in \eq{phi_sum} that are $+1$ eigenvectors of  $\Pi_0^{\occ}$ (as the 0 eigenvectors are annihilated). It follows that $\langle \psi |\Pi_0 |\phi\rangle=0$ since we established above that each such term is orthogonal to $|\psi\rangle$.
\end{proof}

\subsection{Computation of matrix elements between states with legal configurations}
\label{sec:matrix_els_details}

In this Section we compute the matrix elements of
\begin{equation}
  H_{1}\big|_{S_{1}},
  H_{2}\big|_{S_{1}},
  H_{\inn,i}\big|_{S_{1}},
  H_{\out}\big|_{S_{1}}
\label{eq:ops_restriction_S1-1}
\end{equation}
in the basis of $S_{1}$ consisting of the states 
\begin{equation}
  |j,\vec{d},\In(\vec{z}),\vec{a}\rangle
  =\sum_{\vec{x}\in\{0,1\}^{n}} \big(\langle\vec{x} 
   |\bar{U}_{j,d_1}(a_{1}) 
   |\vec{z}\rangle\big)|j,\vec{d},\vec{x},\vec{a}\rangle 
  \label{eq:phi_init_basis_restated}
\end{equation}
for $\vec{z},\vec{a}\in\{0,1\}$, $j\in[M]$, and 
\[
\vec{d}=(d_{1},\ldots,d_{n})\quad\text{with}\quad d_{1}=d_{s(j)}\in\{0,1,2\}\quad\text{and}\quad d_{i}\in\{0,1\}\; i\notin\{1,s(j)\},
\]where
\begin{equation}
  \bar{U}_{j,d_1}(a_1) = \begin{cases}
    U_{j-1}(a_1) U_{j-2}(a_1) \ldots U_1(a_1) & \text{if $d_1 \in \{0,2\}$} \\
    U_j(a_1) U_{j-1}(a_1) \ldots U_1(a_1) & \text{if $d_1=1$}.
  \end{cases}
  \label{eq:ubar_restated}
\end{equation}
Specifically, we prove the results stated in the four boxes in \sec{Matrix-elements-in}. 

\subsubsection{Matrix elements of $H_{1}$}

We begin by computing the matrix elements of 
\[
H_{1}=\sum_{w=1}^{n}h_{1}^{(w)}
\]
in the basis $\mathcal{B}_{\legal}$; then we use them to compute the matrix elements of $H_1$ in the basis \eq{phi_init_basis_restated}.

Note that since $H_{1}$ is symmetric under permutations of the $n$ registers, 
\begin{equation}
H_{1}|j,\vec{d},\vec{z},\vec{a}\rangle=\begin{cases}
\Sym\Big(H_{1}|\rho_{z_{1},a_{1}}^{(1,j,d_{1})}\rangle\underset{i=2}{\overset{n}{\bigotimes}}|\rho_{z_{i},a_{i}}^{F(i,j,d_{i})}\rangle\Big) & d_{1}=d_{s(j)}\in\{0,1\}\\
\Sym\Bigg(H_{1}|T_{z_{1},a_{1},z_{s(j)},a_{s(j)}}^{j}\rangle\underset{\substack{i=2\\
i\neq s(j)
}
}{\overset{n}{\bigotimes}}|\rho_{z_{i},a_{i}}^{F(i,j,d_{i})}\rangle\Bigg) & d_{1}=d_{s(j)}=2
\end{cases}
\label{eq:H1_acting}
\end{equation}
and recall that
\begin{equation}
|T_{z_{1},a_{1},z_{s(j)},a_{s(j)}}^{j}\rangle=\frac{1}{\sqrt{2}}|\rho_{z_{1},a_{1}}^{(1,j,0)}\rangle|\rho_{z_{s(j)},a_{s(j)}}^{(s(j),j,0)}\rangle+\frac{1}{\sqrt{2}}\sum_{x_{1},x_{2}\in\{0,1\}}U_{j}(a_{1})_{x_{1}x_{2},z_{1}z_{s(j)}}|\rho_{x_{1},a_{1}}^{(1,j,1)}\rangle|\rho_{x_{2},a_{s(j)}}^{(s(j),j,1)}\rangle.\label{eq:reminder_T_states}
\end{equation}
To compute $\langle k,\vec{c},\vec{x},\vec{b}|H_{1}|j,\vec{d},\vec{z},\vec{a}\rangle$, we first evaluate the matrix elements of $h_{1}$ between single-particle states of the form
\[
  |\rho_{z,a}^{(1,j,d)}\rangle,\,
  |\rho_{z,a}^{(s(j),j,d)}\rangle,\,
  |\rho_{z,a}^{F(i,j,d)}\rangle
\]
(for $j\in[M]$, $i\in\{2,\ldots,n\}$, and $z,a,d\in\{0,1\}$) that appear in equation \eq{H1_acting}. To evaluate these matrix elements, we use the fact that $h_1$ is of the form \eq{h_edges}, where $\mathcal{E}$ is the set of edges in rows $2,\ldots,n$ that are added to the gate diagram in step 3 of \sec{The-gate-graph}.

We have
\begin{align}
\langle\rho_{x,b}^{F(i,j,0)}|h_{1}|\rho_{z,a}^{F(i,j,0)}\rangle & =\frac{1}{8}\langle\psi_{x,b}^{F(i,j,0)}|h_{1}|\psi_{z,a}^{F(i,j,0)}\rangle=\frac{1}{64}\delta_{x,z}\delta_{a,b}\label{eq:diag_F0}\\
\langle\rho_{x,b}^{F(i,j,1)}|h_{1}|\rho_{z,a}^{F(i,j,1)}\rangle & =\frac{1}{64}\delta_{x,z}\delta_{a,b}\label{eq:diag_F1}
\end{align}
for all $i\in\{2,\ldots,n\}$, $j\in[M]$, and $x,z,a,b\in\{0,1\}$. Similarly,
\begin{equation}
\langle\rho_{x,b}^{F(i,j,0)}|h_{1}|\rho_{z,a}^{F(i,j,1)}\rangle=\langle\rho_{x,b}^{F(i,j,1)}|h_{1}|\rho_{z,a}^{F(i,j,0)}\rangle=\frac{1}{64}\delta_{x,z}\delta_{a,b}\label{eq:F_h1_matels}
\end{equation}
for all $i \in [n]\setminus\{1,s(j)\}$, $j\in[M]$, and $z,x,a,b\in\{0,1\}$. Furthermore,
\begin{equation}
h_{1}|\rho_{z,a}^{(1,j,d)}\rangle=0\label{eq:h1_row1}
\end{equation}
for all $j\in[M]$ and $z,a,d\in\{0,1\}$, and
\begin{align}
\langle\rho_{x,b}^{F(s(j),j,c)}|h_{1}|\rho_{z,a}^{\left(s(j),j,d\right)}\rangle & =\frac{1}{64}\delta_{x,z}\delta_{a,b}\delta_{c,d}\label{eq:F_L_eqn_h1}\\
\langle\rho_{x,b}^{(s(j),j,c)}|h_{1}|\rho_{z,a}^{\left(s(j),j,d\right)}\rangle & =\frac{1}{64}\delta_{x,z}\delta_{a,b}\delta_{c,d}\label{eq:L_L_eqn_h1}
\end{align}
for all $j\in[M]$ and $z,x,a,b,c,d\in\{0,1\}$.

Using equations \eq{diag_F0}, \eq{diag_F1}, \eq{h1_row1}, and \eq{L_L_eqn_h1}, we compute the diagonal matrix elements of $H_{1}$:
\begin{align}
\langle j,\vec{d},\vec{z},\vec{a}|H_{1}|j,\vec{d},\vec{z},\vec{a}\rangle & =\begin{cases}
\sum_{i=2}^{n}\langle\rho_{z_{i},a_{i}}^{F(i,j,d_{i})}|h_{1}|\rho_{z_{i},a_{i}}^{F(i,j,d_{i})}\rangle & d_{1}\in\{0,1\}\\
\langle T_{z_{1},a_{1},z_{s(j)},a_{s(j)}}^{j}|\id\otimes h_{1}|T_{z_{1},a_{1},z_{s(j)},a_{s(j)}}^{j}\rangle+\underset{\substack{i=2\\
i\neq s(j)
}
}{\overset{n}{\sum}}\langle\rho_{z_{i},a_{i}}^{F(i,j,d_{i})}|h_{1}|\rho_{z_{i},a_{i}}^{F(i,j,d_{i})}\rangle & d_{1}=2.
\end{cases}\nonumber \\
 & =\frac{n-1}{64} \label{eq:diag_H1}
\end{align}
where in the last line we used equation \eq{reminder_T_states} and the fact that $U_{j}(a_{1})$ is unitary. We use equations \eq{F_h1_matels} and \eq{F_L_eqn_h1} to compute the nonzero off-diagonal matrix elements of $H_{1}$ between states in $\mathcal{B}_{\legal}$. We get
\begin{equation}
\langle k,\vec{c},\vec{x},\vec{b}|H_{1}|j,\vec{d},\vec{z},\vec{a}\rangle=\delta_{j,k}\delta_{\vec{a},\vec{b}}\cdot\begin{cases}
\frac{1}{64}\delta_{\vec{x},\vec{z}}(n-1) & \vec{c}=\vec{d}\\
\frac{1}{64}\delta_{\vec{x},\vec{z}}\underset{\substack{r=1\\
r\neq i
}
}{\overset{n}{\prod}}\delta_{c_{r},d_{r}} & c_{i} \neq d_{i} \text{ for some } i\in[n]\setminus\{1,s(j)\} \\
\frac{1}{64\sqrt{2}}\delta_{\vec{x},\vec{z}}\underset{\substack{r=2\\
r\neq s(j)
}
}{\overset{n}{\prod}}\delta_{c_{r},d_{r}} & (c_{1},d_{1})\in\{(2,0),(0,2)\}\\
\frac{1}{64\sqrt{2}} U_{j}(a_{1})^{*}_{z_{1}z_{s(j)},x_{1}x_{s(j)}} \underset{\substack{r=2\\
r\neq s(j)
}
}{\overset{n}{\prod}}\delta_{c_{r},d_{r}}\delta_{x_r,z_r} & (c_{1},d_{1})=(2,1)\\
\frac{1}{64\sqrt{2}} U_{j}(a_{1})_{x_{1}x_{s(j)},z_{1}z_{s(j)}} \underset{\substack{r=2\\
r\neq s(j)
}
}{\overset{n}{\prod}}\delta_{c_{r},d_{r}}\delta_{x_r,z_r} & (c_{1},d_{1})=(1,2)\\
0 & \text{otherwise.}
\end{cases}\label{eq:off_diag_H1}
\end{equation}
For the second case we used equation \eq{F_h1_matels}, for the third case we used equation \eq{F_L_eqn_h1} to get
\[
\langle T_{x_{1},b_{1},x_{s(j)},b_{s(j)}}^{j}|\id\otimes h_{1}|\rho_{z_{1},a_{1}}^{(1,j,0)}\rangle|\rho_{z_{s(j)}a_{s(j)}}^{F(s(j),j,0)}\rangle=\delta_{x_1,z_1}\delta_{b_1,a_1}\delta_{x_{s(j)},z_{s(j)}}\delta_{b_{s(j)},a_{s(j)}}\frac{1}{64\sqrt{2}},
\]
and for the fourth and fifth cases we used equation \eq{F_L_eqn_h1} to get
\[
\langle T_{x_{1},b_{1},x_{s(j)},b_{s(j)}}^{j}|\id\otimes h_{1}|\rho_{z_{1},a_{1}}^{(1,j,1)}\rangle|\rho_{z_{s(j)}a_{s(j)}}^{F(s(j),j,1)}\rangle=\delta_{b_1,a_1}\delta_{b_{s(j)},a_{s(j)}}\frac{1}{64\sqrt{2}} U_{j}(a_{1})^{*}_{z_{1}z_{s(j)},x_{1}x_{s(j)}}.
\]
In the remaining case, $(c_1,d_1) \in \{(1,0),(0,1)\}$ and the matrix element is $0$.

We now compute the matrix elements of $H_1$ in the basis \eq{phi_init_basis_restated}. We have
\begin{equation}
\langle j,\vec{c},\In(\vec{x}),\vec{a}|H_1|j,\vec{d},\In(\vec{z}),\vec{a}\rangle
= \sum_{\vec{x}',\vec{z}' \in \{0,1\}^n}
  \langle j,\vec{c},\vec{x}',\vec{a}|H_1|j,\vec{d},\vec{z}',\vec{a}\rangle
  \langle \vec{x}|\bar{U}_{j,d_1}(a_1)^\dag|\vec{x}'\rangle
  \langle \vec{z}'|\bar{U}_{j,d_1}(a_1)|\vec{z}\rangle.
\label{eq:basischange}
\end{equation}
Using this with \eq{off_diag_H1} gives
\[
\langle k,\vec{c},\In(\vec{x}),\vec{b}|H_{1}|j,\vec{d},\In(\vec{z}),\vec{a}\rangle=\delta_{j,k}\delta_{\vec{a},\vec{b}}\delta_{\vec{x},\vec{z}}\cdot\begin{cases}
\frac{n-1}{64} & \vec{c}=\vec{d}\\
\frac{1}{64}\underset{\substack{r=1\\
r\neq i
}
}{\overset{n}{\prod}}\delta_{c_{r},d_{r}} & c_{i}\neq d_{i}\;\text{for some}\; i\in[n]\setminus\{1,s(j)\}\\
\frac{1}{64\sqrt{2}}\underset{\substack{r=2\\
r\neq s(j)
}
}{\overset{n}{\prod}}\delta_{c_{r},d_{r}} & (c_{1},d_{1})\in\{(2,0),(0,2),(2,1),(1,2)\}\\
0 & \text{otherwise}
\end{cases}
\]
as claimed in equation \eq{H1_goodbasis-1}. Note that in the basis $\mathcal{B}_{\legal}$, $H_{1}$ has nonzero matrix elements between states with different values of $\vec{z}$; the basis \eq{phi_init_basis_restated} is convenient because $H_{1}$ only connects basis states with the same value of $\vec{z}$.

\subsubsection{Matrix elements of $H_{2}$}

Recall that
\[
H_{2}=\sum_{w=1}^{n}h_{2}^{(w)}
\]
and note, just as in \eq{H1_acting}, that because $H_2$ is permutation invariant, 
\begin{equation}
H_{2}|j,\vec{d},\vec{z},\vec{a}\rangle=\begin{cases}
\Sym\Big(H_{2}|\rho_{z_{1},a_{1}}^{(1,j,d_{1})}\rangle\underset{i=2}{\overset{n}{\bigotimes}}|\rho_{z_{i},a_{i}}^{F(i,j,d_{i})}\rangle\Big) & d_{1}=d_{s(j)}\in\{0,1\}\\
\Sym\Bigg(H_{2}|T_{z_{1},a_{1},z_{s(j)},a_{s(j)}}^{j}\rangle\underset{\substack{i=2\\
i\neq s(j)
}
}{\overset{n}{\bigotimes}}|\rho_{z_{i},a_{i}}^{F(i,j,d_{i})}\rangle\Bigg) & d_{1}=d_{s(j)}=2.
\end{cases}
\label{eq:actH2}
\end{equation}
Also recall that $h_2$ is of the form \eq{h_edges}, where $\mathcal{E}$ is the set of edges in row 1 that are added in step 3 of \sec{The-gate-graph}.

To compute $\langle k,\vec{c},\vec{x},\vec{b}|H_{2}|j,\vec{d},\vec{z},\vec{a}\rangle$ we first evaluate the matrix elements of $h_{2}$ between the relevant single-particle states $|\rho_{z,a}^{L}\rangle$ with $L\in\mathcal{L}$ and $z,a\in\{0,1\}$. The only such matrix elements that are nonzero are
\begin{equation}
\langle\rho_{x,b}^{(1,j,0)}|h_{2}|\rho_{z,a}^{(1,j,0)}\rangle=\begin{cases}
0 & j=1\\
\frac{1}{64}\delta_{z,x}\delta_{a,b} & j\in\{2,\ldots,M\}
\end{cases}\qquad\langle\rho_{x,b}^{(1,j,1)}|h_{2}|\rho_{z,a}^{(1,j,1)}\rangle=\begin{cases}
\frac{1}{64}\delta_{z,x}\delta_{a,b} & j\in\{1,\ldots,M-1\}\\
0 & j=M
\end{cases}\label{eq:diag_h2}
\end{equation}
for $z,a,x,b\in\{0,1\}$ and
\begin{equation}
\langle\rho_{x,b}^{(1,j-1,1)}|h_{2}|\rho_{z,a}^{(1,j,0)}\rangle=\langle\rho_{z,a}^{(1,j,0)}|h_{2}|\rho_{x,b}^{(1,j-1,1)}\rangle=\frac{1}{64}\delta_{z,x}\delta_{a,b}\label{eq:single_part_offdiag}
\end{equation}
for $j\in\{2,\ldots,M\}$ and $z,x,a,b\in\{0,1\}$.

Using equations \eq{actH2} and \eq{diag_h2}, we compute the diagonal matrix elements of $H_{2}$ in the basis $\mathcal{B}_{\legal}$:
\begin{equation}
\langle j,\vec{d},\vec{z},\vec{a}|H_{2}|j,\vec{d},\vec{z},\vec{a}\rangle=\begin{cases}
0 & d_{1}=0\text{ and }j=1,\text{ or }d_{1}=1\text{ and }j=M\\
\frac{1}{128} & d_{1}=2\text{ and }j\in\{1,M\}\\
\frac{1}{64} & \text{otherwise.}
\end{cases}\label{eq:diag_H2_legalbasis}
\end{equation}
Using equations \eq{actH2} and \eq{single_part_offdiag}, we compute the nonzero off-diagonal matrix elements, which are all of the form 
\[
\langle j-1,\vec{c},\vec{x},\vec{a}|H_{2}|j,\vec{d},\vec{z},\vec{a}\rangle\quad\text{or}\quad
\langle j,\vec{d},\vec{z},\vec{a}|H_{2}|j-1,\vec{c},\vec{x},\vec{a}\rangle
=\big(\langle j-1,\vec{c},\vec{x},\vec{a}|H_{2}|j,\vec{d},\vec{z},\vec{a}\rangle\big)^{*}
\]
for $j\in\{2,\ldots,M\}$, $\vec{x},\vec{z},\vec{a}\in\{0,1\}^{n}$, and
\[
\vec{d}=(d_{1},\ldots,d_{n})\quad\text{with}\quad d_{1}=d_{s(j)}\in\{0,1,2\}\quad\text{and}\quad d_{i}\in\{0,1\} \text{ for all } i\notin\{1,s(j)\}.
\]
We get
\begin{align}
&\langle j-1,\vec{c},\vec{x},\vec{a}|H_{2}|j,\vec{d},\vec{z},\vec{a}\rangle \nonumber\\
&= \underset{\substack{r=2\\
r\notin\{s(j),s(j-1)\}
}
}{\overset{n}{\prod}}\delta_{d_{r},c_{r}} \begin{cases}
\frac{1}{64}\delta_{\vec{x},\vec{z}} & (c_{1},c_{s(j)},d_{1},d_{s(j-1)})=(1,0,0,1)\\
\frac{1}{64\sqrt{2}} U_{j-1}(a_{1})^{*}_{z_{1}z_{s(j-1)},x_{1}x_{s(j-1)}} \underset{\substack{r=2\\
r\neq s(j-1)
}
}{\overset{n}{\prod}}\delta_{x_{r},z_{r}} & (c_{1},c_{s(j)},d_{1},d_{s(j-1)})=(2,0,0,0)\\
\frac{1}{128} U_{j-1}(a_{1})^{*}_{z_{1}z_{s(j-1)},x_{1}x_{s(j-1)}} \underset{\substack{r=2\\
r\neq s(j-1)
}
}{\overset{n}{\prod}}\delta_{x_{r},z_{r}} & (c_{1},c_{s(j)},d_{1},d_{s(j-1)})=(2,1,2,0)\\
\frac{1}{64\sqrt{2}}\delta_{\vec{x},\vec{z}} & (c_{1},c_{s(j)},d_{1},d_{s(j-1)})=(1,1,2,1).
\end{cases}
\label{eq:H2_offdiag_legalbasis}
\end{align}

Now we compute the diagonal matrix elements of $H_{2}$ in the basis \eq{phi_init_basis_restated} using equations \eq{basischange} and \eq{diag_H2_legalbasis}:
\begin{equation}
\langle j,\vec{d},\In(\vec{z}),\vec{a}|H_{2}|j,\vec{d},\In(\vec{z}),\vec{a}\rangle=\begin{cases}
0 & d_{1}=0\text{ and }j=1,\text{ or }d_{1}=1\text{ and }j=M\\
\frac{1}{128} & d_{1}=2\text{ and }j\in\{1,M\}\\
\frac{1}{64} & \text{otherwise.}
\end{cases}\label{eq:H2_diag_goodbasis}
\end{equation}
The nonzero off-diagonal matrix elements are (using equations \eq{basischange} and \eq{H2_offdiag_legalbasis})
\begin{align}
&\langle j-1,\vec{c},\In(\vec{x}),\vec{b}|H_{2}|j,\vec{d},\In(\vec{z}),\vec{a}\rangle
= \langle j,\vec{d},\In(\vec{z}),\vec{a}|H_{2}|j-1,\vec{c},\In(\vec{x}),\vec{b}\rangle \nonumber\\
&\quad=\delta_{\vec{x},\vec{z}}\delta_{\vec{a},\vec{b}}\left(\underset{\substack{r=2\\
r\notin\{s(j),s(j-1)\}
}
}{\overset{n}{\prod}}\delta_{d_{r},c_{r}}\right)\cdot\begin{cases}
\frac{1}{64} & (c_{1},c_{s(j)},d_{1},d_{s(j-1)})=(1,0,0,1)\\
\frac{1}{64\sqrt{2}} & (c_{1},c_{s(j)},d_{1},d_{s(j-1)})=(2,0,0,0)\\
\frac{1}{128} & (c_{1},c_{s(j)},d_{1},d_{s(j-1)})=(2,1,2,0)\\
\frac{1}{64\sqrt{2}} & (c_{1},c_{s(j)},d_{1},d_{s(j-1)})=(1,1,2,1).
\end{cases}\label{eq:H2_offdiag_goodbasis}
\end{align}
Combining equations \eq{H2_diag_goodbasis} and \eq{H2_offdiag_goodbasis} gives the result claimed in equations \eq{H2_formula_with_f}, \eq{H2_diag_goodbasis-1}, and \eq{H2_offdiag_goodbasis-1}.

\subsubsection{Matrix elements of $H_{\inn,i}$ }

We now consider
\[
H_{\inn,i}=\sum_{w=1}^{n}h_{\inn,i}^{(w)}
\]
where $i$ is from the set of indices of the ancilla qubits, i.e., $i\in\{n_{\inn}+1,\ldots,n\}$. Using equation \eq{hin_i} we get 
\[
\langle\rho_{x,b}^{L_{2}}|h_{\inn,i}|\rho_{z,a}^{L_{1}}\rangle=\frac{1}{8}\langle\psi_{x,b}^{L_{2}}|h_{\inn,i}|\psi_{z,a}^{L_{1}}\rangle=\frac{1}{64}\delta_{x,1}\delta_{z,1}\delta_{a,b}\delta_{L_{1},(i,0,1)}\delta_{L_{2},(i,0,1)}
\]
for $L_{1},L_{2}\in\mathcal{L}$ and $a,b,x,z\in\{0,1\}$. Thus $H_{\inn,i}$ is diagonal in the basis $\mathcal{B}_{\legal}$ with entries
\begin{equation}
\langle j,\vec{d},\vec{z},\vec{a}|H_{\inn,i}|j,\vec{d},\vec{z},\vec{a}\rangle=\begin{cases}
\frac{1}{64} & \text{$d_{i}=0$, $z_{i}=1$, and $F(i,j,0)=(i,0,1)$} \\
0 & \text{otherwise}.
\end{cases}\label{eq:Hin_i_legalbasis}
\end{equation}
Note that $F(i,j,0)=\left(i,0,1\right)$ if and only if none of the gates $U_{1},U_{2},\ldots,U_{j-1}$ acts on the $i$th qubit, i.e.,
\[
j\leq j_{\min,i}
\]
where
\[
j_{\min,i}=\min\{ j\in[M]:\: s(j)=i\}.
\]
Now using this fact and equations \eq{basischange} and \eq{Hin_i_legalbasis}, we get the following expression for the nonzero matrix elements of $H_{\inn,i}$ in the basis \eq{phi_init_basis_restated}: 
\begin{align*}
 & \langle j,\vec{d},\In(\vec{x}),\vec{a}|H_{\inn,i}|j,\vec{d},\In(\vec{z}),\vec{a}\rangle\\
 & \quad =\sum_{\vec{w},\vec{y}\in\{0,1\}^{n}}\langle\vec{x}|\bar{U}_{j,d_1}^{\dagger}(a_{1})|\vec{w}\rangle\langle\vec{y}|\bar{U}_{j,d_1}(a_{1})|\vec{z}\rangle\langle j,\vec{d},\vec{w},\vec{a}|H_{\inn,i}|j,\vec{d},\vec{y},\vec{a}\rangle\\
 &\quad = \begin{cases}
\underset{\vec{y}\in\{0,1\}^{n}}{\sum}\langle\vec{x}|\bar{U}_{j,d_1}^{\dagger}(a_{1})|\vec{y}\rangle\langle\vec{y}|\bar{U}_{j,d_1}(a_{1})|\vec{z}\rangle\frac{1}{64}\delta_{y_{i},1} & j\leq j_{\min,i}\text{ and }d_{i}=0\\
0 & \text{otherwise}
\end{cases}\\
 &\quad =\begin{cases}
\langle\vec{x}|\bar{U}_{j,d_1}^{\dagger}(a_{1})|1\rangle\langle1|_{i}\bar{U}_{j,d_1}(a_{1})|\vec{z}\rangle\frac{1}{64} & j\leq j_{\min,i}\text{ and }d_{i}=0\\
0 & \text{otherwise}
\end{cases}\\
 &\quad =\begin{cases}
\langle\vec{x}|U_1^{\dagger}(a_1)\ldots U_{j-1}^{\dagger}(a_1)|1\rangle\langle1|_{i}U_{j-1}(a_1)\ldots U_{1}(a_1)|\vec{z}\rangle\frac{1}{64} & \begin{aligned}
&(j< j_{\min,i}\text{ and }d_{i}=0 \text{ and } d_1\in\{0,2\})\\
&\text{or } (j= j_{\min,i}\text{ and }d_{i}=0)
\end{aligned}\\
\langle\vec{x}|U_1^{\dagger}(a_1)\ldots U_{j}^{\dagger}(a_1)|1\rangle\langle1|_{i}U_{j}(a_1)\ldots U_{1}(a_1)|\vec{z}\rangle\frac{1}{64} & j< j_{\min,i}\text{ and }d_{i}=0 \text{ and } d_1=1\\
0 & \text{otherwise.}
\end{cases}
\end{align*}
In the last line we use the fact that $d_1=d_i$ when $j=j_{\min,i}$ (since $s(j_{\min,i})=i$). Since $[U_{J}(a_{1}),|1\rangle\langle1|_{i}]=0$ for $J<j_{\min,i}$, we have
\begin{align}
\langle j,\vec{d},\In(\vec{x}),\vec{a}|H_{\inn,i}|j,\vec{d},\In(\vec{z}),\vec{a}\rangle & =\begin{cases}
\frac{1}{64}\langle\vec{x}|\left(|1\rangle\langle1|_{i}\right)|\vec{z}\rangle & j\leq j_{\min,i}\text{ and }d_{i}=0\\
0 & \text{otherwise}
\end{cases}\label{eq:mat_els_Hinj}\\
 & =\begin{cases}
\frac{1}{64}\delta_{\vec{x},\vec{z}}\delta_{x_{i},1} & j\leq j_{\min,i}\text{ and }d_{i}=0\\
0 & \text{otherwise}
\end{cases}
\end{align}
(with all other matrix elements equal to zero), which confirms the result stated in equation \eq{Hin_mat_els}.

\subsubsection{Matrix elements of $H_{\out}$}
\label{app:Matrix-elements-of_Hout}

Finally, consider 
\[
H_{\out}=\sum_{w=1}^{n}h_{\out}^{(w)}
\]
where (from equation \eq{hin_i})
\[
\langle\rho_{x,b}^{L_{2}}|h_{\out}|\rho_{z,a}^{L_{1}}\rangle=\frac{1}{64}\delta_{a,b}\delta_{x,0}\delta_{z,0}\delta_{L_{1},(2,M+1,0)}\delta_{L_{2},(2,M+1,0)}
\]
for $L_{1},L_{2}\in\mathcal{L}$ and $z,a,x,b\in\{0,1\}$. From this we see that $H_{\out}$ is diagonal in the basis $\mathcal{B}_{\legal}$, with entries 
\[
\langle j,\vec{d},\vec{z},\vec{a}|H_{\out}|j,\vec{d},\vec{z},\vec{a}\rangle=\begin{cases}
\frac{1}{64} & d_{2}=1\text{, }F(2,j,1)=(2,M+1,0)\text{ and }z_{2}=0\\
0 & \text{otherwise}.
\end{cases}
\]
Note that $F(2,j,1)=(2,M+1,0)$ if and only if $j\geq j_{\max}$, where 
\[
j_{\max}=\max\{j\in[M]:\: s(j)=2\}.
\]
Using this fact we compute the nonzero matrix elements of $H_{\out}$ in the basis \eq{phi_init_basis_restated}:
\begin{align*}
&\langle j,\vec{d},\In(\vec{x}),\vec{a}|H_{\out}|j,\vec{d},\In(\vec{z}),\vec{a}\rangle \nonumber\\
&\quad=\sum_{\vec{w},\vec{y}\in\{0,1\}^{n}}\langle\vec{x}|\bar{U}_{j,d_1}^{\dagger}(a_{1})|\vec{w}\rangle\langle\vec{y}| \bar{U}_{j,d_1}(a_{1})|\vec{z}\rangle\langle j,\vec{d},\vec{w},\vec{a}|H_{\out}|j,\vec{d},\vec{y},\vec{a}\rangle\nonumber \\
&\quad=\begin{cases}
\langle\vec{x}|\bar{U}_{j,d_1}^{\dagger}(a_{1})|0\rangle\langle0|_{2}\bar{U}_{j,d_1}(a_{1})|\vec{z}\rangle\frac{1}{64} & j\geq j_{\max}\text{ and }d_{2}=1\\
0 & \text{otherwise}
\end{cases}\nonumber \\
&\quad=\begin{cases}
\langle\vec{x}|U_{1}^{\dagger}(a_{1})\ldots U_{j}^{\dagger}(a_{1})|0\rangle\langle0|_{2}U_{j}(a_{1})\ldots U_{1}(a_{1})|\vec{z}\rangle\frac{1}{64} & j\geq j_{\max}\text{ and }d_1=d_{2}=1\\
\langle\vec{x}|U_{1}^{\dagger}(a_{1})\ldots U_{j-1}^{\dagger}(a_{1})|0\rangle\langle0|_{2}U_{j-1}(a_{1})\ldots U_{1}(a_{1})|\vec{z}\rangle\frac{1}{64} & j> j_{\max}\text{ and } d_{2}=1 \text{ and } d_1\in \{0,2\} \\
0 & \text{otherwise}
\end{cases}\nonumber\\
&\quad=\begin{cases}
\langle\vec{x}|U_{1}^{\dagger}(a_{1})\ldots U_{M}^{\dagger}(a_{1})|0\rangle\langle0|_{2}U_{M}(a_{1})\ldots U_{1}(a_{1})|\vec{z}\rangle\frac{1}{64} & j\geq j_{\max}\text{ and }d_{2}=1\\
0 & \text{otherwise}
\end{cases}\nonumber \\
&\quad=\begin{cases}
\langle\vec{x}|U_{\mathcal{C}_{X}}^{\dagger}(a_{1})|0\rangle\langle0|_{2}U_{\mathcal{C}_{X}}(a_{1})|\vec{z}\rangle\frac{1}{64} & j\geq j_{\max}\text{ and }d_{2}=1\\
0 & \text{otherwise.}
\end{cases}
\end{align*}
In going from the second to the third equality we use the fact that $j=j_{\max}$ implies $d_1=d_2$ (since $s(j_{\max})=2$).
In the next-to-last line we use the fact that $[U_{J}(a_{1}),|0\rangle\langle0|_{2}]=0$ for $J>j_{\max}$. This confirms the result stated in equation \eq{hout_matels-1}.

\end{document}